\documentclass[floatfix,aps,prx,reprint,longbibliography,nofootinbib]{revtex4-2}
\usepackage{palatino}
\usepackage{multirow} 
\usepackage{makecell} 
\usepackage{graphicx}
\usepackage{amssymb}
\usepackage{amsmath}
\usepackage{swphysics}

\usepackage[dvipsnames]{xcolor} 
\definecolor{tolred}{RGB}{221,110,121}
\definecolor{tolgreen}{RGB}{67,134,62}
\definecolor{tolblue}{RGB}{80,118,166}
\definecolor{tolyellow}{RGB}{201,187,88}
\definecolor{tolsky}{RGB}{128,201,234}
\definecolor{tolpurple}{RGB}{157,60,117}
\definecolor{tolgrey}{RGB}{187,187,187}
\definecolor{tolgray}{RGB}{187,187,187}

\usepackage{hyperref}
\hypersetup{
    colorlinks = true,
    linkcolor = black,%
    citecolor = magenta,
    urlcolor = tolpurple}%
\usepackage[capitalise]{cleveref} 

\Crefname{figure}{Figure}{Figures}
\crefname{figure}{Fig.}{Figs.}
\Crefname{equation}{Equation}{Equations}
\crefname{equation}{Eq.}{Eqs.}
\Crefname{section}{Section}{Sections}
\crefname{section}{Sec.}{Secs.}

\Crefname{appendix}{Appendix}{Appendices}
\crefname{appendix}{Appendix}{Appendices}

\crefformat{equation}{Eq.~(#2\textcolor{MidnightBlue}{#1}#3)}
\Crefformat{equation}{Equation~(#2\textcolor{MidnightBlue}{#1}#3)}

\crefrangeformat{equation}{Eqs.~(#3\textcolor{MidnightBlue}{#1}#4) to~(#5\textcolor{MidnightBlue}{#2}#6)}
\Crefrangeformat{equation}{Equations~(#3\textcolor{MidnightBlue}{#1}#4) to~(#5\textcolor{MidnightBlue}{#2}#6)}

\crefmultiformat{equation}{Eqs.~(#2\textcolor{MidnightBlue}{#1}#3)}%
{ and~(#2\textcolor{MidnightBlue}{#1}#3)}%
{, (#2\textcolor{MidnightBlue}{#1}#3)}%
{, and~(#2\textcolor{MidnightBlue}{#1}#3)}

\Crefmultiformat{equation}{Equations~(#2\textcolor{MidnightBlue}{#1}#3)}%
{ and~(#2\textcolor{MidnightBlue}{#1}#3)}%
{, (#2\textcolor{MidnightBlue}{#1}#3)}%
{, and~(#2\textcolor{MidnightBlue}{#1}#3)}

\makeatletter
\renewcommand{\eqref}[1]{\textup{(\hyperref[#1]{\textcolor{MidnightBlue}{\ref*{#1}}})}}
\makeatother

\makeatletter
\newcommand{\nopareneqref}[1]{\textup{\hyperref[#1]{\textcolor{MidnightBlue}{\ref*{#1}}}}}
\makeatother

\usepackage{mathrsfs} 
\usepackage{swstyles} 
\let\mathfrak\relax 
\newcommand{\mathfrak}[1]{%
  \text{\smash{
  \raisebox{-0.085em}{
  \includegraphics[height=0.85em, trim=0em 0.6em 0em 0.8em, clip]{frak/#1.pdf}%
}}}}

\usepackage{upgreek}

\usepackage{comment}
\usepackage{mathtools}
\usepackage{tikz}
\usetikzlibrary{shapes}
\usetikzlibrary{calc}
\usetikzlibrary{math}
\usetikzlibrary{decorations.markings}
\usetikzlibrary{decorations.markings, shapes.geometric}
\tikzset{
  ->-/.style={
    decoration={
      markings,
      mark=at position 0.5 with {
        \draw[solid, line width=1.0pt, line cap=round, line join=round] 
        (-4pt * #1, 4pt * #1) -- (4pt * #1, 0pt) -- (-4pt * #1, -4pt * #1);
        }
        },
        postaction={decorate}
        },
        ->-/.default={0.5} 
}
\tikzset{
-<-/.style={
decoration={
markings,
mark=at position 0.5 with {
\draw[solid, line width=1.0pt, line cap=round, line join=round]
(4pt * #1, 4pt * #1) -- (-4pt * #1, 0pt) -- (4pt * #1, -4pt * #1);
}
},
postaction={decorate}
},
-<-/.default={0.5} 
}
\newcommand\intikz[2][]{%
    \begin{tikzpicture}[baseline={([yshift=-.5ex] current bounding box.center)}, #1]
        {#2}
    \end{tikzpicture}%
}

\tikzset{
    mathshape/.style={
        draw,
        inner sep=0pt,
        line width=0.1ex 
    }
}

\DeclareRobustCommand{\trilattice}{%
  {\mathord{\text{\tikz[baseline=0.2ex]\node[mathshape, regular polygon, regular polygon sides=3, minimum size=2.3ex, anchor=south] {};}}}%
}

\DeclareRobustCommand{\sqlattice}{%
  {\mathord{\text{\tikz[baseline=0.2ex]\node[mathshape, rectangle, minimum size=1.75ex, anchor=south] {};}}}%
}

\DeclareRobustCommand{\hexlattice}{%
  {\mathord{\text{\tikz[baseline=0.2ex]\node[mathshape, regular polygon, regular polygon sides=6, minimum size=2.0ex, anchor=south] {};}}}%
}

\newsavebox{\myobjbox}

\usepackage{physics} 
\usepackage{float} 

\usepackage{scalerel} 

\usepackage{manfnt} 
\newcommand*\cube{{\mbox{\mancube}}}
\newcommand{\iu}{{\rm i}} 

\usepackage{amsthm}
\usepackage{chngcntr} 
\newtheorem{theorem}{Theorem}
\newtheorem{corollary}{Corollary}
\newtheorem{lemma}{Lemma}
\newtheorem*{theorem*}{Theorem}
\newtheorem{definition}{Definition}

\newtheorem{innercustomthm}{Theorem}
\newenvironment{customthm}[1]
  {\renewcommand\theinnercustomthm{#1}\innercustomthm}
  {\endinnercustomthm}

\newtheorem{innercustomlem}{Lemma}
\newenvironment{customlem}[1]
  {\renewcommand\theinnercustomlem{#1}\innercustomlem}
  {\endinnercustomlem}

\newtheorem{innercustomcor}{Corollary}
\newenvironment{customcor}[1]
  {\renewcommand\theinnercustomcor{#1}\innercustomcor}
  {\endinnercustomcor}

\usepackage[normalem]{ulem} 

\makeatletter
\long\def\@FootNoteText#1{\insert\footins{
  \reset@font\footnotesize
  \interlinepenalty\interfootnotelinepenalty
  \splittopskip\footnotesep
  \splitmaxdepth \dp\strutbox \floatingpenalty \@MM
  \hsize\columnwidth \@parboxrestore
  \def\@currentcounter{footnote}%
  \protected@edef\@currentlabel{%
  \csname p@footnote\endcsname\@thefnmark
  }%
  \color@begingroup
  \@makefntext{%
  \rule\z@\footnotesep\ignorespaces#1\@finalstrut\strutbox}%
  \par
  \color@endgroup}}%

\def\setwidefootnote{\let\@footnotetext=\@FootNoteText}
\makeatother

\makeatletter
\def\l@subsubsection#1#2{}
\makeatother

\makeatletter

\let\originaladdcontentsline\addcontentsline

\newcommand{\appendixaddcontentsline}[3]{%
  \def\firstargument{#1}%
  \def\tocargument{toc}%
  \ifx\firstargument\tocargument
    \originaladdcontentsline{atoc}{#2}{#3}%
  \else
    \originaladdcontentsline{#1}{#2}{#3}%
  \fi
}

\newcommand{\beginappendices}{%
  \clearpage
  \appendix
  \phantomsection
  \originaladdcontentsline{toc}{section}{Appendices}%
  \begingroup
    \let\addcontentsline\@gobblethree
    \section*{Contents of the Appendices}%
  \endgroup
  \@starttoc{atoc}%
  \let\addcontentsline\appendixaddcontentsline
}

\makeatother

\makeatletter
\AtBeginDocument{%
  \global\tocdim@section=2.5em
  \gdef\tocmax@section{2.5em}%
}
\makeatother

\begin{document}
    \title{Mixed-state topological order and error-correction thresholds\\in non-Abelian codes: rigorous results}
	\author{Sun Woo P. Kim}
	\thanks{\texttt{swk34} \texttt{[at]} \texttt{cantab} \texttt{[dot]} \texttt{ac} \texttt{[dot]} \texttt{uk}}
	\affiliation{Department of Physics, King's College London, Strand, London WC2R 2LS, United Kingdom}
	\author{Max McGinley}
    \thanks{\texttt{mm2025} \texttt{[at]} \texttt{cam} \texttt{[dot]} \texttt{ac} \texttt{[dot]} \texttt{uk}}
	\affiliation{T.C.M. Group, Cavendish Laboratory, JJ Thomson Avenue, Cambridge CB3 0HE, United Kingdom}
	
	\begin{abstract}
    We present a versatile and mathematically rigorous technique for bounding recovery thresholds in topological codes subject to noise. Our method captures the effect of applying an arbitrary (possibly non-Pauli) local noise channel to the code state of a broad class of two-dimensional codes, including surface codes, non-Abelian quantum doubles, and string-net codes. In each case, we prove that for noise strengths up to some explicit constant value, any initially encoded logical information can be recovered to high precision, and that the noise-corrupted state exhibits key hallmarks of mixed-state topological order: long-range entanglement and emergent higher-form symmetries. We also describe how these methods can be adapted to higher dimensions and correlated noise models.
	\end{abstract}
	
	\maketitle
	
	
	\section{Introduction}
        Topologically ordered states have long been noted for their potential utility in protecting quantum information against the unwanted effects of noise \cite{kitaev2003fault, dassarma2005, nayak2008, simon2023topological}. This notion can be made concrete through the framework of (topological) error-correcting codes \cite{Lidar2013}. Here, logical information, which one seeks to protect, is encoded into a larger many-body system in terms of `code states',  which lie in the degenerate ground state subspace of a local Hamiltonian within a particular topological phase. For instance, the code states of the celebrated toric and surface codes \cite{kitaev1997quantum, Kitaev1997quantumA, bravyi1998quantum}---arguably the most actively explored codes for realising scalable error-corrected quantum computation---are in an Abelian `$\mathbb{Z}_2$' topological phase, the algebraic structure of which is reflected in the form of its logical operators, syndromes, and transversal gate set. More generally, different topological orders, e.g.~those with non-Abelian anyons, give rise to codes with complementary properties, which may even be used in conjunction to achieve fault-tolerant universal quantum computation \cite{davydova2025, manjunath2026universal}.
        
        
        

        
        For the particular case of $\mathbb{Z}_2$ surface code states subject to local Pauli noise, the performance of the encoding scheme can be analysed very precisely.
        The question of whether the initially encoded information is recoverable can be addressed by mapping properties of the error-corrupted state to a classical statistical mechanics model, namely the random-bond Ising model \cite{dennis2002topological}. When the noise strength exceeds some critical threshold, this model undergoes a transition, signalling the loss of encoded information in the corrupted state. Physically, this abrupt breakdown of quantum memory is tied to a singular change in the intrinsic properties of the corrupted state \cite{fan2024diagnostics}. Such \textit{mixed-state phase transitions} \cite{sang2024, wang2025, sohal2025, ellison2025} are a topic of active exploration.


        Beyond this specific setting, most existing analytical results on recovery thresholds pertain to Pauli stabilizer codes, which only capture Abelian topological orders \cite{bombin2014structure}. In that case, there exists a general stat-mech formalism for determining the threshold of any Pauli stabilizer code with Pauli noise \cite{chubb2021statistical}. A related construction exists for purely coherent errors in the surface code   \cite{bravyi2018correcting,venn2023, behrends2024, cheng2025emergent, eckstein2025learning, bao2026phasesdecodabilitysurfacecode,yan2026nonlinear, bejan2026}, and bounds on thresholds have also been derived for more general noise models \cite{christandl2025}.  In the broader domain of non-Abelian codes, there has been progress in designing and numerically analysing the performance of decoders \cite{wootton2014,brell2014,schotte2022, davydova2025, jing2026, delafuente2026}, and in certain instances the existence of a threshold has been derived \cite{wootton2016, dauphinais2017fault, lyons2026}. However, with regard to the intrinsic breakdown of mixed-state topological order in non-Abelian codes, existing studies have been limited to heuristic mappings via the replica trick \cite{sala2025, sala2025stability, bao2026, vadali2026, wang2025fractional}. 
        Thus, rigorous methods for probing the topological order of noise-corrupted states beyond Abelian codes and Pauli errors are notably lacking.

        In this work, we provide a formal mathematical framework for analysing the effect of noise on code states for a very broad class of topological orders and for arbitrary local noise. Our formalism encompasses non-Abelian codes described by Kitaev's quantum double model \cite{kitaev2003fault} and the string-net model of Levin and Wen \cite{levin2005string}. In brief, we prove that for noise strengths up to some constant value, the encoded information can be recovered to high precision, and the corrupted state continues to exhibit the key features of mixed-state topological order---namely long-range entanglement, emergent higher-form symmetries, and the well-definedness of anyonic excitations.

        The key technical result (\cref{thm:fat})  that allows us to infer these properties (Corollaries \ref{cor:lre}, \ref{cor:anyon-distinguishability}, \ref{cor:recoverability}) pertains to logical operators in the corrupted state. Before noise is applied, logical information can be accessed via particular string operators, e.g.~Pauli-$X$ and -$Z$ strings in the $\mathbb{Z}_2$ surface code. We prove that if the effective strength of the noise channel (\cref{def:noise strength def}) is below some constant value, then there exist thickened logical string operators (see Fig.~\ref{fig:summary}), defined on strips of width $w$, that act in an analogous way on the corrupted state up to errors exponentially small in $w$. Our method, outlined in Section \ref{sec:methods}, is based on upper-bounding the amount of logical information accessible to the environment (i.e.~the degrees of freedom that dilate the noise channel). Via an information-disturbance argument \cite{kretschmann2008information, beny2009}, this allows us to show that certain classes of logical information can be accessed within different parts of the system, thereby giving us a spatially resolved characterisation of where in the system this information is encoded (see \cref{sec:aec-analysis}).

        This paper is structured as follows. In Section \ref{sec:setup}, we describe our setup and define the quantities of interest which capture the recoverability of logical information and various aspects of mixed-state topological order. We present our main results in Section \ref{sec:main results}, starting with the $\mathbb{Z}_2$ surface code under arbitrary separable local noise, generalising to non-Abelian codes later on. 
        Section \ref{sec:methods} outlines the general methodology of our approach, and in Section \ref{sec:main-proof} we provide a high-level proof for our main technical result, Theorem \ref{thm:fat}. Section \ref{sec:quantum-doubles} introduces the quantum double and string-net codes in more detail, and contains the formal statement of our most general results. We outline some further generalisations to higher dimensions with non-matchable errors and correlated errors in Section \ref{sec:further-generalisations}, before finally concluding and highlighting promising future directions in Section \ref{sec:discussion}.
	
	\begin{figure*}[t] 
		\centering
		\includegraphics[width=\linewidth]{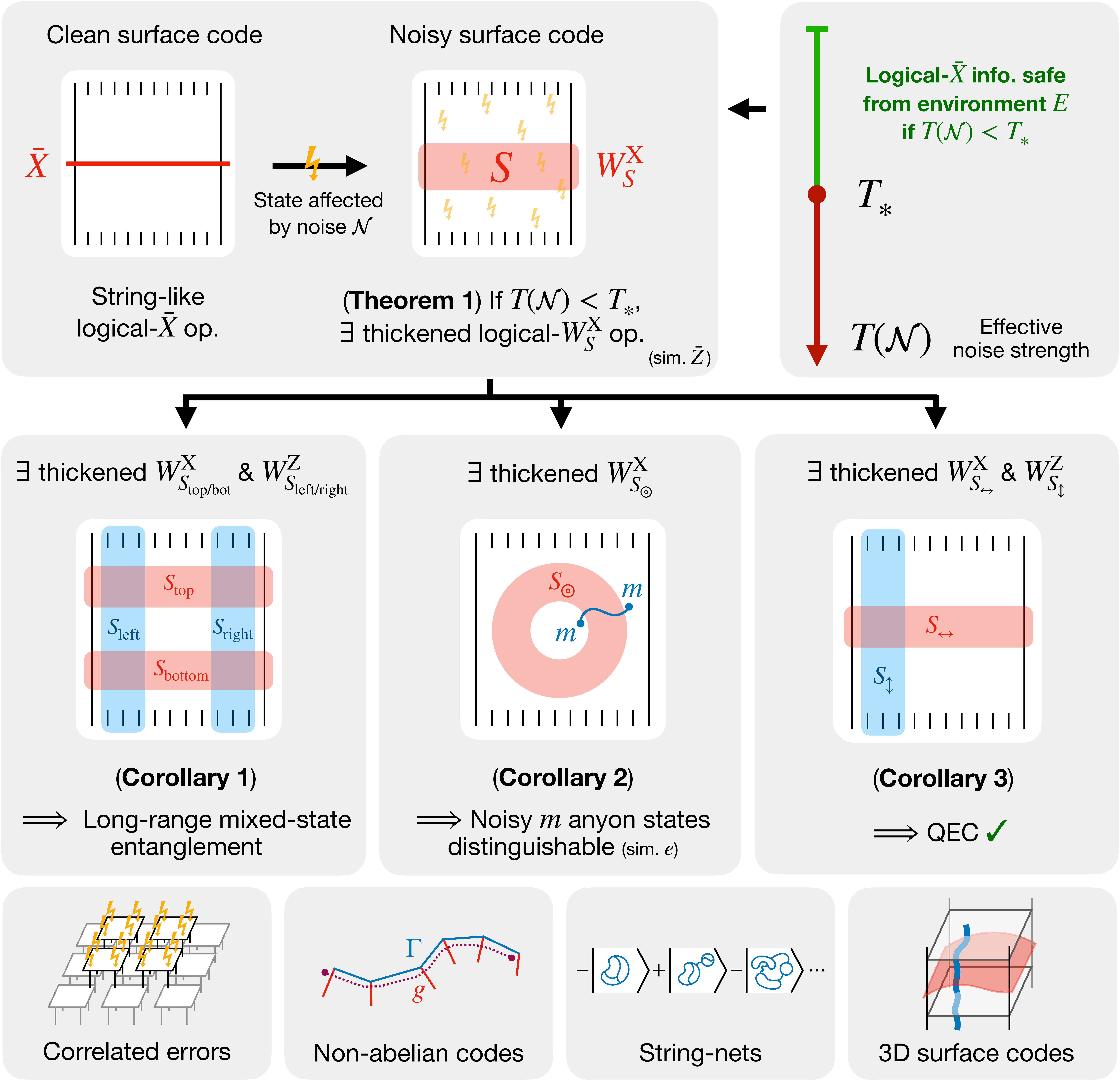}
		\caption{Summary of main results. Top row: Before noise is applied, the code supports `bare' logical operators, e.g.~$\bar{X}$ for the $\mathbb{Z}_2$ surface code. \Cref{thm:fat} shows that for noise strengths below some explicit value, the noise-corrupted state supports thickened logical operators defined on a strip $S$, which have approximately the same action on the encoded information [in the sense of Eq.~\eqref{eq:logical error}]. Middle row: As downstream consequences of \cref{thm:fat}, we can prove the hallmarks of mixed-state topological order: Long-range mixed-state entanglement (\cref{cor:lre}), distinguishability of anyon states (\cref{cor:anyon-distinguishability}), and recoverability of logical quantum information (\cref{cor:recoverability}). Bottom row: generalisations of \cref{thm:fat} to various settings, such as non-Abelian quantum double and string-net codes (\cref{sec:quantum-doubles}),  correlated errors (\cref{sec:correlated}), and 3D surface code (\cref{sec:3d}).}
		\label{fig:summary}
	\end{figure*}


    \setcounter{tocdepth}{2}
	\tableofcontents
    
	\section{Setup and aims}
	\label{sec:setup}
	
	In this work, we are concerned with the effect of generic, local noise on the code states of topological quantum codes for both Abelian and non-Abelian codes. The topological quantum order exhibited by these code states is expected to survive a finite strength of noise up to some critical threshold which signals the breakdown of quantum memory. Our aim is to develop a general and rigorous framework for quantifying this resilience of topological order against noise. We focus on three significant and operationally meaningful hallmarks of mixed state topological order:
	\begin{enumerate}
		\item \textit{Long-range entanglement.}~---~The noise-corrupted mixed state cannot be generated from a product state using low-depth circuits.
		\item \textit{Anyonic excitations and emergent higher-form symmetries.}~---~Inserting a pair of well-separated anyons yields an orthogonal state, which can be distinguished by string-like operators associated with an emergent (categorical) higher-form symmetry.
		\item \textit{Quantum memory.}~---~Logical information embedded in the code space can be approximately recovered after noise is applied.
	\end{enumerate}
	In the remainder of this section, we describe our setup in detail and elaborate on the above diagnostics, before presenting our main results in Sections \ref{sec:main results} and \ref{subsec:nonab results}.
	
	\subsection{Setup}

    We consider systems of $n$ qudits ($q$-level systems) arranged on the links $\ell$ of a regular $L_{\rm x} \times L_{\rm y}$ lattice, with open or periodic boundary conditions. We write $Q$ for the union of all physical qudits, and $\mathscr{H}_Q$ for its Hilbert space. The codes we consider will be defined in terms of stabilizers $A_v$, $B_p$ for each vertex $v$ and plaquette $p$ of the lattice. The code space $\mathscr{H}_{\mathrm{code}} \subset \mathscr{H}_Q$ is the subspace of states $\ket{\psi}_Q$ that satisfy $A_v\ket{\psi}_Q = B_p\ket{\psi}_Q = \ket{\psi}_Q\,\forall v, p$. Its dimension we denote $d_\mathrm{L}$, and its associated projector we denote $\mathbb{\Pi}$. A \textit{logical operator} is an operator on $Q$ that commutes with $\mathbb{\Pi}$, and thus maps code states to code states, but has a nontrivial effect within $\mathscr{H}_{\mathrm{code}}$.
	
	As a concrete example, in the $\mathbb{Z}_2$ surface code, the physical degrees of freedom are qubits ($q = 2$), and one takes open boundary conditions with `rough' and `smooth' edges, as shown in Fig.~\ref{fig:geometries}(a). The stabilisers are products of Pauli operators $\{X_\ell, Y_\ell, Z_\ell\}$ on different sites --- specifically $A_v = \prod_{l \in \partial^\top v} X_\ell$ and $B_p=\prod_{l \in \partial p} Z_\ell$. This structure makes the surface code an example of a \textit{Pauli stabiliser code}, which implies (among other properties) that all stabilisers commute with one another. 
	The projector onto the code space, which has dimension $d_\mathrm{L} = 2$ with the given boundary conditions, can be written $\mathbb{\Pi} = \mathbb{A}_V \mathbb{B}_P$, where $\mathbb{A}_V = \prod_{v} \mathbb{A}_v$, and $\mathbb{A}_v \coloneqq (I+A_v)/2$ is the projector onto the $+1$ eigenspace of $A_v$. Similarly, $\mathbb{B}_P = \prod_{p} \mathbb{B}_p$ with $\mathbb{B}_p \coloneqq (I+B_p)/2$. The logical operators are generated by $\bar{X} = \prod_{l \in \gamma^*} X_\ell$ and $\bar{Z} = \prod_{l \in \gamma} Z_\ell$, which act like the Pauli-$X$ and Pauli-$Z$ operators on the 2-dimensional code space. Here, $\gamma$ ($\gamma^*$) is any path on the primal (dual) lattice that connects one rough (smooth) edge to another, see Fig.~\ref{fig:geometries}.
	
	A $d_L$-dimensional state $\sigma_R$ can be encoded in terms of some physical state $\rho_Q$ supported in $\mathscr{H}_{\mathrm{code}}$. To this state, we apply a local noise channel $\mathcal{N}_{Q' \leftarrow Q}$, yielding a noise-corrupted mixed state $\tilde{\rho}_{Q'} = \mathcal{N}_{Q' \leftarrow Q}[\rho_Q]$. (We allow the output of the channel $Q'$ to differ from the input $Q$, e.g.~if some subset of qudits are traced out, or if new degrees of freedom are tensored in as it is the case for heralded errors.) For most of this work, we take independent noise on each site $\mathcal{N}_{Q' \leftarrow Q} = \bigotimes_{\ell \in Q} \mathcal{N}_\ell$, but we also consider correlated noise channels in Section \ref{sec:correlated}. We make no assumptions on the structure of the applied noise, allowing for operations that are not unital, or representable as Pauli noise channels.

    Our interest is in whether this encoding scheme can protect the logical information against the noise, and whether the corrupted state $\tilde{\rho}_{Q'}$ retains the topological order exhibited by the original code state $\rho_Q$. We now define specific diagnostics that quantify the extent to which this holds.
	
	\begin{figure*}[t] 
    \centering
    \includegraphics[width=1\linewidth]{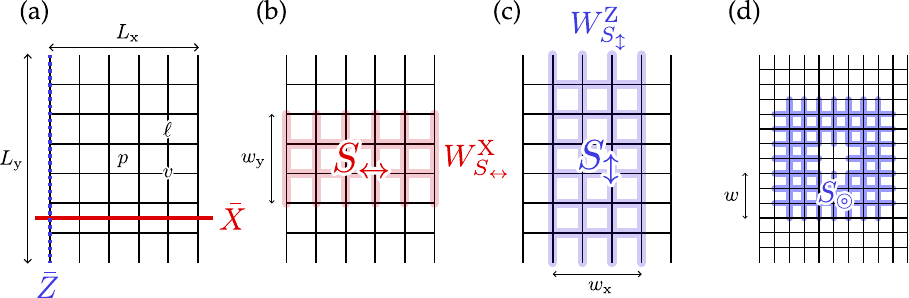}
	\caption{Geometry of the $\mathbb{Z}_2$ surface code. Qubits live on the edges $\{\ell\}$, while stabilisers live on vertices $\{v\}$ and plaquettes $\{p\}$. The maximum horizontal and vertical linear dimensions of the lattice are given by $L_\mathrm{x}$, $L_\mathrm{y}$, respectively, assumed to grow at the same rate, with $L = \mathrm{min}(L_\mathrm{x}, L_\mathrm{y})$. (a) On the clean code state, logical $X$ ($Z$) operators can be implemented by products of Pauli-$X$ ($Z$) operators on a dual (primal) path spanning the system in a horizontal (vertical) direction; see red solid (blue dashed) line. (b,c) Thickened regions $S_{\leftrightarrow}$ and $S_{\updownarrow}$ which support logical $X$ and $Z$ operators on the noise-corrupted state, respectively. (d) Annulus region of width $w$, which supports emergent 1-form symmetry operators.}
    \label{fig:geometries}
\end{figure*}

\subsection{Operational diagnostics of quantum memory and topological order}
\label{subsec:diagnostics}

Unlike symmetry-breaking order, topological properties of the corrupted mixed state $\tilde{\rho}_{Q'}$ cannot be probed via conventional expectation values, i.e.~linear functions of $\rho$. Instead, we must look at non-linear, information-theoretic diagnostics, which elucidate the different patterns of entanglement associated with these phases. The quantities we consider in this work, which are defined subsequently, all have direct operational interpretations that characterise how topologically ordered states can be prepared, manipulated, and used as a quantum memory. Some of these appear in Ref.~\onlinecite{fan2024diagnostics} in various guises. 

\subsubsection{Long-range entanglement}

A pure state (e.g.~a ground state of a local Hamiltonian) is considered to be topologically ordered if it cannot be prepared by low-depth unitary circuits, starting from a product state\footnote{Typically, in definitions of pure state topological order, `low-depth' typically means polylogarithmic in the linear system size $L$.}. Since states that can be prepared by low-depth circuits are said to be short-ranged entangled, or trivial, this obstruction to preparing pure topologically ordered states  is associated with \textit{long-range entanglement}.

In Ref.~\onlinecite{hastings2011topological}, Hastings proposes a natural generalisation of this notion to mixed states, which we present here in a slightly re-framed way; see also Ref.~\onlinecite{mcginley2025lower}. A state $\sigma$ is $(r,t)$-trivial if it can be written as a mixture (i.e.~convex combination) of states that are generated by ancilla-assisted range-$r$, depth-$t$ local unitary circuits, starting from a product state. More precisely, if we associate to each qudit $\ell \in Q$ an ancilla qudit $a(\ell)$ of arbitrary but finite dimension, then a $(r,t)$-trivial state can be written in the form
\begin{align}
	\sigma_\mathrm{SRE} = \sum_i p_i \tr_{\mathrm{Anc}}\big[U_i(\dyad{0}_Q \otimes \dyad{0}_\mathrm{Anc}) U_i^\dagger\big],
	\label{eq:sigma SRE}
\end{align}
where $p_i \geq 0$ are probabilities satisfying $\sum_i p_i = 1$, $\mathrm{Anc} \coloneq \bigcup_{\ell \in Q} a(l)$, and each $U_i$ is a depth-$t$ circuit of where each gate acts locally on the system and ancilla qubits with linear range $r$. The set $\mathrm{SRE}(r,t)$ consisting of all such states includes states prepared by mixtures of local channel circuits \cite{mcginley2025lower}. 

A state $\rho$ is said to be $(r, t, c)$ long-range entangled if every state $\sigma \in \mathrm{SRE}(r,t)$ is at least $c$-far from $\rho$, 
\begin{align}
	\mathsf{T}(\rho, \sigma) & > c & \forall \, \sigma \in \mathrm{SRE}(r, t).
\end{align}
Here, $\mathsf{T}(\rho, \sigma) \coloneqq \frac{1}{2}\|\rho - \sigma\|_1$ is the trace distance, where $\|X\|_1 \coloneqq \tr\big[\sqrt{X^\dagger X}\big]$ is the trace norm, equal to the sum of singular values of $X$. This is a natural measure of distance between two states due to its operational interpretation as the maximum difference between the outcome probability of some measurement when performed on $\rho$ and $\sigma$; however other sensible distance measures between states should behave similarly. Our expectation is that, below threshold, the noise corrupted state will remain long-range entangled.

Note that for mixed states, the above definition of long-range entanglement is distinct from long-range correlations. For instance, the state $\frac{1}{2}(\dyad{0^{\otimes n}} + \dyad{1^{\otimes n}})$ is all-to-all correlated, but is separable, and therefore $(1,1)$-trivial.

\subsubsection{Anyonic excitations and 1-form symmetries \label{subsec:anyonic setup}}

The syndromes (i.e. violations of code stabilisers) of a 2D topological code correspond to the locations of point-like excitations that exhibit anyonic statistics. For example, in the $\mathbb{Z}_2$ surface code, a state that violates a vertex projector $\mathbb{A}_v\ket{\psi} = 0$ hosts an electric ($e$) anyon, while plaquette violations $\mathbb{B}_p\ket{\psi} = 0$ correspond to magnetic ($m$) anyons. Starting from a code state $\rho_Q^{(0)}  \in \mathscr{H}_{\mathrm{code}}$, a pair of $e$ anyons at vertices $(v, v')$ can be formed as $\rho^{(v,v')}_Q = U_\gamma \rho_Q^{(0)} U_\gamma^\dagger$, where $U_{\gamma} \coloneq \prod_{l \in \gamma} Z_\ell$ is a string operator acting along a path $\gamma$ on the primal lattice, whose endpoints are $\partial \gamma = \{v, v'\}$. Analogously, pairs of magnetic anyons can be created and moved around using a complementary set of operators $U_{\gamma^*} \coloneqq \prod_{l \in \gamma^*} X_\ell$, which act along dual paths $\gamma^*$. 

Physically, the states $\rho_Q$ and $\rho^{(v,v')}_Q$ can be distinguished by creating a pair of $m$ anyons, moving one in a loop around the vertex $v$ (but not $v'$), and then fusing them back together (see the middle row of Fig.~\ref{fig:summary}). Since $e$ and $m$ anyons have a non-trivial mutual braiding phase $\pi$, one can in principle detect the relative phase of $e^{\iu \pi}$ between the states $\rho^{(v,v')}_Q$ and $\rho_Q$. This process is mathematically described by the operator $U_{\gamma^*}$ where $\gamma^*$ is a closed loop on the dual lattice encircling $v$. Such closed loop operators, which can detect the presence of anyons in their interior, generate an (anomalous) emergent 1-form symmetry \cite{gaiotto2015generalized, kong2020, mcgreevy2023generalized} which is respected by the clean code state $\rho^{(0)}_Q$. More generally, as described in Section \ref{sec:quantum-doubles}, each topological code will support a family of string-like operators that create anyon pairs, along with a set of 1-form symmetry operators that can detect anyons.

As long as the applied noise is weak enough to remain within the same mixed-state topological phase, we expect all these properties to continue to hold. Following Ref.~\onlinecite{fan2024diagnostics}, one of the consequences is that the noise-corrupted code state $\tilde{\rho}_Q =\mathcal{N}_Q[\rho_Q]$ and the noise-corrupted excited state $\tilde{\rho}^{(v,v')}_Q \coloneqq \mathcal{N}_Q[\rho^{(v,v')}_Q] = \mathcal{N}_Q[U_{\gamma} \rho^{(0)}_Q U_{\gamma}^\dagger]$ remain approximately distinguishable:
\begin{align}
	\mathsf{T}\left(\mathcal{N}_{Q' \leftarrow Q}[\rho_0], \; \mathcal{N}_{Q' \leftarrow Q}[U_{\gamma} \rho_0 U_{\gamma}^\dagger]\right) \geq 1 - \delta,
	\label{eq:anyon dist def}
\end{align}
for some small $\delta$. This implies that the two noise-corrupted states can be distinguished by some two-outcome measurement $\{\Pi, I-\Pi\}$ with probability $1 - \delta/2$. More specifically, we will look for string-like operators supported on annuli that encircle $v$ that can approximately distinguish these states. Such operators generate an emergent 1-form symmetry for the noisy state $\tilde{\rho}_Q$.

\subsubsection{Encoding and recoverability of logical information}
\label{subsec:encoding}

Since the code states form a $d_\mathrm{L}$-dimensional linear subspace, any quantum information in the form of a $d_\mathrm{L}$-dimensional state can be encoded into this space. Let $\mathscr{H}_R$ ($R$ a `reference' qudit) denote the Hilbert space of the information to be encoded, with an orthonormal basis $\{\ket{\alpha}_R\}_{\alpha = 0}^{d - 1}$. If $\{\ket{\bar \alpha}_Q\}_{j=0}^{d-1}$ is a set of orthonormal code states, then we can define an isometric \textit{encoding channel}
\begin{align}
    \begin{aligned}
	\mathcal{C}_{Q \leftarrow R}[\,\cdot\,] & \coloneqq V_{Q \leftarrow R}^{\mathcal{C}} \, \cdot \, (V_{Q \leftarrow R}^{\mathcal{C}})^\dagger, \\
    \; \text{ where } \quad  V_{Q \leftarrow R}^{\mathcal{C}} & \coloneqq \sum_\alpha \ket{\bar \alpha}_Q\bra{\alpha}_R.
    \end{aligned}
	\label{eq:encoding}
\end{align}
This channel associates  each `logical' state on $R$ to a `physical' state on $Q$, lying in the code space.

The ability to encode logical information in a protected subspace is a property of the topological phase as a whole. Thus, below threshold, where topological order remains stable, we expect that the logical information encoded in the initial state should remain (approximately) recoverable after noise is applied, i.e.~the code should serve as a quantum memory robust against $\mathcal{N}_{Q' \leftarrow Q}$. This requires the existence of a recovery channel $\mathcal{R}_{R \leftarrow Q'}$ such that the combined operation of encoding ($\mathcal{C}_{Q \leftarrow R}$), noise ($\mathcal{N}_{Q' \leftarrow Q}$), and recovery ($\mathcal{R}_{R \leftarrow Q'}$) leaves the encoded information approximately undisturbed.

We can quantify the accuracy of a given recovery operation by comparing the concatenated map $\mathcal{R}_{R \leftarrow Q'} \circ \mathcal{N}_{Q' \leftarrow Q} \circ \mathcal{C}_{Q \leftarrow R}$ with the identity channel $\mathrm{id}_R$, which does nothing to the logical state. In particular, we define
\begin{align}
	\epsilon_{\mathrm{rec}} \coloneqq \inf_{\mathcal{R}_{R \leftarrow Q'}} \left\|\mathcal{R}_{R \leftarrow Q'} \circ \mathcal{N}_{Q' \leftarrow Q} \circ \mathcal{C}_{Q \leftarrow R} - \mathrm{id}_R\right\|_\diamond.
	\label{eq:recov error}
\end{align}
Here, $\|\mathcal{T}_R\|_\diamond$ is the diamond norm for any superoperator $\mathcal{T}_R$, equal to $\|\mathcal{T}_R\|_\diamond \coloneqq\sup_{\sigma_{RR'}} \|(\mathcal{T}_R \otimes \mathrm{id}_{R'})[\sigma_{RR'}]\|_1$, where the supremum is over all density matrices on an enlarged Hilbert space $\mathscr{H}_R \otimes \mathscr{H}_{R'}$, with $R'$ having the same Hilbert space dimension as $R$. The interpretation of the diamond distance between two channels is that it measures their distinguishability over all possible experiments (including those that use entangled inputs), just as the trace distance measures the optimal distinguishability between two states over all measurements, see e.g.~Ref.~\onlinecite{watrous2018theory}. Thus, $\epsilon_\mathrm{rec}$ also quantifies the extent to which any disturbance to the logical information could in principle be detected.

One illustrative consequence of \eqref{eq:recov error} is that it guarantees for any state with support in the code space $\rho_{QP} = \mathcal{C}_{Q \leftarrow R}[\rho_{RP}]$ (where $\rho_{RP}$ a joint state with any other Hilbert space $P$) and its corresponding noisy state $\tilde{\rho}_{Q'P} = \mathcal{N}_{Q' \leftarrow Q} \rho_{QP}$, there exists a recovery map $\mathcal{R}_{Q \leftarrow Q'}' = \mathcal{C}_{Q \leftarrow R} \circ \mathcal{R}_{R \leftarrow Q'}$ that approximately restores the clean encoded state, without compromising the encoded information \cite{pkim2026optimal}
\begin{align}
	\mathsf{T}(\rho_{QP}, \mathcal{R}'_{Q \leftarrow Q'}[\tilde{\rho}_{Q'P}]) \leq \frac{\epsilon_\mathrm{rec}}{2}.
\end{align}

\subsection{Previous approaches via the method of replicas}

For the $\mathbb{Z}_2$ surface code subject to stochastic Pauli noise, all the key properties of the corrupted states can be computed using a mapping to the classical random-bond Ising model, originally devised by the authors of Ref.~\onlinecite{dennis2002topological}. This mapping can be used to determine the location of the threshold, and other key properties of the topological phase, giving a very complete picture of how the noise leads to the breakdown of topological order and quantum memory. This procedure has been generalised to any Pauli stabilizer code with Pauli noise \cite{chubb2021statistical}. An analogous construction has also been put forward for the surface code subject to purely coherent errors, followed by stabilizer measurements \cite{bravyi2018correcting,venn2023, behrends2024, cheng2025emergent, eckstein2025learning, bao2026phasesdecodabilitysurfacecode,yan2026nonlinear, bejan2026}.

Beyond these special cases that admit particular solutions in terms of stat-mech mappings, there is much less understood about more general topological codes. A recent line of work \cite{lee2025exact, bao2026, li2025replica, fan2024diagnostics, sala2025, sala2025stability, behrends2025, vadali2026, wang2025fractional} has sought to address this gap using a heuristic approach, where one studies analytically tractable proxy quantities that are related to the operational properties at hand in an indirect way. Specifically, the relationship to operational quantities typically only appears when a certain replica limit to be taken, which generally cannot be done rigorously in practice.

To help explain this approach, we introduce the \textit{coherent information}, which has long been used to the study approximate error correction \cite{schumacher2002approximate,pkim2026optimal}, and has recently gained attention in the context of topological code states \cite{fan2024diagnostics, colmenarez2024, eckstein2024, lee2025exact, li2025replica}. Here, one prepares a logical Bell state $\ket{\Phi_{RQ}} = d_\mathrm{L}^{-1/2}\sum_{\alpha=0}^{d_\mathrm{L}-1} \ket{\alpha}_R \otimes \ket{\bar \alpha}_Q$ [cf.~\cref{eq:encoding}], and applies noise to $Q$, resulting in $\tilde{\rho}_{RQ'} = (\mathrm{id}_R \otimes \mathcal{N}_{Q' \leftarrow Q})[\Phi_{RQ}]$. The coherent information is then
\begin{align}
	\mathsf{I}_\mathrm{C}(R~\rangle~Q') \coloneqq \mathsf{S}(\tilde{\rho}_{Q'}) - \mathsf{S}(\tilde{\rho}_{RQ'}),
\end{align}
where $\mathsf{S}(\rho) = -\tr[\rho \ln \rho]$ is the von Neumann entropy. As shown in Ref.~\onlinecite{schumacher2002approximate}, the coherent information $\mathsf{I}_\mathrm{C}(R~\rangle~Q')$ is close to its maximum value $\ln d_\mathrm{L}$ if and only if $\epsilon_{\mathrm{rec}}$ is correspondingly small. Thus, the approximate preservation of coherent information under noise is a necessary and sufficient condition for approximate recoverability.

Unfortunately, the matrix logarithm in the definition of the von Neumann expression cannot generally be evaluated analytically. Instead, the approach of Refs.~\cite{lee2025exact, bao2026, li2025replica, fan2024diagnostics, sala2025, sala2025stability, behrends2025, vadali2026, wang2025fractional} is to work with a family of R{\'e}nyi generalisations, $\mathsf{I}^{(n)}_\mathrm{C}(R~\rangle~Q') \coloneqq \mathsf{S}^{(n)}(\tilde{\rho}^{Q'}) - \mathsf{S}^{(n)}(\tilde{\rho}^{RQ'})$, where $\mathsf{S}^{(n)}(\rho) \coloneqq (1-n)^{-1}\ln[\rho^n]$ is the R{\'e}nyi entropy, which approaches the von Neumann entropy in the limit $n \rightarrow 1$. For integer $n \geq 2$, this generalised coherent information can sometimes be mapped to a classical statistical mechanics model or field theory, which can then be analysed to gain insight into the stability of the phase. However, these generalised quantities cannot be directly related to correctability, and there is generally no way to rigorously take the necessary replica limit $n \rightarrow 1$.

Our approach will not involve such replica quantities, and will instead employ information-theoretic techniques to put rigorous bounds on the various diagnostics described in the previous section. Even so, to help build intuition for our alternative strategy, it is worth pausing to explore the significance of the coherent information in more detail. By Stinespring's dilation theorem, the channel $\mathcal{N}_{Q' \leftarrow Q}$ can be expressed as an isometry $V^{\mathcal{N}}_{Q'E\leftarrow Q}$, where $E$ represents the `environment' degrees of freedom, followed by tracing over $E$,
\begin{align}
	\mathcal{N}_{Q' \leftarrow Q}[\,\cdot\,] = \tr_{E}\big[V^{\mathcal{N}}_{Q'E\leftarrow Q} \,\cdot\, (V^{\mathcal{N}}_{Q'E\leftarrow Q})^\dagger  \big].
	\label{eq:channel dilation}
\end{align}
In other words, the channel can be synthesised by entangling $Q$ with some extra degrees of freedom, which are then discarded. This allows one to define a \textit{complementary channel}, whose output is the environment $E$, as
\begin{align}
	\widehat{\mathcal{N}}_{E \leftarrow Q}[\,\cdot\,] \coloneqq \tr_{Q'} \big[V^{\mathcal{N}}_{Q'E\leftarrow Q} \,\cdot\, (V^{\mathcal{N}}_{Q'E\leftarrow Q})^\dagger  \big],
	\label{eq:complementary channel}
\end{align}
which mirrors \eqref{eq:channel dilation}. (We use wide hats for complementary channels throughout.) This allows us to define the joint state of the reference and environment, $\tilde{\rho}_{RE} = (\mathrm{id}_R \otimes \widehat{\mathcal{N}}_{E \leftarrow Q})[\Phi_{RQ}]$, in terms of which one finds
\begin{align}
	\mathsf{I}_\mathrm{C}(R~\rangle~Q') = \ln d_\mathrm{L} - \mathsf{I} \mkern2mu (R:E).
	\label{eq:coherent RE}
\end{align}
Here, for a bipartite state $\rho_{AB}$, the mutual information $\mathsf{I} \mkern2mu (A:B) \coloneqq \mathsf{S}(\rho_A) + \mathsf{S}(\rho_B) - \mathsf{S}(\rho_{AB}) \geq 0$ quantifies the amount of correlations between subsystems $A$ and $B$. Thus, (approximate) maximality of $\mathsf{I}_\mathrm{C}(R~\rangle~Q')$ corresponds to the (approximate) absence of correlations between the logical information $R$ and the environment $E$.

This bidirectional relationship between correlations with the system $Q'$ and environment $E$ exemplifies the so-called `information-disturbance tradeoff' \cite{kretschmann2008information}, which states that the absence of quantum information in one subsystem implies that it must be contained in its complement. This principle, refined somewhat in Ref.~\onlinecite{beny2009}, will underpin our approach: we will aim to upper bound the amount of correlations between $R$ and $E$, in order to demonstrate the recoverability of logical information within the system $Q'$. This can be done without resorting to entropic quantities, and without invoking replicas (see \cref{sec:methods} for a more detailed discussion).

\section{Main results \label{sec:main results}}
In this Section, we state our main formal results. \Cref{sec:summary-sc} details all our rigorous findings for the 2D $\mathbb{Z}_2$ surface code under generic single-site errors. The central technical result is \cref{thm:fat}, which we will use to bound the various diagnostics described in \Cref{subsec:diagnostics} (Corollaries \ref{cor:lre}, \ref{cor:anyon-distinguishability}, \ref{cor:recoverability}). As we outline in \Cref{subsec:nonab results brief}, these will all be generalised to non-Abelian codes, as Theorem \ref{thm:fat-nonab} and Corollaries \ref{cor:lre-nonab}, \ref{cor:anyon-distinguishability-nonab}, \ref{cor:recoverability-nonab}. Further generalisations to correlated errors and 3D codes are described later in Section \ref{sec:further-generalisations}. 

\subsection{Summary of results for the $\mathbb{Z}_2$ surface code} \label{sec:summary-sc}

Here we summarise our main technical results for the case of the $\mathbb{Z}_2$ surface code subject to arbitrary (possibly non-Pauli) single-qubit separable noise $\mathcal{N}_Q = \bigotimes_{\ell \in Q}\mathcal{N}_\ell$. Each of these have a natural generalisation to non-Abelian codes and correlated errors, which we will present later on, once we have introduced the necessary formalism. Throughout this section $\mathcal{C}_{Q \leftarrow R}$ will be assumed to be an isometric encoding of the $\mathbb{Z}_2$ surface code. 

The three diagnostics of topological order described in \cref{subsec:diagnostics} can each be characterised using minor adaptations of the following technical result, which posits the existence of `thickened' logical ribbon operators.

\begin{theorem}[Informal, 2D $\mathbb{Z}_2$ surface code]
	\label{thm:fat}
	For the square-lattice $\mathbb{Z}_2$ surface code with rough and smooth boundary conditions as shown in \cref{fig:geometries}(a), consider an arbitrary single-site separable noise channel $\mathcal{N}_Q = \bigotimes_{\ell \in Q}\mathcal{N}_\ell$. There is a function $T(\mathcal{N}_\ell)$, which measures the effective noise strength of $\mathcal{N}_\ell$, such that if
	\begin{align}
		T(\mathcal{N}_\ell) &< T_\star &\forall \ell \in Q,
		\label{eq:informal noise condition}
	\end{align}
	where $T_\star > 0$ is some constant value, then we can define `thickened ribbon operators' $W_S^\mathrm{X}$, supported on horizontal strips of qubits $S \subseteq Q$ of height $w_\mathrm{y}$ [\cref{fig:geometries}(b)], that approximate the logical $X$ operator
	\begin{align}
		W_S^\mathrm{X} \cdot \big( \mathcal{N}_Q\circ \mathcal{C}_{Q \leftarrow R}[\sigma_R]\big) \approx_\epsilon \mathcal{N}_Q\circ \mathcal{C}_{Q \leftarrow R}[X_R \sigma_R]
		\label{eq:logical error}
	\end{align}
	up to error $\epsilon = \mathrm{poly}(L_\mathrm{x}) e^{-\Omega(w_{\rm y})}$.
\end{theorem}
\noindent Naturally, the same holds for vertical strips of width $w_\mathrm{x}$ as shown in \cref{fig:geometries}(b), with $X_R$ replaced by $Z_R$. We defer the definition of the noise strength function $T(\mathcal{N}_{\ell})$ and the formal statement of the theorem to Section \ref{sec:main-proof}, but we remark that \eqref{eq:informal noise condition} is satisfied for all $\mathcal{N}_\ell$ contained within some ball of the identity channel, $\|\mathcal{N}_\ell - \mathrm{id}_{\ell}\|_\diamond \leq c'$ for some constant $c' > 0$. The constant can be made explicit---see Sec.~\ref{sec:main-proof}.

\Cref{eq:logical error} implies that the application of the thickened string operator $W_S^{X}$ has approximately the same effect as the `true' logical operator $X_R$, up to corrections that are exponentially small in the width of the region $S$. Thus, below threshold, a constant target error $\epsilon$ can be achieved using an operator of width $w_\mathrm{x,y} =O\big(\ln(L_\mathrm{y, x}/\epsilon)\big)$. We now state the implications of the above for our diagnostics.

\textit{1. Long-range entanglement.}~---~Since \cref{thm:fat} applies equally for any vertical strip of qubits, it follows that there must exist correlations between ribbon operators $W_{S_\mathrm{left}}^\mathrm{Z}$ and $W_{S_\mathrm{right}}^\mathrm{Z}$ on distant strips $S_\mathrm{left}$, $S_\mathrm{right}$, and similarly for the $X$-type ribbon operators $S_\mathrm{top}$, $S_\mathrm{bottom}$, as left panel of middle row of \cref{fig:summary}. This can be used to show that all noise-corrupted code states are long-range entangled.
\begin{corollary}[2D $\mathbb{Z}_2$ surface code] \label{cor:lre}
	With the assumptions of \cref{thm:fat} and for $T(\mathcal{N}_\ell) < T_\star$, let $\tilde{\rho}_Q = \mathcal{N}_Q \circ \mathcal{C}_{Q \leftarrow R}[\sigma_R]$ be a noise-corrupted code state, with $\mathcal{N}_Q$ satisfying  \cref{eq:informal noise condition}.
	Then, there exists a constant $c_0 > 0$ such that, for sufficiently large $L = \mathrm{min}(L_\mathrm{x}, L_\mathrm{y})$, $\tilde{\rho}_Q$ is $(r,t,c_0)$-long-range entangled for all $rt < L/6$, i.e.~
	\begin{align}
		\mathsf{T}(\tilde{\rho}_Q, \sigma_Q) \geq c_0  \quad \forall \, \sigma_Q \in \mathrm{SRE}(r,t), \quad rt < L/6.
	\end{align}
\end{corollary}

Our formal statement and proof of the above result, presented in \cref{apdx:lre}, follows a similar argument to one outlined Ref.~\onlinecite{hastings2011topological}. The above implies that $\tilde{\rho}_Q$ cannot be in a trivial mixed-state phase.

\textit{2. Anyonic excitations.}~---~From any clean code state $\rho^{(0)}_Q = \mathcal{C}_{Q \leftarrow R}[\sigma_R]$, we can generate a state $\rho^{(v,v')}_Q = U_{\gamma}\rho^{(0)}_Q U_\gamma^\dagger$ that features a pair of electric anyons at vertices $v, v'$ (Sec.~\ref{subsec:anyonic setup}). Under the conditions of \cref{thm:fat}, the corresponding noise-corrupted states remain almost exactly distinguishable.
\begin{corollary}[2D $\mathbb{Z}_2$ surface code] \label{cor:anyon-distinguishability}
	With the assumptions of \cref{thm:fat}, consider an annulus $S$ of thickness $w$ [Fig.~\ref{fig:geometries}(d)], and two vertices $v,v' \notin S$ such that $S$ surrounds $v$ but not $v'$. Then, for $T(\mathcal{N}_\ell) < T_\star$, the noise corrupted states with and without anyons ($\tilde{\rho}_Q^{(0)} = \mathcal{N}_Q[\rho^{(0)}_Q]$ and $\tilde{\rho}_Q^{(v,v')} = \mathcal{N}_Q[\rho^{(v,v')}_Q ]$, respectively) can be approximately distinguished by a POVM $\{K_S, I-K_S\}$ supported on $S$, to accuracy
    $\delta_S = e^{-\Omega(w)}$. In particular, the states are globally $\delta$-distinguishable [Eq.~\eqref{eq:anyon dist def}], with $\delta = e^{-\Omega(\textup{dist}[v,v'])}$. 
	The same holds for magnetic anyons on plaquettes $(p,p')$.
\end{corollary}
To prove the above, we adapt \cref{thm:fat} by modifying $\mathcal{C}_{Q \leftarrow R}$ to be a channel that prepares the code state $\rho_Q^{(0)}$ when the input is in state $\ket{0}_R$, and the state with anyons $\rho^{(v,v')}_Q$ for the input $\ket{1}_R$. The operators $K_S$, which are defined on thickened loops $S$, generate an emergent 1-form symmetry exhibited by the noise-corrupted code states $\tilde{\rho}_Q$, which detect the presence of anyons in the interior of $S$.

\textit{3. Recoverability of logical information.}~---~The existence of approximate logical operators for both $X$ and $Z$ separately naturally implies that the full logical information can be approximately recovered:
\begin{corollary}[2D $\mathbb{Z}_2$ surface code] \label{cor:recoverability}
	With the assumptions of \cref{thm:fat}, and for $T(\mathcal{N}_\ell) < T_\star$, 
		there exists a recovery channel such that, the logical information can be recovered up to error
		\begin{align}
			\epsilon_\mathrm{rec} \leq e^{-\Omega(L)},
		\end{align}
		where $\epsilon_\mathrm{rec}$ measures the distance of the recovered state to the original state, as defined in Eq.~\eqref{eq:recov error}.
\end{corollary}
We note that the approximate preservation of coherent information follows directly from \cref{cor:recoverability}, i.e. $\mathsf{I}_\mathrm{C}(R \, \rangle \, Q') \geq \ln(d_\mathrm{L}) - f(\epsilon_\mathrm{rec})$, where $f(x \rightarrow 0^+) = 0$  [end of Appendix \ref{apdx:full recoverability proof}, Eq.~\eqref{eq:coh inf bound apdx}].

Using Corollary \ref{cor:recoverability} (and the explicit values of $T_\star$, discussed later on) we can derive an explicit lower bound on the intrinsic recovery threshold for a chosen family of noise channels. We compile these bounds for various important types of noise in \cref{tab:beta-ps-and-thresholds}, and compare them to existing numerical estimates where they exist. For the $\mathbb{Z}_2$ case, we are able to optimize our proof technique to increase the value of $T_\star$ (see Section \ref{sec:sow-mapping}); in this case our bounds typically come within a factor of 3 of the numerical values.

Proofs of all these Corollaries, along with their generalisation to non-Abelian codes, are provided in Appendix \ref{apdx:corollaries}. 

\renewcommand{\arraystretch}{1.2}
\begin{table*}[t]
	\begin{center}
		\begin{ruledtabular}
			\setlength{\tabcolsep}{6pt}
			\begin{tabular}{ccccc}
				Code & Channel family & Channel definition $\mathcal{N}_\ell^{(p)}[\sigma]$ & \makecell{Lower bound on $p_\mathrm{th}$ \\ (This work, rigorous)} & Estimate  of $p_\mathrm{th}^\mathrm{est}$ \\
				\hline
				
				\raisebox{-0.9ex}{$\mathbb{Z}_2$ surf.} 
				& \multirow{2}{*}{\raisebox{-3.5ex}{Amplitude damping}}  
				& \multirow{2}{*}{\raisebox{-3.5ex}{\scalebox{1.1}{$\displaystyle \begin{aligned} \smqty[1 & 0 \\ 0 & \sqrt{1-p}]_\ell \, & \sigma  \,\smqty[1 & 0 \\ 0 & \sqrt{1-p}]_\ell \\ + \; \smqty[0 & \sqrt{p} \\ 0 & 0]_\ell \, & \sigma \, \smqty[0 & 0 \\ \sqrt{p} & 0]_\ell \end{aligned}$}}}  
				& \raisebox{-0.9ex}{$13.3\%$} & \raisebox{-0.9ex}{$39\%$ \cite{darmawan2017tensor}} \\ [1.5ex]
				\raisebox{-0.9ex}{3D $\mathbb{Z}_2$} 
				& & & \raisebox{-0.9ex}{$0.16 \%$} & \raisebox{-0.9ex}{?} \\[1.5ex]
				\hline
				
				$\mathbb{Z}_2$ surf. 
				& \multirow{3}{*}{Heralded leakage}  
				& \multirow{3}{*}{\scalebox{1.1}{$\displaystyle \begin{aligned} & \quad \;\, (1-p)\sigma \otimes \ketbra{0}_C \\ + \;  & p \tr_\ell[\sigma] \otimes \ketbra{e}_\ell \otimes \ketbra{1}_C\end{aligned}$}}  
				& $36.5\%$ & $50\%^\dag$ \cite{stace2009thresholds} \\
				$S_3$ 
				& & & $2.7\%$ & ? \\
				$D_4$ 
				& & & $2.6\%$ & ? \\
				\hline
				
				$\mathbb{Z}_2$ surf.
				& \multirow{3}{*}{Bit-flip ($q = 2$)}  
				& \multirow{3}{*}{\scalebox{1.1}{$(1-p) \sigma + p X_\ell \sigma X_\ell $}}  
				& $3.4\%$ & $10.3\%$ \cite{dennis2002topological} \\
				$\mathrm{3D}$ $\mathbb{Z}_2$ 
				& & & $0.061 \%$ & $11\%$ \cite{xu2025phenomenological} \\
				$\mathrm{DSem}$ & & & $0.0054\%$ & ? \\
				\hline
				
				$\mathbb{Z}_2$ surf.
				& \multirow{4}{*}{Partially depolarising}  
				& \multirow{4}{*}{\scalebox{1.1}{$\displaystyle(1-p)\sigma + p \tr_\ell[\sigma] \otimes \frac{I_\ell}{q}$}} 
				& $6.9\%$ & $19\%$ \cite{bombin2012strong} \\
				$S_3$ 
				& & & $0.069\%$ & ? \\
				$D_4$ 
				& & & $0.068\%$ & ? \\
                $\mathrm{Fibonacci}$ 
				& & & $0.0071\%$ & $6.27\%^*$ \cite{schotte2022} \\
			\end{tabular}
		\end{ruledtabular}
		\caption{Bounds on the threshold noise rate $p_{\text{th}}$ for different topological codes and families of single-site separable error channels $\mathcal{N}^{(p)}_{Q' \leftarrow Q} = \prod_\ell \mathcal{N}^{(p)}_\ell$ parametrised by $p$ (see Appendix~\ref{apdx:diamond-norms} for proofs). We consider the $\mathbb{Z}_2$ surface (surf.) code in the geometry of Fig.~\ref{fig:geometries}, Non-Abelian codes $S_3$, $D_4$ based on the quantum double model on the square lattice with the corresponding group, the double semion (\text{DSem}) and Fibonacci codes based on the Levin-Wen string net model on the hexagonal lattice, along with 3D $\mathbb{Z}_2$ surface code in the geometry of Fig.~\ref{fig:3d-geometry}. $q$ is the local Hilbert space dimensions for the physical qudits. In the rightmost column, we list estimates of the thresholds based on numerical simulations, where they exist. $*$: estimate based on a specific decoder. $\dag$: exact value based on percolation threshold. Note that we have only optimised our proof technique for better numerical values in the $\mathbb{Z}_2$ surface code (Sec.~\ref{sec:sow-mapping}).}
		\label{tab:beta-ps-and-thresholds}
	\end{center}
\end{table*}

\subsection{Results for non-Abelian codes \label{subsec:nonab results brief}}

Theorem \ref{thm:fat} and Corollaries \ref{cor:lre}, \ref{cor:anyon-distinguishability}, and \ref{cor:recoverability} for the $\mathbb{Z}_2$ surface code will be generalised to non-Abelian codes. In this work we consider two families of non-Abelian codes: Kitaev's quantum double model for a discrete group $G$ \cite{kitaev2003fault}, and the string-net model of Levin and Wen given a unitary fusion category $\mathscr{C}$ \cite{levin2005string}.

Since the formal statement of these generalised results requires some additional theoretical background yet to be introduced, we defer these until Section \ref{sec:quantum-doubles}. However, as an example of the applicability of our method, we present the following result here.

\begin{theorem}
	For any Kitaev quantum double model or Levin-Wen string net model on a $L_{\rm x} \times L_{\rm y}$ torus, there exists a constant $c_0$ such that, for any tensor product of noise channels $\mathcal{N}_Q = \bigotimes_{\ell \in Q} \mathcal{N}_\ell$ that are at most $c_0$-far from the identity,
	\begin{align}
		\|\mathcal{N}_\ell - \mathrm{id}_\ell\|_\diamond & < c_0, & \forall \ell \in Q,
	\end{align}
	there exists a recovery map $\mathcal{R}_Q$ such that
	\begin{align}
		\norm{\mathcal{R}_Q \circ \mathcal{N}_Q\big[\dyad{\psi}\big] - \dyad{\psi}}_1 \leq e^{-\Omega(L)},
	\end{align}
	for any $\ket{\psi}_Q$ in the ground state subspace, where again $L = \mathrm{min}(L_\mathrm{x}, L_\mathrm{y})$.
    \label{thm:nonab example}
\end{theorem}

The above is actually slightly weaker than the strongest form of our results, but we present it in this form here due to its conceptually transparency.

\section{Overview of methods \label{sec:methods}} 

Our proofs of these various results for different codes all follow the same overall line of reasoning, which involves two key elements. The first is the formalism of approximate operator algebra error correction, which provides sufficient conditions for certain types of logical information to be recoverable after noise is applied. The second is a technical result that describes how two states that cannot be distinguished on local patches become globally indistinguishable after certain channels are applied. We describe these theoretical constructions in the following sections.

\subsection{Approximate operator algebra error correction} \label{sec:aec-analysis}

In brief, the formalism of approximate operator algebra error correction (AOAEC) generalises the formalism of approximate error correction in a way that allows one to determine whether certain subclasses of logical information can be recovered after noise is applied. For instance, in the $\mathbb{Z}_2$ surface code, one can use this technique to separately determine whether $X$-type or $Z$-type logical information can be recovered. We begin by reviewing the necessary mathematical prerequisites.

\subsubsection{Review: operator subalgebras}

For a finite-dimensional Hilbert space $\mathscr{H}_R$, which we consider throughout this work, a subalgebra $\mathfrak{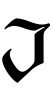}$ is a linear subspace of $\mathfrak{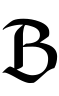}(\mathscr{H}_R)$ (the space of bounded linear operators over $\mathscr{H}_R$) that is also closed under multiplication. Furthermore, $\mathfrak{J}$ is a $\dagger$-algebra if it is also closed under Hermitian conjugation, and it is a \textit{von Neumann algebra} if it also contains the identity for the whole space $I_{\mathscr{H}_R}$ (see e.g.~Ref.~\cite{arveson1976invitation}). Since we are working in finite dimensions, this is equivalent to a $C^*$ algebra. For example, if $d_\mathrm{L} = 2$, then $\mathfrak{B}(\mathscr{H}_R)$ is the space of $2 \times 2$ matrices, spanned by Pauli matrices $\{I, X, Y, Z\}$. The subset of operators $\mathfrak{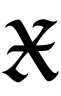} \coloneqq \{\alpha I + \beta X : \alpha, \beta \in \mathbb{C}\}$ forms a von Neumann algebra, which we refer to as the Pauli-$X$ algebra, as does the corresponding Pauli-$Z$ algebra $\mathfrak{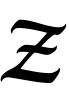}$. 

For every von Neumann algebra $\mathfrak{J} \subseteq \mathfrak{B}(\mathscr{H}_R)$, there is a unique channel $\mathcal{P}^{\mathfrak{J}}_R$ (a CPTP map) which projects operators onto $\mathfrak{J}$. This channel (sometimes referred to as a conditional expectation) satisfies \cite{takesaki2003theory}
\begin{subequations}
	\begin{equation}
		\mathcal{P}^{\mathfrak{J}}_R[O] \in \mathfrak{J} \quad \forall \, O \in \mathfrak{B}(\mathscr{H}_R),
	\end{equation}
	\vspace{-1.5em}
	\begin{equation}
		\mathcal{P}^{\mathfrak{J}}_R \circ \mathcal{P}^{\mathfrak{J}}_R = \mathcal{P}^{\mathfrak{J}}_R,
		\label{eq:super-projector idempotent}
	\end{equation}
	\vspace{-1.5em}
	\begin{equation}
		\mathcal{P}^{\mathfrak{J}}_R[J O J'] = J \mathcal{P}^{\mathfrak{J}}_R[O]J' \quad \forall \, O \in \mathfrak{B}(\mathscr{H}_R); \; J, J' \in \mathfrak{J}.
		\label{eq:super-projector commute}
	\end{equation}
	\label{eq:super-projector}
\end{subequations}
For the Pauli-$X$ subalgebra, the super-projector is $\mathcal{P}^{\mathfrak{X}}_R[\,\cdot\,] = \frac{1}{2}(\,\cdot\, + X\, \cdot \, X)$, which is the completely dephasing channel in the $X$ basis. For the full algebra $\mathfrak{J}= \mathfrak{B}(\mathscr{H}_R)$, we have $\mathcal{P}^{\mathfrak{J}}_R = \mathrm{id}_R$ the identity channel, and for the trivial subalgebra $\mathbb{C}I = \{\alpha I : \alpha \in \mathbb{C}\}$, the super-projector is the completely depolarizing channel $\mathcal{P}^{\mathbb{C}I}_R[\,\cdot\,] = \mathcal{D}_R[\,\cdot\,] = \tr[\,\cdot\,]  I_R/d_R$.

An important concept is the \textit{commutant} of a subalgebra, denoted $C(\mathfrak{J})$, which contains all elements of $\mathfrak{B}(\mathscr{H}_R)$ that commute with every element of $\mathfrak{J}$,
\begin{align}
	C(\mathfrak{J}) = \{O : [O, J] = 0\; \forall J \in \mathfrak{J}\}.
\end{align}
The Pauli-$X$ and -$Z$ algebras are their own commutants, while the trivial algebra $\mathbb{C}I$ and the full algebra $\mathfrak{B}(\mathscr{H}_R)$ are commutants of one another. More generally, the bicommutant theorem states that $C(C(\mathfrak{J})) = \mathfrak{J}$ for all von Neumann algebras $\mathfrak{J}$.

\subsubsection{Recoverability of a subalgebra}

Given an isometric encoding map $\mathcal{C}_{Q \leftarrow R}$ and a noise channel $\mathcal{N}_{Q' \leftarrow Q}$, the relationship between a logical state $\sigma_R$ and its noise-corrupted encoding is described by the concatenation of $\mathcal{C}_{Q \leftarrow R}$ and $\mathcal{N}_{Q' \leftarrow Q}$,
\begin{align}
	\mathcal{N}_{Q' \leftarrow Q} \circ \mathcal{C}_{Q \leftarrow R} \eqqcolon \mathcal{M}_{Q' \leftarrow R}.
	\label{eq:m channel def}
\end{align}
which can be conveniently represented with a single channel $\mathcal{M}_{Q' \leftarrow R}$. A von Neumann subalgebra $\mathfrak{J} \subseteq \mathfrak{B}(\mathscr{H}_R)$ is said to be $\epsilon$-recoverable for $\mathcal{M}_{Q' \leftarrow R}$ if there exists a recovery operation $\mathcal{R}^{\mathfrak{J}}_{R \leftarrow Q'}$ such that
\begin{align}
	\left\| \mathcal{R}^{\mathfrak{J}}_{R \leftarrow Q'} \circ \mathcal{M}_{Q' \leftarrow R} - \mathcal{P}_R^{\mathfrak{J}}\right\|_\diamond \leq \epsilon.
	\label{eq:recov error algebra}
\end{align}
We define $\epsilon^{\mathrm{rec}}_\mathfrak{J}$ as the smallest $\epsilon$ such that the above holds for some optimal $\mathcal{R}^{\mathfrak{J}}_{R \leftarrow Q'}$.

The left hand side of \eqref{eq:recov error algebra} generalises the quantity defined in Eq.~\eqref{eq:recov error}, which governs the correctability of the whole logical space. The difference here is that the map $\mathcal{R}^{\mathfrak{J}}_{R \leftarrow Q'}$ need only restore observables in $\mathfrak{J}$, rather than the full space of operators, since these are the operators within the support of $\mathcal{P}_R^{\mathfrak{J}}$.

Operationally, Equation \eqref{eq:recov error algebra} immediately implies that the expectation value of any Hermitian observable $O_R$ contained in $\mathfrak{J}$ can be approximately inferred by performing a suitably chosen measurement on the noise-corrupted degrees of freedom $Q'$. That is, for each such $O_R \in \mathfrak{J}$, there exists a corresponding operator $\tilde{O}_{Q'}$ supported on $Q'$ such that
\begin{align}
	\left|\tr\big[\tilde{O}_{Q'}  \mathcal{M}_{Q' \leftarrow R}[ \sigma_R]\big] - \tr[O_R \sigma_R]\right| \leq \epsilon. \label{eq:exp value approx}
\end{align}
for any logical state $\sigma_R$\footnote{To prove Eq.~\eqref{eq:exp value approx}, take an isometry $V_{RE' \leftarrow Q'}^{\mathcal{R}^{\mathfrak{J}}}$ that dilates the recovery channel $\mathcal{R}^{\mathfrak{J}}_{R \leftarrow Q'}$, and then choose $\tilde{O}_Q = (V^{\mathcal{R}^{\mathfrak{J}}}_{RE' \leftarrow Q})^\dagger (O_R \otimes I_{E'}) V^{\mathcal{R}^{\mathfrak{J}}}_{RE' \leftarrow Q}$. Then, the left hand side becomes $\abs{\tr[O_R\cdot(\mathcal{R}^{\mathfrak{J}}_{R \leftarrow Q'} \circ \mathcal{M}_{Q' \leftarrow R}[\sigma_R] - \sigma_R)]}$, which by \eqref{eq:recov error algebra} is at most $\epsilon$.}. 

Approximate recoverability also implies that logical operations within $\mathfrak{J}$ can be implemented on the noise-corrupted state, beyond readout of information. In \cref{app:recov-means-log-ops}, we formalise and prove the following result.

\begin{lemma}[Informal] \label{lem:effective operator} 
	If an algebra $\mathfrak{J}$ is $\epsilon$-recoverable for some small $\epsilon$, then for any superoperator $\mathcal{T}^\mathfrak{J}_R$ on the logical space that can be constructed using operators in $\mathfrak{J}$ (i.e.~$\mathcal{T}^\mathfrak{J}_R[O] = \sum_i J_i O J'_i $ for some $J_i, J_i' \in \mathfrak{J}$), there exists a corresponding superoperator $\tilde{\mathcal{T}}^\mathfrak{J}_{Q'}$ acting on the physical space whose action is approximately the same as $\mathcal{T}_R$ \textit{globally}, i.e.~on both system and environment,
	\begin{align}
		(\tilde{\mathcal{T}}^\mathfrak{J}_{Q'} \otimes \mathrm{id}_E) \circ \mathcal{V}^{\mathcal{M}}_{Q'E \leftarrow R} \approx \mathcal{V}^{\mathcal{M}}_{Q'E \leftarrow R} \circ \mathcal{T}^\mathfrak{J}_R.
		\label{eq:effective operator}
	\end{align}
	Here,
	\begin{align}
		\mathcal{V}^{\mathcal{M}}_{Q'E \leftarrow R} = \mathcal{V}^{\mathcal{N}}_{Q' E \leftarrow Q}\circ \mathcal{C}_{Q \leftarrow R} \label{eq:W isom def}
	\end{align}
	is a dilation of $\mathcal{M}_{Q' \leftarrow R}$ and $\mathcal{V}^{\mathcal{N}}_{Q' E \leftarrow Q}[\cdot] = V^{\mathcal{N}}_{Q' E \leftarrow Q} \, \cdot\, (V^{\mathcal{N}}_{Q' E \leftarrow Q})^\dag$ is a dilation of the noise channel $\mathcal{N}_{Q' \leftarrow Q}$, Eq.~\eqref{eq:channel dilation}. 
\end{lemma}
This relation will prove useful in our proof of \cref{thm:fat}, specifically for demonstrating the existence of the thickened string operator $W^Z_S$. 

For the example $\mathfrak{X} = \mathrm{span}\{I_R, X_R\}$, \cref{lem:effective operator} implies that any left- or right-multiplication of projectors $\Pi_{\pm} = (I_R \pm X_R)/2$, and unitary operations of the form $e^{\mathrm{i} \theta X_R} = \cos(\theta) I_R + \mathrm{i} \sin(\theta) X_R$ are approximately implementable.

\subsubsection{Sufficient condition for recoverability from complementary channel \label{subsec:recoverability conditions}}

We will make extensive use of the following result from Ref.~\onlinecite{beny2009}.
Define
\begin{align}
	\delta_{\mathfrak{J}}(E) \coloneqq \left\|  \widehat{\mathcal{M}}_{E \leftarrow R} \circ\Big(\mathrm{id}_R - \mathcal{P}^{C(\mathfrak{J})}_R\Big) \right\|_\diamond, 
	\label{eq:algebra comp error}
\end{align}
where $\mathcal{P}^{C(\mathfrak{J})}_R$ is the super-projector onto the commutant $C(\mathfrak{J})$, and
\begin{align}
	\widehat{\mathcal{M}}_{E \leftarrow R} \coloneqq \tr_{Q'} \circ \mathcal{V}^{\mathcal{M}}_{Q'E \leftarrow R} = \widehat{\mathcal{N}}_{E \leftarrow Q} \circ \mathcal{C}_{Q \leftarrow  R}
\end{align}
is the channel complementary to $\mathcal{M}_{Q' \leftarrow R}$ [Eqs.~(\nopareneqref{eq:m channel def}, \nopareneqref{eq:W isom def})]. If $\delta_{\mathfrak{J}}(E) \leq \epsilon^2/4$, then there exists a recovery operation $\mathcal{R}_{R \leftarrow Q'}$ that satisfies \eqref{eq:recov error algebra}, and hence $\mathfrak{J}$ is $\epsilon$-correctable. Equivalently,
\begin{align}
	\epsilon^{\mathrm{rec}}_\mathfrak{J} \leq 2 \sqrt{\delta_{\mathfrak{J}}(E)}.
	\label{eq:eps rec bound}
\end{align}

Informally, the condition that $\delta_{\mathfrak{J}}(E)$  be 
small implies that the environment learns little about observables outside the span of $C(\mathfrak{J})$,  namely those that do not commute with operators in $\mathfrak{J}$. For the full algebra $\mathfrak{J} = \mathfrak{B}(\mathscr{H}_R)$, the commutant is the trivial algebra, and so the superoperator $\mathrm{id}_R - \mathcal{P}^{C(\mathfrak{J})}_R$ preserves all traceless operators on $R$, in which case \eqref{eq:algebra comp error} will only be small if every traceless observable on $R$ is approximately inaccessible to the environment $E$. This mirrors the observation that the coherent information is maximal if $R$ and $E$ are approximately uncorrelated [Eq.~\eqref{eq:coherent RE}].

As we will see, evaluating $\delta_{\mathfrak{J}}(E)$ will turn out to be a much simpler task than explicitly searching for a recovery map $\mathcal{R}^{\mathfrak{J}}_{R \leftarrow Q'}$ that satisfies \eqref{eq:recov error algebra}.

\subsubsection{Application to topological codes}

In topological codes, different subalgebras of logical operators  are encoded within various different subregions. Reflecting this fact, we will use the formalism described above to obtain a \textit{spatially resolved characterisation} of how different components of the logical information are stored in the noise-corrupted physical state.

For example, in the case of the $\mathbb{Z}_2$ surface code, we will ask whether logical $X$-information can be recovered from a subregion $S \subset Q'$ of the noise-corrupted state $\tilde{\rho}_{Q'}$. This can be addressed within the above formalism as follows. Naturally, we choose $\mathfrak{J} = \mathfrak{X}$, the Pauli-$X$ algebra, and for a given noise channel $\mathcal{N}_{Q' \leftarrow Q}$, we construct the map $\mathcal{N}_{S \leftarrow Q} \coloneqq \tr_{Q' \setminus S} \circ \mathcal{N}_{Q' \leftarrow Q}$. The question is then whether the algebra $\mathfrak{X}$ is recoverable for the composite channel $\mathcal{N}_{S \leftarrow Q} \circ \mathcal{C}_{Q \leftarrow R}$, which describes the relationship between the logical information $R$ and the marginal of the noise-corrupted state on $S$.

To make progress we will need to find a channel complementary to $\mathcal{N}_{S \leftarrow Q}$. To do so, let $\mathcal{V}^{\mathcal{N}}_{Q'E \leftarrow Q}$ be a Stinespring dilation of the global noise channel $\mathcal{N}_{Q' \leftarrow Q}$, we can define $\widehat{\mathcal{N}}_{E_S \leftarrow Q} = \tr_{S} \circ \mathcal{V}^{\mathcal{N}}_{Q'E \leftarrow Q}$, which is complementary to $\mathcal{N}_{S \leftarrow Q}$, and whose output space is
\begin{align}
	E_S \coloneqq (Q' \setminus S) \cup E.
    \label{eq:environment S}
\end{align}
Namely, $E_S$ (the ``environment of $S$'') includes the degrees of freedom in $Q'$ not accessible to the observer $S$, along with all the environment degrees of freedom $E$.

For product noise $\mathcal{N}_{Q' \leftarrow Q} = \bigotimes_{\ell \in Q} \mathcal{N}_{Q'_\ell \leftarrow \ell}$, we can choose $S$ to correspond to a subset of qudits, $S \subset Q$. In this case, we can choose a simpler choice of complementary channels such that the form of the two mutually complementary channels mirrors one another,
\begin{equation}
	\begin{aligned}
		\mathcal{N}_{S \leftarrow Q} &= \tr_{Q \setminus S} \otimes \bigotimes_{\ell \in S}\mathcal{N}_{Q'_\ell \leftarrow \ell}, \\
		\widehat{\mathcal{N}}_{E_S \leftarrow Q} &= \mathrm{id}_{Q \setminus S} \otimes \bigotimes_{\ell \in S}\widehat{\mathcal{N}}_{E_\ell \leftarrow \ell}.
	\end{aligned}
	\label{eq:noise channels subsystem}
\end{equation}
That is, $E_S$ receives all the noise-uncorrupted qudits outside of $S$, along with the environment $E_\ell$ for each $\ell \in S$.

As before, we also define $\mathcal{M}_{S \leftarrow R} \coloneqq \mathcal{N}_{S \leftarrow Q} \circ \mathcal{C}_{Q \leftarrow R}$, and $\widehat{\mathcal{M}}_{E_S \leftarrow R} \coloneqq \widehat{\mathcal{N}}_{E_S \leftarrow Q} \circ \mathcal{C}_{Q \leftarrow R}$. These can be used to define local versions of the global quantities given in Eqs.~(\nopareneqref{eq:recov error algebra}, \nopareneqref{eq:algebra comp error}). The quantity
\begin{align}
    \delta_{\mathfrak{J}}(E_S) \coloneqq \norm{\mathcal{N}_{S \leftarrow Q} \circ \mathcal{C}_{Q \leftarrow R} - \mathcal{N}_{S \leftarrow Q} \circ \mathcal{C}_{Q \leftarrow R} \circ \mathcal{P}^{C(\mathfrak{J})}_R}_\diamond
    \label{eq:delta ES}
\end{align}
can be used to determine whether the subalgebra $\mathfrak{J}$ can be approximately recovered by acting on $S$ only, in the the sense that there exists a recovery channel $\mathcal{R}_{R \leftarrow S}$ satisfying
\begin{align}
    \norm{(\mathcal{R}_{R \leftarrow S}\otimes \mathrm{tr}_{S^c})\circ \mathcal{M}_{Q' \leftarrow R} - \mathcal{P}^{\mathfrak{J}}_R }_\diamond = \epsilon^{\mathrm{rec}}_{\mathfrak{J}}(S)
    \label{eq:recovery local}
\end{align}
where the local recovery error is bounded as \cite{beny2009}
\begin{align}
    \epsilon^{\mathrm{rec}}_{\mathfrak{J}}(S) \leq 2 \sqrt{\delta_{\mathfrak{J}}(E_S)}.
    \label{eq:eps rec bound local}
\end{align}

\subsection{Local and global indistinguishability of code states \label{subsec:local global indist}}

The quantity $\delta_{\mathfrak{J}}(E_S)$ [Eq.~\eqref{eq:delta ES}] measures the distinguishability of the channels $\widehat{\mathcal{N}}_{E_S \leftarrow Q}\circ \mathcal{C}_{Q \leftarrow R}$ and $\widehat{\mathcal{N}}_{E_S \leftarrow Q} \circ \mathcal{C}_{Q \leftarrow R} \circ \mathcal{P}^{C(\mathfrak{J})}_R$. In other words, the size of $\delta_{\mathfrak{J}}(E_S)$ determines whether the code state associated with some logical information $\sigma_R$ could be distinguished from the state associated with $\mathcal{P}^{C(\mathfrak{J})}_R[\sigma_R]$, after the channel $\widehat{\mathcal{N}}_{E_S \leftarrow Q}$ is applied.

To illustrate our general method for upper bounding $\delta_{\mathfrak{J}}(E_S)$, let us work with the example of the $\mathbb{Z}_2$ surface code, with $\mathfrak{J} = \mathfrak{X}$, the Pauli-$X$ algebra. In this case, $\mathcal{P}^{C(\mathfrak{X})}_R$ is the completely dephasing channel in the logical $X$-basis. Suppose that the noise channel $\mathcal{N}_Q$ traces out a subset of qubits $\Sigma \subset S$, and leaves the rest undisturbed. By Eq.~\eqref{eq:noise channels subsystem}, the complementary channel $\widehat{\mathcal{N}}_{E_S \leftarrow Q}$ is then equivalent to tracing out the region $\Sigma^\mathrm{c} \coloneqq S \setminus \Sigma$. If $\Sigma^\mathrm{c}$ supports a logical-$X$ operator, $\bar{X}_{\Sigma^\mathrm{c}}$, then we find that $\tr_{\Sigma^\mathrm{c}}\circ \mathcal{C}_{Q \leftarrow R}$ and $\tr_{\Sigma^\mathrm{c}} \circ \mathcal{C}_{Q \leftarrow R} \circ \mathcal{P}^{C(\mathfrak{X})}_R$ are \textit{exactly} indistinguishable, $\delta_\mathfrak{Z}(E_S) = 0$. This is because the logical dephasing channel $\mathcal{P}^{C(\mathfrak{X})}_R$ can be implemented on the physical level by applying the gate $\bar{Z}_{\Sigma^\mathrm{c}}$ with probability 1/2, and doing nothing otherwise. The environment $E_S$, which does not have access to $\Sigma^\mathrm{c}$, cannot detect whether or not this operation was applied.

The indistinguishability of these channels on the region $\Sigma \cup S^c$ for particular choices of $\Sigma$ we refer to as \textit{local indistinguishability} of code states. The more general non-Abelian codes we consider later on also exhibit natural generalisations of this property, for suitable choices of algebras $\mathfrak{J}$ and subregions $\Sigma \subset S$.

Our key technical innovation is a method to convert this local indistinguishability property into \textit{global indistinguishability} after the channel $\widehat{\mathcal{N}}_{E_S \leftarrow Q}$ is applied. In other words, the fact that $ \mathcal{C}_{Q \leftarrow R}$ and $ \mathcal{C}_{Q \leftarrow R} \circ \mathcal{P}^{C(\mathfrak{J})}_R$ are identical on certain local regions can be used to show that $\widehat{\mathcal{N}}_{E_S \leftarrow Q}\circ \mathcal{C}_{Q \leftarrow R}$ and $\widehat{\mathcal{N}}_{E_S \leftarrow Q} \circ \mathcal{C}_{Q \leftarrow R} \circ \mathcal{P}^{C(\mathfrak{J})}_R$ are approximately equal on the whole of $E_S$. This gives a bound on $\delta_{\mathfrak{J}}(E_S)$ that can be applied for any choice of local noise channel, irrespective of any Pauli stabilizer structure. Our argument, stated formally in Appendix \ref{apdx:bounds} as \cref{lem:bound-norm-connectivity}, captures the following intuition: Local indistinguishability implies that one would have to measure sufficiently non-local observables in order to distinguish the two channels; this becomes difficult because the complementary channel $\widehat{\mathcal{N}}_{E_S \leftarrow Q}$ dampens the norm of operators by a factor that is exponentially small in the size of their support.

This suppression of high-weight observables can be quantified with the following definition of an effective noise strength $T(\mathcal{N}_\ell)$ for a single-site noise channel $\mathcal{N}_\ell$, which appeared in our informal statement of \cref{thm:fat}.

\begin{definition} \label{def:noise strength def}
	For a single-site noise channel $\mathcal{N}_{\ell}$, the effective noise strength $T(\mathcal{N}_{\ell})$ is given by
	\begin{align} \label{eq:noise strength def}
		\norm{\widehat{\mathcal{N}}_{E_\ell \leftarrow \ell}\circ (\mathrm{id}_\ell - \mathcal{D}_\ell)}_\diamond =: e^{-1/T(\mathcal{N}_\ell)}.
	\end{align}
	Here, $\mathcal{D}_{\ell}[\,\cdot\,] = (I_\ell/d_\ell)\tr_\ell[\,\cdot\,]$ is the completely depolarizing channel, and $\widehat{\mathcal{N}}_{E_\ell \leftarrow \ell}$ is any channel complementary to $\mathcal{N}_{\ell}$. We will sometimes use the inverse noise strength $\beta(\mathcal{N}_\ell) = 1/T(\mathcal{N}_\ell)$. 
\end{definition}

Since $(\mathrm{id}_\ell - \mathcal{D}_\ell)[I_\ell] = 0$, and $(\mathrm{id}_\ell - \mathcal{D}_\ell)[O_\ell] = O_\ell$ for any traceless operator $O_\ell$, we can roughly interpret the norm of the superoperator $\widehat{\mathcal{N}}_{E_\ell \leftarrow \ell}\circ (\mathrm{id}_\ell - \mathcal{D}_\ell)$ as a measure of how much information the environment $E_{\ell}$ receives about traceless operators on $\ell$.  As a concrete example, channels that are close to the identity have a small effective noise strength (\cref{apdx:diamond-norms}),
\begin{align}
	\|\mathcal{N}_\ell - \mathrm{id}_\ell\|_\diamond \leq \epsilon \;\; \implies \;\; T(\mathcal{N}_\ell) \leq 4 \sqrt{\epsilon}.
	\label{eq:close to identity}
\end{align}
However, this bound is often not tight. For instance, if $\mathcal{N}_\ell$ is unitary far from identity, then $T(\mathcal{N}_\ell) = 0$ even though the left side of \eqref{eq:close to identity} is large.

This definition of the noise strength quantifies the way in which the complementary channel dampens high weight observables. Indeed, if $\mathcal{N}_Q = \bigotimes_{\ell \in Q} \mathcal{N}_\ell$ is a product noise channel with complement $\widehat{\mathcal{N}}_{E \leftarrow Q}$ and $O$ is a Pauli string that acts non-trivially on each $\ell \in \Sigma$, then we have $\|\widehat{\mathcal{N}}_{E \leftarrow Q}^{\,\dagger}[O]\|_\infty \leq e^{-|\Sigma|/T} \|O\|_\infty$, where $T \coloneqq \min_{\ell \in Q} T(\mathcal{N}_\ell)$ and $\mathcal{N}^\dag$ denotes the adjoint channel \cite{watrous2018theory}. This is exponentially small in the size of the support of $O$.\\

In Appendix \ref{apdx:q double codes}, we state and prove a general result that relates the local indistinguishability property described above to the size of $\delta_{\mathfrak{J}}(E_S)$, via our definition of the noise strength \eqref{eq:noise strength def}. This will be invoked to prove the majority of results (Section \ref{sec:quantum-doubles}). However, for the $\mathbb{Z}_2$ surface code, some additional simplifications arise that will ease our presentation, and allow us to obtain quantitatively better bounds. We therefore specialize to the $\mathbb{Z}_2$ case in the following section.

\section{Proof of Theorem \ref{thm:fat} for $\mathbb{Z}_2$ surface code: existence of thickened logical operators} \label{sec:main-proof}

In this Section, we provide a high-level proof of \cref{thm:fat}, which was originally stated informally in \cref{sec:summary-sc}. Our arguments for non-Abelian codes, which we treat in Section \ref{sec:quantum-doubles}, will follow a very similar structure.

We first provide the formal statement of the Theorem.

\begin{customthm}{1}[2D $\mathbb{Z}_2$ surface code]
	For the 2D $\mathbb{Z}_2$ surface code in the geometry of \cref{fig:geometries}(a, b), suppose that a single-site noise channel $\mathcal{N}_Q = \bigotimes_{\ell \in Q}\mathcal{N}_\ell$ satisfies
	\begin{align}
		T \coloneqq \max_{\ell}T(\mathcal{N}_\ell) < T^{\Lambda_*}_\mathrm{c}.
		\label{eq:temp bound}
	\end{align}
	Here, $T(\mathcal{N}_\ell)$ is defined in \cref{def:noise strength def} and the constant $T^{\Lambda_*}_{\mathrm{c}}$ depends only on the dual lattice $\Lambda_*$.
	
	Then, there exist constants $c_{1,2}(T) > 0$ such that, for any vertical strip of qubits $S \subseteq Q$ of width $w_{\mathrm{x}}$, there exists a recovery channel $\mathcal{R}_{R \leftarrow S}$ acting on $S$ only that approximately recovers the logical-$X$ information with error
	\begin{align}
		\epsilon^{\mathrm{rec}}_{\mathfrak{X}}(S) \leq \varepsilon \coloneqq  c_1(T) \sqrt{L_{\rm y}} e^{-c_2(T)w_{\rm x}},
		\label{eq:epsilon rec thm-fat}
	\end{align}
	where $\epsilon^{\mathrm{rec}}_\mathfrak{X}(S)$ is the error defined in \cref{eq:recovery local}. In particular, there exists a thickened ribbon operator $W^\mathrm{X}_S$ supported on $S$ such that
    \begin{align}
        \norm{W^\mathrm{X}_S \cdot\mathcal{N}_Q \circ \mathcal{C}_{Q \leftarrow R}[\sigma_R] - \mathcal{N}_Q \circ \mathcal{C}_{Q \leftarrow R}[X_R\sigma_R]}_1 \leq \varepsilon.
        \label{eq:effective op thm-fat}
    \end{align}
	The same holds for horizontal strips of height $w_\mathrm{y}$ as shown in \cref{fig:geometries}(c), with $\mathfrak{X}$ replaced by $\mathfrak{Z}$, and the right hand side of \eqref{eq:temp bound} by $T^{\Lambda}_{\mathrm{c}}$.
\end{customthm}

The proof of this result, which is presented in the following, involves upper-bounding $\delta_{\mathfrak{X}}(E_S)$, which by Eq.~\eqref{eq:eps rec bound local} and \cref{lem:effective operator} suffices to prove Eqs.~(\nopareneqref{eq:epsilon rec thm-fat}, \nopareneqref{eq:effective op thm-fat}). In Section \ref{sec:upper bound simple}, we will outline a simpler upper bound on $\delta_{\mathfrak{X}}(E_S)$---one which will generalise in a natural way to the non-Abelian codes considered later. Then, in Section \ref{sec:sow-mapping} we explain how this can be improved by exploiting certain structure specific to the $\mathbb{Z}_2$ surface code, resulting in a quantitatively better bound on the critical noise strength $T_{\rm c}^{\Lambda_\star}$. Specifically, we bound $\delta_{\mathfrak{X}}(E_S)$ in terms of a partition function of a classical statistical mechanics model known as \textit{self-osculating walks}, with the temperature given by the noise strength $T$. This model exhibits a phase transition at the critical value $T_{\rm c}^{\Lambda_*}$, and for $T < T_{\rm c}^{\Lambda_*}$ the partition function becomes exponentially small in $w_{\rm x}$.

We note that though \cref{thm:fat} is stated for the geometries of \cref{fig:geometries} on the square lattice $\sqlattice$, the methods used to prove it can be immediately generalised to other lattices $\Lambda$ like the hexagonal $\hexlattice$ and triangular lattices $\trilattice$ and even to arbitrary planar (dual)~graphs.

\subsection{Upper bound on $\delta_{\mathfrak{X}}$ \label{sec:upper bound simple}}

\begin{lemma}
    With the assumptions of Theorem 1 and $S$ a horizontal strip of qubits of height $w_y < L_y$, for any $T < [1 + \ln(k_\Lambda-1)]^{-1}$ there exist constants $c'_{1,2}(T) > 0$ such that
    \begin{align}
        \delta_{\mathfrak{X}}(E_S) \leq c_1'(T)L_{\rm y} e^{-c_2'(T)w_{\rm y}}.
    \end{align}
    Here, $k_\Lambda$ is the maximum number of edges $\ell' \neq \ell$ that share a vertex with a given edge $\ell$.
    \label{lem:delta bound simple}
\end{lemma}

Combined with Equation \eqref{eq:eps rec bound local} and Lemma \ref{lem:effective operator}, the above result implies \cref{thm:fat}. In the following, we outline our proof of Lemma \ref{lem:delta bound simple}, leaving more technical details to Appendix \ref{apdx:bounds}. This method will generalise naturally to the non-Abelian codes considered later, which is why we highlight it here, but quantitatively better bounds on the threshold noise strength $T$ can be obtained using a modified approach that exploits structure specific to the $\mathbb{Z}_2$ surface code. This is described in Section \ref{sec:sow-mapping}.

\textit{Proof sketch. --- }First, recall that $C(\mathfrak{X}) = \mathfrak{X}$, and so $\mathcal{P}^{C(\mathfrak{X})}_R = \mathcal{P}^{\mathfrak{X}}_R$ is the completely dephasing channel in the $X$-basis, which acts as
\begin{align}
    \mathcal{P}^{\mathfrak{X}}_R[\sigma_R] = \frac{1}{2}(\sigma_R + X_R \sigma_R X_R).
    \label{eq:X proj channel}
\end{align}
Then define
\begin{align}
    \mathcal{\Delta}_{Q \leftarrow R} \coloneqq \mathcal{C}_{Q \leftarrow R} - \mathcal{C}_{Q \leftarrow R} \circ \mathcal{P}^{\mathfrak{X}}_R,
    \label{eq:Delta def}
\end{align}
such that $\delta_\mathfrak{X}(E_S) = \|\widehat{\mathcal{N}}_{E_S \leftarrow Q} \circ \mathcal{\Delta}_{Q \leftarrow R}\|_\diamond$. 
We start with a trivial resolution of the identity channel $\mathrm{id}_\ell \equiv \mathcal{D}_\ell + (\mathrm{id}_\ell - \mathcal{D}_\ell)$,  where $\mathcal{D}_\ell$ is the completely depolarizing channel on qubit $\ell$. Taking a product of this expression over all $\ell \in S$ and expanding, we obtain 
\begin{align}
    \mathcal{\Delta}_{Q \leftarrow R}  &= \sum_{\Sigma \in \mathscr{P}(S)} \mathcal{\Delta}_{Q \leftarrow R}^{\,(\Sigma)} \text{, where} \nonumber\\  \mathcal{\Delta}_{Q \leftarrow R}^{\,(\Sigma)} &\coloneqq \left(\left[\bigotimes_{\ell \in \Sigma}(\mathrm{id}_\ell - \mathcal{D}_\ell) \right]\otimes \left[\bigotimes_{\ell \in \Sigma^\mathrm{c}}  \mathcal{D}_\ell \right] \right) \circ \mathcal{\Delta}_{Q \leftarrow R}.
    \label{eq:delta decomp}
\end{align}
Here, $\mathscr{P}(S) = \{\Sigma : \Sigma \subseteq S\}$ is the power set of $S$, and $\Sigma^\mathrm{c} \coloneqq S \setminus \Sigma$. In words, $\mathcal{\Delta}_{Q \leftarrow R}^{\,(\Sigma)}$ is the component of $\mathcal{\Delta}_{Q \leftarrow R}$ that remains after tracing out $\Sigma^\mathrm{c}$, and projecting onto operators that are traceless on each $\ell \in \Sigma$. As a result,
\begin{align}
    \delta_\mathfrak{X}(E_S) = \norm{ \sum_{\Sigma \in \mathscr{P}(S)} \widehat{\mathcal{N}}_{E_S \leftarrow Q} \circ \mathcal{\Delta}_{Q \leftarrow R}^{\,(\Sigma)}}_\diamond.
    \label{eq:delta diamond}
\end{align}

Crucially, the sum over $\Sigma$ can be restricted to subsets that obey certain graph-theoretic constraints. The most important of these is a consequence of the following observation: 
Define a \textit{spanning path} $W$ to be a sequence of distinct edges $(\ell_1, \ell_2, \ldots, \ell_{|W|})$ such that sequential pairs $(\ell_i$, $\ell_{i+1})$ share a vertex, such that $\ell_1$ on the top boundary, $\ell_{|W|}$ on the bottom boundary. Then, let $\Omega^+ \subset \mathscr{P}(S)$ be the set of $\Sigma$ that contain a spanning path. We have
\begin{align}
    \Sigma \notin \Omega^+ \; \Longrightarrow\; \mathcal{\Delta}_{Q \leftarrow R}^{\,(\Sigma)} = 0.
    \label{eq:sigma constraint}
\end{align}

\begin{figure}
    \centering
    \includegraphics[width=\linewidth]{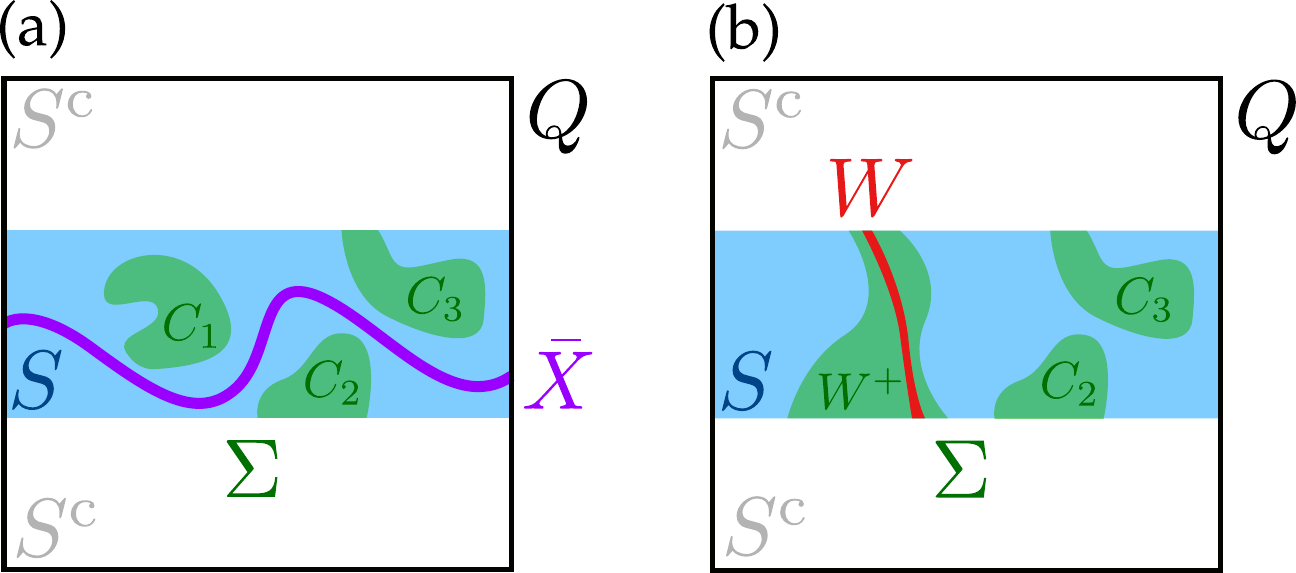}
    \caption{Illustration of conditions for local indistinguishability of code states in the $\mathbb{Z}_2$ surface code, Eq.~\eqref{eq:sigma constraint}. The region $\Sigma$ (green shaded area) is decomposed in terms of maximally connected components $\Gamma[\Sigma] = \{C_1, C_2, \ldots\}$. If none of these components contains a path connecting to the top to the bottom boundary of $S$ (red line), then $\Sigma^c \coloneqq S \setminus \Sigma$ (blue shaded region) contains a logical $\bar{X}$ operator (purple line). This implies that the logical super-projector $\mathcal{P}^{\mathfrak{X}}_R$ appearing in Eq.~\eqref{eq:Delta def} can be implemented by a channel supported in $\Sigma^c$, which suffices for $\mathcal{\Delta}^{(\Sigma)}_{Q \leftarrow R}$ [Eq.~\eqref{eq:delta decomp}] to vanish.}
    \label{fig:local-and-global-approximate-indistinguishability}
\end{figure}

To prove the statement above, we show that the absence of a spanning path in $\Sigma$ implies that $\Sigma^\mathrm{c}$ contains a dual path $\gamma^*$ that supports a logical-$X$ operator, $\bar{X}_{\gamma^*} = \prod_{\ell \in \gamma^*} X_\ell$ (see \cref{fig:local-and-global-approximate-indistinguishability} and Appendix \ref{apdx:bounds}). This can be used to define a channel $\mathcal{P}^{\mathfrak{X}}_{\Sigma^\mathrm{c}}$ acting as $\mathcal{P}^{\mathfrak{X}}_{\Sigma^\mathrm{c}}[\sigma_Q] = \frac{1}{2}(\sigma_Q + \bar{X}_{\gamma^*} \sigma_Q \bar{X}_{\gamma^*})$, which is supported on $\Sigma^\mathrm{c}$, and satisfies $\mathcal{P}^{\mathfrak{X}}_{\Sigma^\mathrm{c}} \circ \mathcal{C}_{Q \leftarrow R} = \mathcal{C}_{Q \leftarrow R} \circ \mathcal{P}^{\mathfrak{X}}_R$, by virtue of Eqs.~(\nopareneqref{eq:super-projector commute}, \nopareneqref{eq:X proj channel}). Thanks to $\mathcal{P}^{\mathfrak{X}}_{\Sigma^\mathrm{c}}$ being supported in $\Sigma^\mathrm{c}$, we have $\mathcal{D}_{\Sigma^\mathrm{c}} \circ \mathcal{P}^{\mathfrak{X}}_{\Sigma^\mathrm{c}} = \mathcal{D}_{\Sigma^\mathrm{c}}$, where $\mathcal{D}_{\Sigma^\mathrm{c}} \coloneqq \bigotimes_{\ell \in \Sigma^\mathrm{c}} \mathcal{D}_\ell$, which together with the definition \eqref{eq:Delta def} proves the desired result.

This already suggests a route towards upper bounding $\delta_{\mathfrak{X}}(E_S)$ by a factor exponentially small in $w_y$: Thanks to the definition of the effective noise strength \eqref{eq:noise strength def}, we have\footnote{In proving Equation \eqref{eq:delta sigma bound}, we use the inequality $\|\mathcal{T} \circ \mathcal{T}'\|_\diamond \leq \|\mathcal{T}\|_\diamond \| \mathcal{T}'\|_\diamond$ for any linear maps $\mathcal{T}, \mathcal{T}'$ compatible with concatenation \cite{aharonov1998quantum}, and the fact that $\|\mathcal{T}\|_\diamond = 1$ for any quantum channel $\mathcal{T}$, along with the triangle inequality.}
\begin{align}
    &\norm{\widehat{\mathcal{N}}_{E_S \leftarrow Q} \circ \mathcal{\Delta}_{Q \leftarrow R}^{\,(\Sigma)}}_\diamond = \Bigg\|\Bigg( \Bigg[\bigotimes_{\ell \in \Sigma}\widehat{\mathcal{N}}_{E_\ell \leftarrow \ell} \circ(\mathrm{id}_\ell - \mathcal{D}_{\ell})\Bigg] \nonumber\\ \otimes& \left[\bigotimes_{\ell \in \Sigma^\mathrm{c}}\widehat{\mathcal{N}}_{E_\ell \leftarrow \ell} \circ\mathcal{D}_{\ell} \right]\Bigg) \circ \mathcal{\Delta}_{Q \leftarrow R}\Bigg\|_\diamond \leq 2e^{-|\Sigma|/T}.
    \label{eq:delta sigma bound}
\end{align}
Since we need only include regions $\Sigma \in \Omega^+$, and such regions have size at least $w_y$, every term will be suppressed by a factor of at least $e^{-w_y/T}$. However, there are $2^{|S|}$ terms in the sum over $\Sigma$ in \eqref{eq:delta decomp}, and a na{\"i}ve application of the triangle inequality $\|\sum_{\Sigma \in P(S)} \widehat{\mathcal{N}}_{E_S \leftarrow Q} \circ \mathcal{\Delta}_{Q \leftarrow R}^{\,(\Sigma)}\|_\diamond \leq \sum_{\Sigma \in P(S)} \|\widehat{\mathcal{N}}_{E_S \leftarrow Q} \circ \mathcal{\Delta}_{Q \leftarrow R}^{\,(\Sigma)}\|_\diamond$ would result in a vacuous bound.

To overcome this abundance of terms, we devise a careful combinatorial decomposition of the sum over $\Sigma$, which allows us to use the triangle inequality in a less wasteful fashion. First, let $\Gamma[\Sigma] = \{C_1, \ldots, C_{|\Gamma[\Sigma]|}\}$ denote the set of maximally connected components of a given subset $\Sigma \subseteq S$ (Fig.~\ref{fig:local-and-global-approximate-indistinguishability}), i.e.~each $C_i$ is a connected subset of $\Sigma$, and is not contained in any larger connected subset of $\Sigma$. In defining connectivity, we consider two edges to be adjacent if and only if they share a vertex. Since paths on the primal lattice are connected with respect to this definition of adjacency, our restriction on $\Sigma$ implies that at least one component $C_i \in \Gamma[\Sigma]$ must contain a spanning path.

With this in mind, let $\mathscr{W}^+$ be the collection of all regions $W^+ \subseteq S$ that are connected, and that contain a spanning path.
Then, for a fixed choice of $W^+$, consider the partial sum
\begin{align}
    \mathcal{\Delta}^{\{W^+\}}_{Q \leftarrow R} \coloneqq \sum_{\substack{\Sigma \in P(S) \\ \Gamma[\Sigma] \ni W^+}} \mathcal{\Delta}^{\,(\Sigma)}_{Q \leftarrow R}
    \label{eq:delta singleton}
\end{align}
which includes all $\Sigma$ that have $W^+$ as one of their connected components. Note that the sum of all such objects
\begin{align}
    \sum_{W^+ \in \mathscr{W}^+} &\widehat{\mathcal{N}}_{E_S \leftarrow Q}\circ \mathcal{\Delta}^{\{W^+\}}_{Q \leftarrow R} \nonumber\\ &= \sum_{W^+ \in \mathscr{W}^+} \sum_{\substack{\Sigma \in P(S) \\ \Gamma[\Sigma] \ni W^+}} \widehat{\mathcal{N}}_{E_S \leftarrow Q}\circ \mathcal{\Delta}^{\,(\Sigma)}_{Q \leftarrow R}
    \label{eq:delta all Cplus}
\end{align}
includes every term $\Sigma \in \Omega^+$ at least once, but with some double-counting. Leaving this double-counting aside for now, we will show that  the norm $\|\widehat{\mathcal{N}}_{E_S \leftarrow Q}\circ \mathcal{\Delta}^{\{W^+\}}_{Q \leftarrow R}\|_\diamond$, which includes a large number of terms $\Sigma$, can be upper bounded in a much more efficient way. 

For a given $\Sigma \in \mathscr{P}(S)$ we observe that $\Gamma[\Sigma] \ni W^+$ if and only if both the following conditions hold: (1) $\ell \in \Sigma$ for every $\ell \in W^+$, and (2) $\ell \notin \Sigma$ for every $\ell \in \partial W^+$, where $\partial W^+$ is the set of edges adjacent to $W^+$, but not contained in $W^+$. The first condition is immediate since $C \subseteq \Sigma$ for any $C \in \Gamma[\Sigma]$. The second condition ensures that $W^+$ is a \textit{maximal} connected component of $\Sigma$, i.e.~there is no additional edge $\ell \in \Sigma$ that can be added to $W^+$ to give a larger connected region. Therefore, we can characterise the sum in \eqref{eq:delta singleton} as being over all choices of $\Sigma$ that contain $W^+$ and exclude $\partial W^+$---a total of $2^{|S| - |W^+|- |\partial W^+|}$ terms corresponding to the choice of including or excluding each $\ell \in S \setminus (W^+ \cup \partial W^+)$. We can thus re-sum the trivial identity $\mathrm{id}_\ell = \mathcal{D}_\ell + (\mathrm{id}_{\ell} - \mathcal{D}_\ell)$ for these unconstrained edges $\ell$ to obtain
\begin{align} 
    \mathcal{\Delta}^{\{W^+\}}_{Q \leftarrow R} =  \left(\left[\bigotimes_{\ell \in W^+}(\mathrm{id}_\ell - \mathcal{D}_\ell) \right]\otimes \left[\bigotimes_{\ell \in \partial W^+}  \mathcal{D}_\ell \right] \right) \circ \mathcal{\Delta}_{Q \leftarrow R}.
    \label{eq:delta Cplus rep}
\end{align}
By analogy to \eqref{eq:delta sigma bound}, the corresponding norm is at most
\begin{align}
    \norm{\widehat{\mathcal{N}}_{E_S \leftarrow Q}\circ \mathcal{\Delta}^{\{W^+\}}_{Q \leftarrow R}}_{\diamond} \leq 2e^{-|W^+|/T}.
    \label{eq:delta singleton bound}
\end{align}
Thus, by applying the triangle inequality to the sum over $W^+$, rather than the full sum over $\Sigma$, the norm of the expression \eqref{eq:delta all Cplus} can be bounded by $2\sum_{W^+ \in \mathscr{W}^+} e^{-|W^+|/T}$. 

The value of this sum depends on how fast the number of regions $W^+$ of a given size $m$ grows with $m$. As a crude upper bound, we can use the fact that the number of connected subsets of $S$ of size $m$ is at most $|S| e^{m[1 + \ln(k_\Lambda - 1)]}$ for $m \geq w_{\rm y}$ [\onlinecite{bollobas2006art}, pp.~130--133]. Importantly, this grows only exponentially with $m$ with a constant base, in contrast to the total number of subsets of $S$ of size $m$, which is ${|S| \choose m} \sim |S|^m$.
Thus, thanks to our assertion that $W^+$ be connected, if $T < [1 + \ln (k_\Lambda-1)]^{-1}$ then the sum over $W^+$ converges, and is upper bounded by $\kappa |S| e^{- \alpha w_{\rm y}}$ for $T$-dependent constants $\kappa, \alpha$.

The inequality we have derived falls short of the desired bound on \eqref{eq:delta diamond}, due to the aforementioned double counting:  if $\Sigma$ contains multiple connected components that include spanning walks, $C_1, C_2, \ldots \in \mathscr{W}^+$, then $\Sigma$ will be included in the sum \eqref{eq:delta all Cplus} an equal number of times, corresponding to the different choices of $W^+ = C_i$. In the Appendix, we explain how a version of the inclusion-exclusion principle can be used to remove these multiplicities, thereby bounding \eqref{eq:delta diamond} as desired.

The fully rigorous proof of \cref{lem:delta bound simple}, which we provide in Appendix \ref{apdx:bounds}, both formalises and generalises these arguments, in a way that will allow us to obtain similar bounds for non-Abelian codes. In brief, with the map $\mathcal{\Delta}_{Q \leftarrow R}$ [Eq.~\eqref{eq:Delta def}] re-defined for the appropriate code and algebra $\mathfrak{J}$, one can always employ the decomposition \eqref{eq:delta decomp}. Then, one must identify the appropriate graph-theoretic constraints of the kind given in Eq.~\eqref{eq:sigma constraint}, which capture the local indistinguishability of code states. All the remaining steps then proceed in much the same way as described above.

\subsection{Quantitative improvement via mapping to self-osculating walks \label{sec:sow-mapping}}

For the square lattice $\mathbb{Z}_2$ surface code, the bound on the threshold noise strength $T_{\rm c}$ implied by \cref{lem:delta bound simple} is $T_{\rm c} = [1 + \ln 5]^{-1} \approx 0.383$. In this section we will improve this threshold by exploiting some additional structure exhibited by the $\mathbb{Z}_2$ surface code.
This results in orders of magnitude improvement on the lower bound for the threshold; for example, for amplitude damping errors, the lower bound for the threshold is improved from $0.5\%$ to $13.3\%$.

The $\mathbb{Z}_2$ surface code is a Calderbank-Shor-Steane (CSS) stabilizer code, since the vertex stabilizers $A_v$ are products of Pauli-$X$s, and the plaquette stabilizers $B_p$ are products of Pauli-$Z$s. As is well-known, when independent Pauli-$X$ and Pauli-$Z$ noise is applied to a CSS code, the recoverability of logical-$X$ information can be determined by considering the syndromes of the $X$-type stabilizers only, and similarly for $Z$ \cite{dennis2002topological}. Here, where we consider non-Pauli noise, we can similarly amend our argument in a way that reflects this structure.

In place of the decomposition \eqref{eq:delta decomp}, which breaks up the output of $\mathcal{\Delta}_{Q \leftarrow R}$ into components that are traceless on each $\ell \in \Sigma$ and proportional to identity on $\ell \in \Sigma^\mathrm{c}$, we instead write
\begin{align}
    \mathcal{\Delta}_{Q \leftarrow R}  &= \sum_{\Sigma \in P(S)} \tilde{\mathcal{\Delta}}_{Q \leftarrow R}^{\,(\Sigma; X)} \text{, where} \nonumber\\  \tilde{\mathcal{\Delta}}_{Q \leftarrow R}^{\,(\Sigma; X)} &\coloneqq \left(\left[\bigotimes_{\ell \in \Sigma}(\mathrm{id}_\ell - \mathcal{E}_\ell^X) \right]\otimes \left[\bigotimes_{\ell \in \Sigma^\mathrm{c}}  \mathcal{E}_\ell^X \right] \right) \circ \mathcal{\Delta}_{Q \leftarrow R},
    \label{eq:delta decomp X}
\end{align}
where $ \mathcal{E}_\ell^X[\sigma_\ell] = \frac{1}{2}(\sigma_\ell + X_\ell \sigma_\ell X_\ell)$ is the complete dephasing channel in the $X$-basis for qubit $\ell$. This decomposition picks out components that are in the span of $\{I_\ell, X_{\ell}\}$ for $\ell \in \Sigma^\mathrm{c}$, and $\{Y_\ell, Z_{\ell}\}$ for $\ell \in \Sigma$.

With $\Omega^+ \subset \mathscr{P}(S)$ defined as before, the natural analog of Eq.~\eqref{eq:sigma constraint} then still continues to hold,
\begin{align}
    \Sigma \in \Omega^+\; \Longrightarrow\; \tilde{\mathcal{\Delta}}_{Q \leftarrow R}^{\,(\Sigma; X)} = 0.
    \label{eq:sigma constraint X}
\end{align}
Thus, if we define $\mathcal{\Delta}^{\{W^+\}, X}_{Q \leftarrow R}$ by analogy to Eq.~\eqref{eq:delta singleton}, then the arguments leading up to Eq.~\eqref{eq:delta Cplus rep} could be followed in an identical fashion. However, the resulting bound \eqref{eq:delta singleton bound} can be made tighter by introducing the following definition.
\begin{definition}
    \label{def:noise strength def xz}
	For a single-site noise channel $\mathcal{N}_{\ell}$, the effective $X$-type noise strength $T_X(\mathcal{N}_{\ell})$ is given by
	\begin{align} \label{eq:noise strength def}
		\norm{\widehat{\mathcal{N}}_{E_\ell \leftarrow \ell}\circ (\mathrm{id}_\ell - \mathcal{E}_\ell^X)}_\diamond =: e^{-1/T_X(\mathcal{N}_\ell)},
	\end{align}
	where $\mathcal{M}^X_{\ell}$ is the complete dephasing channel in the $X$ basis, and $\widehat{\mathcal{N}}_{E_\ell \leftarrow \ell}$ is any channel complementary to $\mathcal{N}_{\ell}$. The $Z$-type noise strength $T_Z(\mathcal{N}_{\ell})$ is defined similarly.
\end{definition}
These quantities provide a sharper measure of the strength of noise applied, since\footnote{Equation \eqref{eq:noise strength ineq} follows from the fact that the map $\mathcal{S}_\ell \coloneqq (\mathrm{id}_\ell - \mathcal{E}_\ell^X)$ is a complete trace norm contraction, i.e.~$\|\mathcal{S}_\ell \otimes \mathrm{id}_{\ell'}[O_{\ell \ell'}]\|_1 \leq \|O_{\ell \ell'}\|_1$ for any ancillary Hilbert space $\ell'$ and operator $O_{\ell \ell'}$, and thus satisfies $\|\mathcal{T}_\ell \circ \mathcal{S}_\ell\|_\diamond \leq \|\mathcal{T}_\ell\|_\diamond$ for any superoperator $\mathcal{T}_\ell$. Then, since the image of $\mathcal{S}_\ell$ is traceless, we have $T_X(\mathcal{N}_\ell) = \|\widehat{\mathcal{N}}_{E_\ell \leftarrow \ell}\circ \mathcal{S}_\ell\|_\diamond = \|\widehat{\mathcal{N}}_{E_\ell \leftarrow \ell}\circ(\mathrm{id}_\ell - \mathcal{D}_\ell)\circ \mathcal{S}_\ell\|_\diamond \leq \|\widehat{\mathcal{N}}_{E_\ell \leftarrow \ell}\circ(\mathrm{id}_\ell - \mathcal{D}_\ell)\|_\diamond = T(\mathcal{N}_\ell)$.}
\begin{align}
    T_{X, Z}(\mathcal{N}_\ell) \leq T(\mathcal{N}_\ell)
    \label{eq:noise strength ineq}
\end{align}
Thus, if we write $T_X \coloneqq \max_\ell T_X(\mathcal{N}_\ell)$, then we have
\begin{align}
    \norm{\widehat{\mathcal{N}}_{E_S \leftarrow Q} \circ \tilde{\mathcal{\Delta}}_{Q \leftarrow R}^{\{W^+\}, X}}_\diamond \leq 2 e^{-|W^+|/T_X},
\end{align}
which improves on Eq.~\eqref{eq:delta singleton bound}. More generally, we expect the quantities $T_{X,Z}(\mathcal{N}_\ell)$  to be useful quantifiers of noise strength for any CSS code.

As well as allowing us to reduce the effective noise rate from $T$ to $T_X$, the  the modified decomposition \eqref{eq:delta decomp X} also has the advantage that it obeys an extra selection rule, in addition to \eqref{eq:sigma constraint X}. Let $V_S$ be the set of vertices $v$ for which every adjacent edge $\ell \in \partial^\top v$ is contained in $S$. Then,
\begin{align}
    \exists v \in V_S : \mathrm{deg}_\Sigma(v) \text { is odd}\; \Longrightarrow\; \tilde{\mathcal{\Delta}}_{Q \leftarrow R}^{\,(\Sigma; X)} = 0.
    \label{eq:sigma constraint X even}
\end{align}
That is, for $\tilde{\mathcal{\Delta}}_{Q \leftarrow R}^{\,(\Sigma; X)}$ to be non-zero, $\Sigma$ must contain an even number of edges surrounding every vertex $v \in V_S$. This greatly constrains the regions we need to include in our combinatorial decomposition.  

As we show in Appendix \ref{apdx:z2 improvements}, the extra constraint \eqref{eq:sigma constraint X even} can be used to decompose $\Sigma$ into non-overlapping components  $C_1, C_2, \ldots$, each of which correspond to the edges visited by a \textit{self-osculating walk} (SOW). An SOW is a path on the primal lattice $\gamma = (v_1, v_2, \ldots)$ with certain rules that restrict how the path can cross itself. In particular, no edge is visited more than once, and each vertex is visited at most twice. For vertices that are visited twice, the path cannot cross through itself, but must `kiss' (or `osculate'), see \cref{fig:self-osculating-walk}.

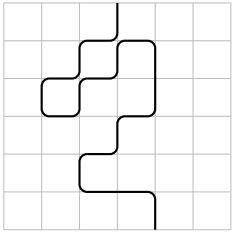
\begin{figure}
    \centering
	\begin{tikzpicture}[scale=0.5, baseline=-0.5ex]
        \foreach \y in {0,...,6} {
		  \draw [color=gray!50,] (0, \y) -- (6, \y);
        }
        \foreach \x in {0,...,6} {
		  \draw [color=gray!50,] (\x, 0) -- (\x, 6);
        }
        \draw [rounded corners=1mm, thick] (3,6)--(3,5)--(2,5)--(2,4)--(1,4)--(1,3)--(2,3)--(2,4)--(3,4)--(3,5)--(4,5)--(4,4)--(4,3)--(3,3)--(3,2)--(2,2)--(2,1)--(3,1)--(4,1)--(4,0);
	\end{tikzpicture}
    \caption{An example of a vertically spanning self-osculating walk (SOW) on the square lattice. Unlike self-avoiding walks, the same vertex may be visited twice if the path can be drawn such that it `kisses' at a vertex \cite{pkim2025existence}.}
    \label{fig:self-osculating-walk}
\end{figure}

Using this approach, we prove the following improvement to \cref{lem:delta bound simple}.
\begin{lemma}
	\label{lem:SOW mapping surface}
	With the assumptions of Theorem \ref{thm:fat}, let $S \subset Q$ be a horizontal strip of qubits of height $w_{\rm y}$, with geometry shown in \cref{fig:geometries}(b). Then,
	\begin{align}
		\delta_{\mathfrak{X}}(E_S) \leq \frac{2}{\ln 2} \times \mathscr{Z}^\mathrm{SOW}_{\Lambda(w_\mathrm{x}, L_\mathrm{y})}(1/T_X),
		\label{eq:surface Z SOWs}
	\end{align}
	where $T_X \coloneqq \max_{\ell \in S} T_X(\mathcal{N}_\ell)$, and
    \begin{align}
	\mathscr{Z}_{\Lambda(w_\mathrm{x}, L_\mathrm{y})}^{\mathrm{SOW}}(\beta) \coloneqq \sum_{\gamma \in \mathrm{SOW}(\Lambda(w_\mathrm{x}, L_\mathrm{y}))} e^{-\beta |\gamma|}
	\label{eq:partition SOW}
    \end{align}
    is the partition function for a fluctuating SOW with inverse temperature $\beta$, defined on a rectangular lattice $\Lambda(L_x , w_\mathrm{y})$, and constrained to start on the top boundary of $S$ and terminate on the bottom. 
\end{lemma}

Using known properties of the SOW partition function, we can convert \cref{lem:SOW mapping surface} into an improved bound on the critical noise strength $T_{\rm c}$ appearing in Theorem \ref{thm:fat}. The number of SOWs of length $r$ with a fixed start point grows as $n_r \leq (\mu_{\Lambda})^{r}$, where $\mu_\Lambda$ is the \textit{connective constant} for SOWs on the lattice $\Lambda$, defined as the smallest constant for which this inequality holds. If $\beta$ exceeds the critical inverse temperature $\ln \mu_{\Lambda}$, then the sum in the partition function \eqref{eq:partition SOW} converges, and is dominated by the shortest-length paths. This results in an upper bound (see Appendix \ref{apdx:z2 improvements})
\begin{align}
    \mathscr{Z}^{\text{SOW}}_{\Lambda(L_x, w_{\rm y})}(\beta) &\leq L_{\rm y} \frac{e^{-(\beta - \ln \mu_\Lambda)w_{\rm y}}}{1 - e^{-(\beta - \ln \mu_\Lambda)}} & \text{for }\beta > \ln \mu_{\Lambda}
    \label{eq:SOW bound}
\end{align}
which is tight up to polynomial prefactors. This implies the following result.

\begin{theorem}
    In Theorem \ref{thm:fat}, the condition \eqref{eq:temp bound} can be replaced by
    \begin{align}
        T_X \coloneqq \max_{\ell \in S} T_X(\mathcal{N}_\ell) < 1/\ln \mu_\Lambda,
    \end{align}
    where $T_X(\mathcal{N}_\ell)$ is given in \cref{def:noise strength def xz}, and $\mu_\Lambda$ is the connective constant for SOWs on the lattice $\Lambda$. For vertical strips, with $\mathfrak{X}$ replaced by $\mathfrak{Z}$, one uses $T_Z$ in place of $T_X$, and $\mu_{\Lambda^*}$ in place of $\mu_\Lambda$. 
\end{theorem}

For the square lattice, the best known rigorous upper bounds for the SOW connective constant is $\mu_\square^\mathrm{SOW} \leq 2.73911$~\cite{pkim2025existence}. Though stated for the square lattice, it is clear that the theorem can be generalised to other lattices or even amorphous planar or cylindrical graphs, as well as regions $S$ that are not strictly rectangular. In this case, $w_\mathrm{x}$ would correspond to the length of the shortest lattice path from one boundary of $S$ to the other. The best known upper bound for the triangular lattice is $\mu_\trilattice^\mathrm{SOW} \leq 4.44931$ and for the hexagonal lattice, the connective constant is the same as that of self-avoiding walks, known exactly as $\mu_{\hexlattice}^\mathrm{SOW} = \mu_{\hexlattice}^\mathrm{SAW} = \sqrt{2+ \sqrt{2}} \approx 1.85$ \cite{duminil2012connective}. These can be used to obtain bounds on the thresholds for various noise channels, which are included in \cref{tab:beta-ps-and-thresholds}. For amorphous planar and cylindrical graphs of maximum degree $k$ on the primal graph $\mathscr{G}$ and $k_*$ on the dual graph $\mathscr{G}_*$, the connective constants can be loosely upper bounded by $\mu_\mathscr{G}^\mathrm{SOW} \leq (k-1)$ and $\mu_{\mathscr{G}_*}^\mathrm{SOW} \leq  (k_*-1)$ respectively, since SOWs are subsets of non-reversing walks.

\section{Non-Abelian codes} \label{sec:quantum-doubles}

In this Section, we introduce families of non-Abelian quantum error correcting codes, and provide bounds on the recovery thresholds after noise is applied. Specifically, we generalise Theorem \ref{thm:fat} and Corollaries \ref{cor:lre}, \ref{cor:anyon-distinguishability}, and \ref{cor:recoverability} to these more general non-Abelian codes. The two code families we consider are Kitaev's quantum double codes \cite{kitaev2003fault}, and the string-net codes of Levin and Wen \cite{levin2005string}.


\subsection{The quantum double model}
Here we describe the model introduced by Kitaev in Ref.~\onlinecite{kitaev2003fault}, which is parametrised by a finite group $G$. Up to a difference in boundary conditions, the case $G = \mathbb{Z}_2$ will correspond to the familiar surface code described above. While the geometry can be chosen freely, to keep matters as simple as possible we will work with a $L_x \times L_y$ square lattice with periodic boundary conditions, with each edge assigned an orientation as below,
\begin{equation}
	\begin{matrix}
		& \vdots & \\
		\cdots & \intikz{
			\foreach \y in {0,...,3} {
							\foreach \x in {0,...,3} {
								\draw[-<-] (\x,\y) -- ({\x+1},\y);
								\draw[->-] (\x,\y) -- (\x,{\y+1});
							}
						}
		} & \cdots \\
		& \vdots &
	\end{matrix}\;\;.
    \label{eq:edge orientation}
\end{equation}
We denote the sets of vertices, edges, and plaquettes as denoted $V, E$, and $P$, respectively. 

Each edge $\ell \in E$ hosts a $|G|$-dimensional qudit, with local Hilbert space $\mathscr{H}_\ell = \mathrm{span}\{\ket{g} : g \in G\}$, where $\{\ket{g} : g \in G\}$ is a chosen computational basis, labelled by elements $g \in G$. On each qudit, we define the left and right shift operators, analogous to Pauli-$X$ operators, by
\begin{align}
	L^{g}_\ell \ket{h}_\ell = \ket{gh}_\ell, \quad R^g_\ell \ket{h} = \ket{h g^{-1}}_\ell.
	\label{eq: operators LR}
\end{align}
We also define projectors onto the basis states
\begin{align}
	T^h_\ell\ket{g}_\ell = \delta_{g,h} \ket{g}_\ell,
\end{align}
which play a similar role to the Pauli-$Z$ operators for the $\mathbb{Z}_2$ case.

The code space projector again has the form $\mathbb{\Pi} = \mathbb{A}_V \mathbb{B}_P$, where $\mathbb{A}_V = \prod_{v \in V} = \mathbb{A}_v$ and $\mathbb{B}_P = \prod_{p \in P} = \mathbb{B}_p$, and $\mathbb{A}_v$ and $\mathbb{B}_p$ are projectors that act only on the edges within a vertex $v$ and plaquette $p$, respectively. Here the vertex projector is given by
\begin{align}
	\mathbb{A}_v = \frac{1}{|G|} \sum_{g \in G} A^g_v,
\end{align}
where the vertex operators are
\begin{align}
	A^g_v = \prod_{\ell \in \partial^\top v} O^{g}_{v, \ell}.
\end{align}
where $O^g_{v, \ell} = L^g_\ell$ if the edge $\ell$ is pointed towards the vertex, and $O^g_{v, \ell} = R^g_\ell$ if it is pointed away. For example, we have
\begin{align}
	A^g_v = R^g_u L^g_r L^g_d R^g_l \; \text{ for edges labelled } \; \intikz{
			\draw[-<-, thick] (0, 1) -- node[right] {$u$} (0, 0);
			\draw[-<-, thick] (0, 0) -- node[below] {$r$} (1, 0);
			\draw[-<-, thick] (0, 0) -- node[left] {$d$} (0, -1);
			\draw[-<-, thick] (-1, 0) -- node[above] {$l$} (0, 0);
			\node[fill=white] (0, 0) {$v$};
	}.
\end{align}
Letting $\ell_1, \ldots, \ell_4 \in \partial p$ denote the four edges adjacent to a plaquette $p$, given in counter-clockwise order, the plaquette projector is given by $\mathbb{B}_p \coloneqq B^\mathrm{id}_p$, where $\mathrm{id}$ is the identity element of $G$ and the operators $\{B^g_p\}_{g \in G}$ act on computational basis states as
\begin{align}
	B^g_p \big(&\ket{g_1}_{\ell_1} \otimes \cdots \otimes \ket{g_4}_{\ell_4} \big) \nonumber\\ &= \delta\left\{{\prod_{i=1}^4} g_{\ell_i}^{o_p(\ell_i)}, g \right\} \ket{g_1}_{\ell_1} \otimes \cdots \otimes \ket{g_4}_{\ell_4}
\end{align}
where the product is ordered right to left, and $o_p(\ell) = +1$ ($-1$) if the orientation of $\ell$ is counter-clockwise (clockwise) around $p$. For example, we have
\begin{equation}
	\begin{aligned}
		B^g_p \ket{g_u g_r g_d g_l} & = \delta\qty{g_u g_l^{-1} g_d^{-1} g_r, g} \ket{g_u g_r g_d g_l} \\ 
		& \; \text{ for edges labelled } \; 
		\intikz{
				\draw[->-, thick] (-1/2, 1/2) -- node[above] {$u$} (1/2, 1/2);
				\draw[->-, thick] (1/2, 1/2) -- node[right] {$r$} (1/2, -1/2);
				\draw[->-, thick] (-1/2, -1/2) -- node[below] {$d$} (1/2, -1/2);
				\draw[->-, thick] (-1/2, 1/2) -- node[left] {$l$} (-1/2, -1/2);
				\draw (0, 0) node {$p$};
		}.
	\end{aligned}
\end{equation}

Analogously to the $\mathbb{Z}_2$ surface code, string-like logical operators can be defined that are supported on particular thickened paths, known as \textit{ribbons} \cite{kitaev2003fault,bombin2008}. Each ribbon $r$ can be associated with a collection of \textit{ribbon operators} $F^{h,g}_r$ labeled by group elements $h,g \in G$. These are described in full generality in Appendix \ref{apdx:q double codes}, but for the specific case of a straight-line path in the $x$-direction $\gamma_{\leftrightarrow} = (v_1, v_2, \ldots, v_{l})$, the corresponding ribbon operators $F^{h,g}_{\gamma_{\leftrightarrow}}$ act as
\tikzset{vdot/.style={circle, fill=black, minimum size=0.5,inner sep=0,outer sep=0}}
\begin{equation}
\label{eq:ribbon cartesian}
	\begin{aligned}
		& F^{h, g}_{\gamma_\leftrightarrow}
		\bket{\intikz[scale=1.75]{
			\draw[-<-] (0, 0) -- node[above] {$x_1$} node[pos=0, vdot] {} (1, 0);
			\draw[->-] (0, -1) -- node[right] {$\, y_1$} (0, 0);
			\draw[-<-] (1, 0) -- node[above] {$x_2$} (2, 0);
			\draw[->-] (1, -1) -- node[right] {$\, y_2$} (1, 0);
			\draw[-<-] (2, 0) -- node[above] {$x_3$} (3, 0);
			\draw[->-] (2, -1) -- node[right] {$\, y_3$} (2, 0);
		} \cdots} \\
		& = \delta_{g, (x_1 x_2 \cdots)} \bket{\intikz[scale=1.75]{
			\draw[-<-] (0, 0) -- node[above] {$x_1$} (1, 0);
			\draw[->-] (0, -1) -- node[right] {$\, h y_1\,$} (0, 0);
			\draw[-<-] (1, 0) -- node[above] {$x_2$} (2, 0);
			\draw[->-] (1, -1) -- node[right] {$\, \mathcal{x}_{1:2}[h] y_2\,$} (1, 0);
			\draw[-<-] (2, 0) -- node[above] {$x_3$} (3, 0);
			\draw[->-] (2, -1) -- node[right] {$\, \mathcal{x}_{1:3}[h] y_3\,$} (2, 0);
		} \cdots},
	\end{aligned}
\end{equation}
where $\mathcal{x}_{1:n}[h] = (x_1 x_2 \cdots x_n)^{-1} h (x_1 x_2 \cdots x_n)$ is the conjugation of $h$ by the product of the $x_i$'s. 

The ribbon operators $F_r^{h,g}$ can be shown to commute with all vertex and plaquette projectors $\mathbb{A}_v$, $\mathbb{B}_p$, except possibly those adjacent to their endpoints $v_1$, $v_l$. They also obey a particular multiplication rule \cite{kitaev2003fault}
\begin{align}
	F^{h_1,g_1}_{r} F^{h_2,g_2}_{r} = \delta_{h_1, g_1 h_2 g_1^{-1}} F^{h_1, g_1 g_2}_{r}
	\label{eq:ribbon product}
\end{align}
As a result, the space of operators $\mathfrak{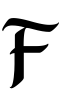}_{r} = \mathrm{span}\{F^{h,g}_{r}\}_{h,g \in G}$ constitutes an algebra. More specifically, \eqref{eq:ribbon product} is precisely the defining property of the algebra $D(G)$, known as the quantum double (or sometimes Drinfeld double) of $G$. Thus, $\mathfrak{F}_{r}$ forms a representation of $D(G)$.

Just as with groups, representations of $D(G)$ can be decomposed in terms of irreducible components (irreps), denoted $\mu \in \mathrm{Irr}(D(G))$. Each irrep $\mu$ can be associated with a unique type of anyon in the code, and \textit{vice-versa}. In particular, as detailed in the appendix, for each $\mu$ and for any open ribbon $r$ with endpoints $\partial_{0,1} r = v, v'$, one can identify a particular linear combination $F^\mu_r = \sum_{h,g} \alpha_\mu^{h,g} F^{h,g}_r$ such that for any code state $\rho_Q$, the state
\begin{align}
    \rho_Q^{(v,v'), \mu} \propto F^\mu_r \rho_Q (F^\mu_r)^\dagger
    \label{eq:anyon state}
\end{align}
features a pair of anyons $(\mu, \bar{\mu})$ at the locations of $(v, v')$. (Here $\bar{\mu}$ denotes the irrep/anyon dual to $\mu$.) In particular, for the trivial anyon $\mu = \varnothing$, we have $\rho_Q^{(v,v'), \varnothing} = \rho_Q$. The operators $F^\mu_r$ generalise the $X$-type and $Z$-type open string operators of the $\mathbb{Z}_2$ surface code, which generate pairs of $e$- and $m$-type anyons. These are defined more thoroughly in Appendix \ref{apdx:q double codes}.

One can similarly construct closed string operators that can be used to implement logical operations. For instance, if $\gamma_{\leftrightarrow} = (v_1, v_2, \ldots, v_l, v_1)$ is a closed non-contractible loop winding around the system in the $x$ direction, then for each anyon type $\mu$ we can find coefficients  $\alpha^\leftrightarrow_\mu(h,g)$ such that the operators
\begin{align}
	K^\mu_{\gamma_{\leftrightarrow}} = \sum_{h,g \in G} \alpha^\leftrightarrow_\mu(h,g) F^{h,g}_{\gamma_{\leftrightarrow}}
\end{align}
are  mutually orthogonal projectors that commute with the full code projector $\mathbb{\Pi}$. Physically, these operators create a $(\mu, \bar{\mu})$ anyon pair, one of which is transported around the torus, which then fuse back to the vacuum. As it turns out, one can construct a complete basis of code states of the form $\dyad{\bar{\mu}^\leftrightarrow}_Q = K^\mu_{\gamma_\leftrightarrow} \mathbb{\Pi}$. This allows one to define an encoding isometry from a reference system $R$, with dimension $d_R = |\mathrm{Irr}(D(G))|$, to the physical space $Q$,
\begin{align}
	\mathcal{C}^\leftrightarrow_{Q \leftarrow R}[\,\cdot\,] &= V^{\mathcal{C}^\leftrightarrow}_{Q \leftarrow R} \,\cdot\, (V^{\mathcal{C}^\leftrightarrow}_{Q \leftarrow R})^\dagger, \nonumber\\  \text{where } V_{Q \leftarrow R}^{\mathcal{C}_\leftrightarrow} &\coloneqq \sum_{\mu \in \mathrm{Irr}(D(G))} \ket{\bar{\mu}^\leftrightarrow}_Q \bra{\mu^\leftrightarrow}_R,
\end{align}
with $\{\ket{\mu^\leftrightarrow}_R\}_{\mu \in \mathrm{Irr}(D(G))}$ an arbitrary basis for $\mathscr{H}_R$. 
For example, if $G = \mathbb{Z}_2$, the code basis states $\ket{\bar{\mu}^\leftrightarrow}_Q$ would be simultaneous eigenstates of $\bar{Z}_{\gamma_\leftrightarrow}$ and $\bar{X}_{\gamma_\leftrightarrow^*}$, namely the product of Pauli-$Z$s and Pauli $X$s along loops in the horizontal direction, sometimes referred to as minimally entangled states \cite{zhang2012}.

Naturally, ribbon operators for other paths $\gamma'$ that are homologically equivalent to $\gamma_\leftrightarrow$ can also be defined, and their action on the code space is identical. However, if we take a ribbon to be a path along the $y$ direction, $\gamma_\updownarrow$, then this defines a non-equivalent set of logical projectors, $K^\mu_{\gamma_\updownarrow}$, of which $\ket{\bar{\mu}^\leftrightarrow}$ are not eigenstates. However, all the same considerations as before apply---these must be orthogonal projectors $K^\mu_{\gamma_\updownarrow}$ that support a new complete set of code states $\ket{\bar{\mu}^\updownarrow}_Q$. The two bases will be related by some unitary matrix $S$,
\begin{align}
	\label{eq:S matrix}
	\ket*{\bar{\mu}^\updownarrow}_{Q} &= \sum_{\nu}S_{\mu \nu}\ket{\bar{\nu}^\leftrightarrow}_{Q}. 
\end{align}
While $S$ can in principle be derived microscopically from the definition of the ribbon operators, it turns out to be a special topological object called the modular $S$-matrix \cite{simon2023topological}. Its structure is closely related to the braiding and fusion rules for the anyons $\mu$, $\nu$. 

This allows us to define two important subalgebras of $\mathfrak{B}(\mathscr{H}_R)$,
\begin{subequations}
\label{eq:subalgebras K}
	\begin{align}
		\mathfrak{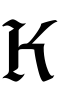}_\leftrightarrow &\coloneqq \mathrm{span}\Big\{\dyad{\mu^\leftrightarrow}_R : \mu \in \mathrm{Irr}(D(G))\Big\},\\
        \mathfrak{K}_\updownarrow &\coloneqq \mathrm{span}\qty{\ketbra*{\mu^\updownarrow}_R : \mu \in \mathrm{Irr}(D(G))},
	\end{align}
\end{subequations}
where we understand $S$ to act the reference states in the same way as \eqref{eq:S matrix}. These algebras recoverable from ribbons of qubits that form loops in the horizontal $\leftrightarrow$ and vertical $\updownarrow$ directions, respectively. In particular, the algebras obey two important properties (which we prove in Appendix \ref{apdx:full recoverability proof}),
\begin{align}
	\mathfrak{K}_\leftrightarrow \cap \mathfrak{K}_\updownarrow &= \mathbb{C} I_R & \big\langle\mathfrak{K}_\leftrightarrow, \mathfrak{K}_\updownarrow \big\rangle &= \mathfrak{B}(\mathscr{H}_R).
	\label{eq:algebra conds}
\end{align}
That is, the intersection of these algebras is the trivial algebra consisting of scalar multiples of the identity, and together they generate the full algebra of operators over $\mathscr{H}_R$.

\subsection{String-net models}
The string-net model of Levin and Wen \cite{levin2005string} is a family of lattice models that allows one to realise a broader range of topological orders. For convenience, these are usually defined on the hexagonal lattice, and again we will take a finite system  with periodic boundary conditions, and dimensions $L_\mathrm{x} \times L_\mathrm{y}$ (measuring the length of the shortest non-contractible path horizontally and vertically). The vertices, links, and plaquettes of this lattice will be written $(V, E, P)$.

Whereas the quantum double model is parametrised by a finite group $G$, here the string net model is parametrised by a unitary fusion category $\mathscr{C}$. The category $\mathscr{C}$ contains simple objects $s \in \mathscr{L}$, of which there are $q = |\mathscr{L}|$, along with fusion coefficients $N^{c}_{ab} \in \mathbb{Z}_{\geq 0}$, labeled by $a,b,c \in \mathscr{L}$. Each edge of the hexagonal lattice, which again is given a fixed but arbitrary orientation, hosts a qudit of dimension $N$, with Hilbert space $\mathscr{H}_\ell = \mathrm{span} \{\ket{s} : s \in \mathscr{L} \}$, where $\{\ket{s}\}_{s \in \mathscr{L}}$ are a set of orthonormal basis states. In the most general case, qudits must also be placed at the vertices of the hexagonal lattice, but this is not necessary if $\mathscr{C}$ is multiplicity-free, namely if $N^{c}_{ab} \in \{0,1\}$. We will describe this simpler case in the main text, though all our arguments apply for categories with fusion multiplicities, $N^{c}_{ab} > 1$.

The code-space projector is again given by a similar structure as the quantum double models, $\mathbb{\Pi} = \mathbb{A}_V \mathbb{B}_P$. The vertex projector, diagonal in the computational basis, is given by
\begin{align}
    \mathbb{A}_v \ket{a, b, c} = \delta\{N^c_{ab}, 1\} \ket{a, b, c},
\end{align}
for simple objects $a, b, c \in \mathscr{L}$. 
The plaquette projectors $\mathbb{B}_p$ are more cumbersome to write down, but importantly we can characterise their support as follows. For a collection of edges $\gamma \subseteq E$, let $\bigstar(\gamma)$ be the set of edges $e$ that are contained in $\gamma$ or adjacent to some $e' \in \gamma$. Then, $\mathbb{B}_p$ is supported on $\bigstar(\partial p)$, where $\partial p$ is the set of edges around the plaquette $p$. The particular matrix elements of $\mathbb{B}_p$ are functions of the so-called $F$-symbols, which are part of the data used to specify $\mathscr{C}$, see Ref.~\onlinecite{levin2005string} for specifics.

We will not need to know the exact form of the stabilizers, but rather the algebraic and spatial properties of the string operators, which generalise the ribbon operators described above. Given a path on the primal lattice, $\gamma = (v_1, v_2, \ldots)$, a `simple' string operator $W_\gamma^s$ can be defined, which is supported on the region $\bigstar(\gamma)$, namely the set of all edges that are either contained in $\gamma$ or adjacent to some edge in $\gamma$,
\begin{equation}
	\begin{aligned}
		& \bigstar \qty(\intikz[scale=0.4]{\newcommand{\zcolor}{tolblue}
\newcommand{\xcolor}{tolred}
\newcommand{\definep}[2]{\coordinate (p) at ($(60:#2)+(120:#2)+(0:3*#1)$);}
\newcommand{\definepn}[3]{\coordinate (p#3) at ($(60:#2)+(120:#2)+(0:3*#1)$);}
\newcommand{\defineq}[2]{\coordinate (q) at ($(60:#2)+(120:#2)+(0:3*#1)+(0:1)+(60:1)$);}
\tikzmath{\X = 2; \Y = 3;}
\foreach \x in {0,...,\X}{
    \foreach \y in {0,...,\Y}{
        \definep{\x}{\y}
        \draw[] ($(p)+(0:1)$)--($(p)+(60:1)$)--($(p)+(120:1)$)--($(p)+(180:1)$);
        \defineq{\x}{\y}
        \draw[] ($(q)+(0:1)$) -- ($(q)+(60:1)$) -- ($(q)+(120:1)$) -- ($(q)+(180:1)$);
    }
}
\foreach \y in {0,...,\Y}{ 
    \definep{0}{\y}
    \draw[] ($(p)+(180:1)$) -- ($(p)+(-120:1)$);
    \defineq{\X}{\y}
    \draw[] ($(q)+(0:1)$) -- ($(q)+(-60:1)$);
}
\foreach \x in {0,...,\X}{
    \definep{\x}{0}
    \draw ($(p)+(180:1)$) -- ($(p)+(-120:1)$) -- ($(p)+(-60:1)$) -- ($(p)+(0:1)$) -- ($(p)+(0:2)$);
    }
\definep{2}{2}
\draw[tolblue, line width=1.2mm, line cap=round] ($(p)+(-120:1)$) -- ($(p)+(-60:1)$);
\definep{2}{2}
\draw[tolblue, line width=1.2mm, line cap=round] ($(p)+(-120:1)$) -- ($(p)+(-180:1)$);
\definep{1.5}{2.5}
\draw[tolblue, line width=1.2mm, line cap=round] ($(p)+(-120:1)$) -- ($(p)+(-60:1)$);
\definep{1}{2}
\node[tolblue] at ($(p)$) {\Large $\gamma$};
\draw[tolblue, line width=1.2mm, line cap=round] ($(p)+(0:1)$) -- ($(p)+(-60:1)$);
\draw[tolblue, line width=1.2mm, line cap=round] ($(p)+(-120:1)$) -- ($(p)+(-60:1)$);
\definep{0.5}{1.5}
\draw[tolblue, line width=1.2mm, line cap=round] ($(p)+(0:1)$) -- ($(p)+(-60:1)$);
\draw[tolblue, line width=1.2mm, line cap=round] ($(p)+(-60:1)$) -- ($(p)+(-120:1)$);
}) \\ & \qquad \qquad \qquad = \; \intikz[scale=0.4]{\newcommand{\zcolor}{tolblue}
\newcommand{\xcolor}{tolred}
\newcommand{\definep}[2]{\coordinate (p) at ($(60:#2)+(120:#2)+(0:3*#1)$);}
\newcommand{\definepn}[3]{\coordinate (p#3) at ($(60:#2)+(120:#2)+(0:3*#1)$);}
\newcommand{\defineq}[2]{\coordinate (q) at ($(60:#2)+(120:#2)+(0:3*#1)+(0:1)+(60:1)$);}
\tikzmath{\X = 2; \Y = 3;}
\foreach \x in {0,...,\X}{
    \foreach \y in {0,...,\Y}{
        \definep{\x}{\y}
        \draw[] ($(p)+(0:1)$)--($(p)+(60:1)$)--($(p)+(120:1)$)--($(p)+(180:1)$);
        \defineq{\x}{\y}
        \draw[] ($(q)+(0:1)$) -- ($(q)+(60:1)$) -- ($(q)+(120:1)$) -- ($(q)+(180:1)$);
    }
}
\foreach \y in {0,...,\Y}{ 
    \definep{0}{\y}
    \draw[] ($(p)+(180:1)$) -- ($(p)+(-120:1)$);
    \defineq{\X}{\y}
    \draw[] ($(q)+(0:1)$) -- ($(q)+(-60:1)$);
}
\foreach \x in {0,...,\X}{
    \definep{\x}{0}
    \draw ($(p)+(180:1)$) -- ($(p)+(-120:1)$) -- ($(p)+(-60:1)$) -- ($(p)+(0:1)$) -- ($(p)+(0:2)$);
    }
\definep{2}{2}
\draw[tolblue, line width=1.2mm, line cap=round] ($(p)+(-120:1)$) -- ($(p)+(-60:1)$);
\draw[tolblue, line width=1.2mm, line cap=round] ($(p)+(-120:1)$) -- ($(p)+(-180:1)$);
\definep{1.5}{2.5}
\draw[tolblue, line width=1.2mm, line cap=round] ($(p)+(-120:1)$) -- ($(p)+(-60:1)$);
\definep{1}{2}
\draw[tolblue, line width=1.2mm, line cap=round] ($(p)+(0:1)$) -- ($(p)+(-60:1)$);
\draw[tolblue, line width=1.2mm, line cap=round] ($(p)+(-120:1)$) -- ($(p)+(-60:1)$);
\definep{0.5}{1.5}
\draw[tolblue, line width=1.2mm, line cap=round] ($(p)+(0:1)$) -- ($(p)+(-60:1)$);
\draw[tolblue, line width=1.2mm, line cap=round] ($(p)+(-60:1)$) -- ($(p)+(-120:1)$);
\definep{2}{2}
\draw[tolblue, line width=1.2mm, line cap=round] ($(p)+(180:1)$) -- ($(p)+(120:1)$);
\draw[tolblue, line width=1.2mm, line cap=round] ($(p)+(0:1)$) -- ($(p)+(-60:1)$);
\definep{2}{1}
\draw[tolblue, line width=1.2mm, line cap=round] ($(p)+(180:1)$) -- ($(p)+(120:1)$);
\draw[tolblue, line width=1.2mm, line cap=round] ($(p)+(0:1)$) -- ($(p)+(60:1)$);
\definep{1.5}{1.5}
\draw[tolblue, line width=1.2mm, line cap=round] ($(p)+(180:1)$) -- ($(p)+(-120:1)$);
\definep{1}{1}
\draw[tolblue, line width=1.2mm, line cap=round] ($(p)+(180:1)$) -- ($(p)+(-120:1)$);
\draw[tolblue, line width=1.2mm, line cap=round] ($(p)+(180:1)$) -- ($(p)+(120:1)$);
\definep{0}{1}
\draw[tolblue, line width=1.2mm, line cap=round] ($(p)+(0:1)$) -- ($(p)+(60:1)$);
\draw[tolblue, line width=1.2mm, line cap=round] ($(p)+(0:1)$) -- ($(p)+(-60:1)$);
\definep{1}{2}
\draw[tolblue, line width=1.2mm, line cap=round] ($(p)+(0:1)$) -- ($(p)+(60:1)$);
\draw[tolblue, line width=1.2mm, line cap=round] ($(p)+(180:1)$) -- ($(p)+(-120:1)$);

} \; \;.
	\end{aligned}
    \label{eq:string net ribbon}
\end{equation}
Products and linear combinations of these simple objects can be used to generate a full algebra of string operators, which turns out to form a representation of the so-called Drinfeld center of $\mathscr{C}$, denoted $Z(\mathscr{C})$. The irreducible components of this representation $\mu \in \mathrm{Irr}(Z(\mathscr{C}))$ label the different anyon types.

In particular, if $\gamma_{\leftrightarrow}$ is a closed loop that wraps around the torus in the horizontal direction, then we can construct a basis of code states $\ket{\bar{\mu}^{\leftrightarrow}}_Q \in \mathscr{H}_{\mathrm{code}}$ that are each eigenstates of the corresponding string operators $W_\gamma^{\mu}$, which are supported on the region $\bigstar(\gamma)$. Namely, we have $W_\gamma^{\mu} \ket{\bar{\nu}^{\leftrightarrow}}_Q = \delta_{\mu, \nu} \ket{\bar{\nu}^{\leftrightarrow}}_Q$. As with the quantum double model, we can also do the same with a loop $\gamma$ that winds around the torus in the vertical direction to obtain a distinct basis of code states $\ket{\bar{\mu}^{\updownarrow}}_Q$, which is again related to the horizontal basis via the modular $S$-matrix, Eq.~\eqref{eq:S matrix}. The associated maximally commutative algebras $\mathfrak{K}_{\leftrightarrow}$ and $\mathfrak{K}_{\updownarrow}$, defined by analogy to Eqs.~\eqref{eq:subalgebras K}, continue to satisfy Eq.~\eqref{eq:algebra conds}.

For categories with fusion multiplicities, logical string operators $W_\gamma^\mu$ can be defined similarly. In this case we should understand the region $\bigstar(\gamma)$ to include the same set of edges, along with all the qudits at the vertices visited by the path $\gamma$. Again, we will not need to know the explicit form of these operators, only their support, and their action on the logical space.  For the purposes of the following, the subregions $S \subset Q$ will be understood to contain all vertex qudits for which every edge adjacent to the vertex is also contained in $S$.

\subsection{Results for non-Abelian codes \label{subsec:nonab results}}

Having characterised the spatial and algebraic structure of logical operators in quantum double and string-net codes, we now have the necessary prerequisites to generalise our arguments for the $\mathbb{Z}_2$ surface code to obtain the following result.

\begin{customthm}{1'}[2D non-abelian codes]
    \label{thm:fat-nonab}
    For the quantum double model $D(G)$ on the square lattice with periodic boundary conditions, let $S \subset Q$ be a horizontal strip of qudits of height $w_{\rm y}$, with rough edges, and take a product noise channel $\mathcal{N}_Q = \bigotimes_\ell \mathcal{N}_\ell$. Then, there exists a constant $T_\star \leq [1 + \ln 7]^{-1}$ such that for $T \coloneqq \max_\ell T(\mathcal{N}_\ell) < T_\star$, there exist operators $W^\mu_S$ supported on $S$ whose action on the noise-corrupted state approximate the logical projectors $\Pi^{x, \mu}_R \coloneqq \dyad{\bar{\mu}^{\leftrightarrow}}_R \in \mathfrak{K}_{\leftrightarrow}$, in that
    \begin{equation}
    \begin{aligned}
        & \norm{ W^\mu_S \mathcal{M}_{Q \leftarrow R}[\sigma_R] - \mathcal{M}_{Q \leftarrow R}[\Pi^{x, \mu}_R\sigma_R]}_1 \\ 
        & \qquad \qquad \qquad \qquad \qquad \qquad \qquad \leq c_1(T)\sqrt{L_{\rm x}} e^{-c_2(T) w_{\rm y}}
    \end{aligned}
    \end{equation}
    for every anyon type $\mu \in \mathrm{Irr}(D(G))$.
    
    The same holds for the string-net model on the hexagonal lattice with input category $\mathscr{C}$, for each $\mu \in \mathrm{Irr}(Z(\mathcal{C}))$, with $T_{\star} \leq [1 + \ln 25]^{-1}$ if no fusion multiplicities are present,  and $T_{\star} \leq [1 + \ln 39]^{-1}$ otherwise.  In this case, $w_\mathrm{y}$ is the length of the shortest lattice path from the top boundary of $S$ to the bottom.
\end{customthm}

The proof of the above result, given in Appendix \ref{apdx:q double codes}, uses a similar line of argument to that presented in Section \ref{sec:upper bound simple}. As before, we proceed by upper bounding the quantity $\delta_{\mathfrak{K}_{\leftrightarrow}}(E_S)$ via the decomposition \eqref{eq:delta decomp}. The first step is to identify a condition on the region $\Sigma \subseteq S$ of the kind given in Eq.~\eqref{eq:sigma constraint}, which captures local indistinguishability of the relevant code states---see \cref{lem:indist q double}. Specifically, by the same logic as before, we can show that $\mathcal{\Delta}^{(\Sigma)}_{Q \leftarrow R} = 0$ if $\Sigma^\mathrm{c}$ supports a family of logical operators that span $\mathfrak{K}_{\leftrightarrow}$, and we establish a simple graph-theoretic condition on $\Sigma$ that suffices for this to be true. Then, we invoke a similar combinatorial decomposition of the sum \eqref{eq:delta decomp}, where $\Sigma$ is decomposed in terms of certain connected components---see \cref{lem:bound-norm-connectivity}. Combining these results gives the desired bound on $\delta_{\mathfrak{K}_{\leftrightarrow}}(E_S)$, and in turn Theorem \ref{thm:fat-nonab}.

As for the $\mathbb{Z}_2$ surface code, the existence of thickened string operators that approximate the action of logicals can be used to infer a number of properties about the noise-robustness of topological order.

\begin{customcor}{1'}[2D non-abelian codes]
\label{cor:lre-nonab}
	Let $\rho_Q$ be any code state of a quantum double or string-net code, and $\tilde{\rho}_Q = \mathcal{N}_Q[\rho_Q]$ the corresponding noise-corrupted state, with $\mathcal{N}_{Q} = \bigotimes_{\ell \in Q}\mathcal{N}_\ell$. Then there exists a constant $c_0 > 0$ such that, for $T  < T_\star$, where $T$, $T_\star$ are given in \cref{thm:fat}, $\tilde{\rho}_Q$ is $(r, t, c_0)$-long-range entangled.
\end{customcor}

\begin{customcor}{2'}[2D non-abelian codes]
\label{cor:anyon-distinguishability-nonab}
	For the quantum double model $D(G)$, let $\rho_Q^{(v,v'), \mu}$ be the states \eqref{eq:anyon state} that feature a $(\mu, \bar{\mu})$ anyon pair at the locations of vertices $v, v'$, and  $\tilde{\rho}_Q^{(v,v'), \mu} \coloneqq \mathcal{N}_Q[\rho_Q^{(v,v'), \mu}]$ the corresponding noise-corrupted states, with $\mathcal{N}_{Q} = \bigotimes_{\ell \in Q}\mathcal{N}_\ell$. Then, if $T < T_\star$, for any annulus $S$ in the geometry of Fig.~\ref{fig:geometries}(d), there exist operators $K^\mu_S$ supported on $S$ such that $|\mathrm{tr}[K^{\mu'}_S \tilde{\rho}_Q^{(v,v'), \mu}] - \delta_{\mu, \mu'}| \leq \epsilon$, with $\epsilon = e^{-\Omega(w_S)}$.
\end{customcor}

\begin{customcor}{3'}[2D non-abelian codes]
 \label{cor:recoverability-nonab}
	For a quantum double model $D(G)$ or string net model $Z(\mathscr{C})$, let $\mathcal{C}_{Q \leftarrow R}$ be the encoding channel \eqref{eq:encoding}, on which a noise channel $\mathcal{N}_{Q} = \bigotimes_{\ell \in Q}\mathcal{N}_\ell$ is applied. If $T  < T_\star$, where $T$, $T_\star$ are given in \cref{thm:fat}, then there exists a channel $\mathcal{R}_{R \leftarrow Q}$ that recovers the full logical information to accuracy $\epsilon_{\textup{rec}} = e^{-\Omega(\mathrm{min}(L_{\rm x}, L_{\rm y}))}$, in the sense of Eq.~\eqref{eq:recov error}.
\end{customcor}

This last result implies \cref{thm:nonab example}, since one can take $\mathcal{R}_Q = \mathcal{C}_{Q \leftarrow R} \circ \mathcal{R}_{R \leftarrow Q}$, which satisfes $\|\mathcal{R}_Q \circ \mathcal{N}_Q \circ  \mathcal{C}_{Q \leftarrow R} - \mathcal{C}_{Q \leftarrow R}\|_\diamond \leq \epsilon_{\mathrm{rec}}$.

\section{Further generalisations} \label{sec:further-generalisations}

Our method is highly adaptable to an even wider range of codes and noise channels. In this section we outline some ways in which the results presented here can be generalised, the complete theories of which we leave for future work. Specifically, we consider certain correlated noise channels (\cref{sec:correlated}), and the 3D $\mathbb{Z}_2$ surface code (\cref{sec:3d}).

\subsection{Correlated noise channels} \label{sec:correlated}

While we have focused on noise that factorises across each qudit $\mathcal{N}_{Q} = \bigotimes_{\ell \in Q} \mathcal{N}_\ell$, our methods can be used to bound thresholds for quasi-local noise channels that do not necessarily factorise, e.g.~low-depth circuit channels.  The intuition given in \cref{subsec:local global indist} still holds in this case---the corresponding complementary channels will suppress the norm of high-weight Pauli strings, which combined with the local indistinguishability property of the code states should suffice to bound the quantity $\delta_{\mathfrak{J}}(E_S)$.

We will illustrate this concretely by presenting bounds for two cases: single-site separable errors after a finite depth unitary circuit of $t$ layers, and plaquette-correlated bit-flip errors.

\subsubsection{Errors following a finite-depth circuit} \label{thm:finite-depth}
Consider the setup discussed in Theorem \ref{thm:fat-nonab} for the most general non-Abelian case, with the modification that the single-site separable error channel $\mathcal{N}_{Q' \leftarrow Q}$ is applied after one applies a unitary circuit composed of $t$ layers of gates, where each gate acts locally on $O(1)$ numbers of sites:
\begin{align}
    \mathcal{N}^\mathrm{fd}_{Q' \leftarrow Q} = \mathcal{N}_{Q' \leftarrow Q} \mathcal{U}_t \cdots \mathcal{U}_1 \eqqcolon \mathcal{N}_{Q' \leftarrow Q} \mathcal{U},
\end{align}
Here $\mathcal{U}_\tau = \prod_i \mathcal{u}_{i, \tau}$ is the channel describing the $\tau$th layer of gates, with each $\mathcal{u}_{i, \tau}[\cdot] = u_{i, \tau} \, \cdot \, u^\dag_{i, \tau}$ a unitary gate. Let $\xi_{\mathcal{U}}(t) = O(t)$ be the light-cone distance of this circuit, i.e.~the maximum distance by which an operator can grow under the action of $\mathcal{U}$. For example, for circuits comprised of gates acting on two neighbouring sites, $\xi(t) \leq 2t$. Then, as we show in Appendix \ref{apdx:finite-depth}, Theorem \ref{thm:fat-nonab} and its Corollaries \ref{cor:lre-nonab}, \ref{cor:anyon-distinguishability-nonab}, and \ref{cor:recoverability-nonab} hold with the modification that the inverse noise strength should be replaced by an effective proxy
\begin{align}
	\beta_{\text{eff}}(\mathcal{N}^{\text{fd}}) = \frac{1}{A[\xi_{\mathcal{U}}(t)]} \beta(\mathcal{N}_\ell) - 2 \ln 2,
    \label{eq:finite depth rescaling}
\end{align}
where $A[\xi_{\mathcal{U}}(t)] = O(t^2)$ is the maximum number of qudits within a ball with radius $\xi_{\mathcal{U}}$.

The idea behind the proof of this result is that, because the output of the map $\mathcal{\Delta}^{\{W^+\}}_{Q \leftarrow R}$ [Eq.~\eqref{eq:delta singleton}] is traceless on many sites [by virtue of Eq.~\eqref{eq:delta Cplus rep}], when the circuit $\mathcal{U}_t \cdots \mathcal{U}_1$ is then applied, the result must still be traceless on proportionally many sites---at least a fraction $1/A(\xi_{\mathcal{U}})$ of the original number.
This is proved in Appendix \ref{apdx:finite-depth}. The rescaling \eqref{eq:finite depth rescaling} is rather severe, in that it captures the worst case over all possible unitary circuits of a given depth. Thus we do not expect the resulting bound on the threshold to be quantitatively tight. Nevertheless this proves the existence of a threshold for any circuit of constant depth $t = O(1)$.

\subsubsection{Plaquette-correlated bit-flips}
We can also consider innately correlated error channels, such as the following example:
We take the 2D $\mathbb{Z}_2$ surface code on the square lattice affected with a particular plaquette-correlated bit-flip error channel 
\begin{align} \label{eq:plaquette-correlated-bit-flips}
	\mathcal{N} = \prod_{p} \mathcal{N}_p,
\end{align}
where
\begin{align}
	\mathcal{N}_p[\rho] = (1-\uppi) \rho + \uppi X_p \rho X_p
\end{align}
and $\uppi$ is the probability to apply a correlated bit-flip around a plaquette $X_p \coloneqq \prod_{\ell \in \partial p} X_\ell$. Here, it is clear that the logical-$\bar X$ operator is preserved, since it commutes with the noise channel. For the logical-$\bar Z$ operator, with the same geometry as \cref{thm:fat}, we can still find an effective noise strength $T(\mathcal{N})$ and the same $T_\star = T^\Lambda_\mathrm{c} = 1/\ln(\mu_\square^\mathrm{SOW})$ below which the thickened $W^\mathrm{Z}_S$ operator can be proven to exist. This results in a lower bound of $p_\mathrm{th} \geq 0.02\%$. See Appendix \ref{apdx:plaquette-correlated-bit-flip} for more details.

\subsection{3D surface code} \label{sec:3d}
\subsubsection{Brief recap of the 3D surface code}
We consider the 3D $\mathbb{Z}_2$ surface code, defined on the cubic lattice ($\cube$) of dimensions $L_\mathrm{x} \times L_\mathrm{y} \times L_\mathrm{z}$ with open boundary conditions as shown in \cref{fig:3d-geometry}.
\begin{figure}
    \centering
    \begin{tikzpicture}[scale=0.5, baseline={([yshift=-.5ex] current bounding box.center)}]
    \foreach \x in {1,...,3} {
    \foreach \z in {1,...,3}{
    \def\nowcolor{black!\grayval}
    \def\nowthick{0.8pt}
    \def\nowopacity{1}
    \ifnum\x=1
        \ifnum\z=3
            \def\nowcolor{blue!\grayval}
            \def\nowthick{1mm}
            \def\nowopacity{0.75}
        \fi
    \fi
    \pgfmathtruncatemacro\grayval{30 + 20*\z} 
    \draw[\nowcolor, line width=\nowthick, opacity=\nowopacity] (\x*2, 2, \z*2) -- (\x*2, 3, \z*2);
    \fill[red!\grayval, opacity=0.75] (\x*2+0.8, 3, \z*2+0.8)
    -- (\x*2+0.8, 3, \z*2-0.8) 
    -- (\x*2-0.8, 3, \z*2-0.8) 
    -- (\x*2-0.8, 3, \z*2+0.8) -- cycle;
    \draw[\nowcolor, line width=\nowthick, opacity=\nowopacity] (\x*2, 3, \z*2) -- (\x*2, 5*2, \z*2);
    }
    }
    \foreach \y in {2,...,4} {
        \foreach \z in {1,...,3}{
            \pgfmathsetmacro\grayval{30 + 20*\z}
            \draw[thick, black!\grayval] (2,\y*2,\z*2) -- (3*2,\y*2,\z*2);
        }
    }
    \foreach \x in {1,...,3} {
        \foreach \y in {2,...,4}{
            \pgfmathsetmacro\grayval{30 + 20*1} 
            \draw[thick, black!\grayval] (\x*2,\y*2,2) -- (\x*2,\y*2,3*2);
        }
    }
\end{tikzpicture}
    \caption{Geometry of the 3D $\mathbb{Z}_2$ surface code. Here, The red surface on the dual lattice denotes a representative of the logical-$\bar X$ operator, and the blue primal path denotes a representative of the logical-$\bar Z$ operator.}
    \label{fig:3d-geometry}
\end{figure}
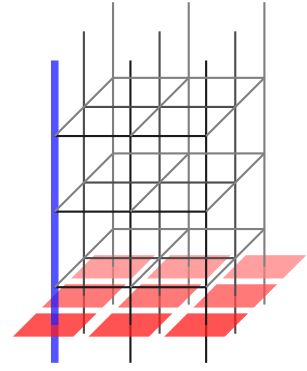
It has the same formal expressions for the code space projector $\mathbb{\Pi}$ in terms of the vertex $A_v$ and plaquette $B_p$ operators but on the 3D cubic lattice, diagrammatically given by
\begin{align}
	A_v = \prod_{\ell \in \partial^\top v} X_\ell = \, \, 
	\intikz{
		\draw[ultra thick,tolred] (-1, 0, 0) -- (1, 0, 0);
		\draw[ultra thick,tolred] (0, -1, 0) -- (0, 1, 0);
		\draw[ultra thick,tolred] (0, 0, -1) -- (0, 0, 1);
		\filldraw[tolred] (0,0,0) circle (0pt) node[above left, black] {$v$};
	} \, ,
\end{align}
\begin{align}
	B_p = \prod_{\ell \in \partial p} Z_\ell = \,
	\, \intikz{
		\draw[ultra thick, tolblue] (0, 0, 0) -- (0, 0, 1) -- (0, 1, 1)  -- (0, 1, 0) -- cycle;
		\node[black] at (0, 0.5, 0.5) {$p$};
	}\, , \;\;
	\, \intikz{
		\draw[ultra thick, tolblue] (0, 0, 0) -- (0, 0, 1) -- (1, 0, 1)  -- (1, 0, 0) -- cycle;
		\node[black] at (0.5, 0, 0.5) {$p$};
	}\, , \;\;
	\, \intikz{
		\draw[ultra thick, tolblue] (0, 0, 0) -- (1, 0, 0) -- (1, 1, 0)  -- (0, 1, 0) -- cycle;
		\node[black] at (0.5, 0.5, 0) {$p$};
	}\, ,
\end{align}
where the red lines ($\color{tolred} \blacksquare$) denote Pauli-$X$ operators and blue lines ($\color{tolblue} \blacksquare$) denote Pauli-$Z$ operators.

The logical-$\bar X$ operators are 2D sheets on the (dual) cubic lattice, whereas logical-$\bar{Z}$ are strings of $Z_\ell$'s spanning vertically on the primal lattice. They are shown in red and blue in \cref{fig:3d-geometry}, respectively.

\subsubsection{Results for the 3D surface code} \label{thm:3d}

\begin{figure}
	\centering
	\includegraphics[width=\linewidth]{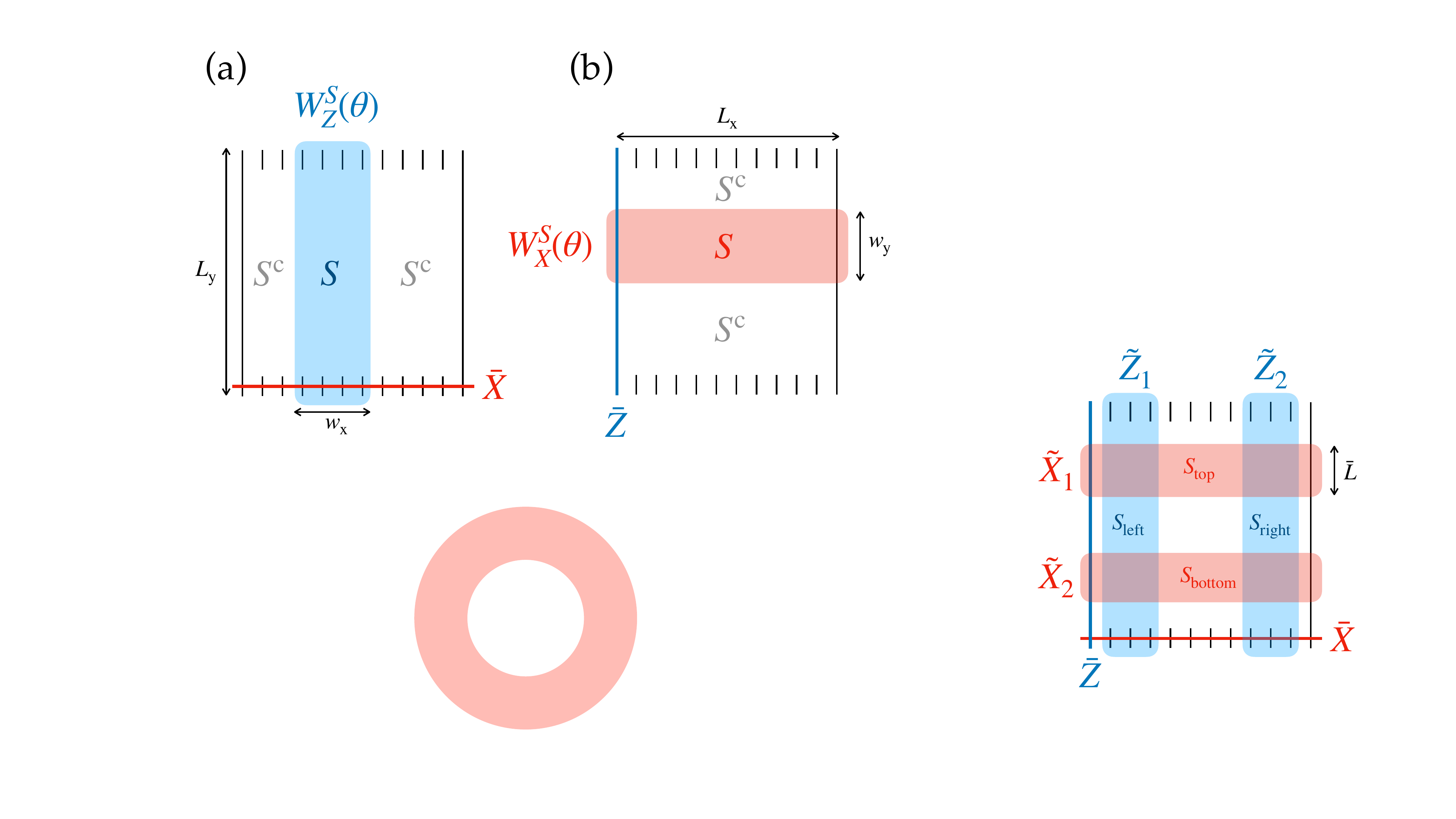}
	\caption{Schematic diagram of supported regions of thickened logical operators for the 3D $\mathbb{Z}_2$ surface code.}
	\label{fig:3d-regions}
\end{figure}

We can adapt the methods devised for \cref{thm:fat} for the 3D $\mathbb{Z}_2$ surface code. In this case, since logical-$\bar X$ operators in form a sheet on a dual lattice, we consider the thickened slab-like region $S$ of thickness $w_\mathrm{z}$ for the support of the approximate logical-$W^\mathrm{X}_S$ operator [\cref{fig:3d-regions}(b)]. On the other hand, logical-$\bar Z$ operators are paths on the primal lattice; therefore we consider tube-like regions $S_\updownarrow$ of diameter $w$ for the support of the approximate logical-$W^\mathrm{Z}_{S_\updownarrow}$ operator [\cref{fig:3d-regions}(a)].

In Appendix \cref{apdx:3d}, we show that there exists $T_\mathrm{c}^\Lambda$ such that for $T < T_\mathrm{c}^\Lambda$, there exists an approximate logical-$W_S^\mathrm{X}$ operator such that
\begin{align}
	& \left\| W_{S}^\mathrm{X} \big( \mathcal{N}_Q\circ \mathcal{C}_{Q \leftarrow R}[\sigma_R]\big) - \mathcal{N}_Q\circ \mathcal{C}_{Q \leftarrow R}[X_R \sigma_R] \right\|_1 \nonumber \\ & \quad \leq c_1(T) \sqrt{L_\mathrm{x} L_\mathrm{y}} e^{-c_2(T) w_\mathrm{z}},
    \label{eq:3d claim X}
\end{align}
and $T_\mathrm{c}^{\Lambda_*}$ such that for $T_Z < T_\mathrm{c}^{\Lambda_*}$, we have
\begin{align}
	& \left\| W_{S_\updownarrow}^\mathrm{Z} \big( \mathcal{N}_Q\circ \mathcal{C}_{Q \leftarrow R}[\sigma_R]\big) - \mathcal{N}_Q\circ \mathcal{C}_{Q \leftarrow R}[Z_R \sigma_R] \right\|_1 \nonumber \\ & \quad \leq c_3(T_Z) \sqrt{L_\mathrm{z}} e^{-c_4(T_Z) w^2},
    \label{eq:3d claim Z}
\end{align}
where $c_{1,2,3,4}(T) > 0$ are constants in $L = \mathrm{min}(L_\mathrm{x}, L_\mathrm{y}, L_\mathrm{z})$. Here, $1/T_\mathrm{c}^\Lambda = \ln \mu^\mathrm{SOS}_\cube$, where $\mu^\mathrm{SOS}_\cube$ is the growth constant of self-osculating surfaces (SOSs) with the loose upper bound $\leq 20.25$ \cite{pkim2025existence}. On the other hand, $T_\mathrm{c}^{\Lambda_*} \leq \ln(9e) \approx \ln(24.5)$.

As examples, we can obtain the lower bound for the threshold for bit-flip and phase-flip errors for the 3D $\mathbb{Z}_2$ surface code, $p^\mathrm{th}_{\cube, \mathrm{bf}} \geq 0.061 \%$ and $p^\mathrm{th}_{\cube, \mathrm{pf}} \geq 0.041 \%$, respectively. We note that for these error channels, the true value of the threshold corresponds to the phase transition point of certain stat-mech models \cite{pkim2026optimal}. In particular, the corresponding stat-mech models are the 3D random-bond Ising model (RBIM) \cite{dennis2002topological} and the 3D $\mathbb{Z}_2$ gauge theory with quenched disorder \cite{xu2025phenomenological}, respectively, along the Nishimori line~\cite{zdeborova2016statistical,p2025planted,pkim2026measurement}. Therefore, our result also provides a rigorous lower bound for the critical point of these stat-mech models.

Note that these results on the 3D $\mathbb{Z}_2$ surface code below demonstrate that our methods can be applied to non-matchable codes. Matchable codes are codes such that a single error results in only a pair of syndromes, as is the case for the 2D $\mathbb{Z}_2$ surface code (both $X$- and $Z$- Pauli errors result in a pair of syndromes). In comparison, for the 3D $\mathbb{Z}_2$ surface code, $X$-errors are non-matchable, since an $X$-error at an edge $\ell$ excites four $Z_p$ stabilizers, namely those that neighbor $\ell$, $\partial^\top \ell = \{p_1, p_2, p_3, p_4 \}$.

\section{Conclusion and outlook \label{sec:discussion}}

In this work we developed techniques for quantifying mixed-state topological order in noise-corrupted code states, both Abelian and non-Abelian, and the associated effect on quantum memory. We used these to put explicit bounds on recovery thresholds for various different codes, and to derive the existence of thickened logical string operators that act on the noise-corrupted states. This in turn was used to show that these states are long-range entangled and exhibit an emergent 1-form symmetry associated with the anyonic excitations of the code.

Given that non-Abelian codes offer a complementary set of properties to the more familiar Pauli stabilizer codes---most importantly the ability to perform non-Clifford gates transversally---we anticipate that the techniques developed here will be useful in the design and analysis of fault-tolerant approaches to universal quantum computation, in the spirit of Refs.~\onlinecite{davydova2025, bauer2026, manjunath2026universal, sajith2026, huang2025generating, huang2026hybridlatticesurgerynonclifford, warman2026constant}. Moreover, the information-theoretic techniques we formulate here are likely to find use in future studies of mixed-state phases more broadly, e.g.~those associated with strong-to-weak symmetry breaking \cite{sang2025mixedstate}. 


A number of concrete and interesting directions for future research immediately present themselves.
Building on our extensions to the 3D $\mathbb{Z}_2$ surface code, there is scope to develop a more complete theory for more general higher-dimensional topological codes, such as color codes \cite{bombin2006}. In addition, one could consider non-geometrically local low-density parity check (LDPC) codes \cite{breuckmann2021,panteleev2022}, non-Abelian examples of which have recently been introduced \cite{christos2026nonabelian, zhu2026nonabelian, williamson2026fast}. The methods outlined in Section \ref{sec:methods} are naturally adaptable to this setting, as are the combinatorial decompositions introduced in Appendix \ref{apdx:bounds}. In addition, there is scope for quantitatively improving the various thresholds we have derived, using similar techniques as in Appendix \ref{apdx:z2 improvements}.

The setup considered in this work involved a clean code state subject to a single application of some particular noise channel. The is a natural choice for understanding the resilience of mixed-state topological order against noise \cite{fan2024diagnostics}. However, in the context of fault-tolerant quantum computing, it is necessary to understand the effect of repeated applications of noise and syndrome measurements \cite{dennis2002topological, hauser2026}. Thanks to our channel-based formalism, it should be possible to use similar techniques to bound the effect of noise across both space and time on the encoded logical information. This would allow one to prove threshold theorems for general noise channels and non-stabilizer codes, thereby generalising some of the results of Ref.~\onlinecite{christandl2025}. Also, while we have focused on demonstrating the existence of recovery operations, which amounts to establishing the existence of some optimal decoder, it would also be interesting to incorporate particular, possibly sub-optimal, decoding operations into this formalism.

While our interest has been in topological phases, the methods developed here also have broader significance in light of the increased interest in information-theoretic approaches to classifying phases of matter. Our usage of the operator algebra formalism could be used to better understand phases such as the classical loop gas highlighted in Ref.~\onlinecite{sang2025mixedstate}, which features strong-to-weak spontaneous symmetry breaking of an emergent 1-form symmetry, as well as intrinsically mixed-state phases \cite{wang2025, sohal2025}. Furthermore, the information-disturbance tradeoff argument \cite{kretschmann2008information, beny2009}, which underpins the relationship between properties of the environment and the encoding of information in the system [Eq.~\eqref{eq:eps rec bound}], could be used to understand mixed-state ordering and the associated constraints on state preparation algorithms \cite{mcginley2025lower, yi2024complexity} in more general settings. 

Finally, having successfully characterised the encoding of logical information in the corrupted state in a spatially-resolved way, it would be appealing to go one step further and demonstrate the existence of local recovery channels that bring the noisy state back to the clean limit in a fully spatially local way, e.g.~as implemented by a low-depth channel circuit. This would provide the most complete picture of mixed-state topological order \cite{Coser2019classificationof, zhang2025stability}, and could potentially be used to compute other fundamental information-theoretic properties such as the Markov length \cite{yi2026universal, sang2025mixedstate}.

\section*{Acknowledgements}
 We would like to thank Michael Levin and Jeongwan Haah for insightful comments. SWPK is supported by Engineering and Physical Sciences Research Council (EPSRC) DTP International Studentship Grant Ref.\ No.\ EP/W524475/1. MM acknowledges support from Trinity College, Cambridge.

\bibliography{bibliography} 

\onecolumngrid

\beginappendices

\counterwithin*{theorem}{section}
\counterwithin*{lemma}{section}
\counterwithin*{corollary}{section}
\renewcommand{\thetheorem}{\Alph{section}\arabic{theorem}}
\renewcommand{\thelemma}{\Alph{section}\arabic{lemma}}
\renewcommand{\thecorollary}{\Alph{section}\arabic{corollary}}

\crefalias{section}{appendix}
\crefalias{subsection}{appendix}
\crefalias{subsubsection}{appendix}
\makeatother


\section{Proof of \cref{lem:effective operator}: Recoverability implies that logical operations can be implemented} \label{app:recov-means-log-ops}
In this Appendix, we prove \cref{lem:effective operator}, which was stated informally in Section \ref{sec:aec-analysis}. The main tool in our proof is Theorem 1 from Ref.~\onlinecite{beny2009}, which we described in \cref{subsec:recoverability conditions}; our arguments essentially improve this result to obtain stronger implications.

From the definition of approximate recoverability of an algebra, \eqref{eq:recov error algebra}, it is relatively straightforward to show that the expectation value of any logical observable $O_R \in \mathfrak{J}$ can be approximately read out by measuring some other observable $\tilde{O}_{Q'}$ on the noise-corrupted system $Q'$, see Eq.~\eqref{eq:exp value approx}. As we will now see, approximate recoverability also implies that logical operations in $\mathfrak{J}$ can be implemented by acting on $Q'$, while ensuring high fidelity globally, i.e.~on the joint-system environment state.

\begin{customlem}{1}[Formal]
	Let $\mathcal{W}_{Q'E \leftarrow R}[\,\cdot\,] = W_{Q'E \leftarrow R}\,\cdot\, (W_{Q'E \leftarrow R})^\dagger$ be an isometric channel with `system' ($Q'$) and `environment' ($E$) output spaces, which simultaneously dilates the mutually complementary channels $\mathcal{M}_{Q' \leftarrow R} = \tr_{E}\circ \mathcal{W}_{Q'E \leftarrow R}$ and $\widehat{\mathcal{M}}_{E \leftarrow R} = \tr_{Q'} \circ \mathcal{W}_{Q'E \leftarrow R}$. Let $\mathfrak{J} \subseteq \mathfrak{B}(\mathscr{H}_R)$ be a von Neumann algebra. Suppose that at least one of the following conditions hold:
	\begin{enumerate}
		\item There exists a channel $\mathcal{R}_{R \leftarrow Q'}$ such that $\|\mathcal{R}_{R \leftarrow Q'} \circ \mathcal{M}_{Q' \leftarrow R} - \mathcal{P}^{\mathfrak{J}}_R\|_\diamond \leq \frac{\epsilon^2}{2}$
		\item $\|\widehat{\mathcal{M}}_{E \leftarrow R} - \widehat{\mathcal{M}}_{E \leftarrow R}\circ \mathcal{P}^{C(\mathfrak{J})}_R\|_\diamond \leq \epsilon$
	\end{enumerate}
	where $\mathcal{P}^{\mathfrak{J}}_R$ and $\mathcal{P}^{C(\mathfrak{J})}_R$ are the super-projectors onto $\mathfrak{J}$ and its commutant $C(\mathfrak{J})$. Then, for any superoperator $\mathcal{T}_R$ whose action can be expressed in the form
	\begin{align}
		\mathcal{T}_R[\sigma_R] &= \sum_{i} J_i\sigma_R J'_i, & \text{where }J_i, J_{i}' \in \mathfrak{J},
		\label{eq:superoperator in algebra}
	\end{align}
	there is a corresponding superoperator $\tilde{\mathcal{T}}_{Q'}$ acting on the system $Q'$ that has approximately the same effect as $\mathcal{T}_R$ on the joint system-environment state,
	\begin{align}
		\left\| \big(\tilde{\mathcal{T}}_{Q'}\otimes \mathrm{id}_E\big) \circ \mathcal{W}_{Q'E \leftarrow R} - \mathcal{W}_{Q'E \leftarrow R} \circ \mathcal{T}_{R}\right\|_\diamond \leq 4\sqrt{\epsilon}\|\mathcal{T}_R\|_\diamond.
		\label{eq:TR approx}
	\end{align}
	If $\mathcal{T}_R$ is a CPTP map, then $\tilde{\mathcal{T}}_{Q'}$ can be chosen to be CPTP and the right hand side of \eqref{eq:TR approx} becomes $4\sqrt{\epsilon}$. Simultaneously, if $\mathcal{T}_R$ has Kraus rank 1, such that $\mathcal{T}_R[\sigma_R] = O_R \sigma_R O'_R$ with $O_R, O_R' \in \mathfrak{J}$, then  $\tilde{\mathcal{T}}_{Q'}[\sigma_{Q'}] = O_{Q'} \sigma_{Q'} O_{Q'}'$, where the operators $O_{Q'}$, $O_{Q'}'$ can be chosen to be of the form $O_{Q'}^{(\prime)} = \mathcal{L}_{Q' \leftarrow R}[O_{R}^{(\prime)}]$ for a completely positive unital map $\mathcal{L}_{Q' \leftarrow R}$. In particular, if $O_R$, $O_R'$ are unitary, then $O_{Q'}$, $O_{Q'}'$ are unitary.
\end{customlem}

The two conditions in this Lemma correspond to the definition of approximate recoverability of an algebra Eq.~\eqref{eq:recov error algebra} and the proxy condition \eqref{eq:algebra comp error} respectively. The fact that the error in Eq.~\eqref{eq:TR approx} is defined with respect to the joint system-environment state $Q'E$, and not just the system $Q'$, will be crucial for our proof of the existence of thickened logical operators, \cref{thm:fat}. \\

\noindent \textit{Proof of Lemma \ref{lem:effective operator}.~---~}Our proof uses the main result of Ref.~\onlinecite{kretschmann2008information}, which was also central to the arguments of Ref.~\onlinecite{beny2009}. To keep the argument self-contained, we state this result here.
\begin{lemma}[Ref.~\onlinecite{kretschmann2008information}]
	If $\mathcal{N}_{A \leftarrow R}^{\vphantom{'}}$ and $\mathcal{N}_{A \leftarrow R}'$ are channels satisfying $\|\mathcal{N}_{A \leftarrow R}^{\vphantom{'}} - \mathcal{N}_{A \leftarrow R}'\|_\diamond \leq \epsilon$, then for any  isometry $V_{AB \leftarrow R}$ that dilates $\mathcal{N}_{A \leftarrow R}^{\vphantom{'}}$, there exists a an isometry $V_{AB \leftarrow R}'$ that dilates $\mathcal{N}_{A \leftarrow R}'$ such that $\|V_{AB \leftarrow R} - V_{AB \leftarrow R}'\|_\infty \leq \sqrt{\epsilon}$. In particular, if $\mathcal{V}_{AB \leftarrow R}$, $\mathcal{V}_{AB \leftarrow R}'$ are the corresponding isometric channels, then $\|\mathcal{V}_{AB \leftarrow R}^{\vphantom{'}} - \mathcal{V}_{AB \leftarrow R}'\|_\diamond \leq 2\sqrt{\epsilon}$.
	\label{lem:stinespring cont}
\end{lemma}

\noindent As shown in Ref.~\onlinecite{beny2009}, this result can be used to show that Condition 1 of \cref{lem:effective operator} implies Condition 2; therefore, without loss of generality, we can assume the latter.
Then, by the same result, we can infer the existence of an isometry $W'_{Q'E \leftarrow R}$  that dilates the composite channel $\big(\widehat{\mathcal{M}}_{E \leftarrow R} \circ \mathcal{P}^{C(\mathfrak{J})}_R\big)$, such that the associated channel $\mathcal{W}'_{Q'E \leftarrow R}$ satisfies
\begin{align}
	\|\mathcal{W}_{Q'E \leftarrow R}^{\vphantom{'}} - \mathcal{W}'_{Q'E \leftarrow R}\|_\diamond \leq 2\sqrt{\epsilon}
	\label{eq:W diff bound}
\end{align}
Equation \eqref{eq:super-projector idempotent} gives $\big(\widehat{\mathcal{M}}_{E \leftarrow R} \circ \mathcal{P}^{C(\mathfrak{J})}_R\big) \circ \mathcal{P}^{C(\mathfrak{J})}_R = \big(\widehat{\mathcal{M}}_{E \leftarrow R} \circ \mathcal{P}^{C(\mathfrak{J})}_R\big)$, which, by the main result of Ref.~\onlinecite{beny2009} [see also Eq.~\eqref{eq:eps rec bound}], implies that the algebra $\mathfrak{J}$ is \textit{exactly} correctable for the channel $\tr_E \circ \mathcal{W}'_{Q'E \leftarrow R}$. That is, there exists a recovery channel $\mathcal{R}_{R \leftarrow Q'}$ such that $(\mathcal{R}_{R \leftarrow Q'} \otimes \tr_E) \circ \mathcal{W}'_{Q'E \leftarrow R} = \mathcal{P}^{\mathfrak{J}}_R$. As a consequence, if $\mathcal{Y}_{RE' \leftarrow Q'}$ is a dilation of $\mathcal{R}_{R \leftarrow Q'}$, then the isometric channel 
\begin{align}
	\mathcal{F}^{\mathfrak{J}}_{REE' \leftarrow R} \coloneqq (\mathcal{Y}_{RE' \leftarrow Q'}\otimes \mathrm{id}_E) \circ \mathcal{W}'_{Q'E \leftarrow R}
	\label{eq:dilation YW}
\end{align}
is an exact dilation of $\mathcal{P}^{\mathfrak{J}}_R$.

We now claim that for any superoperator $\mathcal{T}_R$ that can be written in the form \eqref{eq:superoperator in algebra}, and for any $\mathcal{F}^{\mathfrak{J}}_{RE'' \leftarrow R}$ that dilates $\mathcal{P}^{\mathfrak{J}}_R$, we have
\begin{align}
	\mathcal{F}_{RE'' \leftarrow R}^{\mathfrak{J}} \circ \mathcal{T}_R = (\mathcal{T}_R \otimes \mathrm{id}_{E''})\circ \mathcal{F}_{RE'' \leftarrow R}^{\mathfrak{J}}.
	\label{eq:TR commute}
\end{align}
These claims suffice to to prove \eqref{eq:TR approx} since, given the isometric (and hence reversible) nature of $\mathcal{Y}_{RE' \leftarrow Q'}$ there exists a CPTP map $\tilde{\mathcal{Y}}_{Q' \leftarrow RE'}$ that satisfies
\begin{align}
	\label{eq:channel Y inv}
	\big((\tilde{\mathcal{Y}}_{Q' \leftarrow RE'}\circ\mathcal{Y}_{RE' \leftarrow Q'})\otimes \mathrm{id}_E\big)\circ \mathcal{W}'_{Q'E \leftarrow R} &= \mathcal{W}_{Q'E \leftarrow R}'. 
\end{align}
Then we can choose the superoperator $\mathcal{T}_{Q'}'$ to be
\begin{align}
	\mathcal{T}_{Q'}' = \tilde{\mathcal{Y}}_{Q' \leftarrow RE'}\circ (\mathcal{T}_R \otimes \mathrm{id}_{E'}) \circ \mathcal{Y}_{RE' \leftarrow Q'}.
	\label{eq:TQ def}
\end{align}
which satisfies
\begin{align}
	\big(\mathcal{T}'_{Q'}\otimes \mathrm{id}_E\big) \circ \mathcal{W}_{Q'E \leftarrow R}' &= (\tilde{\mathcal{Y}}_{Q' \leftarrow RE'}\otimes \mathrm{id}_E)\circ (\mathcal{T}_R \otimes \mathrm{id}_{EE'}) \circ  \mathcal{F}_{REE' \leftarrow R}^{\mathfrak{J}} && \text{by Eqs.~(\nopareneqref{eq:dilation YW}, \nopareneqref{eq:TQ def}) }  \nonumber\\ &= (\tilde{\mathcal{Y}}_{Q' \leftarrow RE'}\otimes \mathrm{id}_E)\circ\mathcal{F}_{REE' \leftarrow R}^{\mathfrak{J}} \circ \mathcal{T}_R  && \text{by Eq.~\eqref{eq:TR commute}} \nonumber\\ &= \big((\tilde{\mathcal{Y}}_{Q' \leftarrow RE'}\circ\mathcal{Y}_{RE' \leftarrow Q'})\otimes \mathrm{id}_E\big)\circ \mathcal{W}_{Q'E \leftarrow R}'\circ \mathcal{T}_R. && \text{by Eq.~\eqref{eq:dilation YW}} \nonumber\\ &= \mathcal{W}_{Q'E \leftarrow R}'\circ \mathcal{T}_R && \text{by Eq.~\eqref{eq:channel Y inv}.}
	\label{eq:TQ property}
\end{align}
Combining \eqref{eq:TQ property} with the triangle inequality, the norm in question can then be bounded as
\begin{align}
	\Big\| \big(\mathcal{T}'_{Q'}\otimes \mathrm{id}_E\big) \circ \mathcal{W}_{Q'E \leftarrow R} - \mathcal{W}_{Q'E \leftarrow R} \circ \mathcal{T}_{R}\Big\|_\diamond   \leq 
	\Big\| \big(\mathcal{T}'_{Q'}\otimes &\mathrm{id}_E\big) \circ (\mathcal{W}_{Q'E \leftarrow R} - \mathcal{W}'_{Q'E \leftarrow R})\Big\|_\diamond \nonumber\\ &+ \Big\| (\mathcal{W}_{Q'E \leftarrow R} - \mathcal{W}_{Q'E \leftarrow R}') \circ \mathcal{T}_{R}\Big\|_\diamond.
	\label{eq:TQ bound diff}
\end{align}
By the general inequality $\|\mathcal{T}\circ \mathcal{S}\|_\diamond \leq \|\mathcal{T}\|_\diamond \|\mathcal{S}\|_\diamond$ for superoperators $\mathcal{T}, \mathcal{S}$ (see e.g.~Lemma 12 of Ref.~\onlinecite{aharonov1998quantum}), along with $\|\mathcal{N}\|_\diamond = 1$ for any CPTP map $\mathcal{N}$, each of the norms on the right hand side of \eqref{eq:TQ bound diff} can be upper bounded by $2\sqrt{\epsilon}\|\mathcal{T}_R\|_\diamond$, thanks to Eq.~\eqref{eq:W diff bound}. This proves Equation \eqref{eq:TR approx}.

While there are many choices of $\tilde{\mathcal{Y}}_{Q' \leftarrow RE'}$ that satisfy \eqref{eq:channel Y inv}, for concreteness we can choose
\begin{align}
	\label{eq:channel Y inv CPTP}
	\tilde{\mathcal{Y}}_{Q' \leftarrow RE'}[\sigma_{RE'}] &= (Y_{RE' \leftarrow Q'})^\dagger \sigma_{RE'}Y_{RE' \leftarrow Q'} + \tr[\Pi^\perp_{RE'}\sigma_{RE'}] (I_{Q'}/d_{Q'}) \\  \text{where }\Pi^\perp_{RE'} &\coloneqq I_{RE'} - Y_{RE' \leftarrow Q'}(Y_{RE' \leftarrow Q'})^\dagger 
\end{align}
The first term on the right of \eqref{eq:channel Y inv CPTP} describes the action of $\tilde{\mathcal{Y}}_{Q' \leftarrow RE'}$ on states in the image of $Y_{RE' \leftarrow Q'}$, and the second term describes the action on the complement $\Pi^\perp_{RE'}$. Note that for the choice \eqref{eq:TQ def}, we have that $\mathcal{T}_{Q'}'$ is CPTP whenever $\mathcal{T}_R$ is CPTP, as claimed in the lemma.

We now have to prove \eqref{eq:TR commute}. For this purpose, we will explicitly construct a dilation of $\mathcal{P}^{\mathfrak{J}}_R$, to which all other dilations are equivalent up to isometries on the environment degrees of freedom \cite{watrous2018theory}. 
To do so, we invoke the Wedderburn-Artin theorem (see e.g.~Appendix 11.8 of Ref.~\onlinecite{petz2007quantum}), which implies that all von Neumann algebras over finite-dimensional Hilbert spaces $\mathscr{H}_R$ are classified in the following way. There is a decomposition of Hilbert space as
\begin{align}
	\mathscr{H}_R = \left[\bigoplus_k (\mathscr{A}_k \otimes \mathscr{B}_k)\right] 
	\label{eq:hilbert decomp}
\end{align}
where each $\mathscr{A}_k$, $\mathscr{B}_k$ is a Hilbert space, such that
\begin{align}
	\mathfrak{J} = \left[\bigoplus_k \mathfrak{B}(\mathscr{A}_k) \otimes I_{\mathscr{B}_k}\right] 
\end{align}
where $\mathfrak{B}(\mathscr{A}_k)$ is the algebra of all operators over $\mathscr{A}_k$. 
In other words, $O \in \mathfrak{J}$ if and only if one can write $O = \sum_k (O_{\mathscr{A}_k} \otimes I_{\mathscr{B}_k})$, for some collection of operators $O_{\mathscr{A}_k} \in \mathfrak{B}(\mathscr{A}_k)$. Moreover, the commutant $C(\mathfrak{J})$ is made up of operators of the form $\sum_{k} (I_{\mathscr{A}_k} \otimes O_{\mathscr{B}_k})$.

Using this decomposition, the super-projector can be explicitly constructed as $\mathcal{P}^{\mathfrak{J}}_R[\sigma_R] = \sum_k (\mathrm{id}_{\mathscr{A}_k}\otimes \mathcal{D}_{\mathscr{B}_k})[P_k \sigma P_k]$, where $P_k$ is the projector onto the component  $(\mathscr{A}_k \otimes \mathscr{B}_k) \subseteq \mathscr{H}_R$, and $\mathcal{D}_{\mathscr{B}_k}$ is the completely depolarizing channel on $\mathscr{B}_k$ \cite{beny2009}. We can construct an isometry $F_{RE'' \leftarrow R}$ that dilates this channel using an environment $E''$ whose Hilbert space has the structure $\mathscr{H}_{E''} = \bigoplus_k \mathscr{B}_k^{\otimes 2}$, namely
\begin{align}
	F_{RE'' \leftarrow R} = \sum_k \frac{1}{\sqrt{|\mathscr{B}_k|}} \sum_{a = 1}^{|\mathscr{A}_k|} \sum_{b,b' = 1}^{|\mathscr{B}_k|} \big(\ket{k,b'\otimes b}_{E''} \otimes \ket{k,a\otimes b'}_R\big) \bra{k,a \otimes b}_R.
	\label{eq:dilation projector}
\end{align}
where $\ket{k,a \otimes b}_R$ and $\ket{k,b \otimes b'}_{E''}$ are bases for $\mathscr{H}_R$ and $\mathscr{H}_{E''}$ that naturally respect the structure \eqref{eq:hilbert decomp}. By direct computation of matrix elements in this basis, one can verify that if $O_k \in \mathfrak{J}$ is supported within a single block $k$, then $F_{RE'' \leftarrow R} O_{k} =   \big(O_{k}\otimes I_{E''}\big) F_{RE'' \leftarrow R}$. By linearity, this extends to the whole algebra
\begin{align}
	F_{RE'' \leftarrow R} O_{R} &=   \big(O_{R}\otimes I_{E''}\big) F_{RE'' \leftarrow R} & \forall \, O_R \in \mathfrak{J}.
	\label{eq:FO commute}
\end{align}
This extends property \eqref{eq:super-projector commute} of the super-projector $\mathcal{P}^{\mathfrak{J}}_{R}$ to its dilation $F_{RE'' \leftarrow R}$. Since Stinespring dilations are unique up to isometry on $E''$, and the dilation \eqref{eq:dilation projector} has the minimum possible dimension for the purifying degrees of freedom, the dilation $F_{REE'\leftarrow R}$ appearing in \eqref{eq:dilation YW} must be expressible in the form $F_{REE'\leftarrow R} = (L_{E E' \leftarrow E''}\otimes I_R)F_{RE'' \leftarrow R}$ for some isometry $L_{E E' \leftarrow E''}$. Combining this with \eqref{eq:FO commute}, we see that if $\mathcal{T}_R$ can be written in the form \eqref{eq:superoperator in algebra}, then we have \cref{eq:TR commute}, as claimed.

Finally, we will prove the final statement of \cref{lem:effective operator}. If $\mathcal{T}_R$ is unitary then it also has Kraus rank 1, so we assume that $\mathcal{T}_R[\,\cdot\,] = O_R\,\cdot\, O_R'$ for $O_R, O_R' \in \mathfrak{J}$. In this case, the superoperator $\mathcal{T}'_{Q'}$ [Eq.~\eqref{eq:TQ def}] has the form $\mathcal{T}'_{Q'}[\,\cdot\,] = V_{Q'}\,\cdot\, V_{Q'}'$, where
\begin{align}
	O_{Q'}^{(\prime)} = (Y_{RE' \leftarrow Q'})^\dagger (O_R^{(\prime)}\otimes I_{E'}) Y_{RE' \leftarrow Q'} \eqqcolon \mathcal{L}_{Q' \leftarrow R}[O_R^{(\prime)}].
	\label{eq:Vprime def}
\end{align}
where $O_{R}^{(\prime)}$ is a placeholder for $O_{R}$ or $O_{R}'$, and similar for $O_{Q'}^{(\prime)}$. By Stinespring's dilation theorem \cite{watrous2018theory}, the map $\mathcal{L}_{Q' \leftarrow R}[\,\cdot\,] = (Y_{RE' \leftarrow Q'})^\dagger (\,\cdot\,\otimes I_{E'}) Y_{RE' \leftarrow Q'}$ is completely positive and unital, as claimed. We now consider the case where $O_{R}$, $O_{R}'$ are both unitary. As written in Eq.~\eqref{eq:Vprime def}, $O_{Q'}^{(\prime)}$ may not be unitary, but we do have the identity $(O_{Q'}^{(\prime)}\otimes I_E)W_{Q'E \leftarrow R}' = W_{Q'E \leftarrow R}'O_R^{(\prime)}$, thanks to Eq.~\eqref{eq:TQ property}. Consequently, for any $\ket{\phi}_R \in \mathscr{H}_R$ we have $\|(O_{Q'}^{(\prime)}\otimes I_E)W_{Q'E \leftarrow R}'\ket{\phi}_R\| = \|W_{Q'E \leftarrow R}'\ket{\phi}_R\|$, which means that $(O_{Q'}^{(\prime)}\otimes I_E)$ acts isometrically on the range of $W_{Q'E \leftarrow R}'$,
\begin{align}
	\big((O_{Q'}^{(\prime)})^\dagger O_{Q'}^{(\prime)} \otimes I_E\big)W_{Q'E \leftarrow R}' = W_{Q'E \leftarrow R}'.
	\label{eq:Vprime isom}
\end{align}
Now, let $\{\ket{i}_R\}_{i}$ and $\{\ket{j}_E\}_j$ be orthonormal bases for $\mathscr{H}_R$ and $\mathscr{H}_E$, respectively. Define the subspace
\begin{align}
	\mathscr{F}_{Q'} \coloneqq \mathrm{span}\big\{\ket{\psi_{ij}}_{Q'} \coloneqq
	(I_{Q'}\otimes \bra{j}_E)W_{Q'E \leftarrow R}'\ket{i}_R \}
	\big\}
	\subseteq \mathscr{H}_{Q'},
	\label{eq:FQ subspace}
\end{align}
with associated projector $\Pi_{Q'}^{\mathscr{F}}$. Any state in $\ket{\psi}_{Q'} \in \mathscr{F}_{Q'}$ can be written as $\ket{\psi}_{Q'} = \sum_{ij}c_{ij} \ket{\psi_{ij}}_{Q'}$ for coefficients $c_{ij} \in \mathbb{C}$, and for such a state we have
\begin{align}
	\left\|O_{Q'}^{(\prime)}\ket{\psi}_{Q'}\right\|^2 &= \sum_{i,j,i',j'} c_{ij}c_{i'j'}^* \bigg\langle i'\bigg|(W_{Q'E \leftarrow R}')^\dagger\Big((V_{Q'}^{(\prime)})^\dagger V_{Q'}^{(\prime)} \otimes \ket{j'}\bra{j}_{E}\Big)W_{Q'E \leftarrow R}'\bigg|i \bigg\rangle \nonumber\\ &= \sum_{i,j,i',j'} c_{ij}c_{i'j'}^* \bigg\langle i'\bigg|(W_{Q'E \leftarrow R}')^\dagger\Big(I_{Q'} \otimes \ket{j'}\bra{j}_{E}\Big)\bigg[\Big((V_{Q'}^{(\prime)})^\dagger V_{Q'}^{(\prime)} \otimes I_E\Big) W_{Q'E \leftarrow R}'\bigg]\bigg|i \bigg\rangle  \nonumber\\ 
	&= \sum_{i,j,i',j'} c_{ij}c_{i'j'}^* \bigg\langle i'\bigg|(W_{Q'E \leftarrow R}')^\dagger\Big(I_{Q'} \otimes \ket{j'}\bra{j}_{E}\Big) W_{Q'E \leftarrow R}'\bigg|i \bigg\rangle = \left\|\ket{\psi}_{Q'}\right\|^2,
\end{align}
where we use \eqref{eq:Vprime isom} in going to the last line. Hence, $O_{Q'}^{(\prime)}$ acts isometrically on $\mathscr{F}_{Q'}$, and in particular maps $\mathscr{F}_{Q'}$ to a space $\mathscr{G}_{Q'} = \text{span}\{O_{Q'}^{(\prime)}\ket{\psi}_{Q'} : \ket{\psi}_{Q'} \in \mathscr{F}_{Q'}\}$ of equal dimension. If we choose any isometry $U^\perp_{Q'}$ that maps  $\mathscr{F}_{Q'}^\perp$ to $\mathscr{G}_{Q'}^\perp$, then the modified operator
\begin{align}
	\tilde{O}^{(\prime)}_{Q'} \coloneqq O_{Q'}^{(\prime)} \Pi_{Q'}^{\mathscr{F}} + U^\perp_{Q'}(I_{Q'} - \Pi_{Q'}^{\mathscr{F}})
\end{align}
is unitary, and has the same action as $O_{Q'}^{(\prime)} $ on $\mathscr{F}_{Q'}$. Due to the definition \eqref{eq:FQ subspace}, the range of $W_{Q'E \leftarrow R}'$ is contained in $\mathscr{F}_{Q'} \otimes \mathscr{H}_E$, and thus $(\tilde{O}^{(\prime)}_{Q'}\otimes I_{E})W_{Q'E \leftarrow R}' = (O^{(\prime)}_{Q'}\otimes I_{E})W_{Q'E \leftarrow R}'$. This proves the claim. \hfill $\square$

\section{Further details on the quantum double model \label{apdx:q double codes}}

In Section \ref{sec:quantum-doubles}, we specified the code stabilizers of the quantum double model $D(G)$, and described a particular type of logical operator that is supported on a thickened path following one of the Cartesian axes [Eq.~\eqref{eq:ribbon cartesian}]. Here we describe the full algebra of logical operators on more general paths---known as \textit{ribbons} $r$. These definitions and results appear in Ref.~\onlinecite{bombin2008}. 

\begin{figure}
    \centering
    \includegraphics[width=0.4\linewidth]{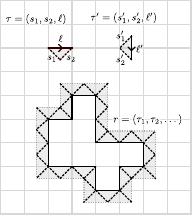}
    \caption{Direct ($\tau$) and dual ($\tau'$) triangles, with arrows showing their orientation. Relative to the edge orientation shown in \eqref{eq:edge orientation}, $\tau$ is anti-aligned and $\tau'$ is aligned. Below is an example of a closed ribbon $r = (\tau_1, \tau_2, \dots)$.}
    \label{fig:ribbons}
\end{figure}

For convenience, let us define a site $s = (v,p)$ as a vertex $v = v(s)$ together with one of its neighboring plaquettes $p = p(s)$. Distinct sites $s, s'$ are said to be directly (dually) adjacent if they share a vertex (plaquette), in which case there is a unique edge $\ell$ shared by the different plaquettes (vertices). If $s_{1,2}$ are a pair of adjacent sites that share the edge $\ell$, then $\tau = (s_1, s_2, \ell)$ defines a direct (dual) \textit{triangle}, whose orientation is determined by the ordering of $s_1$ and $s_2$ [\cref{fig:ribbons}]. If $\tau$ is direct then we say it is aligned if the orientation of the edge matches the orientation of the triangle, i.e.~$\ell = (v({s_1}), v(s_2))$, and we say it is anti-aligned if the opposite is true i.e.~$\ell = (v({s_2}), v(s_1))$. If $\tau$ is dual, then we say $\tau$ is aligned if its direction (i.e.~the vector from the centerpoint of $p(s_1)$ to the centrepoint of $p(s_2)$, denoted by an arrow in \cref{fig:ribbons}) is a $90^\circ$ counterclockwise rotation of the orientation of the edge $\ell$.


A strip $r$ is then a sequence of sites $(s_0, s_1, \ldots, s_l)$ such that $s_{t-1}$ and $s_{t}$ are distinct and either directly or dually adjacent for all $t = 1, \ldots, l$. Each strip can be defined in terms of its triangles $\tau_t = (s_{t-1}, s_t, \ell_{t})$, where $\ell_t$ is the unique edge that ensures $\tau_t$ is a valid triangle. For the purposes of this work, we will use the following definition of a ribbon $r$.
\begin{definition}
	For the quantum double model, a ribbon $r$ is a strip, i.e.~a sequence of triangle-adjacent sites $(s_0, s_1, \ldots, s_l)$, with corresponding triangles $(\tau_1, \tau_2, \ldots, \tau_l)$ such that the following holds. No two distinct triangles $\tau_{t_1}$, $\tau_{t_2}$ appearing in the strip can have the same edge (irrespective of orientation). 
\end{definition}
A slightly more permissive condition is given in Ref.~\onlinecite{bombin2008}, but this will not be needed here. The endpoints of a ribbon $r$ are denoted $\partial_0r = s_0$, $\partial_1 r = s_l$, and $r$ is said to be closed if $\partial_0r = \partial_1 r$, and open otherwise.

The building blocks of ribbon operators are triangle operators $\{F^{h,g}_\tau\}_{h,g \in G}$, which act only on the edge $\ell \in \tau$ of a direct or dual triangle. These are defined as
\begin{equation}
	\begin{aligned}
		F^{h,g}_\tau  &= \delta_{g, 1} L^h_{\ell} && \text{if }\tau \text{ is direct aligned,} \\  F^{h,g}_\tau  &= \delta_{g, 1} R^h_{\ell} && \text{if }\tau \text{ is direct anti-aligned,} \\
		F^{h,g}_\tau  &= \dyad{g}_\ell && \text{if }\tau \text{ is dual aligned}, \\
        F^{h,g}_\tau  &= \dyad{g^{-1}}_\ell && \text{if }\tau \text{ is dual anti-aligned}.
	\end{aligned}
\end{equation}
A single triangle $\tau = (s_1, s_2, \ell)$ constitutes the shortest possible open ribbon $r = (s_1, s_2)$, and the corresponding ribbon operator is defined as $F^{h,g}_r = F^{h,g}_\tau$. If $r = (s_0, \ldots, s_l)$ and $r' = (s'_0, \ldots, s'_{l'})$ are ribbons such that $s_l = s'_0$, and the composition $r'' = (s_0, \ldots, s_l, s'_1, \ldots, s'_{l'})$ also forms a ribbon, then the associated operators are given by
\begin{align}
	F^{h,g}_{r''} = \sum_{k \in G} F^{h,k}_r F^{k^{-1}hk,k^{-1}g}_{r'}
\end{align}
Together with the base case of a single triangle, this gives us a complete recursive definition of any ribbon operator. One can verify that a sequence of alternating dual and direct triangles can chosen to reproduce the straight-line ribbon described in the main text, Eq.~\eqref{eq:ribbon cartesian}. The resulting operators can be shown to satisfy \cite{bombin2008}
\begin{equation}
	\begin{aligned}
		[F^{h,g}_r, A_v] &= 0 &&\forall v \notin \{v(\partial_i r)\}_{i=0,1} \\
		[F^{h,g}_r, B_p] &= 0 &&\forall p \notin \{p(\partial_i r)\}_{i=0,1}
	\end{aligned}
	\label{eq:open ribbon commute}
\end{equation}
As we saw for the `straight path' ribbon in the main text [Eq.~\eqref{eq:ribbon product}], these more general ribbon operators also obey the product rule \cite{kitaev2003fault}
\begin{align}
	F^{h_1,g_1}_r F^{h_2,g_2}_r = \delta_{h_1, g_1 h_2 g_1^{-1}} F^{h_1, g_1g_2}_r
	\label{eq:F product appendix}
\end{align}
The space of operators $\mathfrak{F}_r \coloneqq \mathrm{span}\{F^{h,g}_r\}_{h,g \in G}$ is therefore closed under multiplication. Together with the properties  $(F^{h,g}_r)^\dagger = F^{h^{-1}, g}_r$ and $\sum_g F^{1,g}_r = I$, this implies that $\mathfrak{F}_r$ forms a von Neumann algebra, which, as discussed in the main text, forms a representation of the Drinfeld double $D(G)$.

If $r$ is closed, so that $\partial_0 r = \partial_1 r$, then $\mathfrak{F}_r$ contains a nontrivial subalgebra of codespace-preserving operators
\begin{align}
	\mathfrak{K}_r \coloneqq \big\{F \in \mathfrak{F}_r : [F, \mathbb{A}_v] = [F, \mathbb{B}_p] = 0\, \forall v \in V, p \in P\big\}
\end{align}
Since the projectors $\mathbb{A}_v$, $\mathbb{B}_p$ are Hermitian, this is also a von Neumann subalgebra, which again constitutes a representation of $D(G)$. Our aim is now to decompose $\mathfrak{K}_r$ as a direct sum of irreps $\mu \in \mathrm{Irr}(D(G))$, which have the following structure. Each irrep $\mu$ can be identified with a pair $(C, \pi) \equiv \mu$, where $C \in [G]_{\text{cj}}$ is a conjugacy class of $G$, of which a fixed representative $c \in C$ is chosen, and $\pi$ is a group representation of the centralizer $Z(c) = \{g \in G : gc = cg\}$. For each $x \in C$, we also choose a $q_x \in G$ such that $x = q_x c q_x^{-1}$. Then define the operators
\begin{align}
	K^{\mu}_r \equiv K^{(C, \pi)}_r \coloneqq \frac{d_\pi}{|Z(c)|}\sum_{D \in [Z(c)]_{\text{cj}}} \bar{\chi}_\pi(D) \sum_{d \in D} \sum_{x \in C} F^{q_x d q_x^{-1}, x}_r
\end{align}
where the outer sum is over conjugacy classes of $Z(c)$, and $\chi_{\pi}(D)$ is the character of $\pi$ at $D$. As shown in Appendix B of Ref.~\onlinecite{bombin2008}, these operators are orthogonal projectors, and form a linear basis for $\mathfrak{K}_r$. Thus, $\mathfrak{K}_r$ is a commutative algebra. 

As in the $\mathbb{Z}_2$ case, if $r$ is a contractible loop then operators in $\mathfrak{K}_r$ have a trivial effect on code states, $K^{(C, \pi)}_r \ket{\psi} = \delta_{C, 1} \delta_{\pi, \text{triv}}\ket{\psi},\; \forall \ket{\psi} \in \mathscr{H}_{\mathrm{code}}$. However, if $r$ is a non-contractible loop, then $\mathfrak{K}_{r}$ can act non-trivially on the logical space. In particular, the restriction of $K^{\mu}_r$ on the code space, $K^{\mu}_r \mathbb{\Pi}$, turns out to be a rank-1 projector onto a particular code state $\ket{\bar{\mu}}_{Q}$ (provided $r$ is in the homology class $(1,0)$ or $(0,1)$, corresponding to a single cycle around the $x$- or $y$-direction) \cite{simon2023topological}. 

Open ribbon operators can generate anyons out of the vacuum. In particular, in light of Eq.~\eqref{eq:open ribbon commute}, we see that if $\ket{\psi}$ is a codestate, then $F^{h,g}_r\ket{\psi}$ can only contain anyons at the endpoints of the ribbon $\partial_i r$, $i = 0, 1$. As with the closed ribbon operators, particular linear combinations of the basis operators $F^{h,g}_r$ can be chosen that correspond to a particular $\mu \equiv (C, \pi)  \in \mathrm{Irr}(D(G))$. These operators create a $\mu$-$\bar{\mu}$ pair, where $\bar{\mu}$ is the anyon dual to $\mu$ (corresponding to the inverse conjugacy class $C^{-1}$ and the dual representation $\bar{\pi}$).

More precisely, these operators are denoted $F^{(C, \pi), (c, j),(c',j')}_r$, where $c, c' \in C$ are elements of the conjugacy class $C$, and the indices  $j, j'  \in [d_\pi]$ label basis states within the irrep $\pi$. They have the form \cite{bombin2008}
\begin{align}
	F_r^{(C, \pi), (c, j),(c',j')} = \frac{d_\pi}{|Z(c)|}\sum_{x \in Z(c)} \bar{\pi}(x)_{j,j'} F^{c^{-1}, q_{c^{\vphantom{\prime}}} x q_{c'}^{-1}}_r
\end{align}
When acting on a code state $\ket{\psi} \in \mathscr{H}_{\mathrm{code}}$, the resulting (properly normalised) state
\begin{align}
    \ket{\psi^{\mu, (c, j),(c',j')}_r} \coloneqq\frac{F_r^{\mu, (c, j),(c',j')}\ket{\psi}}{\|F_r^{\mu, (c, j),(c',j')}\ket{\psi}\|_2}
    \label{eq:psi anyon}
\end{align}
has the desired property of hosting a $\mu$-$\bar{\mu}$ at the endpoints $\partial_0 r$, $\partial_1 r$, where $\mu = (C, \pi)$.  One can verify that the norm in the denominator of Eq.~\eqref{eq:psi anyon} takes the value
\begin{align}
    \|F_r^{\mu, (c, j),(c',j')}\ket{\psi}\| = \eta_{\mu} \coloneqq \sqrt{\frac{d_\pi|C|}{|G|^2}}
\end{align}
for $\mu = (C, \pi)$. This is independent of the labels $(c, j)$, $(c',j')$. In particular, for the trivial anyon $\varnothing = (\{1\}, \text{triv})$, we have $\ket{\psi^{\mu, (c, j),(c',j')}_r} =\ket{\psi} \in \mathscr{H}_{\mathrm{code}}$. 

It turns out that, using local transformations, the state \eqref{eq:psi anyon} can be transformed to any other state for which $(c, j)$, $(c',j')$ differ, but $(C, \pi)$ is the same. Thus, the labels $(c, j)$, $(c',j')$ should be thought of as describing `non-topological' local degrees of freedom.  In contrast, the presence of the $\mu$, $\bar{\mu}$ anyon pair is a topological property of the state that cannot be altered by local means. Since our interest is in this latter degree of freedom, we will assume that for each $\mu \in \text{Irr}(D(G))$, a particular choice of $c, j, c', j'$ has been made, and we will simply label the states in question by $\mu$. In particular, the state \eqref{eq:anyon state} referred to in the main text will be given by
\begin{align}
    \rho^{(v, v'), \mu}_Q = \frac{1}{\eta_\mu^2}F_r^{\mu, (c, j),(c',j')} \rho_Q (F_r^{\mu, (c, j),(c',j')})^\dagger
    \label{eq:anyon state full}
\end{align}
for the chosen reference values of $c, j, c', j'$ associated with the given $\mu$. Here, when declaring that the anyons are at the vertices $v, v'$, we mean that the ribbon $r$ has endpoints $\partial_0 r = (v, p)$, $\partial_1 r = (v',p')$, where the plaquettes $p, p'$ are to the immediate North East of $v,v'$, respectively.

If $r$ is a closed, contractible ribbon that encircles the vertex $v$ (which hosts the $\mu$ anyon) but not $v'$ (which hosts the $\bar{\mu}$ anyon), then \cite{bombin2008}
\begin{align}
    K^{\mu'}_{r'} \rho^{(v, v'), \mu}_Q = \delta_{\mu,\mu'} \rho^{(v, v'), \mu}_Q
    \label{eq:anyon distinguish ribbon}
\end{align}
That is, the closed contractible ribbon operators $K^{\mu'}_{r'}$ detect whether there is an anyon of type $\mu'$ in the interior of $r'$. These operators realise a particular 1-form symmetry \cite{mcgreevy2023generalized,cordova2025, song2026strong} that emerges in the quantum double model. Often the symmetry is specified in terms of the transformed operators
\begin{align}
    \tilde{K}^{(C, D)}_r &\coloneqq \sum_{\pi \in \text{Irr}(Z(c))} \frac{|D|}{d_\pi} \chi_{\pi}(D) K^{(C, \pi)}_r & \forall D \in [Z(c)]_{\text{cj}}
\end{align}
in light of the fact that these operators obey the product rule \cite{bombin2008}
\begin{align}
    \tilde{K}^{(C, D)}_r \tilde{K}^{(C, D')}_r = \sum_{D'' \in [Z(c)]_{\text{cj}}} N^{D''}_{D D'} \tilde{K}^{(C, D'')}_r
\end{align}
with the integer-valued fusion coefficients $N^{D''}_{D D'} \geq 0$. This elucidates the non-invertible (or categorical) nature of this 1-form symmetry \cite{kong2020, mcgreevy2023generalized, song2026strong}. 
Corollary \ref{cor:anyon-distinguishability} in the main text demonstrates that, below the threshold noise strength, the noise-corrupted state continues to exhibit an emergent 1-form symmetry, with the bare ribbon operators $K^\mu_r$ replaced by thickened versions.

\section{Proof of Theorem \ref{thm:fat-nonab} for general non-Abelian codes \label{apdx:bounds}}

In this appendix, we present the proof of Theorem 1 in its most general form, covering the $\mathbb{Z}_2$ surface code, quantum double models, and Levin-Wen string net models. At a high level, our approach follows the same overall logic as outlined in Section \ref{sec:upper bound simple}, making use of the methods described in Section \ref{sec:methods}. 

For the quantum double model on a $L_{\rm x} \times L_{\rm y}$ square lattice with periodic boundary conditions, we let $S \subset Q$ be a $L_{\rm x} \times w_{\rm y}$ horizontal strip of qubits with rough boundaries. For a string-net code on a $L_{\rm x} \times L_{\rm y}$ hexagonal lattice, we also take $S \subset Q$ to be a horizontal strip of qudits. If $\mathscr{C}$ contains fusion multiplicities, there are also qudits at the vertices, in which was we understand $S$ to contain all vertex qudits for which all adjacent edges are also in $S$. 
We let $\{\ket{l_\mu}_Q\}_{\mu \in \mathrm{Irr}(D(G))}$ be the basis of logical states that are eigenstates of the logical operators that wind non-trivially in the $x$ direction, labeled by anyon types $\mu \equiv (C, \pi) \in \mathrm{Irr}(D(G))$, of which there are $d_R = |\mathrm{Irr}(D(G))|$. To show that the corresponding subalgebra $\mathfrak{K}_\leftrightarrow$ remains recoverable, we seek to upper bound the quantity 
\begin{align}
	\delta_{\mathfrak{K}_\leftrightarrow}(E_S) \coloneqq  \norm{\widehat{\mathcal{N}}_{E_S \leftarrow Q}\circ \mathcal{C}^{\mathcal{C}}_{Q \leftarrow R} - \widehat{\mathcal{N}}_{E_S \leftarrow Q}\circ \mathcal{C}^{\mathcal{C}}_{Q \leftarrow R}\circ \mathcal{P}_{R}^{\mathfrak{K}_\leftrightarrow}}_\diamond
	\label{eq:code distinguish parameter}
\end{align}
which was introduced in Section \ref{subsec:recoverability conditions}. Combined with \cref{lem:effective operator}, this suffices to prove \cref{thm:fat}.

Our upper bound on $\delta_{\mathfrak{K}_\leftrightarrow}(E_S)$ will be derived by proving two separate results, Lemmas \ref{lem:bound-norm-connectivity} and \ref{lem:indist q double}, which formalise the arguments given in Section \ref{subsec:local global indist}. The latter result captures the fact that the outputs of the channels $\mathcal{C}^{\mathcal{C}}_{Q \leftarrow R}$ and $ \mathcal{C}^{\mathcal{C}}_{Q \leftarrow R}\circ \mathcal{P}_{R}^{\mathfrak{K}_\leftrightarrow}$ are locally indistinguishable on particular (small) subregions of $S$. The former result captures the intuition that the complementary channel $\widehat{\mathcal{N}}_{E_S \leftarrow Q}$ suppresses operators that are supported on large subregions, which can lead to approximate global indistinguishability of $\widehat{\mathcal{N}}_{E_S \leftarrow Q}\circ \mathcal{C}^{\mathcal{C}}_{Q \leftarrow R}$ and $\widehat{\mathcal{N}}_{E_S \leftarrow Q}\circ \mathcal{C}^{\mathcal{C}}_{Q \leftarrow R}\circ \mathcal{P}_{R}^{\mathfrak{K}_\leftrightarrow}$.

\subsection{From local to global indistinguishability}

The individual results outlined above can be more easily formalised once we introduce a particular definition of local indistinguishability of two channels.

\begin{definition}[Local indistinguishability]
	For a $n$-qudit system $Q$, let $A \subset[n]$, $B = [n] \setminus A$ be a bipartition of the qudits, and let $\mathscr{W} \subset \mathscr{P}(A)$ be a collection of subsets of $A$. Two CPTP maps $\mathcal{F}_{Q \leftarrow R}, \mathcal{G}_{Q \leftarrow R}$ with input space $R$ and output space $Q$ are said to be $(B,\mathscr{W})$-indistinguishable if the following statement holds. For any $\Sigma \subset A$ such that $W \nsubseteq \Sigma$ for every $W \in \mathscr{W}$, we have
	\begin{align}
		\tilde{\mathcal{F}}_{\Sigma B \leftarrow R} = \tilde{\mathcal{G}}_{\Sigma B \leftarrow R},
		\label{eq:indistinguishable}
	\end{align}
	where $\tilde{\mathcal{F}}_{\Sigma B \leftarrow R} \coloneqq \tr_{[n] \setminus \Sigma B} \circ \mathcal{F}_{Q \leftarrow R}$, and similar for $\tilde{\mathcal{G}}_{\Sigma B \leftarrow R}$.
	\label{def:locally-indistinguishable}
\end{definition}
The maps appearing Eq.~\eqref{eq:indistinguishable} describe the relationship between the input state on $R$ and the marginal of the output state on $\Sigma B$. Operationally, this condition implies that in a black-box scenario where either $\mathcal{F}_{Q \leftarrow R}$ or $\mathcal{G}_{Q \leftarrow R}$ is applied at some point during an experiment, an observer with access to the outputs $B$ must also gain access to at least one of the regions $W \in \mathscr{W}$ in order to decide which channel was applied. Note that if we choose $R$ to be a trivial Hilbert space, $R = \varnothing$, then the channels are equivalent to density matrices, $\mathcal{F}_{Q \leftarrow R} \equiv \rho_Q$, $\mathcal{G}_{Q \leftarrow R} \equiv \sigma_Q$ for some $\rho_Q, \sigma_Q$. This definition thus also gives us a notion of local indistinguishability for states.

As an example, the relation \eqref{eq:sigma constraint} for the $\mathbb{Z}_2$ surface code is a manifestation of local indistinguishability. In this instance, we set $A = S$, $B = S^c$, and $\mathscr{W}$ contains all subsets of qubits that constitute a spanning walk. Then, the maps $\mathcal{F}_{Q \leftarrow R} = \mathcal{C}_{Q \leftarrow R}$ and $\mathcal{G}_{Q \leftarrow R} = \mathcal{C}_{Q \leftarrow R}\circ \mathcal{P}_R^{\mathfrak{X}}$ are $(S^c, \mathscr{W})$-indistinguishable.\\

We now show that if $\mathcal{F}_{Q \leftarrow R}, \mathcal{G}_{Q \leftarrow R}$ are $(\mathscr{W}, B)$-indistinguishable, then when a sufficiently noisy channel is applied to each output qudit in $A = Q \setminus B$, the resulting composite channels become approximately \textit{globally} indistinguishable.

\begin{lemma} 
	Consider a system of qudits arranged at the vertices of a graph $\mathscr{G} = (V, E)$ of maximum degree $k$, which is divided into complementary regions $A \subset V$ and $B = V \setminus A$. Let $\mathscr{W} \subset \mathscr{P}(A)$ be a collection of subsets of $A$, each of which is connected with respect to $\mathscr{G}$, and define $d_{\textup{min}}(\mathscr{W}) \coloneqq \min_{W \in \mathscr{W}} |W|$. Let $\mathcal{N}_A = \bigotimes_{v \in A} \mathcal{N}_{Q'_v \leftarrow v}$ be a product channel acting on $A$, with output space $Q_v'$ for each qudit $v$, and define
	\begin{align}
		\max_{v \in A}\big\|\mathcal{N}_{Q'_v \leftarrow v}\circ (\mathrm{id}_v - \mathcal{D}_v)\big\|_\diamond \eqqcolon e^{-\beta}.
        \label{eq:beta def lemma}
	\end{align}
	where $\mathcal{D}_v$ is the depolarizing channel on qudit $v$.
	
	Suppose that the pair of channels $\mathcal{F}_{Q \leftarrow R}, \mathcal{G}_{Q \leftarrow R}$ are ($B$, $\mathscr{W}$)-indistinguishable. 
	Then, for any $\beta > \beta^*_k$, where $\beta^*_k \coloneqq 1 + \ln(k-1)$ depends only on $k$, there exists a constant $\kappa > 0$ such that
	\begin{align}
		\label{eq:lemma distinguish}
		\left\|(\mathcal{N}_A\otimes \mathrm{id}_B) \circ \mathcal{F}_{Q \leftarrow R}  - (\mathcal{N}_A\otimes \mathrm{id}_B) \circ \mathcal{G}_{Q \leftarrow R}\right\|_\diamond \leq \kappa |A| e^{-(\beta - \beta^*_k) d_{\textup{min}}(\mathscr{W})}.
	\end{align}\label{lem:bound-norm-connectivity}
\end{lemma}
\noindent Here, the composite channels $(\mathcal{N}_A\otimes \mathrm{id}_B) \circ \mathcal{F}_{Q \leftarrow R}$ and $(\mathcal{N}_A\otimes \mathrm{id}_B) \circ \mathcal{G}_{Q \leftarrow R}$ are approximately globally indistinsguishable if 1) the channel $\mathcal{N}_A$ is sufficiently noisy, as quantified by the parameter $\beta$, and 2) the regions $W \in \mathscr{W}$ are all large, i.e.~the channels cannot be distinguished by looking at $\Sigma B$ for small $\Sigma$. 

This result will be used to upper bound the quantity $\delta_{\mathfrak{J}}(E_S)$ [Eq.~\eqref{eq:algebra comp error}] for the various different codes considered in this work. In each case, we identify a local indistinguishability property that applies to the given code (see \cref{lem:indist q double}), and identify a bounded-degree graph $\mathscr{G}$ whose vertices are the physical qudits, such that each $W \in \mathscr{W}$ is connected with respect to $\mathscr{G}$. This will suffice to bound the threshold for the given code. Indeed, combined with the local indistinguishability property of the $\mathbb{Z}_2$ surface code, Eq.~\eqref{eq:sigma constraint}, the above suffices to prove \cref{lem:delta bound simple}.

While having the advantage of being very general, the numerical value of the threshold $\beta_k^*$ implied by this result may not always be optimal. In Appendix \ref{apdx:z2 improvements}, we show how the argument can be refined to quantitatively improve our bound on the threshold noise strength.

\subsection{Application to quantum double and string-net models}

In order to apply this result to the non-Abelian codes being considered here, we need to show that the noiseless code states in question obey an indistinguishability property of the kind specified in \cref{def:locally-indistinguishable}. Since qudits in the quantum double and string-net models live on the edges (and possibly also vertices) of the square and hexagonal lattices, respectively, we must first specify an graph for each of these, i.e.~a graph whose vertices are the appropriate set of qudits, with some rule for deciding whether two qudits are considered adjacent. We can then identify an appropriate choice for $\mathscr{W}$. 

For the quantum double model, letting $(V_{\mathrm{sq}}, E_{\mathrm{sq}}, F_{\mathrm{sq}})$ be the vertices, edges, and faces of the square lattice, we define the edge graph $G^\#_{\mathrm{sq}} = (E_{\mathrm{sq}}, L_{\mathrm{sq}})$, whose vertices are edges $e \in E_{\mathrm{sq}}$, and whose adjacencies are determined via
\begin{align}
	G^\#_{\mathrm{sq}} = (E_{\mathrm{sq}}, L_{\mathrm{sq}}) \;:\; \{e, e'\} \in L_{\mathrm{sq}} \text{ iff } e \overset{V}{\sim}e' \text{ or }e \overset{F}{\sim}e'  
	\label{eq:edge graph square}
\end{align}
where the relations $\overset{V}{\sim}$ $(\overset{F}{\sim})$ will indicate that two edges share a vertex (face) from here on. We will also use the notation $G^\#_{\mathrm{sq}}|_S$ to denote the induced subgraph of $G^\#_{\mathrm{sq}}$ obtained by removing vertices outside of $S \subset E$. The degree of the edge graph $G^\#_{\mathrm{sq}}$ is $k_{\text{ sq}} = 8$. 

For the string net model without fusion multiplicities, which is defined on the hexagonal lattice $(V_{\mathrm{hex}}, E_{\mathrm{hex}}, F_{\mathrm{hex}})$, we will use a different edge graph whose connectivity is higher than the na{\"i}ve generalisation of \eqref{eq:edge graph square}. This is necessary to deal with the fact that the logical operators for this model are thicker.  To specify this graph, we first introduce the following discrete distance measure between edges on the hexagonal lattice. Let $E^+ \coloneqq E_{\mathrm{hex}} \cup F_{\mathrm{hex}}$ and define the graph $G' = (E^+, D)$ as follows. For two edges $e, e' \in E_{\mathrm{hex}}$, we have $\{e, e'\} \in D$ if $e$ and $e'$ share a vertex. For an edge $e \in E_{\mathrm{hex}}$ and a face $f \in F_{\mathrm{hex}}$,  we have $\{e, f\} \in D$ if $e$ is contained in $f$. For any two faces $f, f' \in F_{\mathrm{hex}}$, we have $\{f, f'\} \notin D$. Then, $d_{G'}(e, e')$ is defined as the graph distance between $e$ and $e'$ within $G'$. Now we define a new edge graph $G^\#_{\mathrm{hex}} = (E_{\mathrm{hex}}, L_{\mathrm{hex}})$ as
\begin{align}
	G^\#_{\mathrm{hex}} = (E_{\mathrm{hex}}, L_{\mathrm{hex}}) \; : \: \{e, e'\} \in L_{\mathrm{hex}}\text{ if } d_{G'}(e, e') \leq 3.
	\label{eq:edge graph hex}
\end{align}
The degree of this edge graph can be shown to be $k_{\rm hex} = 26$.

For string-net models with fusion multiplicities, there are additional qudits at each vertex of the hexagonal lattice. Accordingly, we define an enlarged graph
\begin{align}
    G_{\text{hex, mult}}^\# = \big((V_{\text{hex}} \cup E_{\text{hex}} ) , L'_{\text{hex}}\big)
\end{align}
with $L'_{\text{hex}}$ defined as follows. The adjacency rule for edges $\ell, \ell' \in E_{\text{hex}}$ is the same as for $L_{\mathrm{hex}}$. Two vertices $v,v' \in V_{\text{hex}}$ are adjacent if they are both contained in the same plaquette. An edge $\ell \in E_{\text{hex}}$ and a vertex $v \in V_{\text{hex}}$ are adjacent if there is some other edge $\ell'$ that contains $v$, such that $d_{G'}(\ell, \ell') \leq 2$.  The degree of $G_{\text{hex, mult}}^\#$ can be shown to be $k_{\text{hex, mult}} = 40$.

\begin{lemma}
	Let $\mathcal{C}_{Q \leftarrow R}$ be the isometric encoding channel for the quantum double (string net) model on the square (hexagonal) lattice with periodic boundary conditions, and let $\mathfrak{K}_\leftrightarrow \subset \mathfrak{B}(\mathscr{H}_R)$ be the subalgebra of logical operators that are supported on ribbons that wind in the $x$ direction, with corresponding super-projector $\mathcal{P}^{\mathfrak{K}_\leftrightarrow}_R$.
	Let $S$ be a $L_{\rm x} \times w_{\rm y}$ horizontal strip of qudits with rough edges, with complement $S^c = E_{\mathrm{sq}} \setminus S$. Define $\mathscr{W} \subset P(S)$ to be the collection of graph paths on the relevant subgraph $G^\#_{\mathrm{sq}}|_S$, $G^\#_{\mathrm{hex}}|_S$, $G^\#_{\mathrm{hex, mult}}|_S$ that connect the top and bottom boundaries of $S$. Then, the channels
	\begin{align}
		\mathcal{F}_{Q \leftarrow R} &\coloneqq \mathcal{C}_{Q \leftarrow R} \;\;\text{ and }\;\; \mathcal{G}_{Q \leftarrow R} \coloneqq \mathcal{C}_{Q \leftarrow R} \circ \mathcal{P}^{\mathfrak{K}_\leftrightarrow}_R
		\label{eq:channel choice qd}
	\end{align}
	are $(S^c, \mathscr{W})$-indistinguishable.
	\label{lem:indist q double}
\end{lemma}

Together with \cref{lem:bound-norm-connectivity}, this result swiftly implies \cref{thm:fat} for non-Abelian codes (see Section \ref{subsec:nonab results}), via the following argument. The reduced noise channel $\mathcal{N}_{S \leftarrow Q} = \tr_{S^c} \circ \mathcal{N}_Q$ has complement $\widehat{\mathcal{N}}_{E_S \leftarrow Q} = \left(\bigotimes_{\ell \in S^c} \mathrm{id}_{\ell}\right) \otimes (\bigotimes_{\ell \in S} \widehat{\mathcal{N}}_{E_\ell \leftarrow \ell} )$ [Eq.~\eqref{eq:noise channels subsystem}]. Thus, $\delta_{\mathfrak{K}_\leftrightarrow}$ is precisely in the form of the left hand side of \eqref{eq:lemma distinguish}, with $\mathcal{F}_{Q \leftarrow R}$ and $\mathcal{G}_{Q \leftarrow R}$ specified as in Eq.~\eqref{eq:channel choice qd}. \cref{lem:indist q double} implies that the required indistinguishability condition is met.

For the quantum double model, it is straightforward to show that for the chosen $\mathscr{W}$, we have $d_{\mathrm{min}}(\mathscr{W}) = w_{\rm y}$, and the underlying graph $G^\#_{\mathrm{sq}}$ has degree 8. Thus, for $\beta > \beta_c$, where $\beta_c = 1 + \ln(7)$, we have
\begin{align}
	\delta(\mathfrak{K}_\leftrightarrow) \leq \kappa \mathrm{poly}(L_{\mathrm{x}} w_{\mathrm{y}}) e^{-(\beta - \beta_c)w_y}
    \label{eq:delta kx apdx}
\end{align}
for some constant $\kappa > 0$.
corresponding to a critical noise strength $T_\star^{\mathrm{QD}} = [1 + \ln(7)]^{-1} \approx 0.339$. Similarly, for the string net model without fusion multiplicities, we have $d_{\mathrm{min}}(\mathscr{W}) \geq w_{\rm y}/2$,  and the underlying graph $G^\#_{\mathrm{sq}}$ has degree 26, corresponding to a critical noise strength $T_\star = [1 + \ln(25)]^{-1} \approx 0.237$. With fusion multiplicities, this changes to $T_\star = [1 + \ln(39)]^{-1} \approx 0.214$. \hfill $\square$

\subsection{Proof of \cref{lem:bound-norm-connectivity}}

As in Section \ref{sec:upper bound simple}, we start with the trivial resolution of the identity map on a single qudit,
\begin{align}
    \mathrm{id}_v = \mathcal{D}_v + (\mathrm{id}_{v} - \mathcal{D}_v).
    \label{eq:triv identity app}
\end{align}
and by analogy to Eq.~\eqref{eq:Delta def}, define the map
\begin{align}
    \mathcal{\Delta}_{Q \leftarrow R} \coloneqq \mathcal{F}_{Q \leftarrow R} - \mathcal{G}_{Q \leftarrow R}.
\end{align}
Applying  Eq.~\eqref{eq:triv identity app} to the output of $\mathcal{\Delta}_{Q \leftarrow R}$ for every $v \in A$, and expanding the terms, we obtain
\begin{align}
	\mathcal{\Delta}_{Q \leftarrow R} &= \sum_{\Sigma \in \mathscr{P}(A)} \mathcal{\Delta}_{Q \leftarrow R}^{\,(\Sigma)} & \text{where } \mathcal{\Delta}_{Q \leftarrow R}^{\,(\Sigma)} \coloneqq \left(\left[\bigotimes_{v \in \Sigma} (\mathrm{id}_v - \mathcal{D}_v) \right] \otimes  \left[\bigotimes_{v \in A \setminus \Sigma}  \mathcal{D}_v \right]\right)\circ \mathcal{\Delta}_{Q \leftarrow R} \label{eq:channel sum J}
\end{align}
where $\mathscr{P}(X) = \{Y : Y \subseteq X\}$ denotes the powerset of a set $X$. By the definition of $(B, \mathscr{W})$-indistinguiahsbility and the assumptions of the lemma, we have
\begin{align}
    W \nsubseteq \Sigma \,\forall W \in \mathscr{W} \; \Longrightarrow \; \mathcal{\Delta}_{Q \leftarrow R}^{\,(\Sigma)} = 0,
    \label{eq:sigma indist condition}
\end{align}
since, writing $\pi_{\Sigma^\mathrm{c}} \coloneqq \bigotimes_{v \in \Sigma^\mathrm{c}}(I_v/d_v)$, one can write $\mathcal{\Delta}_{Q \leftarrow R}^{\,(\Sigma)} = \pi_{\Sigma^\mathrm{c}} \otimes [\bigotimes_{v \in \Sigma}(\mathrm{id}_v - \mathcal{D}_v)] \circ (\mathcal{F}_{\Sigma B \leftarrow R} - \mathcal{G}_{\Sigma B \leftarrow R})$. This vanishes for any $\Sigma$ satisfying the condition on the left side of \eqref{eq:sigma indist condition}, by the definition of $(B, \mathscr{W})$-indistinguiahsbility [Eq.~\eqref{eq:indistinguishable}].

For a given $\Sigma \subseteq A$, let $\Gamma[\Sigma] \coloneqq \{C_1, \ldots, C_{|\Gamma[\Sigma]|}\}$ denote the maximally connected components of $\Sigma$ with respect to the graph $G$. This defines a map $\Gamma : \mathscr{P}(A) \rightarrow \mathscr{P}(\mathscr{P}(A))$, whose image we denote $\Omega \coloneqq \text{Im}(\Gamma)$. Now, define $\mathscr{W}^+ \subset \mathscr{P}(A)$ to be the set of connected subregions of $A$ that contain at least one $W \in \mathscr{W}$. We then denote $\Gamma^\mathscr{W}[\Sigma] \coloneqq \Gamma[\Sigma] \cap \mathscr{W}^+$, which contains all the connected components of $\Sigma$ that are supersets of some $W \in \mathscr{W}$; similarly $\Omega^\mathscr{W} \coloneqq \text{Im}(\Gamma^\mathscr{W}) = \Omega \cap \mathscr{P}(\mathscr{W}^+)$. We now claim that $\Gamma^\mathscr{W}[\Sigma] = \varnothing \Rightarrow W \nsubseteq \Sigma \,\forall W \in \mathscr{W}$, which, by the relation \eqref{eq:sigma indist condition} then implies
\begin{align}
	\Gamma^\mathscr{W}[\Sigma] = \varnothing \; \Longrightarrow \; \mathcal{\Delta}_{Q \leftarrow R}^{\,(\Sigma)} = 0
	\label{eq:NJ vanish}
\end{align}
To prove our claim, suppose for a given $\Sigma$ there were some $W \in \mathscr{W}$ such that $W \subseteq \Sigma$. Then by virtue of $W$ being connected, this would imply that one of the connected components $C \in \Gamma[\Sigma]$ would contain $W$. Thus, we would have $C \in \mathscr{W}^+$ and $C \in \Gamma[\Sigma]$, implying $\Gamma^\mathscr{W}[\Sigma] \coloneqq \Gamma[\Sigma] \cap \mathscr{W}^+ \neq \varnothing$. The contrapositive of this relation gives the desired claim.


As a result, for every $\Sigma \in \mathscr{P}(A)$ that contributes to the sum \eqref{eq:channel sum J}, there must be some connected region  $W^+ \in \mathscr{W}^+$ such that $W^+ \in \Gamma^\mathscr{W}[\Sigma]$. This motivates us to consider the partial sum $\mathcal{\Delta}^{\{W^+\}}_{Q \leftarrow R} = \sum_{\Sigma : \Gamma^\mathscr{W}[\Sigma] \ni W^+} \mathcal{\Delta}^{(\Sigma)}_{Q \leftarrow R}$ [the analogue of Eq.~\eqref{eq:delta singleton}], which includes the contribution from all $\Sigma$ that have $W^+$ as one of their connected components. Then, as outlined in Section \ref{sec:upper bound simple}, the sum $\sum_{W^+ \in \mathscr{W}^+} \mathcal{\Delta}^{\{W^+\}}_{Q \leftarrow R}$ [Eq.~\eqref{eq:delta all Cplus}] contains every nonzero $\mathcal{\Delta}^{(\Sigma)}_{Q \leftarrow R}$, but with some double-counting. To fix this double-counting, we must generalise the object $\mathcal{\Delta}^{\{W^+\}}_{Q \leftarrow R}$. For any collection of connected regions $\Gamma' = \{C_1^+, \ldots, C_n^+ \} \in \Omega^{\mathscr{W}}$, we define
\begin{align}
    \mathcal{\Delta}_{Q \leftarrow R}^{\Gamma'} \coloneqq \sum_{\substack{\Sigma \in \mathscr{P}(A)\\ \Gamma[\Sigma] \supseteq \Gamma'}} \mathcal{\Delta}_{Q \leftarrow R}^{(\Sigma)}
    \label{eq:delta superset}
\end{align}
which includes all $\Sigma$ whose connected components are a superset of $\Gamma'$, i.e.~$C_1^+, \ldots, C_n^+ \in \Gamma^{\mathscr{W}}[\Sigma]$. Then, by way of the inclusion-exclusion principle \cite{Cameron1994combinatorics}, we claim that
\begin{align}
    \mathcal{\Delta}_{Q \leftarrow R} = \sum_{\Sigma \in \mathscr{P}(A)}  \mathcal{\Delta}_{Q \leftarrow R}^{(\Sigma)} = \sum_{n=1}^{|A|} (-1)^{n+1} \sum_{\substack{\Gamma' \in \Omega^{\mathscr{W}}\\|\Gamma'|=n }} \mathcal{\Delta}^{\Gamma'}_{Q \leftarrow R}.
	\label{eq:sum J decomp}
\end{align}
Here, the second sum on the right hand side is over all sets $\Gamma' = \{C_1^+, \ldots, C_n^+ \}$ of cardinality $n$. To prove Eq.~\eqref{eq:sum J decomp}, consider an arbitrary $\Sigma \in \mathscr{P}(A)$, and denote $m \coloneqq |\Gamma^{\mathscr{W}}[\Sigma]|$. If $m = 0$, then by Eq.~\eqref{eq:NJ vanish} $\mathcal{\Delta}_{Q \leftarrow R}^{(\Sigma)} = 0$, so we can restrict to the case $m \geq 1$. Inserting the definition \eqref{eq:delta superset} into right hand side of \eqref{eq:sum J decomp}, the resulting sum will feature a contribution $\mathcal{\Delta}_{Q \leftarrow R}^{(\Sigma)}$ for every nonempty $\Gamma' \subseteq \Gamma^{\mathscr{W}}[\Sigma]$. There are ${m \choose n}$ such subsets $\Gamma'$ of a given cardinality $|\Gamma'| = n$, each appearing with sign $(-1)^{n+1}$. Thus, the net multiplicity of the term corresponding to a given $\Sigma$ is $\sum_{n=1}^m {m \choose n}(-1)^{n+1} = 1$, as required.

Now, for a given $\Gamma' \in \Omega^\mathscr{W}$, denote $X[\Gamma'] \coloneqq \bigcup_{W \in \Gamma'} W$ and also define $Y[\Gamma'] \coloneqq \partial(X[\Gamma'])$, where for a subset $Z \subset A$, the boundary $\partial Z \subset A$ is the set of vertices that are not in $Z$, but adjacent to some $v \in Z$. Since $\Gamma' \in \Omega^\mathscr{W}$, there exists a $\Sigma' \in \mathscr{P}(A)$ such that $\Gamma^\mathscr{W}[\Sigma'] = \Gamma'$. Thus, any two distinct elements $W^+_1, W^+_2 \in \Gamma'$ must be non-overlapping and non-adjacent---otherwise $C_1^+ \cup C_2^+$ would be a connected subset of $\Sigma'$, contradicting the fact that $C_1^+, C_2^+$ are maximally connected subsets of $\Sigma'$. As a result, the boundary of $X[\Gamma']$ is precisely the union of the boundaries of each component, $Y[\Gamma'] = \bigcup_{W^+ \in \Gamma'}\partial W^+$. 
We can then derive the following implication for any $\Gamma' \in \Omega^\mathscr{W}$.
\begin{align}
	\Gamma[\Sigma] \supseteq \Gamma'\quad &\Longleftrightarrow \quad X[\Gamma'] \subseteq \Sigma \text{ and } Y[\Gamma'] \cap \Sigma = \varnothing. 
	\label{eq:Gamma J cond}
\end{align}
If we assume $\Gamma' \subseteq \Gamma[\Sigma]$, then each $W^+ \in \Gamma'$ is a connected component of $\Sigma$; thus $X[\Gamma'] = \bigcup_{W^+ \in \Gamma'}W^+ \subseteq \Sigma$. Furthermore, for every $W^+ \in \Gamma[\Sigma]$, there can be no vertex $v \in \partial W^+$ that is contained in contained in $\Sigma$, since otherwise $\{v\}\cup W^+$ would be a connected subset of $\Sigma$, contradicting the fact that $W^+$ is a maximally connected component of $\Sigma$.
For the reverse implication, assume the right hand side of \eqref{eq:Gamma J cond}. Then, for every $W^+ \in \Gamma'$, we have $W^+ \subseteq \Sigma$ and $\partial W^+ \cap \Sigma = \varnothing$. Thus, $W^+$ is a connected subset of $\Sigma$ whose boundary has no overlap with $\Sigma$; this suffices to show that $W^+ \in \Gamma^{\mathscr{W}}[\Sigma]$.

Thanks to the relation \eqref{eq:Gamma J cond} we can write
the map \eqref{eq:delta superset} as
\begin{align}
    \mathcal{\Delta}_{Q \leftarrow R}^{\Gamma'} =\sum_{\substack{\Sigma \in \mathscr{P}(A) \\   \Sigma \supseteq X[\Gamma'] \text{ and } \Sigma \cap Y[\Gamma'] = \varnothing }} \mathcal{\Delta}_{Q \leftarrow R}^{(\Sigma)} &=  \sum_{ \Sigma' \subseteq A \setminus (X[\Gamma'] \cup Y[\Gamma'])} \mathcal{\Delta}_{Q \leftarrow R}^{(\Sigma' \cup X[\Gamma'])} \nonumber\\ &= \left(\left[\bigotimes_{v \in X[\Gamma']} (\mathrm{id}_v - \mathcal{D}_v) \right] \otimes  \left[\bigotimes_{v \in Y[\Gamma']}  \mathcal{D}_v \right]\right) \circ \mathcal{\Delta}_{Q \leftarrow R}.
    \label{eq:delta gamma xy}
\end{align}
where, in going to the last line, we invoke the definition \eqref{eq:delta decomp}, along with the fact that $\sum_{\Sigma' \subseteq Z} \big(\bigotimes_{v \in \Sigma'}(\mathrm{id}_v - \mathcal{D}_v)\big) \otimes \big(\bigotimes_{v \in Z \setminus \Sigma'} \mathcal{D}_v \big) = \mathrm{id}_{Z}$ for a given region $Z$.

Finally, we are ready to upper bound the norm \eqref{eq:lemma distinguish}, which can be written as $\|(\mathcal{N}_A \otimes \mathrm{id}_B)\circ \mathcal{\Delta}_{Q \leftarrow R}\|_\diamond$, where $\mathcal{N}_A  = \bigotimes_{v \in A} \mathcal{N}_{Q'_v \leftarrow v}$. We observe that
\begin{align}
    \norm{(\mathcal{N}_A \otimes \mathrm{id}_B)\circ  \mathcal{\Delta}_{Q \leftarrow R}^{\Gamma'}}_\diamond &= \norm{\left(\bigotimes_{v \in X[\Gamma']}\mathcal{N}_{Q'_v \leftarrow v} (\mathrm{id}_{v} - \mathcal{D}_v)  \right)  \left(\bigotimes_{v \in Y[\Gamma']}\mathcal{N}_{Q'_v \leftarrow v} \mathcal{D}_v  \right) \left(\bigotimes_{\ell \in A \setminus (X[\Gamma']\cup Y[\Gamma'])}\mathcal{N}_{Q'_v \leftarrow v}  \right)\circ \mathcal{\Delta}_{Q \leftarrow R}}_\diamond \nonumber\\ &\leq 2e^{-\beta |X[\Gamma']|}
\end{align}
where in going to the second line we have used the definition of $\beta$ [Eq.~\eqref{eq:beta def lemma}], the sub-multiplicativity of the diamond norm under composition, $\|\mathcal{T} \circ \mathcal{T}'\|_\diamond \leq \|\mathcal{T} \|_\diamond \|\mathcal{T}'\|_\diamond$ for arbitrary composable maps $\mathcal{T}, \mathcal{T}'$,  and the fact that for a CPTP map $\mathcal{T}$ (including $\mathcal{N}_{Q'_v \leftarrow v}$ and $\mathcal{D}_v$), we have $\|\mathcal{T}\|_\diamond = 1$. We also use the triangle inequality in the form $\|\mathcal{\Delta}_{Q \leftarrow R}\|_\diamond = \|\mathcal{F}_{Q \leftarrow R} - \mathcal{G}_{Q \leftarrow R}\|_\diamond \leq \|\mathcal{F}_{Q \leftarrow R}\|_\diamond + \| \mathcal{G}_{Q \leftarrow R}\|_\diamond = 2$. We then insert the decomposition \eqref{eq:sum J decomp} into the norm in question \eqref{eq:lemma distinguish} and apply the triangle inequality, giving
\begin{align}
    \|(\mathcal{N}_A \otimes \mathrm{id}_B)\circ (\mathcal{F}_{Q \leftarrow R} - \mathcal{G}_{Q \leftarrow R})\|_\diamond &\leq \sum_{n=1}^{|A|} \sum_{\substack{\Gamma' \in \Omega^{\mathscr{W}}\\|\Gamma'|=n }} \norm{(\mathcal{N}_A \otimes \mathrm{id}_B) \circ \mathcal{\Delta}^{\Gamma'}_{Q \leftarrow R} }_\diamond \leq 2\sum_{n=1}^{|A|} \sum_{\substack{\Gamma' \in \Omega^{\mathscr{W}}\\|\Gamma'|=n }} e^{-\beta \sum_{W^+ \in \Gamma'} |W^+|}
\end{align}
where we have used the definition $X[\Gamma'] \coloneqq \bigcup_{W^+ \in \Gamma'} W^+$, along with the constraint that the elements of $\Gamma'$ are mutually non-overlapping. The sum over $\Gamma'$ can now be upper bounded by an independent sum over each components $C_1^+, \ldots, C_n^+ \in \mathscr{W}^+$, with a factor of $1/n!$ to account for the fact that different permutations of the components $C_i^+$ correspond to the same set $\{C_1^+, \ldots, C_n^+\}$. Then,
\begin{align}
    \|(\mathcal{N}_A \otimes \mathrm{id}_B)\circ (\mathcal{F}_{Q \leftarrow R} - \mathcal{G}_{Q \leftarrow R})\|_\diamond &\leq 2\sum_{n=1}^{|A|} \frac{1}{n!} \sum_{C_1^+, \ldots, C_n^+ \in \mathscr{W}^+} e^{-\beta \sum_{i=1}^n |C_i^+|} \nonumber\\ &= 2\sum_{n=1}^{|A|} \frac{1}{n!} \left(\sum_{W^+ \in \mathscr{W}^+} e^{-\beta|W^+|} \right)^n \leq 2(e^{\mathscr{Z}^+} -1)
    \label{eq:delta norm Z bound}
\end{align}
where $\mathscr{Z}^+ \coloneqq \sum_{W^+ \in \mathscr{W}^+} e^{-\beta|W^+|}$ is a weighted sum over all components $W^+ \in \mathscr{W}^+$, analogous to the partition function \eqref{eq:partition SOW}. To bound this quantity, we use the fact that on a graph of bounded degree $k$, the number of connected subgraphs with $m$ vertices that pass through a particular point is at most $(e(k-1))^m$ (see Ref.~\onlinecite{bollobas2006art}, p.~130). Since every $W^+ \in \mathscr{W}^+$ has size at least $d_{\rm min}$, by the assumption of the lemma, we have
\begin{align}
    \mathscr{Z}^+ &\leq |A|\sum_{l = d_{\mathrm{min}}}^{|A|} (e(k-1))^{l}e^{-\beta l} \nonumber\\ &\leq \frac{|A|}{1-e^{-(\beta-\beta^*_k)}}e^{-(\beta-\beta^*_k)d_{\mathrm{min}}} & \text{for }\beta > \beta_k^*.
\end{align}
Finally, since the norm itself \eqref{eq:lemma distinguish} is always upper bounded by 2, we can use the fact that $\min(2, 2[e^x -1]) \leq \frac{2}{\ln 2} x$ for all $x > 0$ to prove Eq.~\eqref{eq:lemma distinguish}, with $\kappa = \frac{2}{\ln 2}[1 - e^{-(\beta-\beta^*_k)}]^{-1}$. \hfill $\square$

\subsection{Proof of Lemma \ref{lem:indist q double}}

As described in Appendix \ref{apdx:q double codes}, the logical operators of the quantum double model are defined on ribbons $r$, which are subsets of edges with a particular string-like structure. To show that 

\begin{lemma}
	For the quantum double model or the string net model, let $S$, $S^c$, $\mathscr{W}$ be as in \cref{lem:indist q double}. Then, if $\Sigma \subset S$ satisfies $W \nsubseteq \Sigma\;\forall W \in \mathscr{W}$, then the model in question supports a closed ribbon $r$ in the homology class $(1,0)$, with the corresponding qudits contained in $S \setminus \Sigma$. 
	\label{lem:logical path}
\end{lemma}
The above result, which is proved in the following section, can quickly be used to prove \cref{lem:indist q double}, as we now show. Recall that the super-projector $\mathcal{P}^{\mathfrak{K}_\leftrightarrow}_R$ acts on states $\sigma_R$ as
\begin{align}
	\mathcal{P}^{\mathfrak{K}_\leftrightarrow}_R[\sigma_R] &= \sum_{\mu \in \mathrm{Anyons}} \dyad{x; \mu}_R\sigma_R\dyad{x; \mu}_R.
\end{align}
Since $S \setminus \Sigma$ contains a ribbon $r_\leftrightarrow$ in the homology class $(1,0)$, the corresponding ribbon operator $K_{r_\leftrightarrow}^\mu$ is supported in $S \setminus \Sigma$, and by construction satisfies
\begin{align}
	K_{r_\leftrightarrow}^\mu V_{Q \leftarrow R}^{\mathcal{C}} = V_{Q\leftarrow R}^{\mathcal{C}} \dyad{x; \mu}_R
	\label{eq:ribbon enc action}
\end{align}
where $V_{Q\leftarrow R}^{\mathcal{C}}$ is the encoding isometry appearing in the channel $\mathcal{C}_{Q \leftarrow R}$ [Eq.~\eqref{eq:encoding}].
for any state $\sigma_R$ we have
\begin{align}
	(\tr_{S \setminus \Sigma} \circ \mathcal{C}_{Q \leftarrow R} \circ \mathcal{P}^{\mathfrak{K}_\leftrightarrow}_R)[\sigma_R] &= \sum_{\mu \in \mathrm{Anyons}} \tr_{S \setminus \Sigma}\left[ V^{\mathcal{C}}_{Q \leftarrow R} \dyad{x;\mu}_R \sigma_R \dyad{x;\mu}_R (V^{\mathcal{C}}_{Q \leftarrow R})^\dagger \right] \nonumber\\
	&= \sum_{\mu \in \mathrm{Anyons}} \tr_{S \setminus \Sigma}\left[ K_{r_\leftrightarrow}^\mu V^{\mathcal{C}}_{Q \leftarrow R} \sigma_R (V^{\mathcal{C}}_{Q \leftarrow R})^\dagger K_{r_\leftrightarrow}^\mu \right] & \text{by Eq.~\eqref{eq:ribbon enc action}} \nonumber\\
	&= \sum_{\mu \in \mathrm{Anyons}} \tr_{S \setminus \Sigma}\left[ (K_{r_\leftrightarrow}^\mu)^2 V^{\mathcal{C}}_{Q \leftarrow R} \sigma_R (V^{\mathcal{C}}_{Q \leftarrow R})^\dagger  \right] \nonumber\\
	&= \sum_{\mu \in \mathrm{Anyons}} \tr_{S \setminus \Sigma}\left[ V^{\mathcal{C}}_{Q \leftarrow R} \dyad{x;\mu} \sigma_R (V^{\mathcal{C}}_{Q \leftarrow R})^\dagger  \right] & \text{by Eq.~\eqref{eq:ribbon enc action}} \nonumber\\ &= \tr_{S \setminus \Sigma}\left[ V^{\mathcal{C}}_{Q \leftarrow R} \sigma_R (V^{\mathcal{C}}_{Q \leftarrow R})^\dagger  \right] = \tr_{S \setminus \Sigma} \circ \mathcal{C}_{Q \leftarrow R}[\sigma_R].
    \label{eq:local indist logical}
\end{align}
In going to the third line, we have used the cyclic property of the partial trace $\tr_{S\setminus \Sigma}[ABC] = \tr_{S\setminus \Sigma}[CAB]$ for any operator $C$ supported on $S \setminus \Sigma$, and in going to the final line we use $\sum_{\mu \in \mathrm{Anyons}} \dyad{x; \mu}_R = I_R$, which follows from the fact that $\{\ket{x;\mu}\}_{\mu \in \mathrm{Anyons}}$ forms a basis for $\mathscr{H}_R$. The above implies the necessary equality \eqref{eq:indistinguishable}, with $\mathcal{F}_{Q \leftarrow R}$ and $\mathcal{G}_{Q \leftarrow R}$ specified in Eq.~\eqref{eq:channel choice qd}. \hfill $\square$\\

\textit{Proof of \cref{lem:logical path}.~---~}This result can be intuitively understood via the schematic illustration depicted in the main text as Fig.~\ref{fig:local-and-global-approximate-indistinguishability}: If $\Sigma$ doesn't contain a connected region that joins the top and bottom boundaries of $S$, then there will be some path in its complement $\Sigma^c = S \setminus \Sigma$ that winds around $S$ in the orthogonal direction. Our proof makes this intuition precise at the lattice level.

We start with the quantum double model on the square lattice. The graph $(V, E)$ can be embedded on the torus $T^2 = \mathbb{R}^2 / (L_x \mathbb{Z} \times L_y \mathbb{Z})$ in the standard way, with each vertex $v \in X$ located at a point $\vec{r}_v = (x_v, y_v)$, where $x_v \in \{0, 1, \ldots, L_{\rm x}-1\}$, $ y_v \in \{0, 1, \ldots, L_{\rm y}-1\}$, and each edge $e = \{v_1, v_2\} \in E$ located at the midpoint $\vec{r_e}$ between $\vec{r}_{v_1}$ and $\vec{r}_{v_2}$.  To each edge $e \in E$, we associate its (closed) Voronoi cell $U_e$, namely the set of points that are no closer to any other edge than to $e$. Explicitly, these are `diamond-shaped' regions $U_e = \{\vec{r}\in T^2 : \|\vec{r} - \vec{r}_e\|_1 \leq 1/2\}$, which naturally tesselate the plane.

This construction relates conveniently with the edge graph $G^\#_{\mathrm{sq}}$ defined in Eq.~\eqref{eq:edge graph square}: for any two edges $e_{1,2} \in E_{\rm sq}$ we have
\begin{align}
	\{e_1, e_2\} \in L_{\mathrm{sq}} \;\; \Longleftrightarrow \;\; X_{e_1} \cap X_{e_2}\neq \varnothing.
	\label{eq:edge graph X square}
\end{align}
(Since the regions are closed, the condition on the right includes cases where $X_{e_1}$ and $X_{e_2}$ touch at a corner). Thus, a subset of edges $S' \subset E$ is graph-connected if and only if the associated region $X_{S'} \coloneqq \bigcup_{e \in S'} U_{e'}$ is path-connected. Moreover, if $X_{S'}$ contains a particular continuous path $\mu : [0,1] \rightarrow X_{S'}$, then we can define an associated path $\gamma_\mu$ on the graph $G_{\mathrm{edge}}|_{S'}$, by breaking $\mu$ into continuous segments $\mu = \mu_1 \circ \mu_2 \circ \cdots$ such that the segment $\mu_i$ is contained in a single region $X_{e_i}$ for some $e_i \in {S'}$. We then choose $\gamma_\mu = (e_1, e_2, \ldots)$. 

Our interest is in subsets $\Sigma \subseteq S$, for which the corresponding regions $X_\Sigma$ will all be contained in $X_S \subset T^2$. We write $\partial_{\rm b} X_S$ and $\partial_{\rm t} X_S$ for the bottom and top boundaries of $X_S$, respectively. The condition on $\Sigma$ in the lemma implies that $X_\Sigma$ contains no path connecting $\partial_{\rm b}X_S$ to $\partial_{\rm t}X_S$. Therefore, if $\{Y_i\}_i$ are the connected components of $X_\Sigma$, then no $Y_i$ can intersect both $\partial_{\rm b}X_S$ and $\partial_{\rm t}X_S$. Now define
\begin{align}
	Z &\coloneqq \partial_{b}X_S \cup \bigcup_{i : \delta_i \cap \partial_{b}X_S \neq \varnothing} Y_i
	\label{eq:conn region def}
\end{align}
which is path-connected. In words, $Z$ is the union of the bottom boundary of $X_S$ and the components $Y_i$ that touch the bottom boundary. This region, which contains the line-like subset $\partial_{b}X_S$ and the area-like subsets $Y_i$, is colored yellow in Fig.~\ref{fig:diamonds}(a). Then let $Z^\epsilon \coloneqq \{\vec{r} \in X_{S} : \text{dist}(\vec{r}, Z) < \epsilon\}$ the set of points in $X_S$ that are at most $\epsilon$-from $Z$ for some $0 < \epsilon < 1/\sqrt{2}$. Since $\partial_{\rm b} X \subseteq Z$, we have $\partial_{\rm b} X \subset \text{Int}(Z^\epsilon)$, where $\text{Int}(U) \coloneqq U \setminus \overline{U}$ is the interior of a region $U$, whose closure is denoted $\overline{U}$. Since any two diamonds $X_e$, $X_{e'}$ either overlap or are at least a distance $1/\sqrt{2}$ away, $\overline{Z^\epsilon}$ must have trivial intersection with $\partial_{\rm t} X$, and with any $Y_i$ not appearing in \eqref{eq:conn region def}, namely those for which $Y_i \cap \partial_1 X_S = \varnothing$. As a result, the boundary $\partial Z^{\epsilon}$ (taken relative to the topology of $X_S$) is contained within $X_S \setminus X_{\Sigma}$, and separates $\partial_{\rm b} X$ from $\partial_{\rm t} X$. This suffices to show that one of the connected components of $\partial Z^{\epsilon}$ is a non-contractible loop $\mu$ supported in $X_{S} \setminus X_\Sigma$. Moreover, this path does not pass through any point in $\mathbb{Z}^2$ or $(1/2,1/2) + \mathbb{Z}^2$, namely the locations where two diamonds touch at a corner.

\begin{figure}
	\centering
	\includegraphics[width=\textwidth]{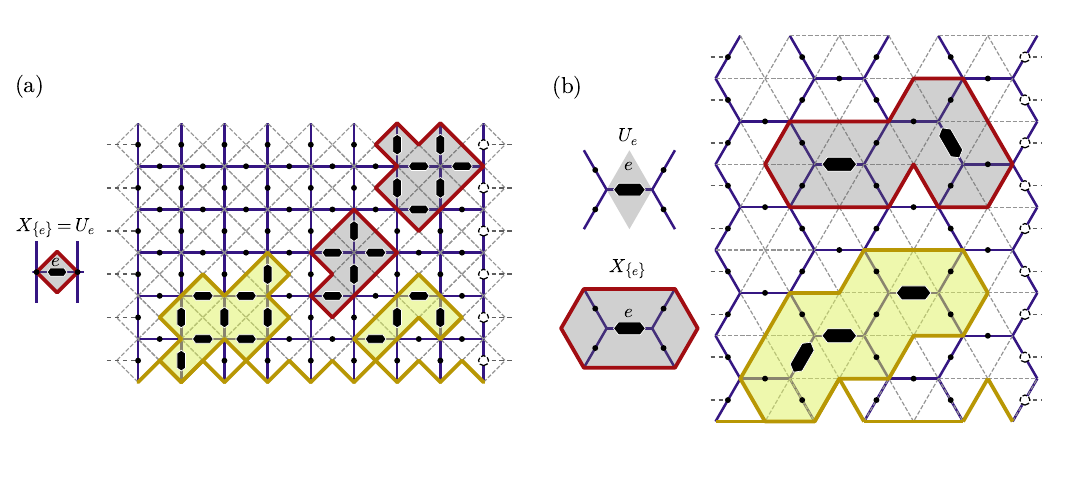}
	\caption{Geometric constructions used in the proof of \cref{lem:logical path}, for (a) the square lattice, corresponding to the quantum double model $D(G)$, and (b) the hexagonal lattice, corresponding to the general Levin-Wen string net model (here without vertex qubits). These illustrate our explicit lattice-level construction of the schematic diagram shown in Fig.~\ref{fig:local-and-global-approximate-indistinguishability}. In each case, we show only the region $S$, which is a strip of edges with rough boundaries on the top and bottom,  winding around the torus in the $x$-direction (here the left and right boundaries are identified). For each edge $e \in S$ (black dots), we identify the Voronoi cell $U_e \subset \mathbb{R}^2$,  boundaries of which are shown as gray dashed lines. Given a subset of qudits $\Sigma \subseteq  E$ (large black bars), we define a region $X_\Sigma$ (shaded regions). For the square lattice, we choose $X_\Sigma = \bigcup_{e \in \Sigma}U_e$, and in particular $X_{\{e\}} = U_e$ [left inset of (a)]. For the hexagonal lattice, we instead choose $X_{\{e\}} = U_e \cup \bigcup_{e' \sim e} U_{e'}$, where the union is over all $e'$ that share a vertex with $e$ [left inset of (b)]. The region $Z \subset \mathbb{R}^2$ (yellow areas and lines), defined in Eq.~\eqref{eq:conn region def} is made up of those connected components of $X_\Sigma$ that touch the bottom boundary $\partial_1 X$, along with $\partial_1 X$ itself.}
	\label{fig:diamonds}
\end{figure}

The loop $\mu(t)$, parametrised by $t \in [0,1]$ can now be divided into segments $\mu_1, \mu_2, \ldots$ such that each $\mu_i(t)$ is contained within a single $X_{e_i}$ for some $e_i \in \Sigma^\mathrm{c}$. As explained above, the start and end points of each $\mu_i$ are on the sides of the given diamond, not the corners. Each segment $\mu_i(t)$ can be continuously deformed such that it passes through $\vec{r}_{e_i}$ exactly once, while staying within $X_{e_i}$, while keeping the start and end points fixed. Now consider the sequence of edges $(e_1, e_2, \ldots )$ passed through by the deformed $\mu$, in order. If any edge $e \in \Sigma^\mathrm{c}$ appears more than once, let $t_1$ and $t_2$ be two parameters at which $\mu(t_1) = \mu(t_2) = \vec{r}_e$, with $t_1 < t_2$. Then, let $\mu'$ be the loop defined by the restriction of $\mu$ to $[t_1, t_2]$, and $\mu''$ the complement, such that $\mu = \mu'' \cdot \mu'$, where $\cdot$ denotes concatenation of loops. Since $\mu$ winds non-trivially in the $y$ direction, the same must be true of at least one of $\mu'$ or $\mu''$. Then replace $\mu$ by one of these homologically non-trivial components. This process can be repeated until no edge appears more than once in the sequence. Define the corresponding graph path $\gamma_\mu = (e_1, e_2, \ldots, e_{|\gamma|}, e_1)$.

Now, for any two sequential edges $e_i, e_{i+1}$, the corresponding diamond regions $X_{e_i}$ and $X_{e_{i+1}}$ must share an edge, since $\mu$ does not pass through any corners between diamonds. This implies that $e_i$ and $e_{i+1}$ share a vertex and a plaquette. Equivalently, one can find a unique site $s_{i+1/2}$ that is shared by both edges. A strip $r$ can then be formed from the sequence of triangles $\tau_i = (s_{i-1/2}, s_{i+1/2}, e_i)$, each of which is supported in $\Sigma^\mathrm{c}$. This forms a valid ribbon since no edge $e_i$ appears multiple times in the sequence.

The argument for the string net model follows a very similar line of reasoning. We embed the hexagonal lattice on the torus in the standard way, placing edges $e = (v_1, v_2) \in E_{\mathrm{hex}}$ at the midpoints between vertices, and define Voronoi cells $U_e$, which are rhombi whose corners are at $\vec{r}_{v_1}, \vec{r}_{f_1}, \vec{r}_{v_2}, \vec{r}_{f_2}$, with $f_{1,2}$ the two faces neighbouring $e$ [see Fig.~\ref{fig:diamonds}(b)]. Rather than choosing $X_{\{e\}} = U_e$ for singletons, this time we choose
\begin{align}
	X_{\{e\}} = U_e \cup \bigcup_{\substack{e' \in S \\e' \overset{V}{\sim} e}} U_{e'}
\end{align}
where again $\overset{V}{\sim}$ denotes vertex-adjacency. This choice also corresponds nicely with the chosen edge graph $G^\#_{\mathrm{hex}}$ [Eq.~\eqref{eq:edge graph hex}], in that the natural analogue of Eq.~\eqref{eq:edge graph X square} holds.  If the string-net model features multiplicities, then we also have to consider qudits at the vertices, and so to each vertex $v$ we associate the region $X_v = \bigcup_{e : v \in e} U_e$, which is in the shape of a regular hexagon centred on $v$. The natural analogue of Eq.~\eqref{eq:edge graph X square} then also holds for $G^\#_{\mathrm{hex, mult}}$.

By identical reasoning, we can then construct a path $\mu$ contained in $X_S \setminus X_{\Sigma}$ that winds nontrivially in the $(1,0)$ direction, never visits any Voronoi cell twice, and never passes through a vertex or the midpoint of a face. Defining $\gamma_\mu = (e_1, e_2, \ldots, e_{|\gamma_\mu|}, e_1)$ as before, each sequential pair of edges $e_t, e_{t+1}$ must have Voronoi cells meeting at a line, not just a corner, and thus they must share a unique vertex $v_{t+1/2}$. This defines a cyclic path on the primal lattice $(v_{1 + 1/2}, v_{2 + 1/2}, \ldots, v_{|\gamma_\mu| + 1/2}, v_{1+1/2})$, which is indeed a path since $(v_{t - 1/2},v_{t+1/2}) = e_t \in E_{\mathrm{hex}}$, that winds non-trivially. By construction, $e_t \notin \Sigma$, and thanks to the structure of $X_{\{e\}}$, we have $e' \notin \Sigma$ for all $e' \overset{V}{\sim} e_t$. Thus, every edge adjacent to $v_t$ is not contained in $\Sigma$. In the case of fusion multiplicities, we can show the same, and also that any vertex in the path is not contained in $\Sigma$ either. Therefore, $\gamma_\mu$ can be used to define a closed string-like operator supported on the region $\bigstar(\gamma_\mu)$ [Eq.~\eqref{eq:string net ribbon}], which is entirely contained in $S \setminus \Sigma$.
\hfill $\square$

\section{Improved bound for $\mathbb{Z}_2$ surface code \label{apdx:z2 improvements}}

As outlined in Section \ref{sec:sow-mapping}, we can obtain a quantitatively improved bound on the recovery threshold for the $\mathcal{Z}_2$ surface code by using the alternative decomposition [given in the main text as Eq.~\eqref{eq:delta decomp X}] 
\begin{align}
    \mathcal{\Delta}_{Q \leftarrow R}  &= \sum_{\Sigma \in P(S)} \tilde{\mathcal{\Delta}}_{Q \leftarrow R}^{\,(\Sigma; X)}, &\text{where }   \tilde{\mathcal{\Delta}}_{Q \leftarrow R}^{\,(\Sigma; X)} \coloneqq \left(\left[\bigotimes_{\ell \in \Sigma}(\mathrm{id}_\ell - \mathcal{E}_\ell^X) \right]\otimes \left[\bigotimes_{\ell \in \Sigma^\mathrm{c}}  \mathcal{E}_\ell^X \right] \right) \circ \mathcal{\Delta}_{Q \leftarrow R}.
    \label{eq:delta decomp X repeat}
\end{align}
in place of Eq.~\eqref{eq:channel sum J}. Here, $ \mathcal{E}_\ell^X[\sigma_\ell] = \frac{1}{2}(\sigma_\ell + X_\ell \sigma_\ell X_\ell)$ is the complete dephasing channel in the $X$-basis for qubit $\ell$. In this appendix, we use this method to prove \cref{lem:SOW mapping surface}.

First, we prove Eqs.~(\nopareneqref{eq:sigma constraint X}, \nopareneqref{eq:sigma constraint X even}), which describe the constraints on the regions $\Sigma \in \mathscr{P}(S)$. Recall the definition $\mathcal{\Delta}_{Q \leftarrow R} \coloneqq \mathcal{C}_{Q \leftarrow R}\circ(\mathrm{id}_R -  \mathcal{P}^{\mathfrak{X}}_R)$. Since the map $(\mathrm{id}_R -  \mathcal{P}^{\mathfrak{X}}_R)$ annihilates the operators $I_R, X_R$, the map $\mathcal{\Delta}_{Q \leftarrow R}$ is fully determined by its action on the span of the operators $Y_R, Z_R$, which is this case gives
\begin{align}
	\mathcal{\Delta}_{Q \leftarrow R}[c_Y Y_R + c_Z Z_R] &= \mathbb{A}_V (\iu c_Y \bar{X} + c_Z) \mathbb{B}_P \bar Z.
    \label{eq:delta Z2 image}
\end{align}
Recall that $\mathbb{B}_P = \prod_{p} \mathbb{B}_p$, with $\mathbb{B}_p  = (I+B_p)/2$, and $B_p = \prod_{\ell \in p}Z_p$ is the stabiliser for plaquette $p$. Here, $\bar{X}$ and $\bar{Z}$ can be any representative logical string operators. For concreteness, we fix $\bar{Z}$ to be the product of Pauli-$Z$ operators on each of the leftmost links. Letting $P_S$ denote the subset of plaquettes whose support intersects $S$, it is helpful to write $\mathbb{B}_P = \mathbb{B}_{P_S} \mathbb{B}_{P \setminus P_S}$, where $\mathbb{B}_{P'} \coloneqq \prod_{p \in P'}\frac{I+B_p}{2}$ is the product of stabilizer projectors in the subset $P' \subseteq P$. We also write $\bar{Z} = \bar{Z}_S\bar{Z}_{S^c}$, where $\bar{Z}_S$ is the product of Pauli-$Z$s on the links on the leftmost column that are contained in $S$, and similar for $\bar{Z}_{S^c}$.

Now, choose a $\Sigma \in \mathscr{P}(S)$. Since $\mathcal{E}^X_\ell[X_\ell O_\ell] = X_\ell\mathcal{E}^X_\ell[ O_\ell]$, and the map inside the brackets in \eqref{eq:delta decomp X repeat} acts only on $S$, the corresponding map $\tilde{\mathcal{\Delta}}_{Q \leftarrow R}^{\,(\Sigma; X)}$ acts as
\begin{align}
    \tilde{\mathcal{\Delta}}_{Q \leftarrow R}^{\,(\Sigma; X)}[c_Y Y_R + c_Z Z_R] = \mathbb{A}_V (\iu c_Y \bar{X} + c_Z)\mathbb{B}_{P \setminus P_S} \bar{Z}_{S^c}\cdot  \left(\left[\bigotimes_{\ell \in \Sigma}(\mathrm{id}_\ell - \mathcal{E}_\ell^X) \right]\otimes \left[\bigotimes_{\ell \in \Sigma^\mathrm{c}}  \mathcal{E}_\ell^X \right] \right)[\mathbb{B}_{P_S} \bar{Z}_S]
    \label{eq:delta sigma X proj}
\end{align}
Now, the operator $\mathbb{B}_{P_S}\bar{Z}_S$ can be decomposed as a sum of $2^{|P_S|}$ distinct $Z$-type Pauli strings by expanding out the product $\mathbb{B}_{P_S} =\prod_{p \in P_S}\frac{I+B_p}{2}$. This can be conveniently represented as a sum over Ising variables $\vec{\sigma} \in \{\pm 1\}^{P_S}$, where $\sigma_p = -1$ ($+1$) indicates that the term $B_p$ ($I$) is chosen. Explicitly, we have $\mathbb{B}_{P_S}\bar{Z}_S = 2^{-|P_S|}\sum_{\vec{\sigma} \in \{\pm 1\}^{P_S}} Z_{\vec{\sigma}}$, where $Z_{\vec{\sigma}} \coloneqq \bar{Z}_{S} \prod_{p \in P_S} B_p^{(1+\sigma_p)/2} = \prod_{\ell \in S[\vec{\sigma}]} Z_\ell$ is a Pauli string supported on a region $S[\vec{\sigma}] \subseteq S$, which can be formally specified as $S[\vec \sigma] \coloneqq \{\ell \in S : \prod_{p \in \partial^\top \ell} \sigma_p = (-1)^{\delta\{\ell \in S_{\mathrm{left}}\}+1}\}$, with $S_{\mathrm{left}}$. Since $\mathcal{E}^X_\ell[Z_\ell] = 0$ and $(\mathrm{id}_\ell - \mathcal{E}^X_\ell)[I_\ell] = 0$, the map \eqref{eq:delta sigma X proj} vanishes unless there is some $\vec{\sigma} \in \{\pm 1 \}^{P_S}$ such that $S[\vec{\sigma}] = \Sigma$. This defines a bijection between regions $\Sigma \in \mathscr{P}(S)$ for which $\tilde{\mathcal{\Delta}}_{Q \leftarrow R}^{\,(\Sigma; X)}$ is nonzero, and spin configurations $\vec{\sigma} \in \mathscr{P}(S)$; we will thus use $\Sigma$ and $\vec \sigma$ interchangeably, where meaning is obvious.

\begin{figure}
    \centering
    \includegraphics[width=0.5\linewidth]{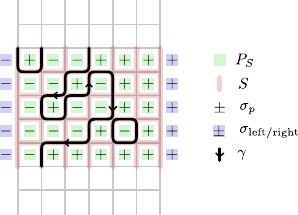}
    \caption{Example of a plaquette spin configuration and the resulting self-osculating walk.}
    \label{fig:sow-ising}
\end{figure}


For a given configuration $\vec{\sigma}$, we can interpret the corresponding region $\Sigma = S[\vec{\sigma}] \subseteq S$ as a collection of domain walls, using the following argument. For each edge $\ell$ not on one of the smooth boundaries of $S$, there are exactly two plaquettes, $p_1, p_2 \in P_S$ that contain $\ell$, in which case we have $\ell \in S[\vec{\sigma}]$ if and only if $ \sigma_{p_1} = - \sigma_{p_2}$. Edges on the left and right boundaries of $S$ only neighbour one plaquette $p_1$, so it is helpful to enlarge the set $P_S$ by introducing an extra fictitious plaquette $p_2(\ell)$, with a corresponding  spin variable $\sigma_{p_2(\ell)} \in \{\pm 1\}$, for each such $\ell$, as illustrated in \cref{fig:sow-ising}. These extra spins are fixed to $-1$ ($+1)$ if $\ell \in S_\mathrm{left}$ ($\ell \in S_\mathrm{right}$). With these included, $S[\vec{\sigma}]$ is precisely the set of edges on which the two neighbouring plaquette spins $\sigma_{p_1}, \sigma_{p_2}$ differ, i.e.~$\ell$ hosts a domain wall. Since domain walls can only terminate on the rough boundaries of $S$, for any configuration $\sigma$, there must be an even number of domain walls around each vertex $v \in V_S$. This suffices to prove the relation Eq.~\eqref{eq:sigma constraint X even}.

Following the well-known contour representation of domain walls in the Ising model \cite{friedli2017}, we now describe a unique way to decompose $\Sigma = S[\vec{\sigma}]$ into a collection of subsets $\Gamma[\vec \sigma] \equiv \Gamma[\Sigma] = \{\gamma_1, \gamma_2, \ldots \}$, such that $S[\vec{\sigma}] = \bigcup_{\gamma \in \Gamma[\vec{\sigma}]} \gamma$. This will refine the decomposition described in Section \ref{sec:upper bound simple}. In particular,  each $\gamma_i$ will correspond to the edges visited by a directed self-osculating walk. 

In the first step, mirroring the procedure described in Section \ref{sec:upper bound simple}, we decompose the domain walls $\Sigma$ into maximally connected components $C_1, C_2, \ldots$, where we consider two edges $\ell, \ell \in S$ to be adjacent if and only if they share a vertex $v \in V_S$. For each such component $C_i$, the extra condition \eqref{eq:sigma constraint X even} implies that every $v \in V_S$ will be contained in an even number of edges in $C_i$, that is $\text{deg}_{C_i}(v) = 0 \mod 2$. If $\text{deg}_{C_i}(v) \in \{0,2\}$ for all $v \in V_S$, then $C_i$ is automatically a path, which is either closed, or terminates at the rough boundaries of $S$. If instead $\text{deg}_{C_i}(v) = 4$ for some vertices $v \in V_S$, then $C_i$ may not uniquely correspond to a path. Nevertheless, we can decompose $C_i$ into a collection of paths using the following deformation rule
\tikzset{midarrow/.style={decorate, decoration={
			markings,
			mark=at position 0.5 with {\arrow{>}}}}}
\begin{equation} \label{eq:domain-wall-ambiguity}
	\begin{tikzpicture}[scale=0.25, baseline=-0.5ex]
		\node[] (bl) at (-1, -1) {$+$};
		\node[] (br) at (1, -1) {$-$};
		\node[] (tl) at (-1, 1) {$-$};
		\node[] (tr) at (1, 1) {$+$};
		\draw [thick] (-2, 0) -- (2, 0);
		\draw [thick] (0, -2) -- (0, 2);
	\end{tikzpicture}
	\quad \longrightarrow \quad
	\begin{tikzpicture}[scale=0.25, baseline=-0.5ex]
		\node[] (bl) at (-1, -1) {$+$};
		\node[] (br) at (1, -1) {$-$};
		\node[] (tl) at (-1, 1) {$-$};
		\node[] (tr) at (1, 1) {$+$};
		\draw [color=gray!50,] (-2, 0) -- (2, 0);
		\draw [color=gray!50,] (0, -2) -- (0, 2);
		\draw [rounded corners=2mm, thick] (-2, 0)--(0, 0)--(0, 2);
		\draw [rounded corners=2mm, thick] (0, -2)--(0, 0)--(2, 0);
	\end{tikzpicture}\,,
\end{equation}
and $90^\circ$ rotations thereof (i.e.~the paths are drawn such that the regions where $\sigma_p = +1$ stay together). Applying this rule at every vertex yields a collection of paths, and we combine all the paths generated from each connected component $C_i$ into a set which we denote $\Gamma[\Sigma] = \{\gamma_1, \gamma_2, \ldots \}$.
Each path in $\Gamma[\Sigma]$ does not self-intersect, but can self-osculate according to the rule \eqref{eq:domain-wall-ambiguity}. Furthermore, for convenience we will assign a direction to each path, such that the Ising spins with $\sigma_p = +1$ are on the left from the point of view of an observer following the directed path.  Note that this construction is distinct from the decomposition $\Gamma[\Sigma] = \{C_1, \ldots , C_{|\Gamma[\Sigma]|}\}$ used in the proof of \cref{lem:bound-norm-connectivity}; however we use the same notation since this will play an identical role in the following.

Due to the fixed and opposite boundary conditions for the spin configuration $\vec \sigma$ on the left and right boundaries, and our choice of deformation rule \eqref{eq:domain-wall-ambiguity}, the set $\Gamma[\Sigma]$ will always contain at least one path $\gamma$ that starts on the top boundary of $S$, and ends on the bottom boundary. To see this, let $P^-_\sigma \subseteq P_S$ be the maximal set of plaquettes that is 1) connected with respect to the dual lattice, 2) has $\sigma_p = -1$ for all $p \in P^-_\sigma$, and 3) contains all the extra plaquettes on the left side. Thanks to our choice of deformation rule \eqref{eq:domain-wall-ambiguity}, each connected component of the boundary $\partial P^-_\sigma \subset S$ will correspond to exactly one element $\gamma \in \Gamma[\Sigma]$ (note that this may not be true for other deformation rules, such as if we fix the osculating rule irregardless of Ising configuration). Because $P^-_\sigma$ contains the leftmost plaquettes and none of the rightmost plaquettes, one of these components must connect the top boundary of $S$ to the bottom boundary of $S$.

These paths, which we distinguish with a superscript $\gamma^+$, are a more refined version of the spanning walks described in the main text (Sec.~\ref{sec:upper bound simple}), and are now endowed with a direction. Mirroring the notation used in our proof of \cref{lem:bound-norm-connectivity}, we denote the set of all possible spanning walks as $\mathscr{W}^+$, and write $\Gamma^{\mathscr{W}}[\Sigma] \coloneqq \Gamma[\Sigma] \cap \mathscr{W}^+$. We can then follow the same steps as the proof of \cref{lem:bound-norm-connectivity} to arrive at the natural analogue of Eq.~\eqref{eq:sum J decomp}. That is, for any $\Gamma' = \{\gamma^+_1, \ldots, \gamma^+_n\} \subset \mathscr{W}^+$, we define $\tilde{\mathcal{\Delta}}^{\Gamma'} = \sum_{\Sigma : \Gamma[\Sigma] \supseteq \Gamma'}\tilde{\mathcal{\Delta}}^{(\Sigma)}$, and
\begin{align}
    \mathcal{\Delta}_{Q \leftarrow R} = \sum_{\Sigma \in \mathscr{P}(A)}  \mathcal{\Delta}_{Q \leftarrow R}^{(\Sigma)} = \sum_{n=1}^{|A|} (-1)^{n+1} \sum_{\substack{\Gamma' \in \Omega^{\mathscr{W}}\\|\Gamma'|=n }} \tilde{\mathcal{\Delta}}^{\Gamma'} .
	\label{eq:sum J decomp z2}
\end{align}
Now, a for a given collection of directed self-osculating walks $\Gamma' = \{\gamma^+_1, \ldots, \gamma^+_n\}$, we can characterise all the configurations $\Sigma$ that satisfy $\Gamma[\Sigma] \supseteq \Gamma'$. For each $\gamma \in \Gamma'$, we must have $\Sigma \supseteq \gamma$. Thus, we define $X[\Gamma'] = \bigcup_{\gamma \in \Gamma'}\gamma$, and by analogy to \eqref{eq:Gamma J cond} we have $\Gamma[\Sigma] \supseteq \Gamma' \Rightarrow X[\Gamma'] \subseteq \Sigma  $. Furthermore, for each $\gamma \in \Gamma'$, and for each vertex $v \in V_S$ visited by $\gamma$, there will be exactly two edges adjacent to $v$ that are not contained in a given walk $\gamma$. If $\gamma$ takes a straight-line path or a left-hand turn through $v$, then these two edges must not be in $\Sigma$. We collect all such edges into a set $Y[\Gamma']$. Note that if $v$ takes a right hand turn then these edges may or may not be contained in $\Sigma$, so we do not include these in $Y[\Gamma']$.

With these definitions, we recover the natural analogue of the relation \eqref{eq:Gamma J cond}, and in turn the expression \eqref{eq:delta gamma xy}, with the appropriate substitutions. The following steps then proceed identically, and we find
\begin{align}
    \norm{\mathcal{N}_{S \leftarrow Q}\circ \mathcal{\Delta}_{Q \leftarrow R} }_\diamond \leq \frac{2}{\ln 2} \mathscr{Z}^{\text{SOW}}_{\Lambda(w_\mathrm{x}, L_\mathrm{y})}
\end{align}
where $\mathscr{Z}^{\text{SOW}}_{\Lambda(w_\mathrm{x}, L_\mathrm{y})}$ is the weighted sum over spanning self-osculating walks $\gamma \in \mathscr{W}^+$. This is equivalent to the partition function of SOWs starting at $(0,y)$ and ending at $(w_\mathrm{x}, y')$ within the rectangular region $S$ for any $y, y' \in [0, L_\mathrm{y})$, as defined in Eq.~\eqref{eq:partition SOW}. This proves \cref{lem:SOW mapping surface}. \hfill $\square$\\

The partition function $\mathscr{Z}^\mathrm{SOW}_{\Lambda(w_\mathrm{x}, L_\mathrm{y})}$ can be upper bounded by partition function of all SOWs of length $w_\mathrm{x}$ or greater and less than some $O(w_\mathrm{x}L_\mathrm{y})$ starting from $(0, 0)$ without further restrictions, multiplied by $L_\mathrm{y}$ representing the entropic factor of all possible starting positions:
\begin{align}
	\mathscr{Z}^\mathrm{SOW}_{\Lambda(w_\mathrm{x}, L_\mathrm{y})} \leq L_\mathrm{y} \sum_{\gamma : |\gamma| \geq w_\mathrm{x}}^{O(w_\mathrm{x} L_\mathrm{y})} e^{-\beta |\gamma|}.
\end{align}
The total number of SOWs of length $n$ on $\Lambda$ can be  upper bounded by $(\mu_\Lambda^\mathrm{SOW})^n$, where $\mu_\Lambda^\mathrm{SOW}$ is the connective constant of SOWs. Therefore we have
\begin{align}
	\mathscr{Z}^\mathrm{SOW}_{\Lambda(w_\mathrm{x}, L_\mathrm{y})} \leq L_\mathrm{y} \sum_{n=w_\mathrm{x}}^{O(w_\mathrm{x} L_\mathrm{y})} (\mu_\Lambda^\mathrm{SOW})^n e^{-\beta n}.
\end{align}
Assuming that $\beta > \beta_\mathrm{c}^\Lambda$, the geometric series on the right hand side converges, and is upper bounded by the identical sum with an upper limit $n = \infty$,
\begin{align}
	\mathscr{Z}^\mathrm{SOW}_{\Lambda(w_\mathrm{x}, L_\mathrm{y})} \leq L_\mathrm{y} \frac{e^{-(\beta - \beta_\mathrm{c}) w_\mathrm{x}}}{1 - e^{-(\beta - \beta_\mathrm{c})}} \quad \text{ for } \quad \beta > \beta_\mathrm{c}^\Lambda,
\end{align}
thus proving Eq.~\eqref{eq:SOW bound}.

\section{Proofs of Corollaries following Theorem \ref{thm:fat}} \label{apdx:corollaries}
In the main text (\cref{sec:summary-sc}), we outlined how the existence of thickened logical operator, proved in \cref{thm:fat}, can be used to rigorously demonstrate that the various hallmarks of topological order defined in Section \ref{subsec:diagnostics} are preserved---namely Corollaries \ref{cor:lre}, \ref{cor:anyon-distinguishability}, and \ref{cor:recoverability}. These were generalised to general group-valued surface codes in Section \ref{sec:quantum-doubles}. Here, we prove the general versions of these Corollaries, using Theorem \ref{thm:fat} as a starting point.

\subsection{Corollary \ref{cor:lre}: Mixed-state long-range entanglement} \label{apdx:lre}


We consider an arbitrary noise-corrupted code state $\tilde{\rho}_Q = \mathcal{M}_{Q \leftarrow R}[\sigma_R]$ corresponding to some logical state $\sigma_R$. Assuming the  conditions stated in Corollary \ref{cor:lre}, we will show that $\tilde{\rho}_Q$ is $(c_0, r, t)$-long-range entangled for some constant $c_0 > 0$. Our proof is by contradiction. We assert that there exists  a state $\sigma_Q$ that satisfies
\begin{align}
	\|\tilde{\rho}_Q - \sigma_Q\|_1 < \delta_0
	\label{eq:rho sigma close}
\end{align}
for some $\delta_0 > 0$, and is also $(r, t)$-short-range entangled. Importantly, this implies $\sigma_Q$ can be written as
\begin{align}
	\sigma_Q = \sum_i p_i \omega_Q^{(i)}
	\label{eq:sigmaQ sum omega}
\end{align}
where $p_i \geq 0$, $\sum_i p_i = 1$, and $\omega_Q^{(i)}$ are short-range correlated states, i.e.~states generated by a local ancilla-assisted circuit of range $r$ and depth $t$. We will show that this leads to conflicting statements when $\delta_0$ is sufficiently small.

For $rt < L/3$, we can identify horizontal strips $S_{x, 1}$, $S_{x, 2}$, each of width $w = L/6$, that are separated by a distance $2rt$. \cref{thm:fat} holds for each of these strips separately. Thus, for each anyon type $\mu \in \mathrm{Anyons}$ and each $\alpha = 1, 2$, there is an operator $X_{\alpha,\mu}$, with support on $S_{x, \alpha}$ that has approximately the same action as the logical projector $\Pi^{x, \mu}_R \coloneqq \dyad{\mu^\leftrightarrow}_R$, in the sense of Eq.~\eqref{eq:effective op thm-fat}. Thanks to \cref{lem:effective operator}, can assume that $X_{\alpha,\mu}$ is Hermitian, and is bounded as $0 \preceq X_{\alpha,\mu} \preceq 1$.  We similarly choose a pair of vertical strips $S_{y, 1}$, $S_{y,2}$, which support operators $Y_{\alpha, \mu}$ that act approximately the same as the projectors $\Pi^{y, \mu}_R = S \Pi^{x, \mu}_R S^\dagger$, where $\Pi^{y, \mu}_R \coloneqq \dyad{\mu^{\updownarrow}}_R$. Here, as in the main text [Eq.~\eqref{eq:S matrix}], $S$ is the modular $S$-matrix.  For the $\mathbb{Z}_2$ surface code, the same constructions apply---in particular we have $\Pi^{x, \mu}_R = \frac{1}{2} (1+(-1)^\mu X_R)$ and $\Pi^{y, \mu}_R = \frac{1}{2} (1+ (-1)^\mu Z_R)$, which are labeled by $\mu \in \{0,1\}$. These are related as $\Pi^{y, \mu}_R = H_R \Pi^{x, \mu}_R H^\dagger_R$, where $H_R = \smqty(1 & 1 \\ 1 & -1)_R$ is the Hadamard matrix. 

Thanks to the geometry of our construction, and the short-range correlated nature of the states $\omega^{(i)}_Q$, for any operators $O_{S_{x, 1}}$ and $O'_{S_{x, 2}}$ supported on $S_{x, 1}$ and $S_{x, 2}$, respectively, we have
\begin{align}
	\langle{O_{S_{x, 1}}O'_{S_{x, 2}}}\rangle_i = \langle{O_{S_{x, 1}}}\rangle_i\langle{O'_{S_{x, 2}}}\rangle_i
\end{align}
where we use the shorthand $\langle{Z_Q}\rangle_i \coloneqq \tr[\omega_Q^{(i)}Z_Q]$ for operators $Z_Q \in \mathfrak{B}(\mathscr{H}_Q)$. 

Now, as a consequence of Eq.~\eqref{eq:effective op thm-fat},  we have
\begin{align}
	\tr[X_{1, \mu}\big(1 - X_{2,\mu}\big)\tilde{\rho}_Q] \leq  \tr[\Pi_R^{x, \mu}(1 - \Pi_R^{x, \mu})\rho_R] + 2\epsilon = 2\epsilon
	\label{eq:x corr bound}
\end{align}
where $\epsilon$ is the error appearing in \cref{thm:fat}, which is at most $c_1(T) \sqrt{L} e^{-c_2(T) L/6}$. Combining this with Eq.~\eqref{eq:rho sigma close}, this implies that the short-ranged correlated states appearing in \eqref{eq:sigmaQ sum omega} must satisfy
\begin{align}
	\sum_i p_i \langle X_{1, \mu} \rangle_i \langle 1 - X_{2,\mu}\rangle_i  \leq 2\epsilon + \delta_0
	\label{eq:sigma exp value simple}
\end{align}
We also have the same with $X_{\alpha, \mu}$ replaced by $Y_{\alpha, \mu}$.  Furthermore, if we define
\begin{align}
	T_\sigma(\mu, \nu) \coloneqq \tr[X_{1,\mu} Y_{2,\nu}(1 - X_{2, \mu}) Y_{1, \nu} X_{1, \mu} \sigma_Q] = \sum_{i} \Big\langle X_{1,\mu} Y_{2,\nu}(1 - X_{2, \mu}) Y_{1, \nu} X_{1, \mu} \Big\rangle_i 
	\label{eq:Tsigma def}
\end{align}
then by the same logic we have
\begin{align}
	\left|T_\rho(\mu, \nu) - T_\sigma(\mu, \nu)  \right| \leq 5 \epsilon + \delta_0 \eqqcolon \epsilon' 
	\label{eq:lre projector prod}
\end{align}
where
\begin{align}
	T_\rho(\mu, \nu) \coloneqq \tr[\Pi^{x, \mu}_R\Pi^{y, \nu}_R(1-\Pi^{x, \mu}_R)\Pi^{y, \nu}_R\Pi^{x, \mu}_R \rho_R] = \tr[\Pi^{x, \mu}_R \rho_R] \sum_{\mu' \neq \mu}  |S_{\mu \nu}|^2 |S_{\mu' \nu}|^2.
	\label{eq:Trho express}
\end{align}

Since $X_{i, \mu}$ is Hermitian, and $0 \preceq X_{\alpha, \mu} \preceq 1$, each of the expectation values in Eq.~\eqref{eq:sigma exp value simple} are contained in the interval $[0,1]$. Using the fact that, for $x_{1,2} \in [0,1]$, we have $x_1(1 - x_1) = x_1(1-x_1)(1-x_2) + x_1(1-x_1)x_2 \leq x_1(1-x_2) + x_2(1 - x_1)$, Eq.~\eqref{eq:sigma exp value simple} then implies
\begin{align}
	\sum_i p_i \mathrm{Var}_i\big[X_{\alpha, \mu}\big] \leq 2\epsilon'
	\label{eq:sre var bound}
\end{align}
where $\mathrm{Var}_i[Z_Q] \coloneqq \langle Z_Q^2 \rangle_i - \langle Z_Q \rangle_i^2$ is the variance of an operator with respect to the state $\omega^{(i)}_Q$. Thus, typical states $\omega^{(i)}_Q$ are approximate eigenstates of $X_{\alpha, \mu}$. The same naturally follows for the operators $Y_{\alpha, \nu}$. If we define
\begin{align}
	\delta_i X_{\alpha, \mu} &= X_{\alpha, \mu} - \langle X_{\alpha, \mu} \rangle_i & \alpha \in \{1,2\};\, \mu \in \mathrm{Anyons},
\end{align}
and similar for $\delta_i Y_{\alpha, \nu}$, then \eqref{eq:sre var bound} can be rewritten as $\sum_i p_i \langle (\delta_i X_{\alpha, \mu})^2 \rangle_i \leq 2 \epsilon'$. The object $T_\sigma(\mu, \nu)$ defined in Eq.~\eqref{eq:Tsigma def} can now be rewritten by substituting the identity $X_{\alpha, \mu} = \delta_i X_{\alpha, \mu} +\langle X_{\alpha, \mu} \rangle_i$, and expanding the resulting product. This will result in 32 terms depending whether the operator $\delta_i X_{\alpha, \mu}$ or the scalar $\langle X_{\alpha, \mu} \rangle_i$ is chosen in each of the 5 factors. Since $\langle \delta_i  X_{\alpha, \mu} \rangle_i = 0$, the five terms that contain a single non-trivial operator will vanish. An example nonzero term is 
\begin{align}
	\sum_i p_i  \langle X_{1, \mu}\rangle_i \langle 1 - X_{2,\mu}\rangle_i   \Big\langle \delta_i Y_{2,\nu} \cdot  \delta_i Y_{1, \nu} \cdot \delta_i X_{1, \mu} \Big\rangle_i
	\label{eq:example term}
\end{align}
The norm of the last expectation value above can be bounded using H{\"o}lder's inequality in the form $\tr[ABC\rho] \leq \|C\rho^{1/2}\|_2 \|\rho^{1/2}A\|_2 \|B\|_\infty = \tr[\rho C^\dagger C]^{1/2} \tr[\rho A^\dagger A]^{1/2} \|B\|_\infty$. We have  $\| \delta_i X_{\alpha, \mu} \|_\infty \leq 2$, and thus this simplifies to
\begin{align}
	\big|\,\text{\eqref{eq:example term}}\, \big| &\leq 2\sum_i p_i \langle (\delta Y_{2,\nu})^2 \rangle_i^{1/2} \langle (\delta X_{1,\mu})^2 \rangle_i^{1/2} \nonumber\\
	&\leq 2 \left( \sum_i p_i \langle (\delta Y_{2,\nu})^2 \rangle_i \right)^{1/2} \left( \sum_i p_i \langle (\delta X_{1,\mu})^2 \rangle_i \right)^{1/2} \leq 4 \epsilon' 
\end{align} 
where we have used the Cauchy-Schwartz inequality in going to the second line. Combining all the terms together, we then find
\begin{align}
	\left| T_\sigma(\mu, \nu) -	\sum_i p_i  \langle X_{1, \mu}\rangle_i \langle Y_{2, \nu}\rangle_i\langle1 - X_{2, \mu}\rangle_i\langle Y_{1, \nu}\rangle_i\langle X_{1, \mu}\rangle_i \langle 1 \right| \leq 58 \epsilon'.
\end{align}
Using Eqs.~(\ref{eq:x corr bound}, \ref{eq:lre projector prod}) and the triangle inequality, we thus obtain
\begin{align}
	T_\rho(\mu, \nu) \leq 60 \epsilon'.
	\label{eq:Trho upper bound}
\end{align}
Now, returning to Eq.~\eqref{eq:Trho express}, we note that for any choice of $\rho_R$, there is always some $\mu^\star \in \mathrm{Anyons}$ for which $\tr[\Pi^{x, \mu}_R \rho_R] \geq d_R^{-1}$, since $\{\Pi^{x, \mu}_R\}_\mu$ form a complete orthogonal set of rank-1 projectors. Then, we choose $\nu = \varnothing$, namely the trivial anyon, for which the $S$-matrix satisfies \cite{simon2023topological}
\begin{align}
	S_{\mu \varnothing} = S_{\varnothing \mu} = \frac{d_\mu}{\mathscr{D}}
	\label{eq:S matrix q dimension}
\end{align}
where $d_\mu \geq 1$ is the quantum dimension of the anyon $\mu$ and $\mathscr{D} \coloneqq (\sum_\mu d_\mu^2)^{1/2}$ is the total quantum dimension. Thus, with these choices of $\mu$ and $\nu$ we have
\begin{align}
	T_\rho(\mu^\star, \varnothing) \geq \frac{1}{d_R\mathscr{D}^4} \eqqcolon c'
\end{align}
with $c'$ being a constant only depending on the anyon content of the model. This contradicts Eq.~\eqref{eq:Trho upper bound} whenever $60c_0 + 300 \epsilon \leq c'$. Since $\epsilon \leq c_1(T)\sqrt{L}e^{-c_2(T) L/6}$, which is less than $c_0/5$ for $L$ exceeding some critical value $L_0$, depending only on $T$, this proves the claim, with $c_0 = c'/61$. \hfill $\square$

\subsection{Corollary \ref{cor:anyon-distinguishability}: Distinguishability of anyons}

The proof of Corollary \ref{cor:anyon-distinguishability}, which we present here, requires some slight adaptation of Theorem \ref{thm:fat}. In particular, whereas we have previously considered the manifold of code states, here we will start from an arbitrary reference code state $\rho_Q$ (the choice of which will not matter) and generate different states with anyonic excitations $\rho^{(v,v'), \mu}_Q$, as defined in Eq.~\eqref{eq:anyon state full}. We can continue to use the formalism of AOAEC to analyse the distinguishability of these states after noise is applied, via the following argument.

Before noise is applied, the states $\rho^{(v,v'), \mu}_Q$ can be perfectly distinguished using the closed ribbon operators $K_{r'}^{\mu'}$ [Eq.~\eqref{eq:anyon distinguish ribbon}]. This can be reframed by defining a (classical) encoding channel $\mathcal{C}^{(v,v')}$ which act as
\begin{align}
    \mathcal{C}^{(v,v')}_{Q \leftarrow R}[\sigma_R] = \sum_\mu \langle\mu|\sigma_R|\mu\rangle_R \rho^{(v,v'), \mu}_Q
    \label{eq:anyon encoding}
\end{align}
where $\{\ket{\mu}_R\}_\mu$ is an arbitrarily chosen basis for a reference Hilbert space $\mathscr{H}_R$ of dimension $d_R = |\text{Irr}(D(G))|$. This channel does not preserve any coherences on $R$, and thus we can view $R$ as a classical register. Conditioned on the classical register being in the state $\mu$, the map $\mathcal{C}^{(v,v')}_{Q \leftarrow R}$ prepares the state $\rho^{(v,v'), \mu}_Q$. Pairwise orthogonality of these states is equivalent to exact recoverability of the diagonal algebra $\mathfrak{K} = \text{span}(\{\dyad{\mu}_R\})$ for the channel $\mathcal{C}^{(v,v')}_{Q \leftarrow R}$.

Recall that the region $S \subset Q$ in \cref{cor:anyon-distinguishability} is an annulus of the form shown in Fig.~\ref{fig:geometries}(d), with thickness $w$, such that the vertices $v$ and $v'$ are in separate connected subsets of $S^c \coloneqq Q \setminus S$ (for concreteness we take $v$ to be in the interior of the annulus). We will prove that after noise is applied, the marginals of the noise corrupted states $\tilde{\rho}^{(v,v'), \mu}_Q =\mathcal{N}_Q[\rho^{(v,v'), \mu}_Q]$ on  $S$  are approximately distinguishable, provided the noise strength is below threshold. Accordingly, we consider the concatenated channel $\mathcal{M}_{S \leftarrow R}^{(v,v')} \coloneqq \tr_{S^c} \circ \mathcal{N}_Q \circ \mathcal{C}^{(v,v')}_{Q \leftarrow R}$ and its complementary channel $\widehat{\mathcal{M}}_{E_S \leftarrow R}^{(v,v')}$, constructed in the same manner as Eqs.~(\nopareneqref{eq:environment S}, \nopareneqref{eq:noise channels subsystem}). Our strategy is to show that the algebra $\mathfrak{K}$ is approximately recoverable for $\mathcal{M}_{S \leftarrow R}^{(v,v')}$ by deriving an upper bound of the form $\delta_{\mathfrak{K}}(E_S) \leq \epsilon$, where  $\delta_{\mathfrak{K}}(E_S)$ is defined in Eq.~\eqref{eq:delta ES}. By \cref{lem:effective operator}, this implies that there are Hermitian operators $\tilde{K}^\mu_S$ satisfying $0 \preceq \tilde{K}^\mu_S \preceq I$ and supported on $S$, such that
\begin{align}
    \left|\tr[(\tilde{K}^{\mu'}_S \otimes I_{S^c}) \tilde{\rho}^{(v,v'),\mu}_Q] - \delta_{\mu, \mu'}\right| \leq \epsilon.
\end{align}
We then have, for any $\mu \neq \nu$,
\begin{align}
    \mathsf{T}\big(\tilde{\rho}^{(v,v'),\mu}_Q, \tilde{\rho}^{(v,v'),\nu}_Q\big) = \max_{0 \preceq K \preceq I}\tr[K(\tilde{\rho}^{(v,v'),\mu}_Q- \tilde{\rho}^{(v,v'),\nu}_Q)] \geq \tr[(\tilde{K}^{\mu}_S \otimes I_{S^c})(\tilde{\rho}^{(v,v'),\mu}_Q- \tilde{\rho}^{(v,v'),\nu}_Q)] \geq 1 - 2\epsilon.
\end{align}
Thus, defining $\delta = 4\epsilon$, the states are pairwise $\delta$-distinguishable (in the sense of \cref{cor:anyon-distinguishability}) if $\mathfrak{K}$ is $\epsilon$-recoverable. 

The remainder of the proof follows the same form as that of \cref{thm:fat}. Again, $\mathfrak{K}$ is a maximally commutative subalgebra of $\mathfrak{B}(\mathscr{H}_R)$, so $C(\mathfrak{K}) = \mathfrak{K}$, and the super-projector acts as a dephasing channel, $\mathcal{P}^{\mathfrak{K}}_R[\sigma_R] = \sum_\mu  \langle \mu|\sigma_R|\mu\rangle_R \dyad{\mu}_R $. Then, we can upper bound $\delta_{\mathfrak{K}}(E_S)$ by showing that $\mathcal{M}_{E_S \leftarrow R}^{(v,v')}$ and $\mathcal{M}_{S \leftarrow R}^{(v,v')} \circ \mathcal{P}^{\mathfrak{K}}_R$ obey a local indistinguishability property, after which \cref{lem:bound-norm-connectivity} will be invoked.

The local indistinguishability property in question has the same structure as that appearing in the proof of \cref{lem:indist q double}. In particular, we let $\mathscr{W}$ be the collection of graph paths on the subgraph $G^\#_{\mathrm{sq}}|_S$ that connect the inner boundary $S_{\text{in}}$ to the outer boundary $S_{\text{out}}$. Then, $\mathcal{C}_{Q \leftarrow R}^{(v,v')}$ and $\mathcal{C}_{Q \leftarrow R}^{(v,v')} \circ \mathcal{P}^{\mathfrak{K}}_R$ are $(S^c, \mathscr{W}$)-indistinguishable in the sense of \cref{def:locally-indistinguishable}. To prove this, we follow the same set of steps as in the proof of \cref{lem:logical path} to show that if $\Sigma$ does not contain some $W \in \mathscr{W}$, then $\Sigma^c$ supports a closed ribbon operator $K^\mu_r$ that encircles the interior of the annulus $S$. Thanks to 
Eq.~\eqref{eq:anyon distinguish ribbon}, such an operator would obey $K^\mu_r \cdot \mathcal{C}^{(v,v')}_{Q \leftarrow R}[\sigma_R] = \mathcal{C}^{(v,v')}_{Q \leftarrow R}[\dyad{\mu}_R \cdot \sigma_R]$ for any operator $\sigma_R \in \mathfrak{B}(\mathscr{H}_R)$. By analogy to Eq.~\eqref{eq:local indist logical}, we then have
\begin{align}
    \mathcal{G}_{\Sigma S^c \leftarrow R}[\sigma_R] &= \sum_{\mu \in \text{Irr}(D(G))}\tr_{\Sigma^\mathrm{c}} \circ \mathcal{C}^{(v,v')}_{Q \leftarrow R}[\ket{\mu} \langle \mu| \sigma_R| \mu\rangle \bra{\mu}_R] \nonumber\\
    &= \sum_{\mu \in \text{Irr}(D(G))}\tr_{\Sigma^\mathrm{c}} \big[ K_r^\mu\cdot  \mathcal{C}^{(v,v')}_{Q \leftarrow R}[ \sigma_R| \mu\rangle \bra{\mu}_R]\big] \nonumber\\
    &= \sum_{\mu \in \text{Irr}(D(G))}\tr_{\Sigma^\mathrm{c}} \big[  \mathcal{C}^{(v,v')}_{Q \leftarrow R}[ \sigma_R| \mu\rangle \bra{\mu}_R] \cdot K_r^\mu\big] \nonumber\\
    &= \sum_{\mu \in \text{Irr}(D(G))}\tr_{\Sigma^\mathrm{c}} \big[  \mathcal{C}^{(v,v')}_{Q \leftarrow R}[ \sigma_R| \mu\rangle \bra{\mu}_R] \big] = \tr_{\Sigma^\mathrm{c}} \circ \mathcal{C}^{(v,v')}_{Q \leftarrow R}[\sigma_R] = \mathcal{F}_{\Sigma S^c \leftarrow R}[\sigma_R].
\end{align}
Since the shortest path in $\mathscr{W}$ has length $\geq w$, we also have $d_{\mathrm{min}}(\mathscr{W}) \geq w$.

Combining all these results with \cref{lem:bound-norm-connectivity}, we find that, if the noise strength $T$ is below the value $T_\star = [1 + \ln 7]^{-1}$, then
\begin{align}
    \delta_{\mathfrak{K}}(E_S) \leq c_1(T)|S| e^{-c_2(T) w_S}
    \label{eq:delta anyon bound}
\end{align}
which gives the desired result.

We treat the $\mathbb{Z}_2$ surface code separately, in order to incorporate the adaptations described in Appendix \ref{apdx:z2 improvements}. 
We will also treat the $e$- and $m$-type anyons separately, reflecting the slightly different statement of the corollary for the $\mathbb{Z}_2$ case. By analogy to Eq.~\eqref{eq:anyon encoding}, we can define an encoding map that acts as $\mathcal{C}^{(v,v')}_{Q \leftarrow R}[\sigma_R] = \rho_Q \langle+|\sigma_R|+\rangle + U_\gamma \rho_Q U_\gamma^\dagger \langle-|\sigma_R|-\rangle$, where $\rho_Q$ is the code state without anyons, $\gamma$ is a path on the primal lattice connecting $v$ to $v'$, and $\ket{\pm}_R = \frac{1}{\sqrt{2}}(\ket{0}_R \pm \ket{1}_R)$ are $X_R$ eigenstates. Distinguishability then corresponds to preservation of the Pauli-$X$ algebra $\mathfrak{X}$. We define $\mathcal{\Delta}_{Q \leftarrow R} = \mathcal{C}_{Q \leftarrow R}^{(v,v')}\circ(\mathrm{id}_R - \mathcal{P}^{\mathfrak{X}}_R)$, and employ the same decomposition as in Section \ref{apdx:z2 improvements} [Eqs.~(\nopareneqref{eq:delta decomp X repeat}, \nopareneqref{eq:delta sigma X proj})], with $\bar{Z}_S$ and $\bar{Z}_{S^c}$ now defined as the parts of $\bar{Z}_\gamma$ that are supported in $S$ and $S^c$, respectively. The subsequent steps  can be repeated by expanding the operator $\mathbb{B}_{P_S} \bar{Z}_S$ in terms of Pauli strings. This again can be represented as a sum over Ising variables $\vec{\sigma} \in \{\pm 1\}^{P_S}$, such that the corresponding region $S[\vec{\sigma}] = \{\ell \in S : \prod_{p \in \partial^\top \ell}\sigma_p = (-1)^{\ell \in \gamma}\}$. We can view $S[\vec{\sigma}]$ as the set of domain walls of the configuration $\vec{\sigma}$ with twisted boundary conditions applied along the path $\gamma$, and again we can keep only the maps $\mathcal{\Delta}^{(\Sigma; X)}_{Q \leftarrow R}$ such that $\Sigma = S[\vec{\sigma}]$ for some $\vec{\sigma} \in \{\pm 1\}^{P_S}$. These regions $\Sigma$ can be decomposed into self-osculating walks as before.

In this setting, we will define a spanning walk as a self-osculating walk that starts on the inner boundary of $S$ and ends on the outer boundary. Due to the twisted boundary conditions, for every $\vec{\sigma} \in \{\pm\}^{P_S}$ the components of the corresponding region $\Sigma$ must feature at least one spanning walk. The rest of the argument then follows identically, with $w$ being the length of the shortest spanning walk. This proves the distinguishability of the noise-corrupted states with and without $e$ anyons. The $m$ anyon case proceeds identically, with Pauli-$Z$s in place of Pauli-$X$s, and the dual and primal lattices exchanged. \hfill $\square$

\subsection{Corollary \ref{cor:recoverability}: Recoverability of logical information \label{apdx:full recoverability proof}}

We can apply Theorem \ref{thm:fat} taking $S$ to be the entire system to show that for the concatenated channel $\mathcal{M}_{Q' \leftarrow R} = \mathcal{N}_{Q' \leftarrow Q} \circ \mathcal{C}_{Q \leftarrow R}$, the algebras $\mathfrak{K}_{\leftrightarrow}$ and $\mathfrak{K}_{\updownarrow}$ are each $\epsilon$-recoverable, with $\epsilon = e^{-\Omega(L)}$ (recalling that $L = \min(L_x, L_y)$). As claimed in the main text [Eq.~\eqref{eq:algebra conds}], the algebras  $\mathfrak{K}_{\leftrightarrow}$ and $\mathfrak{K}_{\updownarrow}$ together generate the whole matrix algebra $\mathfrak{B}(\mathscr{H}_R)$, which suggests that being able to recovery both subalgebras simultaneously should allow one to recover the full logical information. Here we will prove this is the case, and, \textit{inter alia}, will prove Eq.~\eqref{eq:algebra conds}.

Assume that $\delta_{\mathfrak{K}_{\leftrightarrow}}(E), \delta_{\mathfrak{K}_{\updownarrow}}(E) \leq \epsilon$. Then note that
\begin{align}
    \delta_{\mathfrak{B}(\mathscr{H}_R)}(E) &=  \norm{\widehat{\mathcal{M}}_{E_S \leftarrow R}\circ (\mathrm{id}_R - \mathcal{D}_R)}_\diamond \nonumber\\ &= \norm{\widehat{\mathcal{M}}_{E_S \leftarrow R}\circ (\mathrm{id}_R - \mathcal{P}^{\mathfrak{K}_{\leftrightarrow}}_R \circ \mathcal{P}^{\mathfrak{K}_{\updownarrow}}_R) \circ \mathcal{Q}_R}_\diamond \nonumber\\
    &= \norm{\widehat{\mathcal{M}}_{E_S \leftarrow R}\circ \big[(\mathrm{id}_R - \mathcal{P}^{\mathfrak{K}_{\updownarrow}}_R) + (\mathrm{id}_R - \mathcal{P}^{\mathfrak{K}_{\leftrightarrow}}_R)\circ \mathcal{P}^{\mathfrak{K}_{\updownarrow}}_R\big]\circ \mathcal{Q}_R }_\diamond \nonumber\\
    &\leq \left(\norm{\widehat{\mathcal{M}}_{E_S \leftarrow R}\circ (\mathrm{id}_R - \mathcal{P}^{\mathfrak{K}_{\updownarrow}}_R)}_\diamond + \norm{\widehat{\mathcal{M}}_{E_S \leftarrow R}\circ (\mathrm{id}_R - \mathcal{P}^{\mathfrak{K}_{\updownarrow}}_R)}_\diamond\right)\|\mathcal{Q}_R\|_\diamond \leq 2\epsilon \|\mathcal{Q}_R\|_\diamond,
    \label{eq:delta split QR}
\end{align}
where, in going to the second line,logic
we defined the map
\begin{align}
    \mathcal{Q}_R &\coloneqq (\mathrm{id}_R - \mathcal{A}_R)^{-1} (\mathrm{id}_R - \mathcal{D}_R) & \text{where }\mathcal{A}_R \coloneqq \mathcal{P}^{\mathfrak{K}_{\leftrightarrow}}_R \circ \mathcal{P}^{\mathfrak{K}_{\updownarrow}}_R.
    \label{eq:Qr map def}
\end{align}
Here, the inverse need only be defined on the image of $(\mathrm{id}_R - \mathcal{D}_R)$, namely the space of traceless operators---we will show later on that this is indeed well defined.

Observe that the map $\mathcal{A}_R$ annihilates any operators that are purely off-diagonal in the basis $\{\ket{{\mu}^\updownarrow}_R\}_\mu$, and outputs an operator that is diagonal in the basis $\{\ket{{\mu}^\leftrightarrow}_R\}_\mu$. Thus its action can be defined in terms of a matrix whose rows and columns are labeled by anyon times $\mu, \nu \in \mathrm{Anyons}$, with coefficients $q_{\mu \nu}$ chosen such that
\begin{align}
    \mathcal{A}_R\left[\sigma_R\right] = \sum_{\mu \nu} q_{\mu \nu} \langle \mu^\updownarrow|\sigma_R| \mu^{\updownarrow}\rangle \dyad{\nu^{\leftrightarrow}}.
    \label{eq:AR diag action}
\end{align}
In particular, by Eq.~\eqref{eq:S matrix}, we have $q_{\mu \nu} = |S_{\mu \nu}|^2$, where $S$ is the modular $S$ matrix. Since $S$ is unitary, $q_{\mu \nu}$ are the elements of a doubly stochastic matrix. 
The depolarizing channel $\mathcal{D}_R$ also has the same form, with coefficients $q'_{\mu \nu} = \frac{1}{d_L}$ independent of $\mu, \nu$. 

Let us consider consider integer powers of $\mathcal{A}_R$. These can again be written in the form \eqref{eq:AR diag action}, with different coefficients. To express these coefficients, it is convenient to define the matrix $K$ with elements $K_{\mu\nu} = \sum_\lambda q_{\mu \lambda} q_{\nu \lambda}$, i.e.~$K = q q^\top$. Note that $K$ is itself a doubly stochastic matrix. Then,
\begin{align}
    \mathcal{A}_R^n\left[\sigma_R\right] &= \sum_{\mu \nu} [K^{n-1}q]_{\mu \nu} \langle \mu^\updownarrow|\sigma_R| \mu^{\updownarrow}\rangle \dyad{\nu^{\leftrightarrow}}, & n = 1, 2, \ldots
    \label{eq:AR power}
\end{align}
This follows because $\mathcal{A}^n_R$ involves alternating projections of states $\ket{\mu^{\updownarrow}}$ onto the complementary basis $\ket{\nu^\leftrightarrow}$ and back again, each of which is described by the matrix $q$ or $q^\top$. That is, after the first application of $\mathcal{A}_R$, which generates a diagonal operator $\sum_{\nu} p_\nu \dyad{\nu^{\leftrightarrow}}$, each successive application is equivalent to applying the matrix $K$ to the distribution $p = (p_\nu)_\nu$.

Now, we use Eq.~\eqref{eq:S matrix q dimension}, which implies that every anyon $\mu$ has nonzero overlap with the vacuum sector $\varnothing$. In particular, we have $q_{\mu \varnothing} = q_{\varnothing \mu} = d_\mu^2/\mathscr{D}^2$, where $d_\mu$ is the quantum dimension of anyon $\mu$, and $\mathscr{D} = \sqrt{\sum_\mu d_\mu^2}$ is the total quantum dimension. If we define $r_\mu = d_\mu^2/\mathscr{D}^2$, which is a properly normalised probability distribution, then we have
\begin{align}
    K_{\mu \nu} &= \sum_{\lambda} q_{\mu \lambda} q_{\nu \lambda} \geq q_{\mu \varnothing}q_{\nu \varnothing} \geq\alpha r_\mu\;\forall \mu, \nu & \text{where }\alpha \coloneqq \frac{1}{\mathscr{D}^2}.
    \label{eq:markov lower bound}
\end{align}
This implies a uniform ergodicity property for the doubly stochastic matrix $K$ via the following argument (adapted from Theorem 16.2.4 of Ref.~\onlinecite{meyn2012markov}). We can write
\begin{align}
    K_{\mu \nu} &= \alpha K^{(0)}_{\mu \nu} + (1-\alpha)K_{\mu \nu}^{(1)} \nonumber\\
    \text{where }K^{(0)}_{\mu \nu} &= r_\mu, & K_{\mu\nu}^{(1)} = \frac{K_{\mu \nu} - \alpha r_\mu}{1-\alpha}
\end{align}
Both $K^0_{\mu \nu}$ and $K^1_{\mu \nu}$ are valid transition amplitudes for a Markov chain, since all coefficients are non-negative by virtue of \eqref{eq:markov lower bound}, and satisfy $\sum_\nu K^{(a)}_{\mu \nu} = 1$ for $a = 0,1$. Crucially, because $K^{(0)}_{\mu \nu}$ is independent of $\nu$, if $r_\mu$ and $r'_\mu$ are two probability distributions then $K^{(0)}$ annihilates the vector $(r - r')_\mu$, since $\sum_{\nu} K_{\mu \nu}^{(0)}(r_\nu - r'_\nu) = \sum_\nu(r_\nu - r'_\nu) = 0$. Accordingly,
\begin{align}
		\mathcal{A}^n_R &= \mathcal{A}^n_R\circ \mathcal{D}_R + \mathcal{A}^n_R\circ (\mathrm{id}_R -\mathcal{D}_R) \nonumber\\  &= \mathcal{D}_R + (1-\alpha)^{n-1} \big((\mathcal{A}^{(1)}_R)^{n-1} \circ \mathcal{A}_R \circ  (\mathrm{id}_R - \mathcal{D}_R)\big) \label{eq:A channel nth power} \\  \text{where } \mathcal{A}^{(1)}_R[\sigma_R] &\coloneqq \sum_{\mu \nu} [(K^{(1)})^{n-1}q]_{\mu \nu} \langle \nu^{\updownarrow}|\sigma_R|\nu^\updownarrow\rangle \dyad{\mu^{\leftrightarrow}}
        \label{eq:AR1 def}
	\end{align}
    In going to the second line, we use the fact that the identity matrix is a fixed point of $\mathcal{A}_R$, and thus $\mathcal{A}_R \circ \mathcal{D}_R = \mathcal{D}_R$. Then, the output of $\mathcal{A}_R \circ  (\mathrm{id}_R - \mathcal{D}_R)$ can always be written as the difference of two classical distributions, hence the terms featuring $\mathcal{A}^{(0)}_R$ can be discarded.

    Since $\mathcal{A}_R$ is the product of two projectors, the von Neumann alternating projection theorem can be invoked (see e.g.~Ref.~\onlinecite{kayalar1988error}). This implies that $\lim_{n\rightarrow \infty}\|\mathcal{A}_R^n - \mathcal{P}^\cap_R\|_\ast = 0$ for any norm $\|\,\cdot\,\|_\ast$, where $\mathcal{P}^\cap_R$ is the projector onto the intersection of $\mathfrak{K}_{\leftrightarrow}$ and $\mathfrak{K}_{\updownarrow}$. Now, the map $\mathcal{A}^{(1)}_R$ [Eq.~\eqref{eq:AR1 def}], which is constructed from the transition amplitudes $K^1_{\mu \nu}$, is itself a CPTP map, and hence its integer powers have unit diamond norm. Then, since $0 \leq (1-\alpha) < 1$, Eq.~\eqref{eq:A channel nth power}  implies convergence of $\mathcal{A}_R^n$ to $\mathcal{D}_R$ in diamond norm as $n \rightarrow \infty$. Since $\mathcal{D}_R$ is the projector onto the trivial algebra $\mathbb{C}I$, this implies $\mathfrak{K}_{\leftrightarrow} \cap \mathfrak{K}_{\updownarrow} = \mathbb{C}I$. Since (for finite-dimensional Hilbert spaces) the commutant of the intersection of two algebras is equal to the algebra generated by their individual commutants \cite{arveson1976invitation}, and $C(\mathfrak{K}_\leftrightarrow) = \mathfrak{K}_{\leftrightarrow}$ by virtue of $\mathfrak{K}_{\leftrightarrow}$ being a maximal commutative subalgebra (similarly for $\mathfrak{K}_{\updownarrow}$), we also have that $\mathfrak{K}_{\leftrightarrow}$ and $\mathfrak{K}_{\updownarrow}$ together generate the whole logical algebra $\mathfrak{B}(\mathscr{H}_R) = C(\mathbb{C}I)$. This proves Eq.~\eqref{eq:algebra conds}.

    We now have the necessary tools to show that $\mathcal{Q}_R$ is well-defined, and to bound its diamond norm. First, the above arguments make clear that the only fixed point of $\mathcal{A}_R$ is the identity matrix, so the map $(\mathrm{id}_R - \mathcal{A}_R)$ does not annihilate any operator in the traceless subspace (the image of $\mathrm{id}_R - \mathcal{A}_R$). This establishes the existence of the necessary inverse in Eq.~\eqref{eq:Qr map def}. We can formally expand this as a power series
    \begin{align}
        \mathcal{Q}_R &= \sum_{n=0}^\infty \mathcal{A}_R^n \circ (\mathrm{id}_R - \mathcal{D}_R) \nonumber\\
        &= (\mathrm{id}_R-\mathcal{D}_R) + \sum_{n=1}^{\infty} (1-\alpha)^{n-1} \big((\mathcal{A}^{(1)}_R)^{n-1} \circ \mathcal{A}_R \circ  (\mathrm{id}_R - \mathcal{D}_R)\big)
    \end{align}
    where we use Eq.~\eqref{eq:A channel nth power} in going to the second line. Since $0 <1 - \alpha < 1$, this sum converges in diamond norm, which establishes the validity of the power series. Then, taking the diamond norm of both sides and applying the triangle inequality, we obtain
    \begin{align}
        \|\mathcal{Q}_R \|_\diamond &\leq \norm{\mathrm{id}_R - \mathcal{D}_R}_\diamond + \sum_{n=1}^\infty (1-\alpha)^{n-1}\norm{  \big((\mathcal{A}^{(1)}_R)^{n-1} \circ \mathcal{A}_R \circ  (\mathrm{id}_R - \mathcal{D}_R)\big)}_\diamond \nonumber\\
        &\leq \|\mathrm{id}_R - \mathcal{D}_R\|_\diamond\left(1 + \sum_{n=1}^\infty (1-\alpha)^{n-1} \right) \leq 2\left(1 +\frac{1}{\alpha}\right) = 2(1 +\mathscr{D}^2).
    \end{align}
    Combined with Eq.~\eqref{eq:delta split QR}, we infer that $\delta_{\mathfrak{B}(\mathscr{H}_R)}(E) \leq 4\epsilon(1+\mathscr{D}^2)$. From Eq.~\eqref{eq:eps rec bound}, this can be used to show that the full logical algebra $\mathfrak{B}(\mathscr{H}_R)$ is $\epsilon'$-recoverable for suitable $\epsilon'$. In particular, since $\mathscr{D}$ is a system size-independent constant, and $\epsilon = e^{-\Omega(L)}$, we can conclude $\epsilon' = e^{-\Omega(L)}$, thus proving the corollary. \hfill $\square$\\

As explained in the main text, approximate recoverability of the full logical algebra implies that the coherent information $\mathsf{I}_C(R \rangle Q')$ is close to its maximal value of $\ln d_L$. Specifically, we claim that if $\mathfrak{B}(\mathscr{H}_R)$ is $\epsilon_{\mathrm{rec}}$-recoverable in the sense of Eq.~\eqref{eq:recov error} for some $\epsilon_{\mathrm{rec}} < 1$, then
\begin{align}
    0 \leq \ln d_L -\mathsf{I}_C(R \rangle Q') \leq 2\epsilon_{\rm rec}\ln d_L + 4 \epsilon_{\mathrm{rec}}\ln\epsilon_{\rm rec}^{-1} = \mathcal{O}(\epsilon_{\mathrm{rec}} \ln \epsilon_{\mathrm{rec}}^{-1})
    \label{eq:coh inf bound apdx}
\end{align}
We prove this relation now.

First recall the definition of the coherent information $\mathsf{I}_C(R \rangle Q') = \mathsf{S}(\tilde{\rho}^{Q'}) - \mathsf{S}(\tilde{\rho}^{Q'R})$ in terms of the state $\tilde{\rho}_{Q'R} \coloneqq ((\mathcal{N}_{Q' \leftarrow Q} \circ \mathcal{C}_{Q \leftarrow R'}) \otimes \mathrm{id}_R)[\Phi_{R'R}]$, where $\Phi_{R'R}$ is the projector onto a Bell state. This can alternatively be expressed as $\mathsf{I}_C(R \rangle Q') = -\mathsf{S}(\tilde{\rho}^{Q'R}|R)$, where for a bipartite state $\sigma^{AB}$,  $\mathsf{S}(\sigma^{AB}|B) \coloneqq \mathsf{S}(\sigma^{AB}) - \mathsf{S}(\sigma^B)$ is the conditional entropy.

Due to the definition of approximate recoverability, our assumptions imply the existence of a recovery channel $\mathcal{R}_{R'' \leftarrow Q'}$ such that $\|\mathcal{R}_{R'' \leftarrow Q'} \circ \mathcal{N}_{Q' \leftarrow Q} \circ \mathcal{C}_{Q \leftarrow R} - \mathrm{id}_{R'' \leftarrow R}\|_\diamond \leq \epsilon_{\rm rec}$. By the data processing inequality for the coherent information \cite{watrous2018theory}, we have $\mathsf{I}_c(R \rangle Q')_{\tilde{\rho}} \geq \mathsf{I}_c(R \rangle R'')$, 
where $\mathsf{I}_c(R \rangle R'')$ is the coherent information of the state
\begin{align}
    \rho_{R''R} \coloneqq (\mathcal{R}_{R'' \leftarrow Q'}  \otimes \mathrm{id}_R)[\tilde{\rho}^{Q'R}] = \big((\mathcal{R}_{R'' \leftarrow Q'} \circ \mathcal{N}_{Q' \leftarrow Q} \circ \mathcal{C}_{Q \leftarrow R'}) \otimes \mathrm{id}_R\big)[\Phi_{R'R}].
\end{align}
Then, thanks to the definition of the diamond distance, we have $\|\rho_{R''R} - \Phi_{R''R}\|_1 \leq \epsilon_{\mathrm{rec}}$. Finally, we use the following continuity bound on the conditional entropy, from Ref.~\cite{winter2016tight} (Lemma 2 therein),
\begin{align}
    \left|\mathsf{S}(\rho^{AB}|B) - \mathsf{S}(\sigma^{AB}|B)\right| &\leq 2 T\log(d_A) + (1+T)h_2\left( \frac{T}{1+T}\right) & \text{where }T \coloneqq \mathsf{T}(\rho^{AB}, \sigma^{AB}) \equiv \frac{1}{2}\|\rho^{AB} - \sigma^{AB}\|_1
\end{align}
with $d_A$ the Hilbert space dimension of subsystem $A$, and $h_2(x) \coloneqq - x\ln x - (1-x)\ln(1-x)$ the binary entropy function. Applying this to the coherent information, and using the fact that the coherent information for the Bell state $\Phi_{R''R}$ equals $\ln d_L$, we obtain
\begin{align}
    \left|\mathsf{I}_c(R \rangle R'')  - \ln d_L \right| \leq \epsilon_{\mathrm{rec}}\ln d_L - (1 +\epsilon_{\rm rec}/2)h_2 \left(\frac{\epsilon_{\mathrm{rec}}/2}{1+\epsilon_{\mathrm{rec}}} \right).
\end{align}
We finally combine the above with the data processing inequality $\mathsf{I}_c(R \rangle Q')_{\tilde{\rho}} \geq \mathsf{I}_c(R \rangle R'')$, along with the fact that $(1+x)h_2(\frac{x}{1+x}) \leq 4 x\ln(1/x)$ for $0 <x < 1/2$, in order to prove Eq.~\eqref{eq:coh inf bound apdx}.

\section{3D $\mathbb{Z}_2$ surface code} \label{apdx:3d}
In this section, we prove how the methods developed for the 2D $\mathbb{Z}_2$ surface code and the non-abelian codes can be generalised to 3D codes, focusing for concreteness on the 3D $\mathbb{Z}_2$ surface code. 




\subsection{Recoverability of the $\mathfrak{Z}$-algebra}
On the 3D $\mathbb{Z}_2$ surface code, $Z$-errors are matchable. This fact will allow us to use the same optimised proof strategy for the 2D $\mathbb{Z}_2$ surface code (Appendix \ref{apdx:z2 improvements}), with suitable adaptations to the 3D case. Here, we can choose a tube-like region $S_\updownarrow$, chosen to be a cuboidal region of dimension $w \times w \times L_\mathrm{z}$ for concreteness [\cref{fig:3d-regions}(a)].

We follow the stepds described in Appendix \ref{apdx:z2 improvements}, with the $\mathfrak{Z}$-algebra in place of $\mathfrak{X}$, up until the point of \cref{eq:delta sigma X proj}. This results in the expression
\begin{align}
    \tilde{\mathcal{\Delta}}_{Q \leftarrow R}^{\,(\Sigma; Z)}[c_X X_R + c_Y Y_R] = \mathbb{B}_P (c_X + \iu c_Y \bar{Z})\mathbb{A}_{V \setminus V_S} \bar{X}_{S^c_\updownarrow}\cdot  \left(\left[\bigotimes_{\ell \in \Sigma}(\mathrm{id}_\ell - \mathcal{E}_\ell^Z) \right]\otimes \left[\bigotimes_{\ell \in \Sigma^\mathrm{c}}  \mathcal{E}_\ell^Z \right] \right)[\mathbb{A}_{V_{S_\updownarrow}} \bar{X}_{S_\updownarrow}].
\end{align}
Here, $V_{S_\updownarrow}$ is the subset of vertices for which all adjacent edges are in $S_\updownarrow$, and  $\bar{X} = \bar{X}_{S^\mathrm{c}_\updownarrow} \bar {X}_{S_\updownarrow}$ is a sheet of Pauli-$X$ operators in the dual lattice, which we choose to be on the bottom of the lattice [\cref{fig:3d-geometry}], and $\bar {X}_{S_\updownarrow}$, $\bar{X}_{S^\mathrm{c}_\updownarrow}$  are the corresponding restrictions of this sheet to $S_\updownarrow$, $S_\updownarrow^c$.

In a similar fashion to the 2D $\mathbb{Z}_2$ case, we now consider the Pauli string decomposition of the operator $\mathbb{A}_{V_{S_\updownarrow}} \bar{X}_{S_\updownarrow}$. This can be expressed by expanding the product in the definition $\mathbb{A}_{V_{S_\updownarrow}} = \prod_{v \in V_{S_\updownarrow}}\frac{I + A_v}{2}$, which can again be written as a sum over Ising variables $\vec \sigma \in \{\pm 1\}^{V_{S_\updownarrow}}$, which label whether the term $I$ (for $\sigma_v = +1$) or $A_v$ (for $\sigma_v = -1$) is chosen. Since only two vertices neighbour the same edge (a reflection of the matchable nature of $Z$ syndromes), we again have  we may again use the same language of domain walls: If we define additional fixed Ising spins on the 1-valent vertices at the rough boundaries $S$, which are fixed to $\sigma_v = +1$ ($-1$) for the top (bottom) boundary, then the Pauli string for the configuration $\vec{\sigma}$ is supported on the qubits $S[\vec{\sigma}] \subseteq S$ where there is a domain wall. The map $\tilde{\mathcal{\Delta}}_{Q \leftarrow R}^{\,(\Sigma; Z)}$ is then nonzero only if the region $\Sigma$ corresponds to the set of domain walls of a given Ising configuration $\vec{\sigma}$, with the boundary conditions just specified. 



By defining a 3D version of the deformation rule \eqref{eq:domain-wall-ambiguity}, any such $\Sigma$ can be uniquely decomposed into a collection of self-osculating surfaces (SOSs), the precise definition of which is given in Ref.~\onlinecite{pkim2025existence}. Again the rule is chosen such that the $+$ spins always stay together. As an example, we have
\begin{align}
\intikz{
\draw [rounded corners=4mm, ultra thick, blue!75] (0,-1,-1) -- (0,0,-1) -- (1,0,-1);
\draw [rounded corners=4mm, ultra thick, blue!75] (-1,1,0) -- (-1,0,0) -- (-1,0,1);
\draw [rounded corners=4mm, ultra thick, blue!75] (0,-1,1) -- (0,-1,-1);
\draw [gray!50, ultra thick] (0, -1, 0) to (0, 1, 0);
\draw [gray!50, ultra thick] (-1, 0, 0) to (1, 0, 0);
\draw [gray!50, ultra thick] (0, 0, -1) to (0, 0, 1);
\node[circle, inner sep=0, draw=black, fill=white, minimum size=9, thick] at (-0.5, -0.5, -0.5) {};
\node[circle, inner sep=0, draw=black, fill=black, minimum size=9, thick] at (0.5, -0.5, -0.5) {};
\node[circle, inner sep=0, draw=black, fill=white, minimum size=9, thick] at (-0.5, 0.5, -0.5) {};
\node[circle, inner sep=0, draw=black, fill=white, minimum size=9, thick] at (0.5, 0.5, -0.5) {};
\node[circle, inner sep=0, draw=black, fill=white, minimum size=9, thick] at (-0.5, -0.5, 0.5) {};
\node[circle, inner sep=0, draw=black, fill=black, minimum size=9, thick] at (0.5, -0.5, 0.5) {};
\node[circle, inner sep=0, draw=black, fill=black, minimum size=9, thick] at (-0.5, 0.5, 0.5) {};
\node[circle, inner sep=0, draw=black, fill=white, minimum size=9, thick] at (0.5, 0.5, 0.5) {};
\draw[dashed] (1, 1, 1) -- (-1, 1, 1) -- (-1, -1, 1) -- (1, -1, 1) -- (1, 1, 1);
\draw[dashed] (1, 1, 1) -- (1, 1, -1);
\draw[dashed] (-1, 1, 1) -- (-1, 1, -1);
\draw[dashed] (-1, -1, 1) -- (-1, -1, -1);
\draw[dashed] (1, -1, 1) -- (1, -1, -1);
\draw[dashed] (1, 1, -1) -- (-1, 1, -1) -- (-1, -1, -1) -- (1, -1, -1) -- (1, 1, -1);
\draw [rounded corners=4mm, ultra thick, blue!75] (0,-1,1) -- (0,0,1) -- (1,0,1);
\draw [rounded corners=4mm, ultra thick, blue!75] (-1,0,1) -- (0,0,1) -- (0,1,1);
\draw [rounded corners=4mm, ultra thick, blue!75] (0,1,1) -- (0,1,0) -- (-1,1,0);
\draw [rounded corners=4mm, ultra thick, blue!75] (1,0,1) -- (1,0,-1);
\draw [rounded corners=4mm, ultra thick, blue!75] (0+0.12,0-0.12,1) -- (0+0.12,0-0.12,-1);
} \; ,
\end{align}
where the white (black) circles correspond to vertices with Ising spins in the state $\sigma_v = +1$ ($-1$). Importantly, because of the fixed and opposite boundary conditions, the SOSs of any Ising configuration must feature a \textit{spanning surface}, namely a surface that separates the top boundary of $S$ from the bottom boundary. 

The rest of the steps follow in the same way, and $\delta_{\mathfrak{Z}}(E_S)$ can be upper bounded by a sum over spanning SOSs $\gamma$, with weight $e^{-\beta |\gamma|}$. This sum converges if $\beta > \ln \mu_\cube^{\mathrm{SOS}}$, where $\mu_\cube^{\mathrm{SOS}}$ is the growth constant for SOSs on the cubic lattice, which has been (loosely) upper bounded as  $\mu_\cube^\mathrm{SOS} \leq 20.25$ \cite{pkim2025existence}. The minimum size of any spanning surface is $w^2$. Together this gives the result stated in the main text, \cref{eq:3d claim Z}.


Note that we could have used the general method introduced for non-abelian codes, \cref{lem:bound-norm-connectivity}, to show a similar result, but with worse bounds ($11e \approx 30$ vs. $20.25$). To show approximate correctability of the $\mathfrak{X}$-algebra, we use the general method, shown below.

\subsection{Recoverability of the $\mathfrak{X}$-algebra}


In the case of the $\mathfrak{X}$-algebra, we cannot use a similar approach developed for the 2D $\mathbb{Z}_2$ surface code or $\mathfrak{Z}$ for the 3D case. This is because more than two plaquettes neighbor an edge, which means that $X$-errors are not matchable and the language of domain walls cannot be used. We therefore employ the more general technique developed for non-Abelian codes.

As explained in the main text, we take $S$ to be a slab of size $L_\mathrm{x} \times L_\mathrm{y} \times w_\mathrm{z}$ [\cref{fig:3d-regions}(a)]. For this choice, we will upper bound $\delta_{\mathfrak{X}}(E_S)$ by proving a local indistinguishability property for the channels $\mathcal{C}_{Q \leftarrow R}$ and $\mathcal{C}_{Q \leftarrow R}\circ \mathcal{P}^{\mathfrak{X}}_R$, and then invoke \cref{lem:bound-norm-connectivity}. We first define the graph $G^\#_{\text{cubic}}$ to be the standard edge graph for the cubic lattice, where edges are adjacent if they share a vertex; the induced subgraph for the region $S$ is again denoted $G^\#_{\text{cubic}}|_S$. Let $\mathscr{W} \subset P(S)$ be the set of paths in $G^\#_{\text{cubic}}$ that connect the bottom to the top boundary of $S$. Then, we claim that $\mathcal{C}_{Q \leftarrow R}$ and $\mathcal{C}_{Q \leftarrow R}\circ \mathcal{P}^{\mathfrak{X}}_R$ are $(S^c, \mathscr{W})$-indistinguishable. Indeed, if $\Sigma$ does not contain some such path, then $\Sigma^c$ must support a sheet-like $X$-type logical operator. (This statement can be proved using a generalisation of our argument used to prove \cref{lem:logical path}.) Indistinguishability then follows using the same steps as in Eq.~\eqref{eq:local indist logical}.

The graph $G^\#_{\text{cubic}}$ has degree $k=10$, as shown below:
\begin{align}
	\intikz{
		\draw[tolsky, ultra thick, opacity=0.5] (0, -1, 0) -- (0, 2, 0);
		\draw[tolsky, ultra thick, opacity=0.5] (-1, 1, 0) -- (1, 1, 0);
		\draw[tolsky, ultra thick, opacity=0.5] (0, 1, -1) -- (0, 1, 1);
		\draw[tolsky, ultra thick, opacity=0.5] (-1, 0, 0) -- (1, 0, 0);
		\draw[tolsky, ultra thick, opacity=0.5] (0, 0, -1) -- (0, 0, 1);
		\node[black] at (0, 0.5, 0)[circle,fill,inner sep=1.75pt]{};
		\draw (0, 0.5, 0) -- (0, 1.5, 0);
		\draw (0, 0.5, 0) -- (0, -0.5, 0);
		\draw (0, 0.5, 0) -- (0.5, 1, 0);
		\draw (0, 0.5, 0) -- (-0.5, 1, 0);
		\draw (0, 0.5, 0) -- (0, 1, 0.5);
		\draw (0, 0.5, 0) -- (0, 1, -0.5);
		\draw (0, 0.5, 0) -- (0.5, 0, 0);
		\draw (0, 0.5, 0) -- (-0.5, 0, 0);
		\draw (0, 0.5, 0) -- (0, 0, 0.5);
		\draw (0, 0.5, 0) -- (0, 0, -0.5);
		\node[gray] at (0, 1.5, 0)[circle,fill,inner sep=1.5pt]{};
		\node[gray] at (0, -0.5, 0)[circle,fill,inner sep=1.5pt]{};
		\node[gray] at (0.5, 1, 0)[circle,fill,inner sep=1.5pt]{};
		\node[gray] at (-0.5, 1, 0)[circle,fill,inner sep=1.5pt]{};
		\node[gray] at (0, 1, 0.5)[circle,fill,inner sep=1.5pt]{};
		\node[gray] at (0, 1, -0.5)[circle,fill,inner sep=1.5pt]{};
		\node[gray] at (0.5, 0, 0)[circle,fill,inner sep=1.5pt]{};
		\node[gray] at (-0.5, 0, 0)[circle,fill,inner sep=1.5pt]{};
		\node[gray] at (0, 0, 0.5)[circle,fill,inner sep=1.5pt]{};
		\node[gray] at (0, 0, -0.5)[circle,fill,inner sep=1.5pt]{};
	}
\end{align}
and the minimum length of a spanning path is $w_{\mathrm{z}}$. Therefore, via \cref{lem:bound-norm-connectivity}, we have the result quoted in the main text, Eq.~\eqref{eq:3d claim X}.

\section{Correlated errors}
\subsection{Errors after finite depth circuits \label{apdx:finite-depth}}
Here we will consider a variant of \cref{thm:fat} for the general non-Abelian setting with the change of having code states under separable noise channel following a finite-depth unitary circuit,
\begin{align}
    \mathcal{N}^\mathrm{fd}_{Q' \leftarrow Q} = \mathcal{N}_{Q' \leftarrow Q} \circ \mathcal{U}_t \cdots \mathcal{U}_1 \equiv \mathcal{N}_{Q' \leftarrow Q}  \circ \mathcal{U},
\end{align}
where $\mathcal{N}_{Q' \leftarrow Q} = \bigotimes_{\ell \in Q} \mathcal{N}_\ell$ is a separable channel as before, and $\mathcal{U} = \mathcal{U}_t \cdots \mathcal{U}_1 $ is the unitary channel describing the whole circuit. Again we take $S$ to be a horizontal strip of height $w_{\rm y}$. The corresponding localised channel is $\mathcal{N}_{S \leftarrow Q}^\mathrm{fd} = \tr_{Q' \setminus S} \mathcal{N}^\mathrm{fd}_{Q' \leftarrow Q}$, whose complementary channel is
\begin{align}
    \widehat{\mathcal{N}}^\mathrm{fd}_{E_S \leftarrow Q} = \qty(\bigotimes_{\ell \in S} \widehat{\mathcal{N}}_{E_\ell \leftarrow \ell}) \circ  \mathcal{U} \eqqcolon \widehat{\mathcal{N}}_{S} \circ \mathcal{U}.
\end{align}
where, abusing notation, we write $\widehat{\mathcal{N}}_{S} \coloneqq \bigotimes_{\ell \in S}\widehat{\mathcal{N}}_{E_\ell \leftarrow \ell}$. We seek to upper bound $\delta_\mathfrak{J}(E_S)= \norm{\widehat{\mathcal{N}}^\mathrm{fd}_{E_S \leftarrow Q} \mathcal{\Delta}_{Q \leftarrow R}}_\diamond$, where $\mathcal{\Delta}_{Q \leftarrow R} = \mathcal{C}_{Q \leftarrow R}\circ (\mathrm{id}_R - \mathcal{P}^{\mathfrak{K}_{\leftrightarrow}}_R)$.

Let $\xi$ be the (pairwise) lightcone distance of the circuit $\mathscr{U}$, defined such that for any two operators $O$, $O'$ whose supports are separated by a distance $\geq \xi$, the images under the circuit $\mathcal{U}[O]$, $\mathcal{U}[O']$ do not intersect. (Here we use the standard edge graph to define distances between qudits.) We assume $w_{\rm y} > 2\xi$, and define $S' \subset S$ to be the subregion of $S$ in which all qudits within a distance $\xi$ of the top and bottom boundaries of $S$ are removed---this a horizontal strip of height $w_{\rm y} - 2\xi$ and width $L_{\rm x}$. We will use the same steps as in the proof of \cref{lem:bound-norm-connectivity} to decompose the $\mathcal{\Delta}_{Q \leftarrow R}$ in terms of subregions of $S'$ (not of $S$). By \cref{lem:logical path}, the channels in question obey a $(Q \setminus S', \mathscr{W})$-indistinguishability property, where $\mathscr{W}$ is the set of paths that connect the top and bottom  boundaries of $S'$. All the arguments up to  \cref{eq:delta gamma xy} can then be repeated. 
 Specifically, for each $\Gamma' \in \Omega^{\mathscr{W}}$, we can define the corresponding operator $\mathcal{\Delta}^{\Gamma'}_{Q \leftarrow R}$ expressed as \cref{eq:delta gamma xy} (with the suitable adaptations), which is traceless for each qudit $\ell \in \bigcup_{W^+ \in \Gamma'}W^+ \eqqcolon X[\Gamma']$,
\begin{align}
    \mathcal{\Delta}^{\Gamma'}_{Q \leftarrow R} = (\mathrm{id}_\ell - \mathcal{D}_\ell)\circ \mathcal{\Delta}^{\Gamma'}_{Q \leftarrow R} \quad \forall \, \ell \in X[\Gamma'].
    \label{eq:circuit delta traceless}
\end{align}
For notational convenience we will write $X \equiv X[\Gamma']$ where clear from context.
We now seek to upper bound the norm $\|\widehat{\mathcal{N}}^{\mathrm{fd}}_{E_S \leftarrow Q} \circ \mathcal{\Delta}^{\Gamma'}_{Q \leftarrow R} \|_\diamond = \|\widehat{\mathcal{N}}_{S} \circ \mathcal{U} \circ \mathcal{\Delta}^{\Gamma'}_{Q \leftarrow R} \|_\diamond$. Our strategy is to cover $X[\Gamma'] \subseteq S'$ with $n_B$ non-overlapping balls $\{B_i\}_{i=1}^{n_B}$, each of radius $\xi$, and centered on some qudit $\ell_i \in X[\Gamma']$. (Here a ball of radius $r$ is the set of qudits that are at at most a graph distance $r$ from its center qudit.) Thanks to the definition of $S'$, we have $B_i \subseteq S$. Then, if $\mathcal{D}_{B_i} =\bigotimes_{\ell \in B_i} \mathcal{D}_\ell$ is the depolarizing channel on a given ball $B_i$, we claim that, as a consequence of Eq.~\eqref{eq:circuit delta traceless},
\begin{align}
    \mathcal{D}_{B_i} \circ \mathcal{U} \circ  \mathcal{\Delta}^{\Gamma'}_{Q \leftarrow R} &= 0.
    \label{eq:depolarize ball}
\end{align}
In turn, defining the map
\begin{align}
    \mathcal{\Theta}_{B_i} \coloneqq \mathrm{id}_{B_i} - \mathcal{D}_{B_i} = \sum_{M \subseteq B_i: \abs{M} \geq 1} \left(\prod_{{\ell} \in M} (\mathrm{id} - \mathcal{D}_\ell)\right) \left(\prod_{{\ell}' \in B \setminus M} \mathcal{D}_{{\ell}'} \right),
    \label{eq:theta Bi def}
\end{align}
we can freely apply this to the output of $\mathcal{U} \circ \mathcal{\Delta}^{\Gamma'}_{Q \leftarrow R}$ without changing the result. Doing this simultaneously for each $i \in [n_B]$, the quantity in question then becomes
\begin{align}
    \norm{\widehat{\mathcal{N}}_{S} \circ \mathcal{U} \circ \mathcal{\Delta}^{\Gamma'}_{Q \leftarrow R}}_\diamond = \norm{\widehat{\mathcal{N}}_{S} \circ \left(\bigotimes_{i=1}^{n_B}\mathcal{\Theta}_{B_i}\right) \circ  \mathcal{U} \circ \mathcal{\Delta}^{\Gamma'}_{Q \leftarrow R}}_\diamond \leq \|\mathcal{\Delta}^{\Gamma'}_{Q \leftarrow R}\|_\diamond \prod_{i=1}^{n_B} \|\widehat{\mathcal{N}}_{B_i} \circ \mathcal{\Theta}_{B_i} \|_\diamond
\end{align}
where $\widehat{\mathcal{N}}_{B_i} \coloneqq \bigotimes_{\ell \in B_i} \widehat{\mathcal{N}}_{E_\ell \leftarrow \ell}$. To evaluate the diamond norm inside the product on the right hand side, we use the expansion \eqref{eq:theta Bi def} and apply the triangle inequality, giving
\begin{align}
    \|\widehat{\mathcal{N}}_{B_i} \circ \mathcal{\Theta}_{B_i} \|_\diamond &\leq \sum_{M \subseteq B_i: \abs{M} \geq 1} \left(\prod_{\ell \in M} \|\widehat{\mathcal{N}}_{E_\ell \leftarrow \ell} \circ (\mathrm{id}_\ell - \mathcal{D}_\ell)\|_\diamond \right) \left(\prod_{\ell \in B_i \setminus M}\widehat{\mathcal{N}}_{E_\ell \leftarrow \ell} \circ \mathcal{D}_\ell \right) \nonumber\\ &\leq \sum_{M \subseteq B_i: \abs{M} \geq 1} \left(\prod_{\ell \in M}\epsilon\right) = (1+\epsilon)^{|B_i|} -1.
\end{align}
where the last inequality follows from the definition of the effective noise strength $T$ [Eq.~\eqref{eq:noise strength def}], which we are assuming satisfies $e^{-1/T} \leq \epsilon$.

As in the main text, we define $A_\xi = O(\xi^2)$ to be the maximum size of a ball of radius $\xi$. We claim that for any $\Gamma'$, it is possible to find $n_B$ balls with the properties described above, where
\begin{align}
    n_B \geq \big\lfloor|X[\Gamma']/A_\xi|\big \rfloor.
    \label{eq:n ball bound}
\end{align}
Since $\norm{\mathcal{\Delta}^{\Gamma'}_{Q \leftarrow R}}_\diamond \leq 2^{|X[\Gamma']|}$, we then have
\begin{align}
    \norm{\widehat{\mathcal{N}}_{S} \circ \mathcal{U} \circ \mathcal{\Delta}^{\Gamma'}_{Q \leftarrow R}}_\diamond \leq 2^{|X[\Gamma']|} [(1+\epsilon)^{A_\xi} - 1]^{|X[\Gamma']|/A_\xi - 1}
\end{align}
Then, if we define an effective inverse noise rate
\begin{align}
    \beta_{\text{eff}} = -\frac{1}{A_\xi} \ln[(1+\epsilon)^{A_\xi} - 1] - \ln 2 \geq \frac{\beta_0}{A_\xi}-2\ln 2
\end{align}
(with the bound on the right holding for $ \epsilon < 1, A_\xi \geq 2$) then the sum over $\Gamma'$ [cf.~\cref{eq:delta norm Z bound}] will converge if $\beta_{\text{eff}} > (T_\star)^{-1}$, where $T_\star$ is the relevant threshold from Theorem \ref{thm:fat-nonab}, as claimed in the main text. Note that, due to the definition of $S'$, the upper bound of $\delta_{\mathfrak{K}}(E_S)$ will scale as $\sim e^{-(w_{\rm y} - 2\xi)}$.

We now prove Eqs.~\eqref{eq:depolarize ball} and \eqref{eq:n ball bound}, starting with the former. By construction, the center $\ell_i$ of the ball $B_i$ is contained in $X[\Gamma']$, and thus by Eq.~\eqref{eq:circuit delta traceless} we have $(\text{id}_{\ell_i} - \mathcal{D}_{\ell_i})\circ \mathcal{\Delta}^{\Gamma'}_{Q \leftarrow R} = 0$. This means that for any operator $O$ in the image of $\mathcal{\Delta}^{\Gamma'}_{Q \leftarrow R}$, its Pauli decomposition must only feature operator strings that act nontrivially on site $\ell_i$. Furthermore, $\mathcal{D}_{B_i}$ annihilates all operators except for those that act as the identity on $B_i$. If we evolve any such operator backwards under $\mathcal{U}^\dagger$, the result will have support within the set of qudits whose distance to $Q \setminus B_i$ is less than $\xi$. Since $\text{dist}(Q \setminus B_i, \ell_i)\geq \xi$, this must act as the identity on $\ell_i$, and thus have vanishing overlap with $O$.

To prove Eq.~\eqref{eq:n ball bound}, we construct an explicit procedure for constructing balls with the desired properties. Let $X \equiv X[\Gamma']$ be a subset of $S'$, and we will denote the set of center points of the balls as $Y$. Choose an arbitrary $\ell_1 \in X$, which we take to be in $Y$. Then, we update $X \mapsto X \setminus B_1$, where $B_1$ is the ball of redius $\xi$ centered on $\ell_1$. This removes at most $A_\xi$ elements (including $\ell_1$ itself). All the remaining elements of $X$ will all be at most $\xi$ from the first choice $\ell_1$, and so we can pick an arbitrary $\ell_2$ from the remainder, which we place in $Y$. This procedure can be iterated, and in each repetition the size of $Y$ increases by 1, and the size of $X$ decreases by at most $A_\xi$. At the end, there will be at least $\lfloor|\gamma|/A_\xi\rfloor$ elements in $Y$, thereby proving the bound.

\subsection{Plaquette-correlated bit-flips} \label{apdx:plaquette-correlated-bit-flip}
Here, we consider plaquette-correlated bit-flips noise channel considered in the main text. Here, since $X_p$ commutes with $X_\ell$, it is clear that the logical-$\bar X$ operator is trivially preserved. We therefore consider the existence of the $W^\mathrm{Z}_{S_\updownarrow}$ operator. We consider the same geometry as that for single-site separable errors.

The correlated plaquette-flip error channel [\cref{eq:plaquette-correlated-bit-flips}] 
can be easily dilated to
\begin{align}
    \mathcal{V}^{\mathcal{N}}_{QE \leftarrow Q}[\cdot] = \mathcal{U}_{QE} [\cdot \otimes \ketbra{\uppi}^{\otimes |P|}_E],
\end{align}
where $\mathcal{U}_{QE}[\cdot] = U_{QE} \cdot U^\dag_{QE}$ and
\begin{align}
    U_{QE} = \prod_{p} U_p,
\end{align}
where $U_p = X_p \otimes \mathbb{P}^0_{E_p} + I_p \otimes \mathbb{P}^1_{E_p}$, with $\ket{\uppi}_{E_p} = \sqrt{\uppi} \ket{0}_{E_p} + \sqrt{1-\uppi} \ket{1}_{E_p}$.

Now, for the localised channel $\mathcal{N}_{S \leftarrow Q} = \mathrm{tr}_{Q \setminus S} \mathcal{N}_{Q' \leftarrow Q}$,
we choose the complementary channel to be simply given by
\begin{align}
    \widehat{\mathcal{N}}_{E_S \leftarrow Q} = \mathrm{tr}_{S} \mathcal{V}^{\mathcal{N}}_{QE \leftarrow Q}.
\end{align}

Now let's consider the point \cref{eq:sum J decomp z2} in the proof of \cref{thm:fat} for the $\mathbb{Z}_2$ case, where similarly to before, we write $\mathcal{\Delta}_{Q \leftarrow R}$ as a signed sum over $\tilde{\mathcal{\Delta}}^{\Gamma'}_{Q \leftarrow R}$. $\Gamma'$ is a set of self-osculating walks on the dual lattice, $\Gamma' = \{\gamma_*^{(1)}, \gamma_*^{(2)}, \dots\}$. For simplicity, let's consider the case where there is only one self-osculating walk, $\gamma_*$. The generalisation to multiple self-osculating walks will be obvious. Again, for simplicity, denote $\mathcal{\Delta}^{\gamma_*} = \tilde{\mathcal{\Delta}}^{\Gamma'}_{Q \leftarrow R}$.

To upper bound the diamond norm, denote $\Phi_{\gamma_*} = \tilde{\mathcal{\Delta}}^{\Gamma'}_{Q \leftarrow R}[O_{RR'}]$, an output of $\tilde{\mathcal{\Delta}}^{\Gamma'}_{Q \leftarrow R}$ for any input $O_{RR'}$ in the enlarged space, $RR'$, where $R'$ has the same dimensions as $R$. 
For this operator, we have
\begin{align}
    \Phi_{\gamma_*} = \qty(\prod_{\ell \in \gamma_*} (\mathrm{id}_\ell - \mathcal{E}^Z_\ell)) \Phi_{\gamma_*},
    \label{eq:correlated plaquette Z proj}
\end{align}
with $\gamma_*$ the SOW on the dual lattice.
The image of this operator under the complementary channel can then be written as
\begin{align}
    \widehat{\mathcal{N}}_{E_S \leftarrow Q}[\Phi_{\gamma_*}] = \sum_{\vec b, \vec b'} \mathrm{tr}_S \qty[\qty(\prod_{p \in P_S} X^{b_p}) \Phi_\gamma \qty(\prod_{p \in P_S} X^{b'_{p}})] \otimes \qty(\prod_{p \in P_S} \mathbb{P}^{b_p}_{E_p}) \ketbra{\vec \uppi}_E \qty(\prod_{p' \in P_S} \mathbb{P}^{b'_{p}}_{E_p}),
    \label{eq:correlated plaquette b sum}
\end{align}
where $b_p \in \{0, 1\}$, and also define the corresponding Ising variables $\sigma_p = -1 \leftrightarrow b_p = 0$, $\sigma_p = 1 \leftrightarrow b_p = 1$, and $P_S$ is the set of plaquettes for which all edges are contained in $S$. Now the SOW on the links of the dual lattice always goes through a sequence of plaquettes on the primal lattice. Therefore, the plaquettes can be numbered $1, 2, \dots, |\gamma_*| -1$. Here, the same plaquette may appear multiple times in this sequence if the paths osculate. We will refer to this set of plaquettes as $P_{\gamma_*}$. Now, consider the trace over $\gamma_*$,
\begin{align}
    \tr_{\gamma_*} \left[ \left(\prod_p X_p^{b_p}\right) \Phi_{\gamma_*} \left(\prod_{p'} X_{p'}^{b'_{p'}}\right) \right].
\end{align}
Labelling the $\ell \in \gamma_*$ sequentially too, consider $\ell = 1$. By Eq.~\eqref{eq:correlated plaquette Z proj}, if we decompose $\Phi_{\gamma_*}$ in terms of Pauli strings, every string must feature  $X_1$ or $Y_1$ for this qubit. Therefore, we require that $\sigma_1 \sigma'_1 = -1$, for the trace to be non-zero, i.e. the binary variables must disagree. Next, consider $\ell = 2$. Since this link neighbours both $p=1$ and $p=2$, we require $\sigma_1 \sigma'_1 \sigma_2 \sigma'_2 = -1$, which simplifies to $\sigma_2 \sigma'_2 = +1$. This repeats, with the signs of the products $\sigma_i \sigma_i'$ alternating.  For $\Gamma'$ containing multiple SOWs, the same procedure can be applied to each $\gamma \in \Gamma'$ separately. Overall, this leads to a set of constraints on the variables $\sigma_p$, which either cannot be satisfied (in which case the operator \eqref{eq:correlated plaquette b sum} vanishes), or which have $2^{|\gamma_*|}$ solutions, each of which have $\lfloor |\gamma_*|/2 \rfloor$ plaquettes with $b_p = b_p'$, and $\lceil|\gamma_*|/2\rceil$ plaquettes with $b_p \neq b_p'$.

To upper bound the norm for $\widehat{\mathcal{N}}[\Phi_{\gamma_*}]$, we observe that
\begin{align}
    \mathbb{P}^{b_p}_{E_p} \ketbra{\uppi} \mathbb{P}_{E_p}^{b'_p} = \begin{cases}
        \mqty(p & 0 \\ 0 & 1-p) & \text{ if } b_p = b'_p, \\
        \mqty(0 & \sqrt{p(1-p)} \\ \sqrt{p(1-p)} & 0) & \text{ if } b_p \neq b'_p.
    \end{cases}
\end{align}
The trace-norms of the above matrices are $1$ ($b_p = b'_p$) and $2 \sqrt{p(1-p)}$ ($b_p \neq b'_p$), respectively. Then the norm of the operator in question can be bounded as
\begin{align}
    \norm{\widehat{\mathcal{N}}_{E_S \leftarrow Q}[\Phi_{\gamma_*}]}_1 & = \norm{\sum'_{\vec b} \mathrm{tr}_S \qty[\qty(\prod_{p \in P} X^{b_p}) \Phi_\gamma \qty(\prod_{p \in P} X^{b'(b_p)})] \otimes \qty(\prod_{p \in P} \mathbb{P}^{b_p}_{E_p}) \ketbra{\vec \uppi}_E \qty(\prod_{p' \in P} \mathbb{P}^{b'(b_p)}_{E_p})}_1 \\
    & \leq \sum'_{\vec b} \norm{\mathrm{tr}_S \qty[\qty(\prod_{p \in P} X^{b_p}) \Phi_\gamma \qty(\prod_{p \in P} X^{b'(b_p)})]}_1 \qty(2 \sqrt{p(1-p)})^{|\gamma_*|/2} \\
    & \leq 2^{|\gamma_*|} \qty(2\sqrt{p(1-p)})^{|\gamma_*|/2} \norm{\Phi_{\gamma_*}}_1,
\end{align}
where the sums over $\vec{b}$ are over configurations that satisfy the constraints described above. We therefore have
\begin{align}
    \norm{\widehat{\mathcal{N}}_{E_S \leftarrow Q} \mathcal{\Delta}^{\gamma_*}_{Q \leftarrow R} [O_{RR'}]}_1 \leq 2^{|\gamma_*|} \qty(2\sqrt{p(1-p)})^{|\gamma_*|/2} \norm{\mathcal{\Delta}^{\gamma_*}_{Q \leftarrow R}[O_{RR'}]}_1 \leq 2^{|\gamma_*|} \qty(2\sqrt{p(1-p)})^{|\gamma_*|/2} \times 2 \times \norm{O_{RR'}}_1,
\end{align}
repeatedly using the property that $\norm{\mathcal{T}_1 \mathcal{T}_2}_\diamond \leq \norm{\mathcal{T}_1}_\diamond \norm{\mathcal{T}_2}_\diamond$, $\norm{\mathcal{C}}_\diamond =1$ for a channel $\mathcal{C}$, and  $\norm{\mathrm{id}_\ell - \mathcal{E}^Z_\ell}_\diamond = 1$, we have
\begin{align}
    \norm{\widehat{\mathcal{N}}_{E_S \leftarrow Q} \mathcal{\Delta}^{\gamma_*}_{Q \leftarrow R} [O_{RR'}]}_1 \leq 2 \times 2^{|\gamma_*|} \qty(2\sqrt{p(1-p)})^{|\gamma_*|/2} \norm{O_{RR'}}_1
\end{align}
for any input $O_{RR'}$. Therefore by the definition of diamond norm, we have
\begin{align}
    \norm{\widehat{\mathcal{N}}_{E_S \leftarrow Q} \mathcal{\Delta}^{\gamma_*}_{Q \leftarrow R}}_\diamond \leq 2 \times 2^{|\gamma_*|} \qty(2\sqrt{p(1-p)})^{|\gamma_*|/2}.
\end{align}
Looking at the factors that exponentiate $|\gamma_*|$, the above scales as $e^{\beta(p)|\gamma_*|}$, where
\begin{align}
    \beta(p) = \frac{3}{2} \ln 2 + \frac{1}{4} \ln(p(1-p)).
\end{align}
This determines an effective (inverse) noise strength for this correlated noise channel, which implies the threshold described in the main text. 

\section{Derivation of diamond norm expressions} \label{apdx:diamond-norms}
In this section, we will compute the effective noise strengths $T(\mathcal{N}_\ell)$ for various single-site noise channels, including those listed in Table \ref{tab:beta-ps-and-thresholds}. In general, evaluating the diamond norm is challenging due to its definition in terms of an optimization problem. There is a relatively simple relation that one can use to upper bound the norm: by Eq. (3.414) of Ref. \onlinecite{watrous2018theory}, for any superoperator $\mathcal{T}_A$, we have
\begin{align}
    \norm{\mathcal{T}_A}_\diamond \leq d_A \norm{(\mathcal{T}_A \otimes \mathrm{id}_{A'}) \Phi_{AA'}}_1,
\end{align}
where $A'$ has the same dimension as $A$ and $\Phi_{AA'}$ is the Bell state on $A$ and $A'$. However, the factor of $d_A$ is generally quite loose, as we will show by example below. In addition, for Hilbert spaces of low-dimensional qudits, the diamond norms could also be computed numerically via semi-definite programming. Instead, here we will explore settings where the diamond norm can be evaluated analytically, to find:
\begin{lemma} \label{lem:effective-noise-strengths}
    The effective noise strength for the following channels are computed as follows. Here, the diamond norms for $\mathcal{T}^{(p)} = \widehat{\mathcal{N}}_{E_\ell \leftarrow \ell} (\mathrm{id}_\ell - \mathcal{E}^{X,Z}_\ell)$ are given for qubits ($d_\ell = 2$).

\begin{center}
    \setlength{\tabcolsep}{10pt}
    \renewcommand{\arraystretch}{1.8}
    \begin{tabular}{ccc}
        \toprule
        {Channel name \& definition} $\mathcal{N}_\ell^{(p)}[\sigma]$ & {Superoperator $\mathcal{T}^{(p)}$} & {Diamond norm $\norm{\mathcal{T}^{(p)}}_\diamond$} \\
        \hline
        
        \multirow{3}{*}{\shortstack[c]{{Bit-flip error} \\ \\ $\displaystyle (1-p) \sigma + p X_\ell \sigma X_\ell$}} 
        & ${\widehat{\mathcal{N}}_{E_\ell \leftarrow \ell}(\mathrm{id}_\ell - \mathcal{D}_\ell)}$ & $2\sqrt{p(1-p)}$ \\
        & ${\widehat{\mathcal{N}}_{E_\ell \leftarrow \ell}(\mathrm{id}_\ell - \mathcal{E}^X_\ell)}$ & $0$ \\
        & ${\widehat{\mathcal{N}}_{E_\ell \leftarrow \ell}(\mathrm{id}_\ell - \mathcal{E}^Z_\ell)}$ & $2\sqrt{p(1-p)}$ \\
        \hline
        
        \multirow{3}{*}{\shortstack[c]{{Amplitude damping} \\ \\ $\displaystyle \smqty[1 & 0 \\ 0 & \sqrt{1-p}]_\ell \sigma \smqty[1 & 0 \\ 0 & \sqrt{1-p}]_\ell$ \\ $\displaystyle + \; \smqty[0 & \sqrt{p} \\ 0 & 0]_\ell \sigma \smqty[0 & 0 \\ \sqrt{p} & 0]_\ell$}} 
        & ${\widehat{\mathcal{N}}_{E_\ell \leftarrow \ell}(\mathrm{id}_\ell - \mathcal{D}_\ell)}$ & $\sqrt{p} + \frac{p}{2}$ \\
        & ${\widehat{\mathcal{N}}_{E_\ell \leftarrow \ell}(\mathrm{id}_\ell - \mathcal{E}^X_\ell)}$ & $\sqrt{p}$ \\
        & ${\widehat{\mathcal{N}}_{E_\ell \leftarrow \ell}(\mathrm{id}_\ell - \mathcal{E}^Z_\ell)}$ & $\sqrt{p}$ \\
        \hline
        
        \multirow{3}{*}{\shortstack[c]{{Partially depolarising} \\ \\ $\displaystyle (1-p)\sigma + p \frac{I_\ell}{d_\ell}\mathrm{tr}_\ell[\sigma]$}} 
        & ${\widehat{\mathcal{N}}_{E_\ell \leftarrow \ell}(\mathrm{id}_\ell - \mathcal{D}_\ell)}$ & $\frac{p(d_\ell^2-2)}{d_\ell^2} + \frac{\sqrt{p}}{d_\ell} \sqrt{4(d^2_\ell - 1) - p(3d^2_\ell - 4)}$ \\
        & ${\widehat{\mathcal{N}}_{E_\ell \leftarrow \ell}(\mathrm{id}_\ell - \mathcal{E}^X_\ell)}$ & $\sqrt{p(2-p)}$ \\
        & ${\widehat{\mathcal{N}}_{E_\ell \leftarrow \ell}(\mathrm{id}_\ell - \mathcal{E}^Z_\ell)}$ & $\sqrt{p(2-p)}$ \\
        \hline
        
        \multirow{3}{*}{\shortstack[c]{{Heralded erasure} \\ \\ $\displaystyle (1-p)\sigma \otimes \ketbra{0}_C$ \\ $\displaystyle + \; p \tr_\ell[\sigma] \otimes \ketbra{e}_\ell \otimes \ketbra{1}_C$}} 
        & ${\widehat{\mathcal{N}}_{E_\ell \leftarrow \ell}(\mathrm{id}_\ell - \mathcal{D}_\ell)}$ & $2 p \qty(1 - \frac{1}{d_\ell^2})$ \\
        & ${\widehat{\mathcal{N}}_{E_\ell \leftarrow \ell}(\mathrm{id}_\ell - \mathcal{E}^X_\ell)}$ & $p$ \\
        & ${\widehat{\mathcal{N}}_{E_\ell \leftarrow \ell}(\mathrm{id}_\ell - \mathcal{E}^Z_\ell)}$ & $p$ \\
        \botrule
    \end{tabular}
\end{center}
\end{lemma}

We will also prove the following identity upper bounding the information leaked into the environment via the distance to the identity, stated informally in the main text.
\begin{lemma} \label{lem:info-leaked-to-dist-to-id}
	Consider a channel $\mathcal{N}_{B \leftarrow A}$, where $d_B \geq d_A$. Then there exists a choice of Stinespring dilation and therefore complementary channel such that the information leaked into the environment is upper bounded by the diamond distance to the identity channel:
	\begin{align}
		\norm*{\widehat{\mathcal{N}}_{E \leftarrow A}(\mathrm{id}_A - \mathcal{D}_A)}_\diamond \leq 4 \qty(1 - \frac{1}{d^2_A}) \times \norm{\mathcal{N}_{B \leftarrow A} - \mathrm{id}_{B \leftarrow A}}_\diamond^{1/2},
	\end{align}
	where $\mathcal{D}_A$ is the depolarising channel, and $\widehat{\mathcal{E}}_{E \leftarrow A}$ denotes the complementary channel of $\mathcal{E}_{B \leftarrow A}$. Here, the `identity channel' is $\mathrm{id}_{B \leftarrow A}[\cdot] = \mathrm{id}_A$ if $d_A = d_B$, $\mathrm{id}_{B \leftarrow A}[\cdot] = (\cdot) \otimes \sigma_{B \setminus A}$ if $d_B > d_A$ for any $\sigma_{B \setminus A}$.
\end{lemma}

We will also derive a number of important properties of the diamond norm that will be useful in future when evaluating  $T(\mathcal{N}_\ell)$ for different channels. Some of the results may overlap with that of Ref.~\onlinecite{magesan2012characterizing}. \\


It will be useful to consider the set of generalised Pauli operators acting on of Hilbert space of dimension $d$. They are given by
\begin{align} \label{eq:generalised-paulis}
    P = X^{a_P} Z^{b_P},
\end{align}
where $Z$ and $X$ are clock and shift operators, respectively. These are generalisations of the qubit Pauli-$Z$ and $X$ operators. They are unitary but not Hermitian, and have the properties
\begin{align}
    Z^d = X^d = 1,
\end{align}
\begin{align} \label{eq:generalised-pauli-commutation}
    ZX = \omega XZ,
\end{align}
where $\omega = \exp(2 \pi \mathrm{i}/d)$ such that $\omega^d = 1$. They form a complete basis for operators, which means that any operator on a Hilbert space of dimension $d$ can be written in the form $O = \sum_P o_P P$, and traceless except for the identity Pauli operator, i.e. $\mathrm{tr}[P] = d \times \delta_{a_P, 0} \delta_{b_P, 0}$.

The commutation relations \eqref{eq:generalised-pauli-commutation} can be written in a compact form as the following. For Paulis $P$ and $Q$,
\begin{align} \label{eq:generalised-pauli-commutation-compact}
    QP = \omega^{a_P b_Q - a_Q b_P} PQ, \quad P^\dagger Q^\dagger = \omega^{a_Q b_P - a_P b_Q} Q^\dagger P^\dagger.
\end{align}

To calculate the diamond norms exactly, we will rely on the following results and definitions.

\begin{lemma}[Theorem 4.14 of Ref.~\onlinecite{watrous2018theory}] \label{lem:pauli-equivalence}
    Consider the set of generalised Pauli matrices $\{P\}$.

    Then the following three properties of a superoperator $\mathcal{T}$ are equivalent:
    \begin{itemize}
        \item[(a)] $\mathcal{T}$ is a Pauli superoperator, i.e. it can be written in the form $\mathcal{T}[\cdot] = \sum_P c_P P \cdot P^\dagger$, where $P$ is a Pauli and $c_P$ are not necessarily a positive real number.
        \item[(b)] $\mathcal{T}$ is Pauli diagonal, i.e. $\mathcal{T}[P] = \lambda_P P$ for all Pauli $P$.
        \item[(c)] $\mathcal{T}$ is Pauli-to-Pauli covariant, i.e. $\mathcal{T}[P O P^\dagger] = P \mathcal{T}[O] P^\dagger$ for all Pauli $P$.
    \end{itemize}
\end{lemma}

For Pauli superoperators, we may calculate their diamond norms exactly via the next lemma:

\begin{lemma} \label{lem:diamond-norm-is-bell-state}
    For Pauli superoperators, its diamond norm is given by the trace norm
    \begin{align}
        \norm{\mathcal{T}_{B \leftarrow A}}_\diamond = \norm{\mathcal{T}_{B \leftarrow A} [\Phi_{AA'}]}_1 = \sum_P |c_P|,
    \end{align}
    where $\Phi_{AA'}$ is a Bell state such that $A'$ has the same dimensions as $A$.
\end{lemma}

Combining \cref{lem:pauli-equivalence,lem:diamond-norm-is-bell-state}, we can come up with a simple formula for diamond norms for qubit Pauli superoperators:
\begin{lemma} \label{lem:qubit-to-qubit}
    For a qubit-to-qubit Pauli-diagonal superoperator $\mathcal{T}$ its diamond norm is given by
    \begin{align}
        \norm{\mathcal{T}}_\diamond = \frac{1}{4} \Big(|\lambda_I + \lambda_X + \lambda_Y + \lambda_Z| + |\lambda_I - \lambda_X - \lambda_Y + \lambda_Z| + |\lambda_I + \lambda_X - \lambda_Y - \lambda_Z| + |\lambda_I - \lambda_X + \lambda_Y - \lambda_Z|\Big),
    \end{align}
    where $\mathcal{T}[P] = \lambda_P P$ for $P \in \{I, X, Y, Z\}$. 
    
    Note that we use a different convention for Paulis for qubits (i.e. $Y$ vs. $XZ$) compared to the generalised definition \eqref{eq:generalised-paulis}.
\end{lemma}

For channels which change and input and output spaces, such as the heralded leakage channel (due to the output flag classical register) the following Lemma will allow us to calculate the diamond norm exactly.

\begin{lemma} \label{lem:generalised-pauli-invariant}
    Consider a set complete set of unitaries $\{P_A\}$ of size $d_A^2$ such that $\mathrm{tr}_A[P_A^\dag P'_A] = d_A \delta_{PP'}$ (e.g. Pauli superoperators). Let a unitary-set covariant superoperator be one such that
    \begin{align}
        \tilde{P}_B \mathcal{T}_{B \leftarrow A} [\cdot] \tilde{P}_B^\dagger = \mathcal{T}_{B \leftarrow A} [P_A \cdot P^\dag_A] \quad \forall \; P,
    \end{align}
    where $\{\tilde{P}_B\}$ are unitary, but need not form an orthogonal set. We call such superoperators unitary-set invariant. Then, the diamond norm of the superoperator $\mathcal{T}_{B \leftarrow A}$ is given by the trace norm on the Bell state $\Phi_{AA'}$, i.e.
    \begin{align}
        \norm{\mathcal{T}_{B \leftarrow A}}_\diamond = \norm{\mathcal{T}_{B \leftarrow A} \Phi_{AA'}}_1.
    \end{align}
\end{lemma}

\begin{lemma} \label{lem:chan-covariant-means-complement-covariant}
    If a channel (CPTP map) is unitary-set covariant, i.e. that for an complete, orthogonal set of unitaries $\{P_A\}$ such that $P_A^\dagger P_A = I_A$, $\mathrm{tr}_A[P_A^\dagger P'_A] = \delta_{AA'}$, such that
    \begin{align}
        \mathcal{C}_{B \leftarrow A}[P_A \cdot P^\dagger_A] = \tilde{P}_B \mathcal{C}_{B \leftarrow A}[\cdot] \tilde{P}_B^\dagger,
    \end{align}
    where $\{\tilde{P}_B\}$ are unitaries that need not be orthogonal nor complete.
    then its complementary channel $\widehat{\mathcal{C}}_{E \leftarrow A}$ is also unitary-set covariant, i.e.
    \begin{align}
        \widehat{\mathcal{C}}_{E \leftarrow A}[P_A \cdot P^\dagger_A] = \bar{P}_E \widehat{\mathcal{C}}_{E \leftarrow A}[\cdot] \bar{P}_E^\dagger,
    \end{align}
    for some unitary-set $\{\bar{P}_E\}$, which need not to form a covariant set.
\end{lemma}

\subsection{Proofs of useful Lemmas}
\begin{proof}[Proof of \cref{lem:diamond-norm-is-bell-state}]
    Recall that the diamond norm is given by
    \begin{align}
        \norm{\mathcal{T}_{B \leftarrow A}}_\diamond = \sup_{\rho_{AA'}} \norm{(\mathcal{T}_{B \leftarrow A} \otimes \mathrm{id}_{A'}) \rho_{AA'}}_1.
    \end{align}
    
    Firstly, we note that any trace norm for a mixed state is upper bounded by those of pure states. This is because any mixed state $\rho_{AA'}$ can written as a mixture of pure states, $\sum_i p_i \Psi^{(i)}_{AA'}$, for which we have
    \begin{align}
        \norm{\mathcal{T}_{B \leftarrow A} \rho_{AA'}}_1 = \norm{\mathcal{T}_{B \leftarrow A} \sum_i p_i \Psi^{(i)}_{AA'}}_1 \leq \sum_i p_i \norm{\mathcal{T}_{B \leftarrow A} \Psi^{(i)}_{AA'}}_1 \leq \mathrm{max}_i \norm{\mathcal{T}_{B \leftarrow A} \Psi^{(i)}_{AA'}}_1,
    \end{align}
    where the last equality follows since an average is smaller than the maximum. Therefore, the optimal state is given by a pure state.
    
    Now for any state $\rho_{AA'}$, we have
        \begin{align}
            \norm{\sum_P c_P P_{B \leftarrow A} \rho_{AA'} (P_{B \leftarrow A})^\dagger}_1 \leq \sum_P |c_P| \norm{P_{B \leftarrow A} \rho_{AA'} (P_{B \leftarrow A})^\dagger}_1 = \sum_P |c_P|.
        \end{align}
        In the last equality, we have used the fact conjugation by an isometry $W_{B \leftarrow A}$ does not change the trace norm of an operator $X_A$. This is because for any operator $X_A = \sum_\alpha \sigma_\alpha \ketbra{u_\alpha}{v_\alpha}_A$ with trace norm $\norm{X_A}_1 = \sum_\alpha |\sigma_\alpha|$, we have $W_{B \leftarrow A} X_A (W^\dagger)_{A \leftarrow B} = \sum_\alpha \sigma_\alpha (W_{B \leftarrow A} \ket{u_\alpha}_A) (\bra{v_\alpha}_A (W^\dagger)_{A \leftarrow B})$ have the same trace norm since $\bra{u_\beta}_A (W^\dagger)_{A \leftarrow B} W_{B \leftarrow A} \ket{u_\alpha}_A =\bra{u_\beta}_A I_A \ket{u_\alpha}_A = \delta_{\beta \alpha}$, similarly for $W_{B \leftarrow A} \ket{v_\alpha}$ and therefore the above is already written in SVD form.
        
        Now consider its action on the Bell state. Now $\Phi^{(P)}_{BA'} = P_{B \leftarrow A} \Phi_{AA'} (P^\dagger)_{A \leftarrow B}$ are orthogonal to each other. This can be easily seen as follows:
        \begin{align}
            \braket{\Phi^{(P)}_{BA'}}{\Phi^{(P')}_{BA'}} 
            & = \frac{1}{d_A} \sum_{k k'} \bra{kk}_{AA'} (P^\dag)_{A \leftarrow B} P'_{B \leftarrow A} \ket{k'k'}_{AA'} 
            = \frac{1}{d_A}  \sum_k \mel{k}{(P^\dag)_{A \leftarrow B} P'_{B \leftarrow A}}{k}_A \\
            & = \frac{1}{d_A}  \tr_A[(P^\dag)_{A \leftarrow B} P'_{B \leftarrow A}] = \delta_{PP'}.
        \end{align}
        Therefore after application of an isometry-set superoperator, the operator is in block diagonal form. Therefore
        \begin{align}
            \norm{\sum_P c_P P_{B \leftarrow A} \Phi_{AA'} (P^\dagger)_{A \leftarrow B}}_1 = \sum_P \norm{c_P \Phi^{(P)}_{BA'}}_1 = \sum_P |c_P|,
        \end{align}
        and therefore the Bell state maximises the trace norm.
    
\end{proof}

\begin{proof}[Proof of \cref{lem:qubit-to-qubit}]
    By \cref{lem:diamond-norm-is-bell-state}, for a Pauli-diagonal superoperator, we have
    \begin{align}
        \norm{\mathcal{T}_A}_\diamond = \norm{\mathcal{T}_A [\Phi_{AA'}]}_1.
    \end{align}
    For qubits, this can be written as $\Phi_{AA'} = \frac{1}{4} \qty(1 + X_A X_{A'} - Y_A Y_{A'} + Z_A Z_{A'})$. Now applying the Pauli-diagonal superoperator, we have
    \begin{align}
        \mathcal{T}_A [\Phi_{AA'}] & = \frac{1}{4} \qty(\lambda_I + \lambda_X X_A X_{A'} - \lambda_Y Y_A Y_{A'} + \lambda_Z Z_A Z_{A'}) \\
        & = \frac{1}{4} \qty(\lambda_I + \lambda_X X_A X_{A'} + \lambda_Y (X_A X_{A'}) (Z_A Z_{A'}) + \lambda_Z Z_A Z_{A'}).
    \end{align}
    Now since $X_A X_{A'}$ and $Z_A Z_{A'}$ commute, they can be simultaneously diagonalised. Therefore the eigenvalues $\mu$ are given by setting $X_A X_{A'} \leftrightarrow b$, $Z_A Z_{A'} \leftrightarrow b'$, for $b,b' \in \qty{\pm1}$:
    \begin{align}
        \mu_{b, b'} = \frac{1}{4}(\lambda_I + \lambda_X b + \lambda_Y bb' + \lambda_Z b').
    \end{align}
    By definition, the trace norm is found by taking their modulus and summing over them.
\end{proof}

\begin{proof}[Proof of \cref{lem:generalised-pauli-invariant}]
    First, note that any pure state on $AA'$ can be written in the form
    \begin{align}
        \ket{\Psi}_{AA'} = M_{A'} \ket{\Phi}_{AA'},
    \end{align}
    where $\mathrm{tr}_{A'}[M_{A'}^\dagger M_{A'}] = d_A$. This can be shown as follows:
    \begin{align}
        \ket{\Psi}_{AA'} = \sum_\alpha \sqrt{\lambda_\alpha} \ket{u_\alpha}_A \ket{v_\alpha}_{A'} = U_A V_{A'} \sum_\alpha \sqrt{\lambda_\alpha} \ket{\alpha}_A \ket{\alpha}_{A'} = U_A V_{A'} \Lambda_{A'} \frac{1}{\sqrt{d_A}} \sum_\alpha \ket{\alpha}_A \ket{\alpha}_{A'}.
    \end{align}
    Here, $U_A$ and $V_{A'}$ are unitaries which change basis from the orthonormal sets $\{\ket{u_\alpha}\}$, $\{\ket{v_\alpha}\}$ to $\{\ket{\alpha}\}$, and $\Lambda = \mathrm{diag}(\sqrt{d_A/\lambda_\alpha})_\alpha$. Using the transpose trick on Bell states, we have
    \begin{align}
        \ket{\Psi}_{AA'} = V_{A'} \Lambda_{A'} U^\top_{A'} \ket{\Phi}_{AA'} =: M_{A'} \ket{\Phi}_{AA'}.
    \end{align}
    since $M$ is already written in SVD form, it's easy to check that $\mathrm{tr}_{A'}[M_{A'}^\dagger M_{A'}] = d_A$.

    Now recalling from the proof of \cref{lem:diamond-norm-is-bell-state} that only pure states need to be considered to compute the diamond norm, we have
    \begin{align}
        \norm{\mathcal{T}_{B \leftarrow A}}_\diamond = \mathrm{max}_{\Psi_{AA'}} \norm{\mathcal{T}_{B \leftarrow A} \Psi_{AA'}}_1 = \mathrm{max}_{\mathrm{tr}_{A'}[M^\dagger_{A'} M_{A'}]=d_A} \norm{M_{A'}\mathcal{T}_{A}[\Phi_{AA'}] M^\dag_{A'}}_1.
    \end{align}
    Let $J_{BA'} = \mathcal{T}_{B \leftarrow A} \Phi_{AA'}$. Then via covariance and the transpose trick, we have
    \begin{align} \label{eq:twirl-identity}
        J_{BA'} = \tilde{P}_{B} P^*_{A'} J_{BA'} P^\top_{A'} \tilde{P}^\dag_{B}.
    \end{align}
    Now trace norm obeys the property $\norm{X}_1 = \mathrm{max}_{U} |\mathrm{tr}[UX]|$. Therefore
    \begin{align}
        \norm{\mathcal{T}_{B \leftarrow A}}_\diamond & = \mathrm{max}_{\mathrm{tr}_{A'}[M^\dagger_{A'} M_{A'}]=d_A} \mathrm{max}_{U_{BA'}} |\mathrm{tr}[U_{B A'} M_{A'}\mathcal{T}_{B \leftarrow A}[\Phi_{AA'}] M^\dag_{A'}]| \\
        & = \mathrm{max}_{\mathrm{tr}_{A'}[M^\dagger_{A'} M_{A'}]=d_A} \mathrm{max}_{U_{BA'}} |\mathrm{tr}[\underbrace{M^\dag_{A'} U_{B A'} M_{A'}}_{:= W_{B'A}} J_{BA'}]|.
    \end{align}
    Now using the twirling identity \eqref{eq:twirl-identity}, we have
    \begin{align}
        \mathrm{tr}[W_{BA'} J_{BA'}] = \mathrm{tr}[\bar{W}_{BA'} J_{BA'}],
    \end{align}
    where
    \begin{align}
        \bar{W}_{BA'} := \frac{1}{d^2_A} \sum_P \tilde{P}_{B} P^*_{A'} W_{BA'} P^\top_{A'} \tilde{P}^\dag_{B}.
    \end{align}
    Therefore 
    \begin{align} \label{eq:diamond-norm-equal-to-WJ}
        \norm{\mathcal{T}_{B \leftarrow A}}_\diamond = \abs{\mathrm{tr}[\bar{W}_{BA'} J_{BA'}]} \leq \norm{\bar{W}_{BA'}}_\infty \norm{J_{BA'}}_1.
    \end{align}
    Now
    \begin{align}
        \norm{\bar{W}_{BA'}}_\infty = \max_{\norm{\ket{\psi}}_2, \norm{\ket{\phi}}_2 = 1} |\mel{\psi}{\bar W}{\phi}| = \abs{\frac{1}{d^2_A} \sum_P \mel{\psi}{\tilde{P}_B P^*_{A'} M^\dag_{A'} U_{BA'} M_{A'} \tilde{P}_B^\dag P^\top_{A'}}{\phi}}.
    \end{align}
    Let $\ket{\phi^{(P)}} := M_{A'} \tilde{P}_B^\dag P^\top_{A'} \ket{\phi}$, and similarly for $\ket{\psi^{(P)}}$. Then we have
    \begin{align}
        \norm{\bar{W}_{BA'}}_\infty \leq \frac{1}{d^2_A} \sum_P \abs{\mel{\psi^{(P)}}{U_{BA'}}{\phi^{(P)}}} \leq \frac{1}{d^2_A} \sum_P \norm{\ket{\psi^{(P)}}}_2 \norm{\ket{\phi^{(P)}}}_2
    \end{align}
    Now we can interpret $\norm{\ket{\psi^{(P)}}}_2$ as the $P$th element of a vector $v$, and same for $\norm{\ket{\phi^{(P)}}}_2$ as the $P$th element of a vector $u$. Again using the triangular inequality, we have
    \begin{align}
        \norm{\bar{W}_{BA'}}_\infty \leq \frac{1}{d^2_A} \sum_P \abs{\mel{\psi^{(P)}}{U_{BA'}}{\phi^{(P)}}} \leq \sqrt{\frac{1}{d^2_A} \sum_P \norm{\ket{\psi^{(P)}}}_2^2} \sqrt{\frac{1}{d^2_A} \sum_P \norm{\ket{\phi^{(P)}}}_2^2}.
    \end{align}
    Now we have
    \begin{align}
        \frac{1}{d^2_A} \sum_P \norm{\ket{\phi}}_2^2 = \frac{1}{d^2_A} \sum_P \mel{\phi}{\tilde{P}_B \tilde{P}_B^\dag P^* M^\dag_A M_A P^\top_A}{\phi}.
    \end{align}
    Firstly, since $\tilde{P}_B$ is unitary $\tilde{P}_B \tilde{P}_B^\dag = I_B$, secondly, we can bring the sum inside the inner product to find
    \begin{align}
        \frac{1}{d^2_A} \sum_P \norm{\ket{\phi}}_2^2 = \mel{\phi}{\frac{1}{d^2_A} \sum_P P^* M^\dag_A M_A P^\top_A}{\phi} = \mel{\phi}{\mathcal{D}_{A}[M^\dag_A M_A]}{\phi}.
    \end{align}
    Here we used the fact the the fully depolarising channel can be written Pauli form, $\mathcal{D}_A[\cdot] = \frac{I_A}{d_A} \mathrm{tr}_A[\cdot] = \frac{1}{d^2_A} \sum_P P_A \cdot P^\dag_A$. Now $\mathrm{tr}_A[M^\dag_A M_A] = d_A$, and therefore we have
    \begin{align}
        \frac{1}{d^2_A} \sum_P \norm{\ket{\phi}}_2^2 = \mel{\phi}{\frac{I_A d_A}{d_A}}{\phi} = 1
    \end{align}
    Applying similar logic to $\ket{\psi}$, we therefore have
    \begin{align}
        \norm{\bar{W}_{BA'}}_\infty \leq 1,
    \end{align}
    thereby proving that $\norm{\mathcal{T}_A}_\diamond \leq 1 \times \norm{\mathcal{T}_A \Phi_{AA'}}_1$ via \cref{eq:diamond-norm-equal-to-WJ}. Combining with the usual bound $\norm{\mathcal{T}_A \Phi_{AA'}}_1 \leq \norm{\mathcal{T}_A}_\diamond$, we have the desired result.
\end{proof}

\begin{proof}[Proof of \cref{lem:chan-covariant-means-complement-covariant}]
    Expressing $\mathcal{N}_{B\leftarrow A}[\cdot]= \mathrm{tr}_E\qty[V_{BE \leftarrow A} \cdot V^\dagger_{A \leftarrow BE}]$, we have that
    \begin{align}
        \mathrm{tr}_E \qty[\tilde{P}_B V_{BE \leftarrow A} \cdot V^\dagger_{A \leftarrow BE}\tilde{P}^\dagger_B] = \mathrm{tr}_E \qty[V_{BE \leftarrow A} P_A \cdot P^\dagger_A V^\dagger_{A \leftarrow BE}],
    \end{align}
    where we put $\tilde{B}_B$ inside of the trace as they commute. Therefore the isometries $V^{(1)}_{BE \leftarrow A} = \tilde{P}_B V_{BE \leftarrow A}$ and $V^{(2)}_{BE \leftarrow A} = V_{BE \leftarrow A} P_A$ are isometries that implement the same channel, which means that theey must be related by an isometry $V^{(2)}_{BE \leftarrow A} = \bar{P}_E V^{(1)}_{BE \leftarrow A}$. Then consider the complementary channel
    \begin{align}
        \widehat{\mathcal{N}}_{E \leftarrow A}[P_A \cdot P_A^\dagger] 
        & = \mathrm{tr}_B \qty[V^{(2)}_{BE \leftarrow A} (V^{(2)\dagger})_{A \leftarrow BE}]
        = \bar{P}_E \mathrm{tr}_B \qty[V^{(1)}_{BE \leftarrow A} (V^{(1)\dagger})_{A \leftarrow BE}] \bar{P}_E^\dagger \\
        & = \bar{P}_E \mathrm{tr}_B \qty[\tilde{P}_B V_{BE \leftarrow A} (V^\dagger)_{A \leftarrow BE} \tilde{P}_B^\dagger] \bar{P}_E^\dagger
        = \bar{P}_E \mathrm{tr}_B \qty[V_{BE \leftarrow A} (V^\dagger)_{A \leftarrow BE} \tilde{P}_B^\dagger \tilde{P}_B] \bar{P}_E^\dagger \\
        & = \bar{P}_E \mathrm{tr}_B \qty[V_{BE \leftarrow A} (V^\dagger)_{A \leftarrow BE}] \bar{P}_E^\dagger
        = \bar{P}_E \widehat{\mathcal{N}}_{E \leftarrow A}[\cdot] \bar{P}_E^\dagger.
    \end{align}
\end{proof}

\subsection{Proof of \cref{lem:info-leaked-to-dist-to-id}: upper bound on information leaked into the environment by distance to identity channel}

\begin{proof}[Proof of \cref{lem:info-leaked-to-dist-to-id}]
    First, let us derive the complementary channel for $\mathrm{id}_{B \leftarrow A}[\cdot] = \mathrm{id}_A[\cdot] \otimes \sigma_{B \setminus A}$. To do this, write $\sigma_{B \setminus A} = \sum_i p_i \ketbra{\phi^{(i)}}_{B \setminus A}$. Purifying, the isometry is given by
    \begin{align}
        V^{\mathrm{id}}_{BE \leftarrow A} \ket{\psi} = \sum_i \sqrt{p_i} \ket{\phi^{(i)}}_{B \setminus A} \ket{e_i}_{E} \ket{\psi}.
    \end{align}
    Then the complementary channel is given by
    \begin{align}
        \widehat{\mathrm{id}}_{E \leftarrow A}[\cdot] = \mathrm{tr}_B\qty[\sum_{ij} \sqrt{p_i p_j} \ketbra{\phi^{(i)}}{\phi^{(j)}}_{B \setminus A} \otimes \cdot] \ketbra{e_i}{e_j}_E = \qty(\sum_i p_i \ketbra{e_i}_E) \mathrm{tr}_A[\cdot],
    \end{align}
    i.e. a replacement channel.

	Next, notice that $\widehat{\mathrm{id}}_{E \leftarrow A}(\mathrm{id}_A - \mathcal{D}_A) = 0$, since the output of $(\mathrm{id}_A - \mathcal{D}_A)$ is always traceless and $\widehat{\mathrm{id}}_{E \leftarrow A}$ includes a trace.
    
    Therefore we may freely subtract $\widehat{\mathrm{id}}_{E \leftarrow A}(\mathrm{id}_A - \mathcal{D}_A)$ from the superoperator inside of the diamond expression and retrieve
	\begin{align}
		\norm*{\widehat{\mathcal{N}}_{E \leftarrow A} (\mathrm{id}_A - \mathcal{D}_A)}_\diamond 
        = \norm*{(\widehat{\mathcal{N}}_{E \leftarrow A} - \widehat{\mathrm{id}}_{E \leftarrow A}) (\mathrm{id}_A - \mathcal{D}_A)}_\diamond 
        \leq \norm*{\widehat{\mathcal{N}}_{E \leftarrow A} - \widehat{\mathrm{id}}_{E \leftarrow A}}_\diamond \norm*{\mathrm{id}_A - \mathcal{D}_A}_\diamond. 
	\end{align}
	Now, we will show in \cref{lem:id-D} below that  $\norm*{\mathrm{id} - \mathcal{D}}_\diamond = 2(1-1/d^2_A)$. In addition, by data processing inequality, we have
	\begin{align}
		\norm{\widehat{\mathcal{N}}_{E \leftarrow A} - \widehat{\mathrm{id}}_{E \leftarrow A}}_\diamond \leq \norm{\mathcal{V}^{{\mathcal{N}}}_{E B \leftarrow A} - \mathcal{V}^{{\mathrm{id}}}_{E B \leftarrow A}}_\diamond
	\end{align}
	for any dilations of channel $\mathcal{E}$, $\mathcal{V}^{\mathcal{E}}_{EB \leftarrow A}[\cdot] = V^{\mathcal{E}}_{EB \leftarrow A}(\cdot)(V^{\mathcal{E}}_{EB \leftarrow A})^\dagger$. By Lemma 12 of Ref.~\onlinecite{aharonov1998quantum}, for any isometric channel we have
	\begin{align}
		\norm{\mathcal{V}^{{\mathcal{N}}}_{E B \leftarrow A} - \mathcal{V}^{{\mathrm{id}}}_{E B \leftarrow A}}_\diamond \leq 2 \norm{V^{{\mathcal{N}}}_{E B \leftarrow A} - V^{{\mathrm{id}}}_{E B \leftarrow A}}_\infty.
	\end{align}
	Now let us assume that we have chosen the optimal isometry as Theorem 1 of Ref.~\onlinecite{kretschmann2008information} (adapting from the dual channel in their definition to the primal channel) to show
	\begin{align}
		\norm{V^{{\mathcal{N}}}_{E B \leftarrow A} - V^{{\mathrm{id}}}_{E B \leftarrow A}}_\infty \leq \norm{\mathcal{N}_{B \leftarrow A} - \mathrm{id}_{B \leftarrow A}}_\diamond^{1/2}.
	\end{align}
	Bringing everything together, we have the advertised result.
\end{proof}

Now we state a lemma used for the proof.

\begin{lemma} \label{lem:id-D}
	\begin{align}
		\norm{\mathrm{id}_A - \mathcal{D}_A}_\diamond = 2\qty(1-\frac{1}{d_A^2}).
	\end{align}
\end{lemma}
\begin{proof}
    Now both $\mathrm{id}_A$ and $\mathcal{D}_A$ are Pauli channels, and therefore
    \begin{align}
        \norm{\mathrm{id}_A - \mathcal{D}_A}_\diamond 
        & = \norm{(\mathrm{id}_A - \mathcal{D}_A)\Phi_{AA'}}_1 
        = \norm{\Phi_{AA'} - \frac{I_A}{d_A} \otimes \frac{I_{A'}}{d_A}}_1.
	\end{align}
	Now $\Phi_{AA'}$ is a rank-1 matrix with one eigenvalue $1$ and $d^2_A - 1$ eigenvalues $0$. Since $\frac{I_A}{d_A} \otimes \frac{I_{A'}}{d_A}$ is proportional to the identity, it is diagonalisable in any basis, and therefore $\Phi_{AA'} - \frac{I_A}{d_A} \otimes \frac{I_{A'}}{d_A}$ has one eigenvalue $1 - 1/d_A^2$ and $d_A^2 -1$ eigenvalues of value $-1/d_A^2$. Therefore its trace norm is
    \begin{align}
        \abs{1 - \frac{1}{d_A^2}} + (d_A^2 -1) \abs{-\frac{1}{d_A^2}} = 2 \qty(1 - \frac{1}{d_A^2}),
    \end{align}
    as advertised.
\end{proof}

\subsection{Proof of Lemma \ref{lem:effective-noise-strengths}: expressions of effective noise strengths}
In this section,  we compute the exact analytic expressions for diamond norms corresponding to the effective noise strengths. Our results are summarised as follows:

Once these diamond norms are known, the threshold can be found by finding where parameter $p$ is such that $\norm{\mathcal{T}^{(p)}}_\diamond \mu = 1$, where $\mu$ is any upper bound for the appropriate connective / growth constant, which is used to obtain the lower bounds for the thresholds, shown in \cref{tab:beta-ps-and-thresholds}. For $S_3$, we have $\mu = 7 e$ and $d_\ell = 6$. For $D_4$, we have $\mu = 7 e$ and $d_\ell = 8$. For $\mathrm{Fibonacci}$ and $\mathrm{DSem}$, we have $\mu = 25 e$ and $d_\ell = 2$. For $\mathrm{3D}$ $\mathbb{Z}_2$, we have $\mu_\Lambda, \mu_{\Lambda_*} = 20.25, 24.5$.

\subsubsection{Bit-flip errors}
A single-qubit bit-flip channel is defined as
\begin{align} \label{eq:bit-flip-channel}
    \mathcal{N}_\ell[\cdot] = (1-p) (\cdot) + p X_\ell (\cdot) X_\ell.
\end{align}
After dilation, the bit-flip error channel becomes a $\mathrm{CNOT}$ gate controlled by $\ket{p}_{E_\ell} = \sqrt{p} \ket{1}_{E_\ell} + \sqrt{1-p} \ket{0}_{E_\ell}$ and acting on the physical qubit. It is given by
\begin{align}
\mathrm{CNOT}^{E_\ell \rightarrow \ell}_{E_\ell \ell} = \ketbra{0}_{E_\ell} \otimes I_{\ell} + \ketbra{1}_{E_\ell} \otimes X_{\ell}.
\end{align}
After tracing over $\ell$, we obtain $\widehat{\mathcal{N}}_{E_\ell}$, which acts as $O_\ell \rightarrow \tr_\ell [\mathrm{CNOT}^{E_\ell \rightarrow \ell}_{E_\ell \ell}(\ketbra{p}_{E_\ell} \otimes O_\ell)]$. We will consider its action on the Pauli operators, which will be sufficient to calculate the desired diamond norm. We find that
\begin{align} \label{eq:bf-pauli-outputs}
    \widehat{\mathcal{N}}_{E_\ell \leftarrow \ell}[I_\ell] = 2 \mathrm{diag}(p, 1-p)_{E_\ell}, \quad
    \widehat{\mathcal{N}}_{E_\ell \leftarrow \ell}[X_\ell] = 2 \sqrt{p(1-p)} X_{E_\ell}, \quad
    \widehat{\mathcal{N}}_{E_\ell \leftarrow \ell}[Y_{E_\ell}] = 0, \quad
    \widehat{\mathcal{N}}_{E_\ell \leftarrow \ell}[Z_{E_\ell}] = 0.
\end{align}

Now consider the super-operators considered in the main text: $\widehat{\mathcal{N}}_{E_\ell \leftarrow \ell} (\mathrm{id}_\ell - \mathcal{D}_\ell)$, $\widehat{\mathcal{N}}_{E_\ell \leftarrow \ell} (\mathrm{id}_\ell - \mathcal{E}^X_\ell)$, $\widehat{\mathcal{N}}_{E_\ell \leftarrow \ell} (\mathrm{id}_\ell - \mathcal{E}^Z_\ell)$. Because all $(\mathrm{id}_\ell - \mathcal{D}_\ell)[I_\ell] = 0 = (\mathrm{id}_\ell - \mathcal{E}^{X, Z}_\ell)[I_\ell] = 0$, the only non-Pauli output for the identity in \eqref{eq:bf-pauli-outputs} is removed in that we can set $\lambda_I = 0$ and the superoperator is Pauli. 

\paragraph{$\norm{\widehat{\mathcal{N}}_{E_\ell \leftarrow \ell} (\mathrm{id}_\ell - \mathcal{D}_\ell)}_\diamond$.} We note that
\begin{align}
    (\mathrm{id}_\ell - \mathcal{D}_\ell)[I_\ell] = 0, \quad (\mathrm{id}_\ell - \mathcal{D}_\ell)[X_\ell] = X_\ell, \quad (\mathrm{id}_\ell - \mathcal{D}_\ell)[Y_\ell] = Y_\ell, \quad (\mathrm{id}_\ell - \mathcal{D}_\ell)[Z_\ell] = Z_\ell.
\end{align}
Therefore, we have $\lambda_I = \lambda_Y = \lambda_Z = 0$, $\lambda_X = 2 \sqrt{p(1-p)}$. Therefore using \cref{lem:qubit-to-qubit}, we have $\norm{\widehat{\mathcal{N}}_{E_\ell \leftarrow \ell} (\mathrm{id}_\ell - \mathcal{D}_\ell)}_\diamond = \frac{1}{4}(|0+2 \sqrt{p(1-p)}+0+0| + |0-2 \sqrt{p(1-p)}-0+0)| + |0+2 \sqrt{p(1-p)}-0-0| + |0-2 \sqrt{p(1-p)}+0-0|$ and therefore
\begin{align}
    \norm{\widehat{\mathcal{N}}_{E_\ell \leftarrow \ell} (\mathrm{id}_\ell - \mathcal{D}_\ell)}_\diamond = 2 \sqrt{p(1-p)}.
\end{align}

\paragraph{$\norm{\widehat{\mathcal{N}}_{E_\ell \leftarrow \ell} (\mathrm{id}_\ell - \mathcal{E}^X_\ell)}_\diamond$.}
Here we have
\begin{align}
    (\mathrm{id}_\ell - \mathcal{E}^X_\ell)[I_\ell] = 0, \quad (\mathrm{id}_\ell - \mathcal{E}^X_\ell)[X_\ell] = 0, \quad (\mathrm{id}_\ell - \mathcal{E}^X_\ell)[Y_\ell] = Y_\ell, \quad (\mathrm{id}_\ell - \mathcal{E}^X_\ell)[Z_\ell] = Z_\ell.
\end{align}
Then combined with the Pauli outputs of the complementary channel, we have $\lambda_I = \lambda_X = \lambda_Y = \lambda_Z = 0$ and therefore
\begin{align}
    \norm{\widehat{\mathcal{N}}_{E_\ell \leftarrow \ell} (\mathrm{id}_\ell - \mathcal{E}^X_\ell)}_\diamond = 0.
\end{align}
This means that the logical-$X$ operator always remains under bit-flip errors.

\paragraph{$\norm{\widehat{\mathcal{N}}_{E_\ell \leftarrow \ell} (\mathrm{id}_\ell - \mathcal{E}^Z_\ell)}_\diamond$.}
Here we have
\begin{align}
    (\mathrm{id}_\ell - \mathcal{E}^Z_\ell)[I_\ell] = 0, \quad (\mathrm{id}_\ell - \mathcal{E}^Z_\ell)[X_\ell] = X_\ell, \quad (\mathrm{id}_\ell - \mathcal{E}^Z_\ell)[Y_\ell] = Y_\ell, \quad (\mathrm{id}_\ell - \mathcal{E}^Z_\ell)[Z_\ell] = 0.
\end{align}
Therefore we have $\lambda_I = \lambda_Y = \lambda_Z = 0$ and $\lambda_X = 2 \sqrt{p(1-p)}$ and therefore similarly to $(\mathrm{id}_\ell - \mathcal{D}_\ell)$, we have
\begin{align}
    \norm{\widehat{\mathcal{N}}_{E_\ell \leftarrow \ell} (\mathrm{id}_\ell - \mathcal{E}^Z_\ell)}_\diamond = 2 \sqrt{p(1-p)}.
\end{align}

\subsubsection{Amplitude damping channel}
The amplitude damping channel is given by
\begin{align}
    \mathcal{N}_\ell[\cdot] = K_\ell^{(0)} \cdot K_\ell^{(0)\dagger} + K_\ell^{(1)} \cdot K_\ell^{(1)\dagger},
\end{align}
where
\begin{align}
    K^{(0)}_\ell = \begin{pmatrix} 1 & 0 \\ 0 & \sqrt{1-p} \end{pmatrix}_\ell, \quad K^{(1)}_\ell = \begin{pmatrix}
        0 & \sqrt{p} \\ 0 & 0
    \end{pmatrix}_\ell.
\end{align}
It can be written as an isometry as
\begin{align}
    V \ket{0}_\ell = \ket{00}_{E_\ell \ell}, \quad V \ket{1}_\ell = \sqrt{1-p} \ket{01}_{E_\ell \ell} + \sqrt{p} \ket{10}_{E_\ell \ell}.
\end{align}
One can see that if we trace over the left environment qubit, then we get the error channel as we need. We can also write it in matrix form as
\begin{align}
    V = 
    \begin{pmatrix}
        1 & 0 \\
        0 & \sqrt{1-p} \\
        0 & \sqrt{p} \\
        0 & 0
    \end{pmatrix}_{E_\ell \ell \leftarrow \ell}.
\end{align}
To extend it to a unitary, we just need to pad this isometry column vectors that are orthogonal to the ones already there. A natural choice is then
\begin{align}
    U = 
    \begin{pmatrix}
        1 & 0 & 0 & 0\\
        0 & \sqrt{1-p} & \sqrt{p} & 0 \\
        0 & \sqrt{p} & - \sqrt{1-p} & 0 \\
        0 & 0 & 0 & 1
    \end{pmatrix}_{E_\ell \ell}.
\end{align}
Then if we choose the (left) environment input qubit to be $\ket{0}_{E_\ell}$, the we retrieve the isometry above. 

To deduce the action of the complementary channel we trace out the physical qubit after acting the channel the isometry on the four basis operators then tracing out the physical qubit:
\begin{align}
    I_\ell & = \ketbra{0}_\ell + \ketbra{1}_\ell \\
    & \mathop{\rightarrow}^{V} \ketbra{00}{00}_{E_\ell \ell} + \left(\sqrt{1-p} \ket{01}_{E_\ell \ell} + \sqrt{p} \ket{10}_{E_\ell \ell}\right) \left(\sqrt{1-p} \ket{01}_{E_\ell \ell} + \sqrt{p} \ket{10}_{E_\ell \ell}\right) \\
    & \mathop{\rightarrow}^{\tr_\ell} (2 - p)\ketbra{0}_{E_\ell} + p \ketbra{1}_{E_\ell}
\end{align}
\begin{align}
    Z_\ell & = \ketbra{0}{0}_\ell - \ketbra{1}{1}_\ell \\
    & \mathop{\rightarrow}^{V} \ketbra{00}{00}_{E_\ell \ell} - \left(\sqrt{1-p} \ket{01}_{E_\ell \ell} + \sqrt{p} \ket{10}_{E_\ell \ell}\right) \left(\sqrt{1-p} \ket{01}_{E_\ell \ell} + \sqrt{p} \ket{10}_{E_\ell \ell}\right) \\
    & \mathop{\rightarrow}^{\tr_\ell} p \ketbra{0}{0}_{E_\ell} - p \ketbra{1}{1}_{E_\ell} = p Z_{E_\ell},
\end{align}
\begin{align}
    X_\ell & = \ketbra{0}{1}_\ell + \ketbra{1}{0}_\ell \\
    & \mathop{\rightarrow}^{V} \ket{00}_{E_\ell \ell} \left(\sqrt{1-p} \bra{01}_{EP} + \sqrt{p} \bra{10}_{EP}\right) + \left(\sqrt{1-p} \ket{01}_{E_\ell \ell} + \sqrt{p} \ket{10}_{E_\ell \ell}\right) \bra{00}_{E_\ell \ell} \\
    & \mathop{\rightarrow}^{\tr_\ell} \sqrt{p} \ketbra{0}{1}_{E_\ell} + \sqrt{p}  \ketbra{1}{0}_{E_\ell} = \sqrt{p} X_{E_\ell},
\end{align}
\begin{align}
    Y_\ell & = -\mathrm{i} \ketbra{0}{1}_\ell + \mathrm{i} \ketbra{1}{0}_{\ell} \\
    & \mathop{\rightarrow}^{V} -\mathrm{i} \ket{00}_{E_\ell \ell} \left(\sqrt{1-p} \bra{01}_{E_\ell \ell} + \sqrt{p} \bra{10}_{EP}\right) + \mathrm{i} \left(\sqrt{1-p} \ket{01}_{E_\ell \ell} + \sqrt{p} \ket{10}_{E_\ell \ell}\right) \bra{00}_{E_\ell \ell} \\
    & \mathop{\rightarrow}^{\tr_\ell} -\mathrm{i} \sqrt{p} \ketbra{0}{1}_{E_\ell} + \mathrm{i} \sqrt{p}  \ketbra{1}{0}_{E_\ell} = \sqrt{p} Y_{E_\ell}.
\end{align}

As in the bit-flip case, this is a Pauli super-operator once the identity output is killed, similarly to the bit-flip case.

\paragraph{$\norm{\widehat{\mathcal{N}}_{E_\ell \leftarrow \ell} (\mathrm{id}_\ell - \mathcal{D}_\ell)}_\diamond$.} Upon contraction with $(\mathrm{id}_\ell - \mathcal{D}_\ell)$, we have we have $\lambda_I =0 0$, $\lambda_X = \lambda_Y = \sqrt{p}$, and $\lambda_Z = p$. Therefore using \cref{lem:qubit-to-qubit}, we have 
\begin{align}
    \norm{\widehat{\mathcal{N}}_{E_\ell \leftarrow \ell} (\mathrm{id}_\ell - \mathcal{D}_\ell)}_\diamond & = \frac{1}{4}\qty(|0+\sqrt{p}+\sqrt{p}+p| + |0-\sqrt{p}-\sqrt{p}+p| + |0+\sqrt{p}-\sqrt{p}-p| + |0-\sqrt{p}+\sqrt{p}-p|) \\ & = \frac{1}{4}\qty((2\sqrt{p} + p) + (2\sqrt{p}-p) + p + p)
\end{align}
and therefore
\begin{align}
    \norm{\widehat{\mathcal{N}}_{E_\ell \leftarrow \ell} (\mathrm{id}_\ell - \mathcal{D}_\ell)}_\diamond = \sqrt{p} + \frac{p}{2}.
\end{align}

\paragraph{$\norm{\widehat{\mathcal{N}}_{E_\ell \leftarrow \ell} (\mathrm{id}_\ell - \mathcal{E}^X_\ell)}_\diamond$.}
Upon contraction with $(\mathrm{id}_\ell - \mathcal{E}^X_\ell)$, we have $\lambda_I = \lambda_X = 0$, $\lambda_Y = \sqrt{p}$, $\lambda_Z = p$. Therefore we have 
\begin{align}
    \norm{\widehat{\mathcal{N}}_{E_\ell \leftarrow \ell} (\mathrm{id}_\ell - \mathcal{E}^X_\ell)}_\diamond & = \frac{1}{4}\qty(|0+0+\sqrt{p}+p| + |0-0-\sqrt{p}+p| + |0+0-\sqrt{p}-p| + |0-0+\sqrt{p}-p|) \\
    & = \frac{1}{4}\qty((\sqrt{p} + p) + (\sqrt{p}-p) + (\sqrt{p}+p) + (\sqrt{p}-p))
\end{align}
therefore
\begin{align}
    \norm{\widehat{\mathcal{N}}_{E_\ell \leftarrow \ell} (\mathrm{id}_\ell - \mathcal{E}^X_\ell)}_\diamond = \sqrt{p}.
\end{align}

\paragraph{$\norm{\widehat{\mathcal{N}}_{E_\ell \leftarrow \ell} (\mathrm{id}_\ell - \mathcal{E}^Z_\ell)}_\diamond$.}
Upon contraction with $(\mathrm{id}_\ell - \mathcal{E}^Z_\ell)$ we have $\lambda_I = 0$, $\lambda_{X} = \lambda_Y = \sqrt{p}$ and $\lambda_Z = 0$. Therefore upon we have
\begin{align}
    \norm{\widehat{\mathcal{N}}_{E_\ell \leftarrow \ell} (\mathrm{id}_\ell - \mathcal{E}^Z_\ell)}_\diamond & = \frac{1}{4}\qty(|0+\sqrt{p}+\sqrt{p}+0| + |0-\sqrt{p}-\sqrt{p}+0| + |0+\sqrt{p}-\sqrt{p}-0| + |0-\sqrt{p}+\sqrt{p}-0|) \\ & = \frac{1}{4}\qty(2\sqrt{p} + 2\sqrt{p} + 0 + 0)
\end{align}
therefore
\begin{align}
    \norm{\widehat{\mathcal{N}}_{E_\ell \leftarrow \ell} (\mathrm{id}_\ell - \mathcal{E}^Z_\ell)}_\diamond = \sqrt{p}.
\end{align}

\subsubsection{Partially depolarising noise}
The partially depolarising channel is given by
\begin{align}
    \mathcal{N}_\ell[\cdot] = (1-p) \rho + p \frac{I_\ell}{d_\ell}\mathrm{tr}_\ell[\cdot].
\end{align}
This is equivalently given by
\begin{align}
    \mathcal{N}_\ell[\rho] = (1-p) \rho + p \times \frac{1}{d^2_\ell} \sum_P P_\ell \rho P^\dag_\ell = \qty(1 - p + \frac{p}{d^2_\ell}) \rho + \frac{p}{d^2_\ell} \sum_{P \neq 0} P_\ell \rho P_\ell^\dagger.
\end{align}
Therefore we have the Kraus operators
\begin{align}
    K^{(I)}_\ell = \sqrt{1-p+\frac{p}{d_\ell^2}} I_\ell =: \sqrt{c_I} I_\ell
\end{align}
with the rest
\begin{align}
    K^{(P\neq I)}_\ell = \frac{\sqrt{p}}{d_\ell} P_\ell =: \sqrt{c_{\neq I}} P_\ell.
\end{align}
Therefore the isometry can be constructed as
\begin{align}
    V_{E_\ell \ell \leftarrow \ell} \ket{\psi} = \sum_P K^{(P)}_\ell \ket{\psi} \otimes \ket{P}_{E_\ell}.
\end{align}
Therefore the complementary channel can be written as
\begin{align}
    \widehat{\mathcal{N}}_{E_\ell \leftarrow \ell}[\cdot] = \sum_{PP'} \mathrm{tr}_\ell [K^{(P)}_\ell \cdot K^{(P')\dagger}_\ell] \otimes \ketbra{P}{P'}_{E_\ell}.
\end{align}

\paragraph{$\norm{\widehat{\mathcal{N}}_{E_\ell \leftarrow \ell} (\mathrm{id}_\ell - \mathcal{D}_\ell)}_\diamond$.} Since $\mathcal{N}_{\ell}$ is unitary-set covariant, by \cref{lem:chan-covariant-means-complement-covariant}, $\widehat{\mathcal{N}}_{E_\ell \leftarrow \ell}$ is also unitary-set covariant. Therefore the diamond norm of $\widehat{\mathcal{N}}_{E_\ell \leftarrow \ell}(\mathrm{id}_\ell - \mathcal{D}_\ell)$ is given by $\norm{\widehat{\mathcal{N}}_{E_\ell \leftarrow \ell}(\mathrm{id}_\ell - \mathcal{D}_\ell)}_\diamond = \norm{\widehat{\mathcal{N}}_{E_\ell \leftarrow \ell}(\mathrm{id}_\ell - \mathcal{D}_\ell) \Phi_{\ell \ell'}}_1$. Now

\begin{align}
\widehat{\mathcal{N}}_{E_\ell \leftarrow \ell}(\mathrm{id}_\ell - \mathcal{D}_\ell) \Phi_{\ell \ell'} 
& = \frac{1}{d_\ell} \widehat{\mathcal{N}}_{E_\ell \leftarrow \ell}\qty[\sum_{\alpha \beta} \ketbra{\alpha\alpha}{\beta \beta}_{\ell \ell'} - \frac{I_\ell I_{\ell'}}{d_\ell}] \\
& = \frac{1}{d_\ell} \sum_{PP'} \ketbra{P}{P'}_{E_\ell} \qty(\sum_{\alpha \beta} \mathrm{tr}_\ell[K^{(P)}_\ell \ketbra{\alpha}{\beta}_\ell K^{(P')\dagger}_\ell] \ketbra{\alpha}{ \beta}_{\ell'} - \mathrm{tr}_\ell[K^{(P)}_\ell K^{(P')\dagger}_\ell] \frac{I_{\ell'}}{d_\ell}) \\
& = \frac{1}{d_\ell} \sum_{PP'} \ketbra{P}{P'}_{E_\ell} \qty(\sum_{\alpha \beta} (K^{(P')\dagger} K^{(P)})_{\beta \alpha} \ketbra{\alpha}{\beta}_{\ell'} - \mathrm{tr}_\ell[K^{(P)}_\ell K^{(P')\dagger}_\ell] \frac{I_{\ell'}}{d_\ell}) \\
& = \frac{1}{d_\ell} \sum_{PP'} \ketbra{P}{P'}_{E_\ell} \qty(K^{(P)\top}_{\ell'} K^{(P')*}_{\ell'} - \mathrm{tr}_\ell[K^{(P)}_\ell K^{(P')\dagger}_\ell] \frac{I_{\ell'}}{d_\ell}).
\end{align}

Defining $K^{(P)} = \sqrt{c_P} P$, we have

\begin{align}
\widehat{\mathcal{N}}_{E_\ell \leftarrow \ell}(\mathrm{id}_\ell - \mathcal{D}_\ell) \Phi_{\ell \ell'} 
& = \frac{1}{d_\ell} \sum_{PP'} \ketbra{P}{P'}_{E_\ell} \sqrt{c_P c_{P'}} \qty( P^{\top}_{\ell'} P^{\prime *}_{\ell'} - \delta_{PP'} I_{\ell'}) =: M.
\end{align}

To diagonalise this matrix, first consider the set of vectors
\begin{align}
    \ket{v_1^{(\alpha)}} = \ket{I}_{E_\ell} \ket{\alpha}_{\ell'},
\end{align}
and
\begin{align}
    \ket{v_2^{(\alpha)}} = \frac{1}{\sqrt{d^2_\ell - 1}} \sum_{P \neq I} \ket{P}_{E_\ell} P^\top \ket{\alpha}_{\ell'}.
\end{align}
It's clear that they are orthogonal to each other because the vectors $\ket{I}$ and $\ket{P \neq I}$ are orthogonal to each other.

Now apply $M$ to $\ket{v_1^{(\alpha)}}$:
\begin{align}
    M \ket{v^{(\alpha)}_1} & = \frac{1}{d_\ell} \sum_P \sqrt{c_P c_I} \ket{P}_{E_\ell} (P^\top_{\ell'} - \delta_{P,I}) \ket{\alpha}_{\ell'} \\
    & = \frac{1}{d_\ell} c_I (1 - 1) \ket{0}_{E_\ell} \ket{\alpha}_{\ell'} + \frac{1}{d_\ell} \sqrt{c_{\neq I} c_I} \sum_{P \neq I} \ket{P}_{E_\ell} P^\top_{\ell'} \ket{\alpha}_{\ell'} \\
    & = 0 \ket{v_1^{(\alpha)}} + \frac{\sqrt{c_{\neq I} c_I(d^2_\ell - 1)}}{d_\ell} \ket{v^{(\alpha)}_2}.
\end{align}
Here, we used the fact that $c_{P} = c_{P'} = c_{\neq I}$ for all $P, P' \neq I$.

Now apply $M$ on $\ket{v_2^{(\alpha)}}$:
\begin{align}
    M \ket{v_2^{(\alpha)}} 
    & = \frac{1}{d_\ell} \sum_{PP'} \ketbra{P}{P'}_{E_\ell} \sqrt{c_P c_{P'}} (P^\top_{\ell'} P^{\prime *}_{\ell'} - \delta_{PP'}) \frac{1}{\sqrt{d^2_\ell - 1}} \sum_{P'' \neq I} \ket{P''}_{E_\ell} P^{\prime \prime\top}_{\ell'} \ket{\alpha}_{\ell'} \\
    & = \frac{1}{d_\ell} \frac{1}{\sqrt{d^2_\ell - 1}}  \sum_{P} \sum_{P' \neq I} \ket{P}_{E_\ell} \sqrt{c_P c_{P'}} (P^\top_{\ell'} P^{\prime *}_{\ell'} - \delta_{PP'}) P^{\prime \top}_{\ell'} \ket{\alpha}_{\ell'} \\
    & = \frac{1}{d_\ell} \frac{1}{\sqrt{d^2_\ell - 1}} \sum_{P' \neq I} \ket{I}_{E_\ell} \sqrt{c_I c_{\neq I}} I^\top_{\ell'} P^{\prime \top}_{\ell'} P^{\prime *}_{\ell'} \ket{\alpha}_{\ell'}
    + \frac{1}{d_\ell} \frac{1}{\sqrt{d^2_\ell - 1}}  \sum_{P \neq I} \sum_{P' \neq I} \ket{P}_{E_\ell} \sqrt{c_P c_{P'}} (P^\top_{\ell'} P^{\prime *}_{\ell'} - \delta_{PP'}) P^{\prime \top}_{\ell'} \ket{\alpha}_{\ell'} \\
    & = \frac{1}{d_\ell} \sqrt{c_I c_{\neq I}(d^2_\ell-1)} \ket{I}_{E_\ell} \ket{\alpha}_{\ell'}
    + \frac{c_{\neq I} (d^2_\ell  - 2)}{d_\ell} \frac{1}{\sqrt{d_\ell^2 - 1}} \sum_{P \neq I} \ket{P}_{E_\ell} P^\top_{\ell'} \ket{\alpha}_{\ell'} \\
    & = \frac{\sqrt{c_I c_{\neq I} (d_\ell^2 -1 )}}{d_\ell} \ket{v^{(\alpha)}_1} + \frac{c_{\neq I} (d^2_\ell - 2)}{d_\ell} \ket{v^{(\alpha)}_2}.
\end{align}
Therefore we have to block diagonal matrix
\begin{align}
    M^{(\alpha)} = \begin{pmatrix}
        0 & \frac{\sqrt{c_I c_{\neq I} (d^2_\ell -1)}}{d_\ell} \\
        \frac{\sqrt{c_{\neq I} c_I(d^2_\ell - 1)}}{d_\ell} & \frac{c_{\neq I} (d^2_\ell - 2))}{d_\ell}
    \end{pmatrix}.
\end{align}
The trace norm of this block can be solved to be
\begin{align}
    \norm{M^{(\alpha)}}_1 = \frac{\sqrt{p}}{d^2_\ell} \sqrt{4(d^2_\ell - 1) - p(3d^2_\ell - 4)}.
\end{align}

Now consider all vectors orthogonal to $\{\ket{v^{(\alpha)}_{1,2}}\}_\alpha$. It is clear that it must be written in the form
\begin{align}
    \ket{v} = \sum_{P \neq I} \ket{P}_{E_\ell} \ket{\psi_P}_{\ell'},
\end{align}
where the sum is over $P \neq I$ since it must be orthogonal to $\{\ket{v^{(\alpha)}_1}\}_\alpha$. Taking an inner product with $\ket{v^{(\alpha)}_2}$, we have
\begin{align} \label{eq:annhilation-of-v}
    \sum_{P} \bra{\alpha}_{\ell'} P^*_{\ell'} \ket{\psi_P}_{\ell'} = 0 \forall \alpha \implies \sum_{P} P^*_{\ell'} \ket{\psi_P}_{\ell'} = 0.
\end{align}
Now if we apply $M$ on $\ket{v}$, we have
\begin{align}
    M \ket{v} = \frac{1}{d_\ell} \sum_{P P'} \ketbra{P}{P'}_{E_\ell} \sqrt{c_P c_{P'}} (P^\top_{\ell'} P^{\prime *}_{\ell'} - \delta_{PP'} I_{\ell'}) \sum_{P'' \neq I} \ket{P''}_{E_{\ell}} \ket{\psi_{P''}}_{\ell'}.
\end{align}
Now the first part of $M$ annhilates $\ket{v}$, due to \cref{eq:annhilation-of-v}. Therefore we have
\begin{align}
    M \ket{v} = \frac{-p}{d_\ell^3} \ket{v}.
\end{align}
Therefore, the rest of the Hilbert space is highly degenerate with eigenvalue $-p/d^3_\ell$. The dimension of the space $\perp$ to $\{\ket{v^{(\alpha)}_{1,2}}\}_\alpha$, $\mathscr{H}^\perp$, is $d^2_\ell - 2 d_\ell$. Therefore the trace norm is

\begin{align}
    \norm{M}_1 & = \sum_\alpha \norm{M^{(\alpha)}}_1 + \norm{M^\perp}_1 = \sum_\alpha \frac{\sqrt{p}}{d^2_\ell} \sqrt{4(d^2_\ell - 1) - p(3d^2_\ell - 4)} + (d^2_\ell - 2 d_\ell) \times \abs{-p/d^3_\ell}
\end{align}
therefore
\begin{align}
    \norm{\widehat{\mathcal{N}}_\ell(\mathrm{id}_\ell - \mathcal{D}_\ell)}_\diamond = \frac{p(d_\ell^2-2)}{d_\ell^2} + \frac{\sqrt{p}}{d_\ell} \sqrt{4(d^2_\ell - 1) - p(3d^2_\ell - 4)}.
\end{align}

\paragraph{$\norm{\widehat{\mathcal{N}}_{E_\ell \leftarrow \ell} (\mathrm{id}_\ell - \mathcal{E}^{X,Z}_\ell)}_\diamond$.} Now specialising to qubits, consider
\begin{align}
    \norm{\widehat{\mathcal{N}}_\ell(\mathrm{id}_\ell - \mathcal{E}^{X}_\ell)}_\diamond.
\end{align}
Again, this is a unitary-set covariant superoperator, and therefore is given by the trace norm on the Bell state.

Now applying $\mathrm{id}_\ell - \mathcal{E}^X_\ell$ on the Bell state, we have
\begin{align}
    (\mathrm{id}_\ell - \mathcal{E}^X_\ell) \Phi_{\ell \ell'} = \frac{1}{4} Y_\ell Y^\top_{\ell'} Z_\ell Z^\top_{\ell'}.
\end{align}
Therefore, we have
\begin{align}
    \widehat{\mathcal{N}}_{E_\ell \leftarrow \ell} (\mathrm{id}_\ell - \mathcal{E}^X_\ell)\Phi_{\ell \ell'} = \frac{1}{4} \sum_{PP'} \sum_{P'' \in \{Y, Z\}} \sqrt{c_P c_{P'}} \ketbra{P}{P'}_{E_\ell} \mathrm{tr}_\ell [P_\ell P''_\ell P^{\prime \dag}_\ell] P^{\prime \prime \top}_{\ell'},
\end{align}
which, upon considering all of the traces, can be written as
\begin{align}
    = \frac{1}{2} \left(\begin{tabular}{c|cccc}
    &$\bra{I}$&$\bra{X}$&$\bra{Y}$&$\bra{Z}$\\
    \hline
    $\ket{I}$ &$0$  &$0$  &$\sqrt{c_I c_{\neq I}}$  &$0$  \\
    $\ket{X}$ &$0$  &$0$  &$0$  &$\mathrm{i}c_{\neq I}$  \\
    $\ket{Y}$ &$\sqrt{c_I c_{\neq I}}$  &$0$  &$0$  &$0$  \\
    $\ket{Z}$ &$0$  &$-\mathrm{i}c_{\neq I}$  &$0$  &$0$  \\
    \end{tabular}
    \right)_{E_\ell} Y^\top_{\ell'}
    + \frac{1}{4} \left(\begin{tabular}{c|cccc}
    &$\bra{I}$&$\bra{X}$&$\bra{Y}$&$\bra{Z}$\\
    \hline
    $\ket{I}$ &$0$  &$0$  &$0$  &$\sqrt{c_I c_{\neq I}}$  \\
    $\ket{X}$ &$0$  &$0$  &$-\mathrm{i}c_{\neq I}$  &$0$  \\
    $\ket{Y}$ &$0$  &$\mathrm{i}c_{\neq I}$  &$0$  &$0$  \\
    $\ket{Z}$ &$\sqrt{c_I c_{\neq I}}$  &$0$  &$0$  &$0$  \\
    \end{tabular}
    \right)_{E_\ell} Z^\top_{\ell'}.
\end{align}
Explicitly evaluating its eigenvalues leads to
\begin{align}
    \vec{\lambda} = \left(0,0,0,0,-\frac{1}{4} \sqrt{2 p-p^2},-\frac{1}{4} \sqrt{2 p-p^2},\frac{1}{4} \sqrt{2 p-p^2},\frac{1}{4} \sqrt{2 p-p^2}\right),
\end{align}
and therefore
\begin{align}
    \norm{\widehat{\mathcal{N}}_{E_\ell \leftarrow \ell} (\mathrm{id}_\ell - \mathcal{E}^X_\ell)}_\diamond = \norm{\widehat{\mathcal{N}}_{E_\ell \leftarrow \ell} (\mathrm{id}_\ell - \mathcal{D}_\ell) \Phi_{\ell \ell'}}_1 = \sqrt{p(2-p)}.
\end{align}
Similarly, we also have
\begin{align}
    \norm{\widehat{\mathcal{N}}_{E_\ell \leftarrow \ell} (\mathrm{id}_\ell - \mathcal{E}^Z_\ell)}_\diamond = \sqrt{p(2-p)}.
\end{align}

\subsubsection{Heralded erasure channel}
The heralded leakage, or erasure, channel is given by
\begin{align}
    \mathcal{N}_\ell[\rho] = p \tr_\ell[\rho] \otimes \ketbra{e}_\ell \otimes \ketbra{1}_C + (1-p) \rho \otimes \ketbra{0}_C.
\end{align}

This channel can be written in terms of Kraus operators $K_a: \ell \rightarrow \ell' C$:

\begin{align}
    K_0 & = \sqrt{p} \ \ketbra{e}{0}_{\ell'} \otimes \ket{1}_C, \\
    K_1 & = \sqrt{p} \ \ketbra{e}{1}_{\ell'} \otimes \ket{1}_C \\
    K_2 & = \sqrt{1-p} \ I_{\ell' \leftarrow \ell} \otimes \ket{0}_C.
\end{align}

Tracing over the classical register $C$ and the system $\ell'$, we have

\begin{align}
    \widehat{\mathcal{N}}_{E_\ell \leftarrow \ell}[\rho_Q] \rightarrow (1-p) \tr_\ell[\rho_Q] \otimes \ketbra{2}{2}_{E_\ell} + p \rho_{(Q \setminus \ell) E_\ell},
\end{align}
where $\rho_{(Q \setminus \ell) E_\ell} = I_{E_\ell \leftarrow \ell} \rho I^\dagger_{\ell \leftarrow E_\ell} = \mathcal{I}_{E_\ell \leftarrow \ell}[\rho_Q]$ is just an embedding of $\rho_Q$ in the environment space. Here, $I_{E_\ell \leftarrow \ell} = \sum_{\sigma=0,1} \ket{\sigma}_{E_\ell} \bra{\sigma}_\ell$

\paragraph{$\norm{\widehat{\mathcal{N}}_{E_\ell \leftarrow \ell} (\mathrm{id}_\ell - \mathcal{D}_\ell)}_\diamond$.}
Since for any operator $O_\ell$, $\mathrm{tr}_\ell[(\mathrm{id}_\ell - \mathcal{D}_\ell)[O_\ell]] = 0 \; \forall \; O_\ell$ is traceless on $\ell$, we have
\begin{align}
    \widehat{\mathcal{N}}_{E_\ell \leftarrow \ell}(\mathrm{id}_\ell - \mathcal{D}_\ell) = p \mathcal{I}_{E_\ell \leftarrow \ell} (\mathrm{id}_\ell - \mathcal{D}_\ell)
\end{align}
and therefore we have
\begin{align}
    \norm{\widehat{\mathcal{N}}_{E_\ell \leftarrow \ell} (\mathrm{id}_\ell - \mathcal{D}_\ell)}_\diamond = p \norm{\mathcal{I}_{E_\ell \leftarrow \ell} (\mathrm{id}_\ell - \mathcal{D}_\ell)}_\diamond.
\end{align}
Now trace norms are invariant under conjugation by an isometry, and therefore by \cref{lem:id-D}, we have
\begin{align}
    \norm{\widehat{\mathcal{N}}_{E_\ell \leftarrow \ell} (\mathrm{id}_\ell - \mathcal{D}_\ell)}_\diamond = p \norm{\mathrm{id}_\ell - \mathcal{D}_\ell}_\diamond = 2 p \qty(1 - \frac{1}{d_\ell^2}).
\end{align}

\paragraph{$\norm{\widehat{\mathcal{N}}_{E_\ell \leftarrow \ell} (\mathrm{id}_\ell - \mathcal{E}^{X,Z}_\ell)}_\diamond$.}
Here, specialise to qubits. First, we use the fact that again, any output of $(\mathrm{id}_\ell - \mathcal{E}^{X, Z}_\ell)$ is traceless and therefore the corresponding diamond norm simplifies to
\begin{align}
    \norm{\widehat{\mathcal{N}}_{E_\ell \leftarrow \ell} (\mathrm{id}_\ell - \mathcal{E}^{X, Z}_\ell)}_\diamond = p \norm{\mathcal{I}_{E_\ell \leftarrow \ell} (\mathrm{id}_\ell - \mathcal{E}^{X, Z}_\ell)}_\diamond = p \norm{\mathrm{id}_\ell - \mathcal{E}^{X, Z}_\ell}_\diamond.
\end{align}
Next, we use the fact that $(\mathrm{id}_\ell - \mathcal{E}^{X, Z}_\ell)$ is a Pauli superoperator, which means that its diamond norm can be explicitly calculated via trace norm on the Bell state, which yields $\norm{\mathrm{id}_\ell - \mathcal{E}^{X, Z}_\ell}_\diamond = 1$, and therefore
\begin{align}
    \norm{\widehat{\mathcal{N}}_{E_\ell \leftarrow \ell} (\mathrm{id}_\ell - \mathcal{E}^{X, Z}_\ell)}_\diamond = p.
\end{align}

\end{document}